\documentclass[a4paper,11pt]{article} 
\usepackage[english]{babel}
\usepackage[a4paper]{geometry}
\usepackage{amsmath}
\usepackage{amsthm}
\usepackage{amssymb}
\usepackage{color}

\usepackage{hyperref}         %%%% click to contents

\def\bR{\mathbb{R}}

\def\bN{\mathbb{N}}

\def\bZ{\mathbb{Z}}
\def\cC{\mathcal{C}}

\def\cQ{\mathcal{Q}}
\def\cD{\mathcal{D}}
\def\cA{\mathcal{A}}
\def\cM{\mathcal{M}}

\def\cV{\mathcal{V}}
\def\cO{\mathcal{O}}
\def\cF{\mathcal{F}}
\def\cG{\mathcal{G}}
\def\cL{\mathcal{L}}
\def\cJ{\mathcal{J}}
\def\cN{\mathcal{N}}
\def\cE{\mathcal{E}}
\def\cK{\mathcal{K}}
\def\cH{\mathcal{H}}
\def\eps{\varepsilon}
\def\ph{\varphi}

\def\wt{\widetilde}

\def\indic{\hbox{\raise-2pt \hbox{\indbf 1}}}

\let\dpr=\partial

\let\==\equiv

\let\0=\noindent

\def\*{{\hfill\break\null\hfill\break}}

\def\bmedia#1{{\bigl\langle#1\bigr\rangle}}

\def\tende#1{\,\vtop{\ialign{##\crcr\rightarrowfill\crcr
             \noalign{\kern-1pt\nointerlineskip}
             \hskip3.pt${\scriptstyle #1}$\hskip3.pt\crcr}}\,}
\def\otto{\,{\kern-1.truept\leftarrow\kern-5.truept\to\kern-1.truept}\,}

\def\tr{{\rm tr}}

\newtheorem{theorem}{Theorem}[section]  % use thm for %Theorems to keep numbering consistent
\newtheorem{cor}[theorem]{Corollary}
\newtheorem{prop}[theorem]{Proposition}
\newtheorem{lemma}[theorem]{Lemma}

\numberwithin{equation}{section}

\def\tl#1{{\tilde{#1}}}
\def\be{\begin{equation}}
\def\ee{\end{equation}}

\newcommand{\hc}{\mbox{h.c.}}

\let\a=\alpha     \let\g=\gamma     \let\d=\delta     
        \let\k=\kappa     \let\l=\lambda
                  \let\p=\pi        \let\r=\rho
\let\s=\sigma \let\t=\tau     \let\f=\phi    \let\ph=\varphi   
\let\ps=\psi        
\let\G=\Gamma \let\D=\Delta       \let\L=\Lambda

\begin{document}

\title{Bogoliubov Theory in the Gross-Pitaevskii Limit} 

\author{Chiara Boccato$^1$, Christian Brennecke$^2$, Serena Cenatiempo$^3$, Benjamin Schlein$^2$ \\
\\
IST Austria, Am Campus 1\\
3400 Klosterneuburg,
Austria$^1$\\
\\
Institute of Mathematics, University of Zurich\\
Winterthurerstrasse 190, 8057 Zurich, 
Switzerland$^2$ \\
\\
Gran Sasso Science Institute, Viale Francesco Crispi 7 \\ 
67100 L'Aquila, Italy$^3$}

\maketitle

\begin{abstract}	
We consider Bose gases consisting of $N$ particles trapped in a box with volume one and interacting through a repulsive potential with scattering length of the order $N^{-1}$(Gross-Pitaevskii regime). We determine the ground state energy and the low-energy excitation spectrum, up to errors vanishing as $N \to \infty$. Our results confirm Bogoliubov's predictions.
\end{abstract}

\section{Introduction and main result}

We consider systems of $N$ bosons in the three dimensional box $\Lambda = [-1/2 ; 1/2]^{\times 3}$  with periodic boundary conditions. In the Gross-Pitaevskii regime, the Hamilton operator has the form 
\begin{equation}\label{eq:Ham0} H_N = \sum_{j=1}^N -\Delta_{x_j} +  \sum_{i<j}^N N^2 V (N (x_i -x_j)) \end{equation}
and acts on the Hilbert space $L^2_s (\Lambda^N)$, the subspace of $L^2 (\Lambda^N)$ consisting of functions that are symmetric with respect to permutations of the $N$ particles. We require $V \in L^3 (\bR^3)$ to be non-negative, radial, compactly supported and to have scattering length~$\frak{a}_0 $. 

Recall that the scattering length of the interaction potential is defined through the zero-energy scattering equation 
\begin{equation}\label{eq:0en} \left[ - \Delta + \frac{1}{2} V (x)  \right] f (x) = 0 \end{equation}
with the boundary condition $f (x) \to 1$, as $|x| \to \infty$. For $|x|$ large enough (outside the support of $V$), we have 
\[ f(x) = 1- \frac{\frak{a}_0 }{|x|} \]
for an appropriate constant $\frak{a}_0 $, which is known as the scattering length of $V$. Equivalently, $\frak{a}_0 $ can be recovered by 
\begin{equation}\label{eq:a0}  8 \pi \frak{a}_0  =  \int V(x) f(x) dx \,  \end{equation}
By scaling, we obtain from (\ref{eq:0en}) that 
\[ \left[ - \Delta + \frac{1}{2} N^2 V (Nx) \right] f (Nx) = 0 \]
and therefore that the scattering length of $N^2 V (N x)$ is given by $\frak{a}_0 /N$.  

It follows from the works of Lieb, Yngvason \cite{LY}, of Lieb, Yngvason and Seiringer \cite{LSY} and, more recently, of Nam, Rougerie and Seiringer \cite{NRS}, that the ground state energy of the Hamilton operator (\ref{eq:Ham0}) is given, to leading order in $N$, by
\begin{equation}\label{eq:LSY} E_N = 4 \pi \frak{a}_0  N + o (N) \, 
\end{equation}
In \cite{LS}, Lieb and Seiringer also showed that the ground state of (\ref{eq:Ham0}) exhibits complete Bose-Einstein condensation. In other words, if $\gamma^{(1)}_N$ denotes the one-particle reduced density associated with a normalized ground state wave function $\psi_N \in L^2_s (\Lambda^N)$, then it was shown in \cite{LS} that 
\begin{equation}\label{eq:BEC0}  \gamma_N^{(1)} \to |\ph_0 \rangle \langle \ph_0 | \end{equation}
as $N \to \infty$, where $\ph_0 (x) = 1$ for all $x \in \Lambda$ is the one-particle zero-momentum mode. This convergence was then extended by Lieb and Seiringer \cite{LS2} and, with different techniques, by Nam, Rougerie and Seiringer \cite{NRS}, to sequences of approximate ground states, ie. states whose energy satisfies (\ref{eq:LSY}) (more precisely, the results of \cite{LSY,LS,LS2,NRS} were not 
restricted to bosons moving in a box $\Lambda = [0;1]^3$ with periodic boundary conditions; they 
applied instead to systems trapped by an external confining potential). The interpretation of 
(\ref{eq:BEC0}) is straightforward: all particles in the system, up to a fraction that vanishes in the limit of large $N$, occupy the same one-particle state $\ph_0$. It is, however, important to observe that (\ref{eq:BEC0}) does not mean that the factorized wave function $\ph_0^{\otimes N}$ is a good approximation for the ground state of (\ref{eq:Ham0}). In fact, a simple computation shows that 
\begin{equation}\label{eq:wrong} \langle \ph^{\otimes N}_0, H_N \ph^{\otimes N}_0 \rangle = \frac{(N-1)}{2}  \widehat{V} (0) \end{equation} 
which is always much larger than (\ref{eq:LSY}), with an error of order $N$ (this follows from (\ref{eq:a0}), because $f < 1$ on the support of $V$). In contrast to $\ph_0^{\otimes N}$, the ground state of (\ref{eq:Ham0}) is characterized by a correlation structure, varying on the length scale $N^{-1}$, which is responsible for lowering the energy to (\ref{eq:LSY}). 

We recently improved (\ref{eq:LSY}) and (\ref{eq:BEC0}) in \cite{BBCS3}, where we showed that 
\begin{equation}\label{eq:gse} E_N = 4 \pi \frak{a}_0  N + \cO (1) 
\end{equation}
and that the one-particle reduced density $\gamma^{(1)}_N$ associated with the ground state of (\ref{eq:Ham0}) is such that   
\begin{equation}\label{eq:BBCS1} 1- \langle \ph_0 , \gamma^{(1)}_N \ph_0 \rangle \leq C N^{-1} \end{equation}
for a constant $C > 0$ (we previously obtained these results in \cite{BBCS1}, for sufficiently small interaction potentials). Eq. (\ref{eq:gse}) determines the ground state energy up to an error of order one, independent of $N$. As for Eq. (\ref{eq:BBCS1}), it shows that the number of excitations of the Bose-Einstein condensate remains bounded, uniformly in $N$. Although (\ref{eq:gse}) and (\ref{eq:BBCS1}) substantially improve previous results, these bounds are still not enough to resolve low-lying excited eigenvalues of (\ref{eq:Ham0}), which play a fundamental role in the understanding of the low-temperature  physics of Bose gases.

In this paper, we go beyond (\ref{eq:gse}), computing the ground state energy and the low-lying excitation spectrum of (\ref{eq:Ham0}), up to errors vanishing in the limit $N \to \infty$. This is the content of our main theorem.
\begin{theorem}\label{thm:main}
Let $V \in L^3 (\bR^3)$ be non-negative, spherically symmetric and compactly supported. Then, in the limit $N \to \infty$, the ground state energy $E_N$ of the Hamilton operator $H_N$ defined in (\ref{eq:Ham0}) is given by 
\begin{equation}
\begin{split}\label{1.groundstate}
E_{N} = \; &4\pi (N-1) \frak{a}_0 + e_\Lambda \frak{a}_0^2 \\ & - \frac{1}{2}\sum_{p\in\Lambda^*_+} \left[ p^2+8\pi \frak{a}_0  - \sqrt{|p|^4 + 16 \pi \frak{a}_0  p^2} - \frac{(8\pi \frak{a}_0 )^2}{2p^2}\right] + \cO (N^{-1/4})
    \end{split}
    \end{equation}      
 Here we introduced the notation $\Lambda_+^* = 2\pi \bZ^3 \backslash \{0 \}$ and we defined    \begin{equation}\label{eq:eLambda0}
   e_\Lambda = 2 - \lim_{M \to \infty} \sum_{\substack{p \in \bZ^3 \backslash \{ 0 \} : \\ |p_1|, |p_2|, |p_3| \leq M}} \frac{\cos (|p|)}{p^2} \end{equation}
where, in particular, the limit exists. Moreover, the spectrum of $H_N-E_{N}$ below a threshold $\zeta$ consists of eigenvalues given, in the limit $N \to \infty$, by 
    \begin{equation}
    \begin{split}\label{1.excitationSpectrum}
    \sum_{p\in\Lambda^*_+} n_p \sqrt{|p|^4+ 16 \pi \frak{a}_0  p^2}+ \cO (N^{-1/4} (1+ \zeta^3))
    \end{split}
    \end{equation}
Here $n_p \in \bN$ for all $p\in\Lambda^*_+$ and $n_p \not = 0$ for finitely many $p\in \Lambda^*_+$ only.  
\end{theorem}

{\it Remark:} The term $e_\Lambda \frak{a}_0^2$ in (\ref{1.groundstate}) arises as a correction to 
the scattering length $\frak{a}_0$ (defined in (\ref{eq:0en}), through an equation on $\bR^3$), due to the finiteness of the box $\Lambda$. For small interaction potential, we can define a finite volume scattering length $a_N$ through the convergent Born series 
\[  8\pi \frak{a}_N =   \widehat{V} (0)
+\sum_{k=1}^{\infty}\frac{(-1)^{k}}{(2N)^{k}} \sum_{p_1, \dots, p_{k}\in\Lambda^*_+ } \frac{\widehat{V} (p_1 /N)}{p_1^2} \left(\prod_{i=1}^{k-1}\frac{\widehat{V} ((p_i-p_{i+1})/N)}{p^2_{i+1}}\right)  \widehat{V} (p_{k}/N) \]
In this case, one can check that   
\[ \lim_{N \to \infty} 4\pi (N-1) \left[ \frak{a}_0 -  \frak{a}_N \right] = e_\Lambda \frak{a}_0^2 \] 

Theorem \ref{thm:main} gives precise information on the low-lying eigenvalues of (\ref{eq:Ham0}). Our approach, combined with standard arguments, also gives information on the corresponding eigenvectors. In (\ref{eq:eigvec}) and (\ref{eq:gs-appro}) we provide a norm approximation of eigenvectors associated with the low-energy spectrum of (\ref{eq:Ham0}) (we postpone the precise statement of this result, because it requires additional notation that will be introduced in the next sections). As an application, we can compute the  condensate depletion, i.e. the expected number of excitations of the condensate, in the ground state $\psi_N$ of (\ref{eq:Ham0}); if $\gamma_N^{(1)}$ denotes the one-particle reduced density associated with $\psi_N$, we find
\begin{equation}\label{eq:depl0} 1 - \langle \ph_0 , \gamma_N^{(1)} \ph_0 \rangle =  \frac{1}{N} \sum_{p \in \Lambda^*_+} \left[ \; \frac{p^2  + 8\pi \frak{a}_0  - \sqrt{p^4 + 16 \pi \frak{a}_0  p^2}}{2 \sqrt{p^4 + 16 \pi \frak{a}_0  p^2}}  \; \right] +  \cO(N^{-9/8})\,. \end{equation}
The proof of (\ref{eq:depl0}), which is based on the approximation (\ref{eq:gs-appro}) of the ground state vector and on some additional bounds from Section \ref{sec:GN}, is deferred to Appendix \ref{sec:deple}.

Multiplying lengths by $N$, it is easy to check that the Gross-Pitaevskii regime considered in this paper is equivalent (up to a trivial rescaling) to an extended gas of $N$ particles moving in a box with volume $\text{Vol} = N^3$ and interacting through a fixed potential $V$ with scattering length $\frak{a}_0$ of order one, independent of $N$. In this sense, the Gross-Pitaevskii regime describes an  ultradilute gas of $N$ particles, with density $\rho = N / \text{Vol} = N^{-2}$ converging to zero, as $N \to \infty$. It is interesting to compare the Gross-Pitaevskii regime with the thermodynamic limit, where a Bose gas of $N$ particles interacting through a fixed potential with scattering length $\frak{a}_0$ is confined in a box with volume $\text{Vol}$ so that $N, \text{Vol} \to \infty$ while the density $\rho = N/ \text{Vol}$ is kept fixed.  At low-density, we expect the ground state energy $E (N , \text{Vol})$, divided by the number of particle $N$, to converge, in the thermodynamic limit, towards the famous Lee-Huang-Yang formula 
\begin{equation}\label{eq:LHY}  \lim_{\substack{N,\text{Vol} \to \infty \\ \rho = N / \text{Vol}}} \frac{E (N,\text{Vol})}{N} = 4\pi \rho \frak{a}_0 \left[ 1 + \frac{128}{15 \sqrt{\pi}} (\rho \frak{a}_0^3)^{1/2} + o ((\rho \frak{a}_0^3)^{1/2} ) \right] \end{equation}
where the error is of lower order in the density $\rho$, as $\rho \to 0$ (the order of the limit is important: first, we let $N, \text{Vol} \to \infty$ keeping $\rho = N/ \text{Vol}$ fixed and only afterwards we focus on small $\rho$). 

To date, there is no mathematically rigorous proof of (\ref{eq:LHY}). There are, however, some partial results. In \cite{Dy}, Dyson gave an upper bound to  the ground state energy matching the leading contribution $4\pi \rho \frak{a}_0$ in (\ref{eq:LHY}) (for particles interacting through a hard-sphere potential). About forty years later, Lieb and Yngvason showed in \cite{LY} the corresponding lower bound, establishing the validity of the leading term on the r.h.s. of (\ref{eq:LHY}) (for general  repulsive potentials). More recently, Yau and Yin proved in \cite{YY} an upper bound for the ground state energy per particles in the thermodynamic limit coinciding with (\ref{eq:LHY}) up to second order. They reached this goal by modifying a previous trial state constructed in \cite{ESY} by Erd\H os-Schlein-Yau, reproducing the Lee-Huang-Yang prediction up to errors that are sub-leading for small potentials. 

Eq. (\ref{eq:LHY}) can be compared with our result (\ref{1.groundstate}) for the ground state energy in the Gross-Pitaevskii regime where, as explained above, $\rho = N^{-2}$ (and where the energy has to be multiplied by an additional factor $N^2$ to make up for rescaling lengths). To leading order, the two formulas give the same result. The second order corrections do not agree. It should be observed, however, that if in (\ref{1.groundstate}) we replaced sums over discrete momenta $p \in \Lambda^*_+$ by integrals over continuous variables $p \in \bR^3$, we would obtain exactly (\ref{eq:LHY}). Theorem \ref{thm:main} establishes therefore the analogon of the Lee-Huang-Yang formula for the ground state energy in the Gross-Pitaevskii regime. 

The properties of low-energy states in dilute Bose gases have already been studied in the pioneering work of Bogoliubov, see \cite{B}. Bogoliubov rewrote the Hamilton operator (\ref{eq:Ham0}) in momentum space, using the formalism of second quantization. Since he expected low-energy states to exhibit Bose-Einstein condensation, he replaced all creation and annihilation operators associated with the zero-momentum mode $\ph_0$ by factors of $N^{1/2}$. The resulting Hamiltonian contains constant terms (describing the interaction among particles in the condensate), terms that are quadratic in creation and annihilation operators associated with modes with momentum $p \not = 0$ (describing the kinetic energy of the excitations as well as the interaction between excitations and the condensate) and terms that are cubic and quartic (describing interactions among excitations). Neglecting all cubic and quartic contributions, Bogoliubov obtained a quadratic Hamiltonian that he could diagonalize explicitly. At the end, he recognised (famously, with the help of Landau) that certain expressions  appearing in his formulas were just the first and second Born approximations of the scattering length and he replaced them by $\frak{a}_0 $. This procedure led him essentially to the Lee-Huang-Yang formula (\ref{eq:LHY}) and to expressions similar to (\ref{1.excitationSpectrum}) and to (\ref{eq:depl0}) (but with continuous momenta $p \in \bR^3$, since he considered the thermodynamic limit rather than the Gross-Pitaevskii regime) for the excitation spectrum and the condensate depletion of the dilute Bose gas. 

Let us stress the fact that replacing first and second Born approximations with the full scattering length, we produce an error to the ground state energy that is comparable with the leading term in (\ref{eq:LHY}). In the Gross-Pitaevskii regime, we already discussed this issue; the difference between (\ref{eq:wrong}) and (\ref{eq:LSY}) (where the first Born approximation $\widehat{V} (0)$ is replaced by $8\pi \frak{a}_0$) is of the order $N$, i.e. of the same order as the full ground state energy. The reason why Bogoliubov nevertheless ended up with correct results is that his final replacement compensated exactly for all  terms (cubic and quartic in creation and annihilation operators) that he neglected in his analysis.  
 
Mathematically, the validity of Bogoliubov's approach in three dimensional Bose gases has been first established by Lieb and Solovej for the computation of the ground state energy of bosonic jellium in \cite{LSo} and of the two-component charged Bose gas in \cite{LSo2} (upper bounds were  later given by Solovej in \cite{So}). Extending the ideas of \cite{LSo,LSo2}, Giuliani and Seiringer established in \cite{GiuS} the validity of the Lee-Huang-Yang formula (\ref{eq:LHY}) for Bose gases interacting through potentials scaling with the density to approach a simultaneous weak coupling and high density limit. This result has been recently improved by Brietzke and Solovej in \cite{BriS} to include a certain class of weak coupling and low density limits. In the regimes considered in \cite{LSo, LSo2, GiuS, BriS} the difference between first and second Born approximation and the full scattering length is small and it only gives negligible contributions to the energy; this is crucial to make Bogoliubov's approach rigorous.  An ambitious long-term project consisting in proving Bose-Einstein condensation and the validity of Bogoliubov theory for dilute Bose gases in the thermodynamic limit by means of renormalization group analysis is currently being pursued by Balaban-Feldman-Kn\"orrer-Trubowitz; see \cite{Balaban} for recent progress. 

In the last years, rigorous versions of Bogoliubov's approach have been used to establish ground state energy and excitation spectrum for mean-field models describing systems of $N$ trapped bosons interacting weakly through a potential whose range is comparable with the size of the trap. 
The first results in this direction have been obtained by Seiringer in \cite{Sei} for the translation invariant case (where particles are confined in a box with volume one and periodic boundary conditions are imposed). In \cite{GS}, Grech and Seiringer extended this result to mean-field systems confined by non-homogeneous external potentials. In this paper, they also conjectured the form of the excitation spectrum in the Gross-Pitaevskii regime (the expression in \cite[Conjecture 1]{GS} coincides with (\ref{1.excitationSpectrum}) in the translation invariant case). Lewin, Nam, Serfaty and Solovej obtained in \cite{LNSS} an alternative derivation of the low-energy spectrum of mean-field bosons (the techniques of this paper play a central role in our analysis). A different approach, valid in a combined mean-field and infinite volume limit, was proposed by Derezinski and Napiorkowski in \cite{DN}. Furthermore, in \cite{P1,P2,P3}, Pizzo obtained an expansion of the ground state energy and of the ground state function, for a mean field Hamiltonian, imposing an ultraviolet cutoff. 

Recently, these results have been extended in \cite{BBCS2} to systems of $N$ bosons in the box 
$\Lambda = [-\frac12;\frac12]^{\times 3}$, described by the Hamilton operator
\begin{equation}\label{eq:ham-beta}  H_N^\beta = \sum_{j=1}^N -\Delta_{x_j} + \frac{\kappa}{N} \sum_{i<j}^N N^{3\beta} V (N^\beta (x_i -x_j)) \end{equation}
for $\beta \in (0;1)$, a sufficiently small coupling constant $\kappa > 0$ and a short range potential $V \geq 0$. Hamilton operators of the form (\ref{eq:ham-beta}) interpolate between the mean-field regime associated with $\beta = 0$ and the Gross-Pitaevskii Hamiltonian (\ref{eq:Ham0}) corresponding to $\beta = 1$. In \cite{BBCS2}, the dispersion law of the excitations has the form \begin{equation}\label{eq:epsp} \eps_\beta (p) = \sqrt{|p|^4 + 2\kappa \widehat{V} (0) p^2}, \end{equation} independently of $\beta \in (0;1)$, because the difference between the scattering length of the interaction in (\ref{eq:ham-beta}) and its first Born approximation $\kappa \widehat{V} (0)$ is of the order $N^{\beta -1}$ and vanishes in the limit $N \to \infty$. Moreover, a simple computation shows that, in the regime described by the Hamilton operator (\ref{eq:ham-beta}), replacing first and second Born approximations with the scattering length produces an error in the ground state energy of the order $N^{2\beta-1}$. Hence, while Bogoliubov's approximation can be rigorously justified in the mean field limit $\beta =0$ (as done in \cite{Sei,GS,LNSS,DN}), for $\beta > 1/2$ it certainly fails, complicating the analysis of \cite{BBCS2}. 

Our goal here is to extend the results of \cite{BBCS2} to the physically more interesting and mathematically more challenging Gross-Pitaevskii regime, where $\beta = 1$. The fact that the dispersion relation of the excitations changes from (\ref{eq:epsp}) to $\eps (p) = \sqrt{|p|^4 + 16 \pi \frak{a}_0 p^2}$ in (\ref{1.excitationSpectrum}) is a first hint to the fact that the step from 
$0 < \beta < 1$ to $\beta =1$ is quite delicate; as we will explain below, it requires completely new ideas and it makes the analysis substantially more involved. 

It is worth noticing that the expression \eqref{1.excitationSpectrum} for the excitation spectrum of \eqref{eq:Ham0} has important consequences from the point of view of physics. It shows that the dispersion of bosons described by \eqref{eq:Ham0} is linear for small momenta, in sharp contrast with the quadratic dispersion of free particles. This observation was used by Bogoliubov in \cite{B} to
explain the emergence of superfuidity, via the so-called Landau criterion \cite{Lan}.

Let us now briefly sketch the main ideas in the proof of Theorem \ref{thm:main}. We start with the observation, due to Lewin, Nam, Serfaty and Solovej in \cite{LNSS}, that every $N$-particle wave function $\psi_N \in L^2_s (\Lambda^N)$ can be decomposed uniquely as 
\[ \psi_N =  \sum_{n=0}^N \alpha^{(n)} \otimes_s \ph_0^{\otimes (N-n)} 
\]
where $\alpha^{(n)} \in L^2_\perp (\Lambda)^{\otimes_s n}$ for all $n =1 , \dots , N$, $L^2_\perp (\Lambda)$ is the orthogonal complement of the condensate wave function $\ph_0$ and where $\otimes_s$ denotes the symmetric tensor product. Recall that the symmetric tensor product of $\psi_k \in L^2 (\Lambda)^{\otimes k}$ and $\psi_\ell \in L^2 (\Lambda)^{\otimes \ell}$ is defined by 
\begin{multline*}
\psi_k\otimes_s\psi_\ell \, (x_1,...,x_{k+\ell}) \\
=\frac{1}{\sqrt{k!\ell!(k+\ell)!}}\sum_{\sigma\in\mathfrak{S}_{k+\ell}}\psi_k(x_{\sigma(1)},...,x_{\sigma(k)})\psi_\ell(x_{\sigma(k+1)},...,x_{\sigma(k+\ell)}).
\end{multline*}
This remark allows us to define a unitary map 
\begin{equation}\label{eq:UN-intro} U_N : L^2_s (\Lambda^N) \to \cF_+^{\leq N} = \bigoplus_{n=0}^N L^2_\perp (\Lambda)^{\otimes_s n} \end{equation}
through $U_N \psi_N = \{ \alpha^{(0)}, \alpha^{(1)}, \dots , \alpha^{(N)} \}$. Here $\cF_+^{\leq N}$ is a truncated Fock space, constructed on the orthogonal complement of $\ph_0$ in $L^2 (\Lambda)$. The map $U_N$ factors out the condensate and allows us to focus on its orthogonal excitations. 

With the map $U_N$ we can construct the excitation Hamiltonian $\cL_N = U_N H_N U_N^* : \cF_+^{\leq N} \to \cF_+^{\leq N}$. As we will discuss in Section \ref{sec:ex}, conjugation with $U_N$ is reminiscent of the Bogoliubov approximation described above; it produces constant contributions and also terms that are quadratic, cubic and quartic in creation and annihilation operators $a_p^*, a_p$ associated with momenta $p \in \L^*_+ = 2\pi \bZ^3 \backslash \{ 0 \}$. In contrast with what Bogoliubov did and in contrast with what was done in \cite{LNSS} in the mean-field regime, here we cannot neglect cubic and quartic terms resulting from conjugation with $U_N$; they are large and they have to be taken into account to obtain a rigorous proof of Theorem \ref{thm:main}. 

The reason why, in the Gross-Pitaevskii regime, cubic and quartic terms are still important is that conjugation with $U_N$ factors out products of the condensate wave function $\ph_0$, while it does not affect correlations. Hence, the correlation structure that, as discussed around (\ref{eq:wrong}),  carries an energy of order $N$ and characterizes all low-energy states $\psi_N \in L^2_s (\Lambda^N)$ is left in the corresponding excitation vector $U_N \psi_N \in \cF^{\leq N}_+$. To extract the large contributions to the energy that are still hidden in cubic and quartic terms, we have to conjugate the excitation Hamiltonian $\cL_N$ with a unitary map generating the correct correlation structure. To reach this goal, we will introduce so-called generalized Bogoliubov transformations having the form 
\begin{equation}\label{eq:defT} T = \exp \left[ \frac{1}{2} \sum_{p \in \Lambda^*_+} \eta_p  \left( b_p^* b_{-p}^* -  b_p b_{-p} \right) \right] \end{equation}
where, for $p \in \Lambda^*_+$, $b_p^*, b_p$ are modified creation and annihilation operators 
acting on $\cF_+^{\leq N}$ by creating and, respectively, annihilating a particle with momentum 
$p$ while preserving the total number of particles $N$ (by removing or adding a particle in the condensate). The normalization of the operators $b_p^*, b_p$ is chosen so that, on states exhibiting Bose-Einstein condensation, their action is close to that of the standard creation and annihilation operators. Hence, although the action of $T$ on creation and annihilation operators is not explicit, we will show that 
\begin{equation}\label{eq:TbT} \begin{split} T^* b_p T &= \cosh (\eta_p) b_p + \sinh (\eta_p) b_{-p}^* + d_p \\ T^* b_p^* T &= \cosh (\eta_p) b_p^* + \sinh (\eta_p) b_{-p} + d^*_{p} \end{split} \end{equation}
for remainder operators $d_p$ that are small on states with few excitations. 

Using the generalized Bogoliubov transformation $T$, we can define a new, renormalized, excitation Hamiltonian $\cG_N = T^* \cL_N T = T^* U_N H_N U_N^* T : \cF^{\leq N}_+ \to \cF^{\leq N}_+$. With the appropriate choice of the coefficients $\eta_p$ (related with a modification of the solution of the zero-energy scattering equation (\ref{eq:0en})), we find that   
\begin{equation}\label{eq:GN-intro1} \cG_N = E_N + \cH_N + \Delta_N \end{equation}
where $\cH_N = \cK+ \cV_N$ is the Hamiltonian $H_N$ restricted on the excitation space $\cF_+^{\leq N}$, with 
\[ \cK = \sum_{p \in \L^*_+} p^2 a_p^* a_p, \qquad \cV_N = \frac{1}{2N} \sum_{p,q \in \L_+^* ,r \in \L^* : r \not = -p,-q}  \widehat{V} (r) a^*_{p+r} a_q^* a_{q+r} a_p \, , \] 
indicating the kinetic and, respectively, the potential energy operators, while $\Delta_N$ is an error term with the property that, for every $\delta > 0$ there exists $C > 0$ with 
\begin{equation}\label{eq:pmEN-intro} \pm \Delta_N \leq \delta \cH_N + C (\cN_+ + 1) \end{equation}
where $\cN_+$ is the number of particles operator on $\cF_+^{\leq N}$ (it measures the number of excitations). While this bound is enough to show Bose-Einstein condensation with optimal rates for sufficiently small potentials (because the constant $C$ on the r.h.s. of (\ref{eq:pmEN-intro}) 
can be chosen proportionally to the size of the interaction), in general it is not. So, a crucial ingredient in our proof of Theorem \ref{thm:main} is the additional estimate
\begin{equation}\label{eq:BBCS3} \cG_N - E_N \geq c \cN_+ - C \end{equation}
which follows from the analysis in \cite{BBCS3}, and also makes use of  the result (\ref{eq:BEC0}) from \cite{LS,NRS}. Eq. (\ref{eq:BBCS3}) controls the number of excitations in terms of the excess energy $\cG_N - E_N$. Combined with (\ref{eq:GN-intro1}), it also allows us to control the energy of the excitations, showing that
\begin{equation}\label{eq:GNHN} \cG_N - E_N \geq c \cH_N - C \end{equation}
for appropriate constants $C,c > 0$. In fact, combining (\ref{eq:GNHN}) with a bound similar to (\ref{eq:GN-intro1}) for commutators of $\cG_N$ with $\cN_+$, we can go one step further and show that, if $\psi_N \in L^2_s (\Lambda^N)$ is such that $\psi_N =\chi (H_N -E_N \leq \zeta) \psi_N$ (i.e. if $\psi_N$ belongs to a low-energy spectral subspace of $H_N$),  the corresponding excitation vector $\xi_N = T^* U_N \psi_N$ satisfies the strong a-priori bound 
\begin{equation}\label{eq:HNk-intro} \left\langle \xi_N , \left[ (\cH_N +1) (\cN_+ + 1) + (\cN_+ + 1)^3 \right]  \xi_N \right\rangle \leq C (1+ \zeta^3) 
\end{equation}
uniformly in $N$.  

Armed with this estimate, we can have a second look at the renormalized excitation Hamiltonian $\cG_N$ and we can prove that several terms contributing to $\cG_N$ are negligible on low-energy states. We find that 
\begin{equation}\label{eq:GNdec-intro} \cG_N = C_{\cG_N} + \cQ_{\cG_N} + \cC_N + \cH_N + \cE_{\cG_N} \end{equation}
where $C_{\cG_N}$ is a constant, $\cQ_{\cG_N}$ is quadratic in (generalized) creation and annihilation operators, 
%$\cV_N$ is the potential energy operator, restricted on $\cF_+^{\leq N}$, 
$\cC_N$ is the cubic term 
\begin{equation} \label{eq:cCN-intro} \cC_N = \frac{1}{\sqrt{N}} \sum_{p,q \in \Lambda^*_+ : q \not = -p} \widehat{V} (p/N) \left[ b_{p+q}^* b_{-p}^* ( b_q \cosh (\eta_q) + b_{-q}^* \sinh (\eta_q)) + \text{h.c.} \right]  \end{equation}
and $\cE_{\cG_N}$ is an error term that can be estimated by  
\begin{equation}\label{eq:theta-intro} \pm \cE_{\cG_N} \leq C N^{-1/2} \Big[ (\cH_N + 1) (\cN_+ + 1)  + (\cN_+ + 1)^3 \Big] \end{equation}
and thus, by (\ref{eq:HNk-intro}), is negligible on low-energy states. 

The presence, in (\ref{eq:GNdec-intro}), of the cubic term $\cC_N$ and of the quartic interaction $\cV_N$ (hidden in $\cH_N$) is one of the main new challenges, compared with our analysis in \cite{BBCS2}, where we determined the ground state energy and the low-energy excitation spectrum for the Hamiltonian (\ref{eq:ham-beta}), for $0 < \beta < 1$ (and for sufficiently small interaction potentials). For $\beta < 1$, these terms were small (on low-energy states) and they could be included in the error $\cE_{\cG_N}$. For $\beta = 1$, this is no longer the case; it is easy to find normalized $\xi_N \in \cF_+^{\leq N}$ satisfying (\ref{eq:HNk-intro}), and with $\langle \xi_N, \cC_N \xi_N \rangle$ and $\langle \xi_N , \cV_N \xi_N \rangle$ of order one (not vanishing in the limit $N \to \infty$). 

It is important to notice that cubic and quartic terms do not improve with different choices of the coefficients $\eta_p$. This is related with the observation, going back to the work of Erd\H os-Schlein-Yau in \cite{ESY} and more recently to the papers \cite{NRS1,NRS2} of Napiorkowski-Reuvers-Solovej that quasi-free states can only approximate the ground state energy of a dilute Bose gas in the Gross-Pitaevskii regime up to errors of order one (to be more precise, \cite{ESY,NRS1,NRS2} study the ground state energy of an extended dilute Bose gas in the thermodynamic limit, but it is clear how to translate those results to the Gross-Pitaevskii regime). 

To extract the missing energy from the cubic and quartic terms in  (\ref{eq:GNdec-intro}), we are  going to conjugate the excitation Hamiltonian $\cG_N$ with a unitary operator of the form $S = e^{A}$, where $A$ is an antisymmetric operator, cubic in (generalized) creation and annihilation operators. Observe that a similar idea, formulated however with a different language and in a different setting, was used by Yau-Yin in \cite{YY} to find an upper bound to the ground state energy of a dilute Bose gas in the thermodynamic limit matching the Lee-Huang-Yang prediction up to second order.

With $S$, we define yet another (cubically renormalized) excitation Hamiltonian 
\[ \cJ_N = S^* \cG_N S = S^* T^* U_N H_N U_N^* T S : \cF_+^{\leq N} \to \cF_+^{\leq N} \]
With the appropriate choice of $A$, we show that 
\begin{equation}\label{eq:cJN-intro} \cJ_N = C_{\cJ_N} + \cQ_{\cJ_N} + \cV_N + \cE_{\cJ_N} \end{equation}
where $C_{\cJ_N}$ and $\cQ_{\cJ_N}$ are new constant and quadratic terms, while $\cE_{\cJ_N}$ is an error term, satisfying an estimate similar to (\ref{eq:theta-intro}) and thus negligible on low-energy states. The important difference with respect to (\ref{eq:GNdec-intro}) is that now, on the r.h.s. of  (\ref{eq:cJN-intro}), there is no cubic term! The quartic interaction term $\cV_N$ is still there, but this is a positive operator and therefore it can be ignored, at least for proving lower bounds. 

Let us quickly explain the mechanism we use to eliminate the cubic term $\cC_N$. 
Expanding to second order, we find 
\begin{equation}\label{eq:mech} \cJ_N = S^* \cG_N S = e^{-A} \cG_N e^A \simeq \cG_N + [\cG_N , A] + \frac{1}{2} [[ \cG_N , A ] , A ] + \dots \end{equation}
From the canonical commutation relations (ignoring the fact that $A$ is cubic in generalized, rather than standard, field operators) we conclude that $[\cK, A]$ and $[\cV_N, A]$ are cubic and quintic in creation and annihilation operators, respectively. Some of the terms contributing to $[\cV_N, A]$ are not in normal order, i.e. they contain creation operators lying to the right of annihilation operators. When we rearrange creation and annihilation operators to restore normal order, we generate an additional cubic contribution. There are therefore two cubic contributions arising from the first commutator $[\cG_N, A]$ on the r.h.s. of (\ref{eq:mech}). We choose $A$ so that these two terms renormalize the cubic operator (\ref{eq:cCN}) in $\cG_N$, making it small on low energy states. While the generalized Bogoliubov transformation $T$ used in the definition of $\cG_N$ described scattering processes involving two excitations with momenta $p$ and $-p$ and two particles in the condensate (i.e. two particles with zero momentum), we find that the appropriate choice of the cubic operator $A$ corresponds to processes involving two excitations with large momenta $p$ and $-p+v$, an excitation with small momentum $v$ and a particle in the condensate. It turns out that, with this choice of $A$, the only other terms generated through conjugation with $S=e^A$ that give a non-negligible contribution to $\cJ_N$ are $[\cC_N, A]$ and the second commutator $[[\cH_N, A], A]$. These two terms produce constant and quadratic contributions that transform $C_{\cG_N}$ and $\cQ_{\cG_N}$ in (\ref{eq:GNdec-intro}) into $C_{\cJ_N}$ and $\cQ_{\cJ_N}$ on the r.h.s. of (\ref{eq:cJN-intro}) (in fact, $\cQ_{\cJ_N}$ also absorbs the kinetic energy operator that was excluded from $\cQ_{\cG_N}$). 

Conjugating $\cJ_N$ with a last generalized Bogoliubov transformation $R$ to diagonalize the quadratic operator $\cQ_{\cJ_N}$, we obtain a final excitation Hamiltonian 
\[ \cM_N = R^* \cJ_N R = R^* S^* T^* U_N H_N U_N^* T S R : \cF_+^{\leq N} \to \cF_+^{\leq N}  \]
which can be written as 
\begin{equation}\label{eq:cMN-intro} \begin{split} \cM_N = \; &4 \pi  (N-1)\frak{a}_0 +e_\Lambda\mathfrak{a}_0^2- \frac{1}{2} \sum_{p \in \Lambda^*_+} \left[ p^2 + 8\pi \frak{a}_0  - \sqrt{|p|^4 + 16 \pi \frak{a}_0  p^2} - \frac{(8\pi \frak{a}_0 )^2}{2p^2} \right] \\ &+ \sum_{p \in \Lambda^*_+} \sqrt{|p|^4 + 16 \pi \frak{a}_0  p^2} \, a_p^* a_p + \cV_N  + \cE_{\cM_N} \end{split} \end{equation}
with an error term $\cE_{\cM_N}$ which satisfies   
\[ \pm \cE_{\cM_N} \leq C N^{-1/4} \Big[ (\cH_N + 1) (\cN_+ + 1)  + (\cN_+ + 1)^3 \Big] \]
and is therefore negligible on low-energy states. With (\ref{eq:cMN-intro}), Theorem \ref{thm:main} follows comparing the eigenvalues of $\cM_N$ with those of its quadratic part, by means of the min-max principle. To prove lower bounds, we can ignore the quartic interaction $\cV_N$. To prove upper bounds, on the other hand, it is enough to control the values of $\cV_N$ on low-energy eigenspaces of the quadratic operator; they turn out to be negligible. 

The plan of the paper is as follows. In Section \ref{sec:fock}, we briefly review the formalism of second quantization. In particular, we define and study the properties of generalized Bogoliubov transformations that play a very important role in our analysis. In Section~\ref{sec:ex} we introduce the excitation Hamiltonian $\cL_N$, the renormalized excitation Hamiltonian $\cG_N$ and the excitation Hamiltonian $\cJ_N$ with renormalized cubic term, and we study their properties. In particular, Prop. \ref{prop:gene} provides important bounds on $\cG_N$ while Prop.~\ref{thm:tGNo1} gives a precise description of $\cJ_N$. In Section \ref{sec:BEC}, we prove estimates for the excitation vectors associated with low-energy many-body wave functions. Section \ref{sec:diag} is devoted to the diagonalization of the quadratic part of $\cJ_N$ and Section \ref{sec:proof} applies the min-max principle to conclude the proof of Theorem \ref{thm:main}. Finally, Section \ref{sec:GN} and Section \ref{sec:wtGN} contain the proof of Prop. \ref{prop:gene} and, respectively, of Prop. \ref{thm:tGNo1}. \\

{\it Acknowledgements.} B.S. gratefully acknowledge support from the NCCR SwissMAP and from the Swiss National Foundation of Science through the SNF Grant ``Effective equations from quantum dynamics'' and the SNF Grant ``Dynamical and energetic properties of Bose-Einstein condensates''.

\section{Fock space}
\label{sec:fock}

The bosonic Fock space over $L^2 (\Lambda)$ is defined as 
\[ \cF = \bigoplus_{n \geq 0} L^2_s (\Lambda^{n}) = \bigoplus_{n \geq 0} L^2 (\Lambda)^{\otimes_s n} \]
where $L^2_s (\Lambda^{n})$ is the subspace of $L^2 (\Lambda^n)$ consisting of wave functions that are symmetric w.r.t. permutations. 
The vacuum vector in $\cF$ will be indicated with $\Omega = \{ 1, 0, \dots \} \in \cF$. 

For $g \in L^2 (\Lambda)$, the creation operator $a^* (g)$ and the annihilation operator $a(g)$ are defined by  
\[ \begin{split} 
(a^* (g) \Psi)^{(n)} (x_1, \dots , x_n) &= \frac{1}{\sqrt{n}} \sum_{j=1}^n g (x_j) \Psi^{(n-1)} (x_1, \dots , x_{j-1}, x_{j+1} , \dots , x_n) 
\\
(a (g) \Psi)^{(n)} (x_1, \dots , x_n) &= \sqrt{n+1} \int_\Lambda  \bar{g} (x) \Psi^{(n+1)} (x,x_1, \dots , x_n) \, dx \end{split} \]
Observe that $a^* (g)$ is the adjoint of $a(g)$ and that the canonical commutation relations
\begin{equation*}%\label{eq:ccr} 
[a (g), a^* (h) ] = \langle g,h \rangle , \quad [ a(g), a(h)] = [a^* (g), a^* (h) ] = 0 \end{equation*}
hold true for all $g,h \in L^2 (\Lambda)$ ($\langle g,h \rangle$ is the inner product on $L^2 (\Lambda)$). 

Due to translation invariance of our system, it will be convenient to work in the momentum space $\Lambda^* = 2\pi \bZ^3$. For $p \in \Lambda^*$, we consider the plane wave $\ph_p  (x) = e^{-ip\cdot x}$ in $L^2 (\Lambda)$. We define the operators 
\begin{equation*}%\label{eq:ap} 
a^*_p = a^* (\ph_p), \quad \text{and } \quad  a_p = a (\ph_p) \end{equation*} 
creating and, respectively, annihilating a particle with momentum $p$. 

For some parts of our analysis, we will switch to position space (where it is easier to use the positivity of the potential $V(x)$). To this end, we introduce operator valued distributions $\check{a}_x, \check{a}_x^*$ defined by 
\begin{equation*}%\label{eq:axf}
 a(f) = \int \bar{f} (x) \,  \check{a}_x \, dx , \quad a^* (f) = \int f(x) \, \check{a}_x^* \, dx  \end{equation*}

On $\cF$, we introduce the number of particles operator $\cN$ defined by $(\cN\Psi)^{(n)} = n \Psi^{(n)}$. Notice that 
\[ \cN = \sum_{p \in \Lambda^*} a_p^* a_p  = \int \check{a}^*_x \check{a}_x \, dx \,  \]
It is useful to notice that creation and annihilation operators are bounded by the square root of the number of particles operator, i.e.   
\begin{equation}\label{eq:abd} 
\| a (f) \Psi \| \leq \| f \| \| \cN^{1/2} \Psi \|, \quad \| a^* (f) \Psi \| \leq \| f \| \| (\cN+1)^{1/2} \Psi \| 
\end{equation}
for all $f \in L^2 (\Lambda)$. 

\medskip

Recall that $\ph_0$ denotes the zero-momentum mode in $L^2 (\Lambda)$, defined by $\ph_0 (x) = 1$ for all $x \in \Lambda$. We define $L^2_{\perp} (\Lambda)$ to be the orthogonal complement of the one dimensional space spanned by $\ph_0$ in $L^2 (\Lambda)$. The Fock space over $L^2_\perp (\Lambda)$ will be denoted by  
\[ \cF_{+} = \bigoplus_{n \geq 0} L^2_{\perp} (\Lambda)^{\otimes_s n} \,  \]
This Hilbert space is generated by creation and annihilation operators $a_p^*, a_p$, with $p \in \Lambda^*_+ := 2\pi \bZ^3 \backslash \{ 0 \}$. On $\cF_+$, the number of particles operator will be denoted by  
\[ \cN_+ = \sum_{p \in \Lambda^*_+} a_p^* a_p \]
For $N \in \bN$, we also define the truncated Fock space  
\[ \cF_{+}^{\leq N} = \bigoplus_{n=0}^N L^2_{\perp} (\Lambda)^{\otimes_s n} \,  \]
On $\cF_+^{\leq N}$, we introduce modified creation and annihilation operators. For $f \in L^2_\perp (\Lambda)$, we set 
\[ b (f) = \sqrt{\frac{N- \cN_+}{N}} \, a (f), \qquad \text{and } \quad  b^* (f) = a^* (f) \, \sqrt{\frac{N-\cN_+}{N}} \]
We have $b(f), b^* (f) : \cF_+^{\leq N} \to \cF_+^{\leq N}$. The interpretation of these fields becomes clear if we conjugate them with the unitary map $U_N$ defined in (\ref{eq:UN-intro}). We find 
\begin{equation*}%\label{eq:UaU} 
U_N a^* (f) \frac{a(\ph_0)}{\sqrt{N}} U_N^* = a^* (f) \sqrt{\frac{N-\cN_+}{N}}  =  b^* (f) 
\end{equation*}
which means that $b^* (f)$ excites a particle from the condensate into its orthogonal complement and, similarly, that $b_p$ annihilates an excitation back into the condensate. Compared with the standard fields $a^*,a$, the modified operators $b^*,b$ have an important advantage; they create (or annihilate) excitations but, at the same time, they preserve the total number of particles. As a consequence, again in contrast with the standard fields $a^*, a$, the modified operators $b^*, b$ leave the truncated Fock space $\cF_+^{\leq N}$ invariant. 

It is also convenient to define modified creation and annihilation operators in momentum space and operator valued modified creation and annihilation operators in position space, putting 
\begin{equation}\label{eq:bp-de} b_p = \sqrt{\frac{N-\cN_+}{N}} \, a_p, \qquad \text{and } \quad  b^*_p = a^*_p \, \sqrt{\frac{N-\cN_+}{N}} \end{equation}
for all $p \in \Lambda^*_+$, and   
\[ \check{b}_x = \sqrt{\frac{N-\cN_+}{N}} \, \check{a}_x, \qquad \text{and } \quad  \check{b}^*_x = \check{a}^*_x \, \sqrt{\frac{N-\cN_+}{N}} \]
for all $x \in \Lambda$. 

Modified creation and annihilation operators satisfy the commutation relations 
\begin{equation}\label{eq:comm-bp} \begin{split} [ b_p, b_q^* ] &= \left( 1 - \frac{\cN_+}{N} \right) \delta_{p,q} - \frac{1}{N} a_q^* a_p 
\\ [ b_p, b_q ] &= [b_p^* , b_q^*] = 0 
\end{split} \end{equation}
and, in position space, 
\begin{equation*}%\label{eq:comm-b}
\begin{split}  [ \check{b}_x, \check{b}_y^* ] &= \left( 1 - \frac{\cN_+}{N} \right) \delta (x-y) - \frac{1}{N} \check{a}_y^* \check{a}_x \\ 
[ \check{b}_x, \check{b}_y ] &= [ \check{b}_x^* , \check{b}_y^*] 
= 0 
\end{split} \end{equation*}
Furthermore 
\begin{equation*}%\label{eq:comm-b2}
\begin{split}
[\check{b}_x, \check{a}_y^* \check{a}_z] &=\delta (x-y)\check{b}_z, \qquad 
[\check{b}_x^*, \check{a}_y^* \check{a}_z] = -\delta (x-z) \check{b}_y^*
\end{split} \end{equation*}
It follows that $[ \check{b}_x, \cN_+ ] = \check{b}_x$, $[ \check{b}_x^* , \cN_+ ] = - \check{b}_x^*$ and, in momentum space, $[b_p , \cN_+] = b_p$, $[b_p^*, \cN_+] = - b_p^*$. With (\ref{eq:abd}), we obtain 
\begin{equation*}%\label{lm:bbds}
\begin{split} 
\| b(f) \xi \| &\leq \| f \| \left\| \cN_+^{1/2} \left( \frac{N+1-\cN_+}{N} \right)^{1/2} \xi \right\| \leq \| f \| \| \cN_+^{1/2} \xi \|  \\ 
\| b^* (f) \xi \| &\leq \| f \| \left\| (\cN_+ +1)^{1/2} \left( \frac{N-\cN_+ }{N} \right)^{1/2} \xi \right\| \leq \| f \| \| (\cN_+ + 1)^{1/2} \xi \| 
\end{split} \end{equation*}
for all $f \in L^2_\perp (\Lambda)$ and $\xi \in \cF^{\leq N}_+$. Since $\cN_+  \leq N$ on $\cF_+^{\leq N}$, $b(f), b^* (f)$ are bounded operators with $\| b(f) \|, \| b^* (f) \| \leq (N+1)^{1/2} \| f \|$. 

\medskip

Next, we introduce generalized Bogoliubov transformations and we discuss their properties. For $\eta \in \ell^2 (\Lambda^*_+)$ with $\eta_{-p} = \eta_{p}$ for all $p \in \Lambda^*_+$, we define 
\begin{equation}\label{eq:defB} 
B(\eta) = \frac{1}{2} \sum_{p\in \Lambda^*_+}  \left( \eta_p b_p^* b_{-p}^* - \bar{\eta}_p b_p b_{-p} \right) \, \end{equation}
and we consider 
\begin{equation}\label{eq:eBeta} 
e^{B(\eta)} = \exp \left[ \frac{1}{2} \sum_{p \in \Lambda^*_+}   \left( \eta_p b_p^* b_{-p}^* - \bar{\eta}_p  b_p b_{-p} \right) \right] 
\end{equation}
We refer to unitary operators of the form (\ref{eq:eBeta}) as generalized Bogoliubov transformations, in analogy with the standard Bogoliubov transformations 
\begin{equation}\label{eq:wteBeta} e^{\wt{B} (\eta)} = \exp \left[  \frac{1}{2} \sum_{p\in \Lambda^*_+}  \left( \eta_p a_p^* a_{-p}^* - \bar{\eta}_p a_p a_{-p} \right) \right] \end{equation}
defined by means of the standard creation and annihilation operators. In this paper, we will work with (\ref{eq:eBeta}), rather than (\ref{eq:wteBeta}), because the generalized Bogoliubov transformations, in contrast with the standard transformations, leave the truncated Fock space $\cF_+^{\leq N}$ invariant. 
The price we will have to pay is the fact that, while the action of standard Bogoliubov transformation 
on creation and annihilation operators is explicitly given by 
\begin{equation}\label{eq:act-Bog} e^{-\wt{B} (\eta)} a_p e^{\wt{B} (\eta)}  = \cosh (\eta_p) a_p + \sinh (\eta_p) a_{-p}^* \,   \end{equation}
there is no such formula describing the action of generalized Bogoliubov transformations. An important part of our analysis is therefore devoted to the control of the action of (\ref{eq:eBeta}). A first important observation in this direction is the following lemma, whose proof can be found in \cite[Lemma 3.1]{BS} (a similar result has been previously established in \cite{Sei}).
\begin{lemma}\label{lm:Ngrow}
For every $n \in \bN$ there exists a constant $C > 0$ such that, on $\cF_+^{\leq N}$, 
\begin{equation}\label{eq:bd-Beta} e^{-B(\eta)} (\cN_+ +1)^{n} e^{B(\eta)} \leq C e^{C \| \eta \|} (\cN_+ +1)^{n}  
\end{equation}
for all $\eta \in \ell^2 (\L^*)$.
\end{lemma}
Unfortunately, controlling the change of the number of particles operator is not enough for our purposes. To obtain more precise information we expand, for any $p \in \Lambda^*_+$,  
\[\begin{split} e^{-B(\eta)} \, b_p \, e^{B(\eta)} &= b_p + \int_0^1 ds \, \frac{d}{ds}  e^{-sB(\eta)} b_p e^{sB(\eta)} \\ &= b_p - \int_0^1 ds \, e^{-sB(\eta)} [B(\eta), b_p] e^{s B(\eta)} \\ &= b_p - [B(\eta),b_p] + \int_0^1 ds_1 \int_0^{s_1} ds_2 \, e^{-s_2 B(\eta)} [B(\eta), [B(\eta),b_p]] e^{s_2 B(\eta)} \end{split} \]
Iterating $m$ times, we find 
\begin{equation}\label{eq:BCH} \begin{split} 
e^{-B(\eta)} b_p e^{B(\eta)} = &\sum_{n=1}^{m-1} (-1)^n \frac{\text{ad}^{(n)}_{B(\eta)} (b_p)}{n!} \\ &+ \int_0^{1} ds_1 \int_0^{s_1} ds_2 \dots \int_0^{s_{m-1}} ds_m \, e^{-s_m B(\eta)} \text{ad}^{(m)}_{B(\eta)} (b_p) e^{s_m B(\eta)} \end{split} \end{equation}
where we recursively defined \[ \text{ad}_{B(\eta)}^{(0)} (A) = A \quad \text{and } \quad \text{ad}^{(n)}_{B(\eta)} (A) = [B(\eta), \text{ad}^{(n-1)}_{B(\eta)} (A) ]  \]
We are going to expand the nested commutators $\text{ad}_{B(\eta)}^{(n)} (b_p)$ and   
$\text{ad}_{B(\eta)}^{(n)} (b^*_p)$. To this end, we need to introduce some additional notation. 
We follow here \cite{BS,BBCS1,BBCS2}. For $f_1, \dots , f_n \in \ell_2 (\Lambda^*_+)$, $\sharp = (\sharp_1, \dots , \sharp_n), \flat = (\flat_0, \dots , \flat_{n-1}) \in \{ \cdot, * \}^n$, we set 
\begin{equation}\label{eq:Pi2}
\begin{split}  
\Pi^{(2)}_{\sharp, \flat} &(f_1, \dots , f_n) \\ &= \sum_{p_1, \dots , p_n \in \Lambda^*}  b^{\flat_0}_{\alpha_0 p_1} a_{\beta_1 p_1}^{\sharp_1} a_{\alpha_1 p_2}^{\flat_1} a_{\beta_2 p_2}^{\sharp_2} a_{\alpha_2 p_3}^{\flat_2} \dots  a_{\beta_{n-1} p_{n-1}}^{\sharp_{n-1}} a_{\alpha_{n-1} p_n}^{\flat_{n-1}} b^{\sharp_n}_{\beta_n p_n} \, \prod_{\ell=1}^n f_\ell (p_\ell)  \end{split} \end{equation}
where, for $\ell=0,1, \dots , n$, we define $\alpha_\ell = 1$ if $\flat_\ell = *$, $\alpha_\ell =    -1$ if $\flat_\ell = \cdot$, $\beta_\ell = 1$ if $\sharp_\ell = \cdot$ and $\beta_\ell = -1$ if $\sharp_\ell = *$. In (\ref{eq:Pi2}), we require that, for every $j=1,\dots, n-1$, we have either $\sharp_j = \cdot$ and $\flat_j = *$ or $\sharp_j = *$ and $\flat_j = \cdot$ (so that the product $a_{\beta_\ell p_\ell}^{\sharp_\ell} a_{\alpha_\ell p_{\ell+1}}^{\flat_\ell}$ always preserves {} the number of particles, for all $\ell =1, \dots , n-1$). With this assumption, we find that the operator $\Pi^{(2)}_{\sharp,\flat} (f_1, \dots , f_n)$ maps $\cF^{\leq N}_+$ into itself. If, for some $\ell=1, \dots , n$, $\flat_{\ell-1} = \cdot$ and $\sharp_\ell = *$ (i.e. if the product $a_{\alpha_{\ell-1} p_\ell}^{\flat_{\ell-1}} a_{\beta_\ell p_\ell}^{\sharp_\ell}$ for $\ell=2,\dots , n$, or the product $b_{\alpha_0 p_1}^{\flat_0} a_{\beta_1 p_1}^{\sharp_1}$ for $\ell=1$, is not normally ordered) we require additionally that $f_\ell  \in \ell^1 (\Lambda^*_+)$. In position space, the same operator can be written as 
\begin{equation}\label{eq:Pi2-pos} \Pi^{(2)}_{\sharp, \flat} (f_1, \dots , f_n) = \int   \check{b}^{\flat_0}_{x_1} \check{a}_{y_1}^{\sharp_1} \check{a}_{x_2}^{\flat_1} \check{a}_{y_2}^{\sharp_2} \check{a}_{x_3}^{\flat_2} \dots  \check{a}_{y_{n-1}}^{\sharp_{n-1}} \check{a}_{x_n}^{\flat_{n-1}} \check{b}^{\sharp_n}_{y_n} \, \prod_{\ell=1}^n \check{f}_\ell (x_\ell - y_\ell) \, dx_\ell dy_\ell \end{equation}
An operator of the form (\ref{eq:Pi2}), (\ref{eq:Pi2-pos}) with all the properties listed above, will be called a $\Pi^{(2)}$-operator of order $n$.

For $g, f_1, \dots , f_n \in \ell_2 (\Lambda^*_+)$, $\sharp = (\sharp_1, \dots , \sharp_n)\in \{ \cdot, * \}^n$, $\flat = (\flat_0, \dots , \flat_{n}) \in \{ \cdot, * \}^{n+1}$, we also define the operator 
\begin{equation}\label{eq:Pi1}
\begin{split} \Pi^{(1)}_{\sharp,\flat} &(f_1, \dots , f_n;g) \\ &= \sum_{p_1, \dots , p_n \in \Lambda^*}  b^{\flat_0}_{\alpha_0, p_1} a_{\beta_1 p_1}^{\sharp_1} a_{\alpha_1 p_2}^{\flat_1} a_{\beta_2 p_2}^{\sharp_2} a_{\alpha_2 p_3}^{\flat_2} \dots a_{\beta_{n-1} p_{n-1}}^{\sharp_{n-1}} a_{\alpha_{n-1} p_n}^{\flat_{n-1}} a^{\sharp_n}_{\beta_n p_n} a^{\flat n} (g) \, \prod_{\ell=1}^n f_\ell (p_\ell) \end{split} \end{equation}
where $\alpha_\ell$ and $\beta_\ell$ are defined as above. Also here, we impose the condition that, for all $\ell = 1, \dots , n$, either $\sharp_\ell = \cdot$ and $\flat_\ell = *$ or $\sharp_\ell = *$ and $\flat_\ell = \cdot$. This implies that $\Pi^{(1)}_{\sharp,\flat} (f_1, \dots , f_n;g)$ maps $\cF^{\leq N}_+$ back into $\cF_+^{\leq N}$. Additionally, we assume that $f_\ell \in \ell^1 (\Lambda^*_+)$ if $\flat_{\ell-1} = \cdot$ and $\sharp_\ell = *$ for some $\ell = 1,\dots , n$ (i.e. if the pair $a_{\alpha_{\ell-1} p_\ell}^{\flat_{\ell-1}} a^{\sharp_\ell}_{\beta_\ell p_\ell}$ is not normally ordered). In position space, the same operator can be written as
\begin{equation}\label{eq:Pi1-pos} \Pi^{(1)}_{\sharp,\flat} (f_1, \dots ,f_n;g) = \int \check{b}^{\flat_0}_{x_1} \check{a}_{y_1}^{\sharp_1} \check{a}_{x_2}^{\flat_1} \check{a}_{y_2}^{\sharp_2} \check{a}_{x_3}^{\flat_2} \dots  \check{a}_{y_{n-1}}^{\sharp_{n-1}} \check{a}_{x_n}^{\flat_{n-1}} \check{a}^{\sharp_n}_{y_n} \check{a}^{\flat n} (g) \, \prod_{\ell=1}^n \check{f}_\ell (x_\ell - y_\ell) \, dx_\ell dy_\ell \end{equation}
An operator of the form (\ref{eq:Pi1}), (\ref{eq:Pi1-pos}) will be called a $\Pi^{(1)}$-operator of order $n$. Operators of the form $b(f)$, $b^* (f)$, for a $f \in \ell^2 (\Lambda^*_+)$, will be called $\Pi^{(1)}$-operators of order zero. 

The next lemma gives a detailed analysis of the nested commutators $\text{ad}^{(n)}_{B(\eta)} (b_p)$ and $\text{ad}^{(n)}_{B(\eta)} (b^*_p)$ for $n \in \bN$; the proof can be found in \cite[Lemma 2.5]{BBCS1}(it is a translation to momentum space of \cite[Lemma 3.2]{BS}). 
\begin{lemma}\label{lm:indu}
Let $\eta \in \ell^2 (\Lambda^*_+)$ be such that $\eta_p = \eta_{-p}$ for all $p \in \ell^2 (\Lambda^*)$. To simplify the notation,  assume also $\eta$ to be real-valued (as it will be in applications). Let $B(\eta)$ be defined as in (\ref{eq:defB}), $n \in \bN$ and $p \in \Lambda^*$. Then the nested commutator $\text{ad}^{(n)}_{B(\eta)} (b_p)$ can be written as the sum of exactly $2^n n!$ terms, with the following properties. 
\begin{itemize}
\item[i)] Possibly up to a sign, each term has the form
\begin{equation}\label{eq:Lambdas} \Lambda_1 \Lambda_2 \dots \Lambda_i \, N^{-k} \Pi^{(1)}_{\sharp,\flat} (\eta^{j_1}, \dots , \eta^{j_k} ; \eta^{s}_p \ph_{\alpha p}) 
\end{equation}
for some $i,k,s \in \bN$, $j_1, \dots ,j_k \in \bN \backslash \{ 0 \}$, $\sharp \in \{ \cdot, * \}^k$, $ \flat \in \{ \cdot, * \}^{k+1}$ and $\alpha \in \{ \pm 1 \}$ chosen so that $\alpha = 1$ if $\flat_k = \cdot$ and $\alpha = -1$ if $\flat_k = *$ (recall here that $\ph_p (x) = e^{-ip \cdot x}$). In (\ref{eq:Lambdas}), each operator $\Lambda_w : \cF^{\leq N} \to \cF^{\leq N}$, $w=1, \dots , i$, is either a factor $(N-\cN_+ )/N$, a factor $(N-(\cN_+ -1))/N$ or an operator of the form
\begin{equation}\label{eq:Pi2-ind} N^{-h} \Pi^{(2)}_{\sharp',\flat'} (\eta^{z_1}, \eta^{z_2},\dots , \eta^{z_h}) \end{equation}
for some $h, z_1, \dots , z_h \in \bN \backslash \{ 0 \}$, $\sharp,\flat  \in \{ \cdot , *\}^h$. 
\item[ii)] If a term of the form (\ref{eq:Lambdas}) contains $m \in \bN$ factors $(N-\cN_+ )/N$ or $(N-(\cN_+ -1))/N$ and $j \in \bN$ factors of the form (\ref{eq:Pi2-ind}) with $\Pi^{(2)}$-operators of order $h_1, \dots , h_j \in \bN \backslash \{ 0 \}$, then 
we have
\begin{equation*}%\label{eq:totalb}
 m + (h_1 + 1)+ \dots + (h_j+1) + (k+1) = n+1 \end{equation*}
\item[iii)] If a term of the form (\ref{eq:Lambdas}) contains (considering all $\Lambda$-operators and the $\Pi^{(1)}$-operator) the arguments $\eta^{i_1}, \dots , \eta^{i_m}$ and the factor $\eta^{s}_p$ for some $m, s \in \bN$, and $i_1, \dots , i_m \in \bN \backslash \{ 0 \}$, then \[ i_1 + \dots + i_m + s = n .\]
\item[iv)] There is exactly one term having of the form (\ref{eq:Lambdas}) with $k=0$ and such that all $\Lambda$-operators are factors of $(N-\cN_+ )/N$ or of $(N+1-\cN_+ )/N$. It is given by 
\begin{equation*}%\label{eq:iv1} 
\left(\frac{N-\cN_+ }{N} \right)^{n/2} \left(\frac{N+1-\cN_+ }{N} \right)^{n/2} \eta^{n}_p b_p 
\end{equation*}
if $n$ is even, and by 
\begin{equation*} %\label{eq:iv2} 
- \left(\frac{N-\cN_+ }{N} \right)^{(n+1)/2} \left(\frac{N+1-\cN_+ }{N} \right)^{(n-1)/2} \eta^{n}_p b^*_{-p}  \end{equation*}
if $n$ is odd.
\item[v)] If the $\Pi^{(1)}$-operator in (\ref{eq:Lambdas}) is of order $k \in \bN \backslash \{ 0 \}$, it has either the form  
\[ \sum_{p_1, \dots , p_k}  b^{\flat_0}_{\alpha_0 p_1} \prod_{i=1}^{k-1} a^{\sharp_i}_{\beta_i p_{i}} a^{\flat_i}_{\alpha_i p_{i+1}}  a^*_{-p_k} \eta^{2r}_p  a_p \prod_{i=1}^k \eta^{j_i}_{p_i}  \]
or the form 
\[\sum_{p_1, \dots , p_k} b^{\flat_0}_{\alpha_0 p_1} \prod_{i=1}^{k-1} a^{\sharp_i}_{\beta_i p_{i}} a^{\flat_i}_{\alpha_i p_{i+1}}  a_{p_k} \eta^{2r+1}_p a^*_p \prod_{i=1}^k \eta^{j_i}_{p_i}  \]
for some $r \in \bN$, $j_1, \dots , j_k \in \bN \backslash \{ 0 \}$. If it is of order $k=0$, then it is either given by $\eta^{2r}_p b_p$ or by $\eta^{2r+1}_p b_{-p}^*$, for some $r \in \bN$. 
\item[vi)] For every non-normally ordered term of the form 
\[ \begin{split} &\sum_{q \in \Lambda^*} \eta^{i}_q a_q a_q^* , \quad \sum_{q \in \Lambda^*} \, \eta^{i}_q b_q a_q^* \\  &\sum_{q \in \Lambda^*} \, \eta^{i}_q a_q b_q^*, \quad \text{or } \quad \sum_{q \in \Lambda^*} \, \eta^{i}_q b_q b_q^*  \end{split} \]
appearing either in the $\Lambda$-operators or in the $\Pi^{(1)}$-operator in (\ref{eq:Lambdas}), we have $i \geq 2$.
\end{itemize}
\end{lemma}
With Lemma \ref{lm:indu}, it follows from (\ref{eq:BCH}) that, if $\| \eta \|$ is sufficiently small, 
\begin{equation}\label{eq:conv-serie}
\begin{split} e^{-B(\eta)} b_p e^{B (\eta)} &= \sum_{n=0}^\infty \frac{(-1)^n}{n!} \text{ad}_{B(\eta)}^{(n)} (b_p) \\
e^{-B(\eta)} b^*_p e^{B (\eta)} &= \sum_{n=0}^\infty \frac{(-1)^n}{n!} \text{ad}_{B(\eta)}^{(n)} (b^*_p) \end{split} \end{equation}
where the series converge absolutely (the proof is a translation to momentum space of \cite[Lemma 3.3]{BS}). 

In our analysis, we will use the fact that, on states with $\cN_+ \ll N$, the action of the generalized Bogoliubov transformation (\ref{eq:eBeta}) can be approximated by the action of the standard Bogoliubov transformation (\ref{eq:wteBeta}), which is explicitly given by (\ref{eq:act-Bog}) (from the definition (\ref{eq:bp-de}), we expect that $b_p \simeq a_p$ and $b_p^* \simeq a_p^*$ on states with $\cN_+ \ll N$). To make this statement more precise we define, under the assumption that $\| \eta \|$ is small enough,  the remainder operators 
\begin{equation} \label{eq:defD}
d_q =\sum_{m\geq 0}\frac{1}{m!} \Big [\text{ad}_{-B(\eta)}^{(m)}(b_q) - \eta_q^m b_{\alpha_m q}^{\sharp_m }  \Big],\hspace{0.5cm} d^*_q =\sum_{m\geq 0}\frac{1}{m!} \Big [\text{ad}_{-B(\eta)}^{(m)}(b^*_q) - \eta_q^m b_{\alpha_m q}^{\sharp_{m+1}}  \Big]\end{equation}
where $q \in \L^*_+$, $ (\sharp_m, \alpha_m) = (\cdot, +1)$ if $m$ is even and $(\sharp_m, \alpha_m) = (*, -1)$ if $m$ is odd. It follows then from (\ref{eq:conv-serie}) that 
\begin{equation}\label{eq:ebe} e^{-B(\eta)} b_q e^{B(\eta)} = \gamma_q  b_q +\s_q b^*_{-q} + d_q, \hspace{1cm} e^{-B(\eta)} b^*_q e^{B(\eta)} = \g_q b^*_q +\s_q b_{-q} + d^*_q  \end{equation} 
where we introduced the notation $\g_q = \cosh (\eta_q)$ and $\s_q = \sinh (\eta_q)$. It will also be useful to introduce remainder operators in position space. For $x \in \Lambda$, we define the operator valued distributions $\check{d}_x, \check{d}^*_x$ through
\[ e^{-B(\eta)} \check{b}_x e^{B(\eta)} = b ( \check{\g}_x)  +  b^* (\check{\s}_x) + \check{d}_x, \qquad 
e^{-B(\eta)} \check{b}^*_x e^{B(\eta)} = b^* ( \check{\gamma}_x)  +  b (\check{\s}_x) + \check{d}^*_x
\]
where $\check{\gamma}_x (y) = \sum_{q \in \Lambda^*} \cosh (\eta_q) e^{iq \cdot (x-y)}$ and $\check{\s}_x (y) = \sum_{q \in \Lambda^*} \sinh (\eta_q) e^{iq \cdot (x-y)}$.  

The next lemma confirms the intuition that remainder operators are small, on states with $\cN_+ \ll N$. 
This Lemma is the result that will be used in the rest of the paper (in particular in Section \ref{sec:GN}) to control the action of generalized Bogoliubov transformations. 
\begin{lemma}\label{lm:dp}
Let $\eta \in \ell^2 (\Lambda_+^*)$, $n \in \bZ$. Let the remainder operators be defined as in (\ref{eq:defD}). Then, if $\| \eta \|$ is small enough, there exists $C > 0$ such that  
\begin{equation}\label{eq:d-bds}
\begin{split} 
\| (\cN_+ + 1)^{n/2} d_p \xi \| &\leq \frac{C}{N} \left[ |\eta_p| \| (\cN_+ + 1)^{(n+3)/2} \xi \| + \| b_p (\cN_+ + 1)^{(n+2)/2} \xi \| \right], \\ 
\| (\cN_+ + 1)^{n/2} d_p^* \xi \| &\leq \frac{C}{N} \| (\cN_+ +1)^{(n+3)/2} \xi \| \end{split}  \end{equation}
for all $p \in \L^*_+$ and, in position space, such that 
\begin{equation}\label{eq:dxy-bds} \begin{split}  \| (\cN_+ + 1)^{n/2} \check{d}_x \xi \| \leq \; &\frac{C}{N} \left[ \| (\cN_+ + 1)^{(n+3)/2} \xi \| + \| \check{a}_x (\cN_+ + 1)^{(n+2)/2} \xi \| \right] \\
 \| (\cN_+ + 1)^{n/2} \check{a}_y \check{d}_x \xi \| \leq \; &\frac{C}{N} \left[ \| \check{a}_x (\cN_+ + 1)^{(n+1)/2} \xi \| + (1 + |\check{\eta} (x-y)|) \| (\cN_+ + 1)^{(n+2)/2} \xi \| \right. \\ &\hspace{1cm} \left. + \| \check{a}_y (\cN_+ + 1)^{(n+3)/2} \xi \| + \| \check{a}_x \check{a}_y (\cN_+ + 1)^{(n+2)/2} \xi \| 
 \right] \\ 
\| (\cN_+ + 1)^{n/2} \check{d}_x \check{d}_y \xi \| \leq \; &\frac{C}{N^2} \left[ \| (\cN_++ 1)^{(n+6)/2} \xi \| + |\check{\eta} (x-y)| \| (\cN_+ + 1)^{(n+4)/2}  \xi \| \right. \\ &\hspace{1cm} + \| \check{a}_x (\cN_+ + 1)^{(n+5)/2} \xi \| + \| \check{a}_y (\cN_+ + 1)^{(n+5)/2} \xi \| \\ &\hspace{1cm} \left. + \| \check{a}_x \check{a}_y (\cN_+ +  1)^{(n+4)/2} \xi \| \right] \end{split} \end{equation}
for all $x, y \in \Lambda$, in the sense of distributions.
\end{lemma}
\begin{proof}
To prove the first bound in (\ref{eq:d-bds}), we notice that, from (\ref{eq:defD}) and from the triangle inequality (for simplicity, we focus on $n=0$, powers of $\cN_+$ can be easily commuted through the operators $d_p$), 
\begin{equation} \label{eq:d-sum}\| d_q \xi \| \leq \sum_{m \geq 0} \frac{1}{m!} \left\| \left[ \text{ad}^{(m)}_{-B(\eta)} (b_q) - \eta_q^m b^{\sharp_m}_{\alpha_m p} \right] \xi \right\| \end{equation} 
From Lemma \ref{lm:indu}, we can bound the norm $\| [ \text{ad}^{(m)}_{-B(\eta)} (b_q) - \eta_q^m b^{\sharp_m}_{\alpha_m p} ] \xi \|$ by the sum of one term of the form
\begin{equation}\label{eq:N-term} 
\left\| \left[ \left( \frac{N- \cN_+}{N} \right)^{\frac{m+ (1-\alpha_m)/2}{2}} \left( \frac{N+1-\cN_+}{N} \right)^{\frac{m-(1-\alpha_m)/2}{2}} - 1 \right] \eta_p^m b^{\sharp_m}_{\alpha_m p} \xi \right\| \end{equation}
and of exactly $2^m m! - 1$ terms of the form
\begin{equation}\label{eq:L-term} \left\| \Lambda_1 \dots \Lambda_{i_1} N^{-k_1} \Pi^{(1)}_{\sharp,\flat} (\eta^{j_1} , \dots , \eta^{j_{k_1}} ; \eta^{\ell_1}_p \ph_{\alpha_{\ell_1} p}) \xi \right\| \end{equation}
where $i_1, k_1, \ell_1 \in \bN$, $j_1, \dots , j_{k_1} \in \bN \backslash \{ 0 \}$ and where each $\Lambda_r$-operator is either a factor $(N-\cN_+ )/N$, a factor $(N+1-\cN_+ )/N$ or a $\Pi^{(2)}$-operator of the form 
\begin{equation}\label{eq:Pi2-ex}
N^{-h} \Pi^{(2)}_{\underline{\sharp}, \underline{\flat}} (\eta^{z_1} , \dots, \eta^{z_h}) 
\end{equation}
with $h, z_1, \dots , z_h \in \bN \backslash \{ 0 \}$. Furthermore, since we are considering the term (\ref{eq:N-term}) separately, each term of the form (\ref{eq:L-term}) must have either $k_1 > 0$ or it must contain at least one $\Lambda$-operator having the form (\ref{eq:Pi2-ex}). Since (\ref{eq:N-term}) vanishes for $m=0$, it is easy to bound
\[ \begin{split} &\left\| \left[ \left( \frac{N- \cN_+}{N} \right)^{\frac{m+ (1-\alpha_m)/2}{2}} \left( \frac{N+1-\cN_+}{N} \right)^{\frac{m-(1-\alpha_m)/2}{2}} - 1 \right] \eta_p^m b^{\sharp_m}_{\alpha_m p} \xi \right\| \\ & \hspace{8cm} \leq C^m \| \eta \|^{m-1} N^{-1} |\eta_p| \| (\cN_+ + 1)^{3/2} \xi \| \end{split} \]
On the other hand, distinguishing the cases $\ell_1 = 0$ and $\ell_1 > 0$, we can bound
\[  \begin{split}  &\left\| \Lambda_1 \dots \Lambda_{i_1} N^{-k_1} \Pi^{(1)}_{\sharp,\flat} (\eta^{j_1} , \dots , \eta^{j_{k_1}} ; \eta^{\ell_1}_p \ph_{\alpha_{\ell_1} p}) \xi \right\| \\ & \hspace{4cm} \leq C^m \| \eta \|^{m-1} N^{-1} \left[  |\eta_p| \| (\cN_+ + 1)^{3/2} \xi \| + \| b_p (\cN_+ + 1) \xi \| \right] \end{split} \]
Inserting the last two bounds in (\ref{eq:d-sum}) and summing over $m$ under the assumption that $\| \eta \| $ is small enough, we arrive at the first estimate (\ref{eq:d-bds}). The second estimate in (\ref{eq:d-bds}) can be proven similarly (notice that, when dealing with the second estimate in (\ref{eq:d-bds}), contributions of the form (\ref{eq:L-term}) with $\ell_1 = 0$, can only be bounded by $\| b_p^* (\cN_+ +1) \xi \| \leq \| (\cN_+ + 1)^{3/2} \xi \|$). Also the bounds in (\ref{eq:dxy-bds}) can be 
shown analogously, using \cite[Lemma 7.2]{BBCS2}. 
\end{proof}

\section{Excitation Hamiltonians} 
\label{sec:ex}

Recall the definition (\ref{eq:UN-intro}) of the unitary operator $U_N: L^2_s (\Lambda^N) \to \cF_+^{\leq N}$, first introduced in \cite{LNSS}.  In terms of creation and annihilation operators, $U_N$ is given by 
\[ U_N \, \psi_N = \bigoplus_{n=0}^N  (1-|\ph_0 \rangle \langle \ph_0|)^{\otimes n} \frac{a(\ph_0)^{N-n}}{\sqrt{(N-n)!}} \, \psi_N \]
for all $\psi_N \in L^2_s (\Lambda^N)$ (on the r.h.s. we identify $\psi_N \in L^2_s (\Lambda^N)$ with $\{ 0, \dots , 0, \psi_N, 0, \dots \} \in \cF$). The map $U_N^* : \cF_{+}^{\leq N} \to L^2_s (\Lambda^N)$ is given, on the other hand, by 
\[ U_N^* \, \{ \alpha^{(0)}, \dots , \alpha^{(N)} \} = \sum_{n=0}^N \frac{a^* (\ph_0)^{N-n}}{\sqrt{(N-n)!}} \, \alpha^{(n)} \]

It is instructive to compute the action of $U_N$ on products of a creation and an annihilation operator (products of the form $a_p^* a_q$ can be thought of as operators mapping $L^2_s (\Lambda^N)$ to itself). For any $p,q \in \Lambda^*_+ = 2\pi \bZ^3 \backslash \{ 0 \}$, we find (see \cite{LNSS}):
\begin{equation}\label{eq:U-rules}
\begin{split} 
U_N \, a^*_0 a_0 \, U_N^* &= N- \cN_+  \\
U_N \, a^*_p a_0 \, U_N^* &= a^*_p \sqrt{N-\cN_+ } \\
U_N \, a^*_0 a_p \, U_N^* &= \sqrt{N-\cN_+ } \, a_p \\
U_N \, a^*_p a_q \, U_N^* &= a^*_p a_q 
\end{split} \end{equation}
Writing (\ref{eq:Ham0}) in momentum space and using the formalism of second quantization, we find 
\begin{equation}\label{eq:Hmom} H_N = \sum_{p \in \Lambda^*} p^2 a_p^* a_p + \frac{1}{2N} \sum_{p,q,r \in \Lambda^*} \widehat{V} (r/N) a_{p+r}^* a_q^* a_{p} a_{q+r} 
\end{equation}
where 
\[ \widehat{V} (k) = \int_{\bR^3} V (x) e^{-i k \cdot x} dx \] 
is the Fourier transform of $V$, defined for all $k \in \bR^3$. With (\ref{eq:U-rules}), we can compute the excitation Hamiltonian $\cL_N = U_N H_N U_N^* : \cF_+^{\leq N} \to \cF_+^{\leq N}$. We obtain 
\begin{equation}\label{eq:cLN} \cL_N =  \cL^{(0)}_{N} + \cL^{(2)}_{N} + \cL^{(3)}_{N} + \cL^{(4)}_{N} \end{equation}
with
\begin{equation}\label{eq:cLNj} \begin{split} 
\cL_{N}^{(0)} =\;& \frac{(N-1)}{2N} \widehat{V} (0) (N-\cN_+ ) + \frac{\widehat{V} (0)}{2N} \cN_+  (N-\cN_+ ) \\
\cL^{(2)}_{N} =\; &\sum_{p \in \Lambda^*_+} p^2 a_p^* a_p + \sum_{p \in \Lambda_+^*} \widehat{V} (p/N) \left[ b_p^* b_p - \frac{1}{N} a_p^* a_p \right] \\ &+ \frac{1}{2} \sum_{p \in \Lambda^*_+} \widehat{V} (p/N) \left[ b_p^* b_{-p}^* + b_p b_{-p} \right] \\
\cL^{(3)}_{N} =\; &\frac{1}{\sqrt{N}} \sum_{p,q \in \Lambda_+^* : p+q \not = 0} \widehat{V} (p/N) \left[ b^*_{p+q} a^*_{-p} a_q  + a_q^* a_{-p} b_{p+q} \right] \\
\cL^{(4)}_{N} =\; & \frac{1}{2N} \sum_{\substack{p,q \in \Lambda_+^*, r \in \Lambda^*: \\ r \not = -p,-q}} \widehat{V} (r/N) a^*_{p+r} a^*_q a_p a_{q+r} 
\end{split} \end{equation}

Conjugation with $U_N$ extracts, from the original quartic interaction, some constant and quadratic contributions, collected in $\cL^{(0)}_N$ and $\cL^{(2)}_N$. In the Gross-Pitevskii regime, however, this is not enough; there are still important (order $N$) contributions to the ground state energy and to the energy of low-lying excitations that are hidden in the cubic and quartic terms. In other words, in contrast with the mean-field regime, here we cannot expect $\cL^{(3)}_N$ and $\cL^{(4)}_N$ to be small. To extract the relevant contributions from $\cL_N^{(3)}$ and $\cL_N^{(4)}$, we are going to conjugate $\cL_N$ with a generalized Bogoliubov transformation of the form (\ref{eq:eBeta}). 

To choose the function $\eta \in \ell^2 (\Lambda^*_+)$ entering the generalized Bogoliubov transformation (\ref{eq:eBeta}), we consider the ground state solution of the Neumann problem 
\begin{equation}\label{eq:scatl} \left[ -\Delta + \frac{1}{2} V \right] f_{\ell} = \lambda_{\ell} f_\ell \end{equation}
on the  ball $|x| \leq N\ell$ (we omit the $N$-dependence in the notation for $f_\ell$ and for $\lambda_\ell$; notice that $\lambda_\ell$ scales as $N^{-3}$), with the normalization $f_\ell (x) = 1$ if $|x| = N \ell$. It is also useful to define $w_\ell = 1-f_\ell$ (so that $w_\ell (x) = 0$ if $|x| > N \ell$). By scaling, we observe that $f_\ell (N.)$ satisfies the equation 
\[ \left[ -\Delta + \frac{N^2}{2} V (Nx) \right] f_\ell (Nx) = N^2 \lambda_\ell f_\ell (Nx) \]
on the ball $|x| \leq \ell$. We choose $0 < \ell < 1/2$, so that the ball of radius $\ell$ is contained in the box $\Lambda= [-1/2 ; 1/2]^3$. We extend then $f_\ell (N.)$ to $\Lambda$, by choosing $f_\ell (Nx) = 1$ for all $|x| > \ell$. Then  
\begin{equation}\label{eq:scatlN}
 \left[ -\Delta + \frac{N^2}{2} V (Nx) \right] f_\ell (Nx) = N^2 \lambda_\ell f_\ell (Nx) \chi_\ell (x) 
\end{equation}
where $\chi_\ell$ is the characteristic function of the ball of radius $\ell$. It follows that the functions $x \to f_\ell (Nx)$ and also $x \to w_\ell (Nx) = 1 - f_\ell (Nx)$ can be extended as periodic functions on the torus $\Lambda$. The Fourier coefficients of the function $x \to w_\ell (Nx)$ are given by 
\[  \int_{\Lambda} w_\ell (Nx) e^{-i p \cdot x} dx = \frac{1}{N^3} \widehat{w}_\ell (p/N) \]
where \[ \widehat{w}_\ell (p) = \int_{\bR^3} w_\ell (x) e^{-ip \cdot x} dx \] is the Fourier transform of the (compactly supported) function $w_\ell$. The Fourier coefficients of $x \to f_\ell (Nx)$ are then given by 
\begin{equation}\label{eq:fellN} \widehat{f}_{\ell,N} (p) := \int_\Lambda f_\ell (Nx) e^{-i p \cdot x} dx = \delta_{p,0} - \frac{1}{N^3} \widehat{w}_\ell (p/N) \end{equation}
for all $p \in \Lambda^*$. {F}rom (\ref{eq:scatlN}), we derive 
\begin{equation}\label{eq:wellp}
\begin{split}  - p^2 \widehat{w}_\ell (p/N) +  \frac{N^2}{2} \sum_{q \in \L^*} \widehat{V} ((p-q)/N) \widehat{f}_{\ell,N} (q) = N^5 \lambda_\ell \sum_{q \in \L^*} \widehat{\chi}_\ell (p-q) \widehat{f}_{\ell,N} (q) \end{split} \end{equation}

In the next lemma we collect some important properties of $w_\ell, f_\ell$. The proof of the lemma can be found in Appendix \ref{appx:sceq}. 
% Notice that this lemma is the reason why we require that $V \in L^3 
%(\bR^3)$; for the rest of the analysis $V \in L^2 (\bR^3)$ would be %enough. 
\begin{lemma} \label{3.0.sceqlemma}
Let $V \in L^3 (\bR^3)$ be non-negative, compactly supported and spherically symmetric. Fix $\ell> 0$ and let $f_\ell$ denote the solution of \eqref{eq:scatl}. 
\begin{enumerate}
\item [i)] We have 
\begin{equation}\label{eq:lambdaell} 
  \lambda_\ell = \frac{3\frak{a}_0 }{(\ell N)^3} \left(1 + \frac{9}{5}\frac{\frak{a}_0}{\ell N}+\mathcal{O} \big(\frak{a}_0^2  / (\ell N)^2\big) \right)
\end{equation}
\item[ii)] We have $0\leq f_\ell, w_\ell \leq 1$ and there exists a constant $C > 0$ such that 
%\begin{equation}\label{eq:Vfa0} \left|  \int  V(x) f_\ell (x) dx - 8\pi \frak{a}_0  \right| \leq \frac{C \frak{a}_0^2}{\ell N} \, . 
%\end{equation}    
\begin{equation} \label{eq:Vfa0} 
\left|  \int_{\bR^3}  V(x) f_\ell (x) dx - 8\pi \frak{a}_0 \Big(1 + \frac 3 2 \frac{\frak{a}_0}{ \ell N}  \Big)  \right| \leq \frac{C \frak{a}_0^3}{(\ell N)^2} \, . 
\end{equation}
for all $\ell \in (0;1/2)$, $N \in \bN$.
\item[iii)] There exists a constant $C>0$ such that 
	\begin{equation}\label{3.0.scbounds1} 
	w_\ell(x)\leq \frac{C}{|x|+1} \quad\text{ and }\quad |\nabla w_\ell(x)|\leq \frac{C }{x^2+1}. 
	\end{equation}
for all $x\in \bR^3$, $\ell \in (0;1/2)$ and $N \in \bN$ large enough. Moreover,
\be \label{eq:intw}
\Big| \frac{1}{(N \ell)^2} \int_{\bR^3} w_{\ell}(x) dx - \frac 2 5 \pi \frak{a}_0   \Big| \leq   \frac{C \mathfrak{a}_0^2}{N \ell}
\ee
for all $\ell \in (0;1/2)$ and $N \in \bN$ large enough.
\item[iv)] There exists a constant $C > 0$ such that 
\[ |\widehat{w}_\ell (p/N)| \leq \frac{CN^2}{p^2} \]
for all $p \in \Lambda^*_+$, $\ell \in (0;1/2)$ and $N \in \bN$ large enough. 
\end{enumerate}        
\end{lemma}

We define $\eta : \Lambda^* \to \bR$ through  
\begin{equation}\label{eq:ktdef} \eta_p = - \frac{1}{N^2}  \widehat{w}_\ell (p/N) 
\end{equation}
{F}rom (\ref{eq:wellp}), we find that these coefficients satisfy the relation  
\begin{equation}\label{eq:eta-scat0}
\begin{split} 
p^2 \eta_p + \frac{1}{2} \widehat{V} (p/N) & + \frac{1}{2N} \sum_{q \in \Lambda^*} \widehat{V} ((p-q)/N) \eta_q = N^3 \lambda_\ell \sum_{q \in \L^*} \widehat{\chi}_\ell (p-q) \widehat{f}_{\ell,N} (q) 
\end{split} \end{equation}
or equivalently, expressing also the r.h.s. through the coefficients $\eta_p$, 
\begin{equation}\label{eq:eta-scat}
\begin{split} 
p^2 \eta_p + \frac{1}{2} \widehat{V} (p/N) & + \frac{1}{2N} \sum_{q \in \Lambda^*} \widehat{V} ((p-q)/N) \eta_q \\ &\hspace{2cm} = N^3 \lambda_\ell \widehat{\chi}_\ell (p) + N^2 \lambda_\ell \sum_{q \in \Lambda^*} \widehat{\chi}_\ell (p-q) \eta_q
\end{split} \end{equation}
With Lemma \ref{3.0.sceqlemma}, we can bound  
\begin{equation}\label{eq:etap} |\eta_p | \leq \frac{C}{p^2} \end{equation}
for all $p \in \Lambda^*_+ = 2\pi \bZ^3 \backslash \{ 0 \}$. Eq. (\ref{eq:etap}) implies that $\eta \in \ell^2 (\Lambda^*_+)$, with norm bounded uniformly in $N$. In fact, denoting by $\check{\eta} \in L^2 (\Lambda)$ the function with Fourier coefficients $\eta_p$, and using the first bound in (\ref{3.0.scbounds1}), we even find 
\begin{equation}\label{eq:ellbdeta} \| \eta \|^2 = \| \check{\eta} \|^2 = \| N w_\ell (N.) \| = N^2 \int_{|x| \leq \ell} |w_\ell (Nx)|^2  dx = \int_{|x| \leq \ell} \frac{1}{|x|^2} dx \leq C \ell  \end{equation}
which implies that $\| \eta \|$ can be made arbitrarily small, by choosing $\ell \in (0;1/2)$ small enough (this remark will allow us to use the expansions (\ref{eq:conv-serie}) and the bounds in Lemma \ref{lm:dp}). Notice that $\check{\eta} (x)$ has a singularity of the form $|x|^{-1}$ at $x=0$, regularized only on the scale $N^{-1}$. In particular, from (\ref{3.0.scbounds1}), we obtain that 
\begin{equation}\label{eq:checketainf} \| \check{\eta} \|_\infty \leq C N \end{equation}
and that 
\begin{equation}\label{eq:etapN} 
\| \nabla \check{\eta} \|_2^2 = \sum_{p \in \Lambda^*_+} p^2 |\eta_p|^2  \leq C N 
\end{equation}

We will mostly use the coefficients $\eta_p$ with $p \not = 0$. Sometimes, however, it will also be useful to have an estimate for $\eta_0$ (because the equation (\ref{eq:eta-scat}) involves $\eta_0$). {F}rom  Lemma~\ref{3.0.sceqlemma}, part iii), we find 
\begin{equation*}%\label{eq:wteta0} 
 |\eta_0| \leq N^{-2} \int_{\bR^3} w_\ell (x) dx \leq C \ell^2 \end{equation*}

Because of (\ref{eq:ebe}), it will also be useful to have bounds for the quantities $\s_q = \sinh (\eta_q)$ and $\g_q = \cosh (\eta_q)$, and, in position space, for $\check{\s} (x) = \sum_{q \in \L^*} \sinh (\eta_q) e^{iq \cdot x}$ and $\check{\g} (x) = \sum_{q \in \L^*} \cosh (\eta_q) e^{iq \cdot x} = \delta (x) + \check{r} (x)$, with $\check{r} (x) = \sum_{q \in \L^*} [ \cosh (\eta_q)  - 1] \, e^{iq \cdot x}$. In momentum space, we find the pointwise bounds 
\begin{equation}\label{eq:bdsg-mom} |\s_q| \leq C  |q|^{-2} , \quad |\s_q - \eta_q| \leq C |q|^{-6}, \quad |\g_q| \leq C, \quad |\g_q - 1| \leq C |q|^{-4} \end{equation}
for all $q \in \L^*_+$. In position space, we obtain from (\ref{eq:checketainf}) the estimates  
\begin{equation}\label{eq:bdsg-x} \|  \check{\s} \|_2 \leq C , \quad \| \check{\s} \|_\infty \leq C N, \quad \| \check{\s} * \check{\g} \|_\infty \leq C N \end{equation}

%%%%%%%%%%%%%%%%%%%%%%%%%%%%%%%%%%%%%%%%%%%%%%%%%%%%%%%%%%%
%%%%%%%%%%%%%%%%%%%%%%%%%%%%%%%%%%%%%%%%%%%%%%%%%%%%%

With $\eta \in \ell^2 (\Lambda^*_+)$, we construct the generalized Bogoliubov transformation $e^{B(\eta)} : \cF_+^{\leq N} \to \cF^{\leq N}_+$, defined as in (\ref{eq:eBeta}). Furthermore, we define a new, renormalized, excitation Hamiltonian $\cG_N : \cF^{\leq N}_+ \to \cF^{\leq N}_+$ by setting 
\begin{equation}\label{eq:GN} \cG_N = e^{-B(\eta)} \cL_N e^{B(\eta)} = e^{-B(\eta)} U_N H_N U_N^* e^{B(\eta)} : \cF_+^{\leq N} \to \cF_+^{\leq N}
\end{equation}

In the next proposition, we collect some important properties of the renormalized excitation Hamiltonian $\cG_N$. Here and in the following, we will use the notation 
\begin{equation}\label{eq:KcVN}  
\cK = \sum_{p \in \Lambda^*_+} p^2 a_p^* a_p \qquad \text{and } \quad \cV_N = \frac{1}{2N} \sum_{p,q \in \Lambda_+^*, r \in \Lambda^* : r \not = -p, -q} \widehat{V} (r/N)  a_{p+r}^* a_q^* a_{q+r} a_p \end{equation}
for the kinetic and potential energy operators, restricted on $\cF_+^{\leq N}$. Furthermore, we will write $\cH_N = \cK + \cV_N$. 

\begin{prop}\label{prop:gene}
Let $V \in L^3 (\bR^3)$ be non-negative, compactly supported and spherically symmetric. Let $\cG_N$ be defined as in (\ref{eq:GN}), with $\ell \in (0;1/2)$ small enough. Let $E_N$ be the ground state energy of the Hamilton operator (\ref{eq:Hmom}). 
\begin{itemize}
\item[a)] We have 
\begin{equation}\label{eq:GDelta} \cG_N - E_N = \cH_N + \Delta_N \end{equation}
where the error term $\Delta_N$ is such that for every $\delta >0$, there exists $C > 0$ with 
\begin{equation}\label{eq:Delta-bd} \pm \Delta_N \leq \delta \cH_N + C (\cN_+ + 1) \end{equation}
Furthermore, for every $k \in \bN$ there exists a $C > 0$ such that 
\begin{equation}\label{eq:adkN}
\pm \text{ad}^{\, (k)}_{\, i\cN_+} (\cG_N) = \pm \text{ad}^{\, (k)}_{\, i\cN_+} (\Delta_N) = \pm \big[ i \cN_+ , \dots \big[ i \cN_+ , \Delta_N \big] \dots \big] \leq C (\cH_N + 1) \end{equation}
\item[b)] For $p \in \Lambda^*_+$, we use the notation, already introduced in (\ref{eq:ebe}), $\sigma_p = \sinh \eta_p$ and $\gamma_p = \cosh \eta_p$. Let 
\begin{equation}\label{eq:CN}
\begin{split} C_{\cG_N} =& \;\frac{(N-1)}{2} \widehat{V} (0) +\sum_{p\in\L^*_+}\Big[p^2\sigma_p^2+ \widehat{V}(p/N)\left(\s_p\g_p +\s_p^2\right)\Big]
\\ &+\frac{1}{2N} \sum_{p,q\in\Lambda_+^*}\widehat{V} ((p-q)/N)\sigma_q\gamma_q\sigma_p\gamma_p\, + \frac1N\sum_{p\in\Lambda^*}\Big[p^2 \eta_p^{2} + \frac{1}{2N}\big(\widehat{V} (\cdot/N)\ast \eta \big)_p \eta_p\Big]\;\\
        & - \frac 1 N \sum_{q \in \L^*} \hat V(q/N) \eta_q  \sum_{p \in \L_+^*} \s_p^2
\end{split} \end{equation}
For every $p \in \Lambda^*_+$, let   
\begin{equation}\label{eq:phi}
\begin{split}
  \Phi_p&=2p^2\s_p^2+\widehat{V}(p/N)\left(\g_p+\s_p\right)^2+\frac{2}{N}\g_p\s_p\sum_{q\in\L^*}\widehat{V}((p-q)/N) \eta_q\\
  &\quad-(\g_p^2+\s_p^2)\frac{1}{N}\sum_{q\in\L^*}\widehat{V}(q/N) \eta_q
\end{split}
\end{equation}
and
\begin{equation}\label{eq:gamma}
\begin{split} 
 \Gamma_p&=2p^2\s_p\g_p+\widehat{V}(p/N)(\g_p+\s_p)^2+(\g_p^2+\s_p^2)\frac{1}{N}\sum_{q\in\L^*}\widehat{V}((p-q)/N)\eta_q\\
 &\quad-2\g_p\s_p\frac{1}{N}\sum_{q\in\L^*}\widehat{V}(q/N)\eta_q
\end{split}
\end{equation}
Using $\Phi_{p},\G_{p}$ we construct the operator 
\begin{equation}\label{eq:QN}
 \cQ_{\cG_N} =\sum_{p\in\Lambda^*_+}\Phi_pb^*_pb_p+\frac{1}{2}\sum_{p\in\Lambda^*_+}\Gamma_p(b^*_pb^*_{-p}+b^*_pb^*_{-p})
\end{equation}
Moreover, we define
\be \label{eq:cCN}
\cC_N = \frac{1}{\sqrt N} \sum_{\substack{p,q \in \L^*_+ \\ q \neq -p}} \widehat V(p/N)\,\Big[ b^*_{p+q} b^*_{-p}  \big(  \g_q b_q + \s_q b^*_{-q} \big) + \hc \Big] 
\ee
Then, we have 
\be\label{eq:deco2-GN}
\cG_N = C_{\cG_N} + \cQ_{\cG_N} + \cH_N + \cC_N  + \cE_{\cG_N} 
\ee
with an error term $\cE_{\cG_N}$ satisfying, on $\cF_+^{\leq N}$, the bound   
\begin{equation}\label{eq:error}  \pm \cE_{\cG_N} \leq \frac{C}{\sqrt{N}} \,  (\cH_N+\cN_+^2+1)(\cN_++1)\end{equation}
\end{itemize}
\end{prop}

For the Hamilton operator (\ref{eq:ham-beta}) with parameter $\beta \in (0;1)$, a result similar to Proposition \ref{prop:gene} has been recently established in Theorem 3.2 of \cite{BBCS2}. The main 
difference between Prop. \ref{prop:gene} and previous results for $\beta < 1$ is the emergence, in (\ref{eq:deco2-GN}), of a cubic and a quartic term in the generalized creation and annihilation operators (the quartic term $\cV_N$ is included in the Hamiltonian $\cH_N$). As explained in the introduction, for $\beta < 1$, the cubic and the quartic parts of $\cG_N$ were negligible and could be absorbed in the error $\cE_{\cG_N}$. In the Gross-Pitaevskii regime, on the other hand, this is not possible. It is easy to find normalized $\xi \in \cF_+^{\leq N}$ with bounded expectation of $(\cN_+ + 1) ( \cH_N + \cN_+^2 + 1)$ such that $\langle \xi, \cC_N \xi \rangle$ and $\langle \xi, \cV_N \xi \rangle$ are of order one and do not tend to zero, as $N \to \infty$. 

To extract the important contributions that are still hidden in the cubic and in the quartic terms on the r.h.s. of (\ref{eq:deco2-GN}), we conjugate the renormalized excitation Hamiltonian $\cG_N$ with a unitary operator obtained by exponentiating a cubic expression in creation and annihilation operators. 

More precisely, we define the skew-symmetric operator $A:\cF_+^{\leq N}\to\cF_+^{\leq N}$ by   
\begin{equation}\label{eq:defA}
        \begin{split}
       A &= \frac1{\sqrt{N}} \sum_{\substack{ r\in P_H, v\in P_L }} \eta_r \big[ \sigma_v b^*_{r+v} b^*_{-r}b^*_{-v} +  \gamma_v b^*_{r+v} b^*_{-r}b_{v} - \text{h.c.} \big] \\
        &=: A_\sigma +A_\gamma - \text{h.c.}
        \end{split}
\end{equation}
where $P_L = \{\, p \in \L_+^* : |p|\leq N^{1/2} \}$ corresponds to low momenta and $P_H =\L_+^* \setminus P_L$ to high momenta (by definition $r + v \not = 0$). The coefficients $\eta_p$ are defined in (\ref{eq:ktdef}); they are the same as those used in the definition of the generalized Bogoliubov transformation $\exp (B(\eta))$ appearing in $\cG_N$. We then define the cubically renormalized excitation Hamiltonian
\begin{equation}\label{eq:cJN-def}
\cJ_N := e^{-A} e^{-B(\eta)} U_N H_N U_N^* e^{B(\eta)} e^{A} = e^{-A} \mathcal{G}_N e^{A} : \cF_+^{\leq N} \to \cF_+^{\leq N}.
\end{equation}
In the next proposition, we collect important properties of 
$\cJ_N $.  
\begin{prop}\label{thm:tGNo1}
 Let $V \in L^3 (\bR^3)$ be non-negative, compactly supported and spherically symmetric. Let $\cJ_N$ be defined as in (\ref{eq:cJN-def}). For $p \in \Lambda^*_+$, we use again the notation $\sigma_p = \sinh (\eta_p)$, $\gamma_p = \cosh (\eta_p)$ and we recall the notation $\widehat{f}_{\ell,N}$ from (\ref{eq:fellN}). Let 
\begin{equation}
        \begin{split}  \label{eq:tildeCN}
       C_{\cJ_N} :=&\;\frac{(N-1)}{2}  \widehat{V} (0) +\sum_{p\in\L^*_+}\Big[p^2\sigma_p^2+ \widehat{V}(p/N)\sigma_p\g_p + \big(\widehat{V} (\cdot/N)\ast\widehat{f}_{\ell,N}\big)_p\sigma_p^2\Big]\\
       &+ \frac{1}{2N}\sum_{p,q\in\Lambda_+^*}\widehat{V} ((p-q)/N)\sigma_q\gamma_q\sigma_p\gamma_p + \frac1N\sum_{p\in\Lambda^*}\Big[p^2 \eta_p^{\;2} + \frac{1}{2N}\big(\widehat{V} (\cdot/N)\ast\eta\big)_p \eta_p\Big]
           \end{split}
        \end{equation}
        Moreover, for every $p \in \L_+^*$ we define
        \begin{equation} \label{eq:defFpGp}
\begin{split}
F_p:=&\;p^2(\sigma_p^2+\gamma_p^2)+ \big(\widehat{V} (\cdot/N)\ast\widehat{f}_{\ell,N}\big)_p(\gamma_p+\sigma_p)^2;\\
G_p:=&\;2p^2\sigma_p\gamma_p+ \big(\widehat{V} (\cdot/N)\ast\widehat{f}_{\ell,N}\big)_p(\gamma_p+\sigma_p)^2
\end{split}
\end{equation}
With the coefficients $F_p$ and $G_p$, we construct the operator
        \begin{equation*}
\begin{split}
       \cQ_{\cJ_N} :=\;\sum_{p\in\L^*_+}\Big[F_pb^*_pb_p+\frac{1}{2}G_p\big(b^*_pb^*_{-p}1+b_pb_{-p}\big)\Big]
\end{split}
\end{equation*}
quadratic in the $b,b^*$-fields. Then, we have
\begin{equation*}%\label{eq:finaltGN}
        \begin{split}
       \cJ_N = C_{\cJ_N} + \cQ_{\cJ_N}  + \cV_N + \cE_{\cJ_N} 
        \end{split}
        \end{equation*}
for an error term $\cE_{\cJ_N}$ satisfying, on $\cF_+^{\leq N}$,
         \be \label{eq:tildeEN}
         \pm \cE_{\cJ_N}  \leq C N^{-1/4} \Big[(\cH_N+1)(\cN_++1) + (\cN_++1)^3\Big]\,.
         \ee
\end{prop}

The proof of Proposition \ref{prop:gene} is deferred to Section \ref{sec:GN}. Proposition \ref{thm:tGNo1} will then be proved in Section \ref{sec:wtGN}. In the next three sections, 
on the other hand, we show how to use these two propositions to complete the proof of 
Theorem \ref{thm:main}.

\section{Bounds on excitations vectors}
\label{sec:BEC}

To make use of the bounds (\ref{eq:error}) and (\ref{eq:tildeEN}), we need to prove that excitation vectors associated with many-body wave functions $\psi_N \in L^2_s (\Lambda^N)$ with small excitation energy, defined either as $e^{B(\eta)} U_N \, \psi_N$ (if we want to apply (\ref{eq:error})) or as $e^A e^{B(\eta)} U_N \, \psi_N$ (if we want to apply (\ref{eq:tildeEN})) have finite expectations of the operator $(\cH_N + 1) ( \cN_+ + 1) + (\cN_+ + 1)^3$. This is the goal of this section. 

We start with estimates on the excitation vector $\xi_N = e^{B(\eta)} U_N \, \psi_N$, that are relevant to bound errors arising before conjugation with the cubic exponential $\exp (A)$.
\begin{prop}\label{prop:hpN}
Let $V \in L^3 (\bR^3)$ be non-negative, compactly supported and spherically symmetric. Let $E_N$ be the ground state energy of the Hamiltonian $H_N$ defined in (\ref{eq:Hmom}) (or, equivalently, in (\ref{eq:Ham0})). Let $\psi_N \in L^2_s (\Lambda^N)$ with $\| \psi_N \| = 1$ belong to the spectral subspace of $H_N$ with energies below $E_N + \zeta$, for some $\zeta > 0$, i.e. 
\begin{equation}\label{eq:psi-space} \psi_N = {\bf 1}_{(-\infty ; E_N + \zeta]} (H_N) \psi_N 
\end{equation}
Let $\xi_N = e^{-B(\eta)} U_N \psi_N$ be the renormalized excitation vector associated with $\psi_N$. Then, for any $k\in\mathbb N$ there exists a constant $C > 0$ such that 
\begin{equation}\label{eq:hpN} \langle \xi_N, (\cN_+ +1)^k (\cH_N+1) \xi_N \rangle \leq C (1 + \zeta^{k+1}) \end{equation}
\end{prop}
{\it Remark:} As shown in \cite{BBCS3}, the bound $\langle \xi_N, \cN_+ \xi_N \rangle \leq C (\zeta +1)$ which follows from (\ref{eq:hpN}) taking $k =0$ (because $\cN_+ \leq C \cH_N$), immediately implies that normalized many-body wave functions $\psi_N \in L^2_s (\L^N)$ satisfying (\ref{eq:psi-space}) exhibit  complete Bose-Einstein condensation in the zero-momentum mode $\ph_0$ with optimal rate. In other words, it implies that the one-particle reduced density $\gamma^{(1)}_N = \tr_{2, \dots , N} |\psi_N \rangle \langle \psi_N|$ associated with $\psi_N$ is such that 
\begin{equation}\label{eq:BEC-opt} 1 - \langle \ph_0 , \gamma^{(1)}_N \ph_0 \rangle \leq \frac{C ( \zeta + 1)}{N} \end{equation}
\begin{proof}[Proof of Prop. \ref{prop:hpN}]
First of all, taking the vacuum expectation, (\ref{eq:GDelta}), (\ref{eq:Delta-bd}) imply that there exists a constant $C > 0$ such that $E_N \leq 4 \pi \frak{a}_0 N + C$ for all $N$ large enough. Let now $Q_\zeta = \text{ran } {\bf 1}_{(-\infty ; E_N + \zeta]} (\cG_{N})$. We claim that, for every $k \in \bN$, there exists a constant $D_k > 0$ with 
\begin{equation}\label{eq:claimk}
\sup_{\xi_N \in Q_\zeta\backslash \{0\}} \frac{\langle \xi_N , (\cN_+ + 1)^k (\cH_N + 1) \xi_N \rangle}{\| \xi_N \|^2}  \leq D_k (1+\zeta)^{k+1}   \end{equation}
for all $\zeta > 0$. Clearly, (\ref{eq:claimk}) implies (\ref{eq:hpN}). We prove (\ref{eq:claimk}) by induction over $k \in \bN$. Let us first consider $k=0$. To show (\ref{eq:claimk}) with $k=0$ we combine (\ref{eq:GDelta}), (\ref{eq:Delta-bd}) with the results of \cite{BBCS3}. In \cite{BBCS3}, we consider the excitation Hamiltonian $\cG_{N,\ell} = e^{-B(\eta_H)} U_N H_N U_N^* e^{B(\eta_H)}$, renormalized through a generalized Bogoliubov transformation with coefficients $\eta_H (p) = \eta_p \chi (|p| \geq \ell^{-\a})$ for all $p \in \L_+^*$, where $\eta_p$ is defined as in (\ref{eq:ktdef}) (and $\ell$ is, as in (\ref{eq:scatl}), the radius of the ball on which we solve the Neumann problem used to define $\eta$). Using \cite[Prop. 6.1]{BBCS3} and the observation $E_N \leq 4 \pi \frak{a}_0 N + C$, we conclude that for each $\a > 3$ and $\ell \in (0;1/2)$ small enough, there exist constants $c,C > 0$ with 
\[ \cG_{N,\ell} \geq E_N + c \cN_+ - C \]
Using Lemma \ref{lm:Ngrow} (and the fact that $\| \eta_H \| \leq \| \eta \| \leq C$, uniformly in $N$), we can translate this bound to an analogous estimate for the  excitation Hamiltonian $\cG_N$ defined in (\ref{eq:GN}). We obtain
\[ \begin{split} \cG_N &= e^{-B(\eta)} e^{B(\eta_H)} \cG_{N,\ell} e^{-B(\eta_H)} e^{B(\eta)} \geq E_N + c\,  \cN_+ - C \end{split} \]
with new constants $c,C> 0$.  Combining the last equation with (\ref{eq:GDelta}), (\ref{eq:Delta-bd}), we also find constants $c,C > 0$ such that  
\begin{equation}\label{eq:condH2} \cG_N - E_N \geq c \cH_N - C \end{equation}
for all $N$ large enough. Thus, for any $\xi_N \in Q_\zeta$ we have
\[ \langle \xi_N, (\cH_N +1) \xi_N \rangle \leq C \langle \xi_N, (\cG_N - E_N) \xi_N \rangle + C \| \xi_N \|^2 \leq C (1+\zeta) \| \xi_N \|^2 \]
which implies (\ref{eq:claimk}), for $k =0$. 

Let us now consider the induction step. We assume (\ref{eq:claimk}) holds true for a $k \in \bN$ and we show it for $k$ replaced by $(k+1)$. We use the shorthand notation $\cG'_{N} = \cG_{N} - E_N$. 
Let $\xi_N \in Q_\zeta$. {F}rom \eqref{eq:condH2}, we find 
        \begin{equation}\label{eq:bHG}\begin{split}
\langle \xi_N, (\cN_+ + 1)^{k+1} &(\cH_N + 1) \xi_N \rangle \\ = \; &   \langle \xi_N, (\cN_++1)^{(k+1)/2} (\cH_N+1) (\cN_+ + 1)^{(k+1)/2}\xi_N\rangle \\
\leq \; &C \langle \xi_N, (\cN_++1)^{(k+1)/2}(\cG'_{N}+C) (\cN_++1)^{(k+1)/2} \xi_N\rangle
\\ = \; &C\langle  \xi_N,(\cN_++1)^{(k+1)} (\cG'_{N} + C) \xi_N \rangle \\ &+C  \langle \xi_N, (\cN_++1)^{(k+1)/2}\Big[\cG'_{N}, (\cN_++1)^{(k+1)/2}\Big] \xi_N \rangle \\ =: \; &\text{I} + \text{II} 
        \end{split}\end{equation}
We can bound the first term by 
\begin{equation}\label{eq:I-apri} \begin{split} 
\text{I} \leq \; &C \langle \xi_N , (\cN_+ + 1)^{k+1} \xi_N \rangle^{1/2} \langle (\cG'_{N} + C) \xi_N, (\cN_+ + 1)^{k+1} (\cG'_{N} + C) \xi_N \rangle^{1/2} \\ \leq \; &C \langle \xi_N , (\cN_+ + 1)^{k} (\cH_N + 1) \xi_N \rangle^{1/2} \\ &\hspace{3cm} \times \langle (\cG'_{N} + C) \xi_N, (\cN_+ + 1)^{k} (\cH_N + 1) (\cG'_{N} + C) \xi_N \rangle^{1/2} \\ \leq \; &C \| \xi_N \| \| (\cG'_{N} +C)\xi_N \|  \sup_{\wt\xi_N \in Q_\zeta\backslash \{0\}} \frac{\langle \wt\xi_N, (\cN_+ + 1)^k (\cH_N + 1) \wt\xi_N \rangle}{\| \wt\xi_N \|^2} \\ \leq \; &C (1+\zeta) \| \xi_N \|^2  \sup_{\wt\xi_N \in Q_\zeta} \frac{\langle \wt\xi_N, (\cN_+ + 1)^k (\cH_N + 1) \wt\xi_N \rangle}{\| \wt\xi_N \|^2} \\ \leq \; &C  D_k  (1+\zeta)^{k+2} \| \xi_N \|^2 \end{split} \end{equation}
using the induction assumption and the fact that $(\cG'_{N}+C) \xi_N \in Q_\zeta$ with $\| (\cG'_{N}+C) \xi_N \| \leq (\zeta +C) \| \xi_N \|$ for all $\xi_N \in Q_\zeta$. To bound the second term on the r.h.s. of (\ref{eq:bHG}), we use 
        \[ \frac{1}{\sqrt{z}} = \frac{1}{\pi} \int_0^\infty   \frac{dt}{\sqrt{t}} \, \frac{1}{t+z}   \]
to write 
\begin{equation*}
\begin{split} 
\text{II} &= \frac{1}{\pi} \int_0^\infty  dt \, \sqrt{t} \;  \Big\langle \xi_N, \frac{(\cN_++1)^{(k+1)/2}}{t+(\cN_+ +1)^{k+1}} \, \Big[\cG'_{N}, (\cN_++1)^{k+1}\Big] \, \frac{1}{t+(\cN_+ +1)^{k+1}} \xi_N \Big\rangle    
        \end{split}\end{equation*} 
With the identity 
\[ \Big[\cG'_{N}, (\cN_++1)^{k+1}\Big]  = - \sum_{j = 1}^{k+1} {k+1 \choose j} \text{ad}^{\, (j)}_{\cN_+} (\cG'_{N}) \, (\cN_+ + 1)^{k+1-j} \]
which can be proven by induction over $k$, we obtain 
         \[\begin{split} \text{II} & = \frac{1}{\pi} \sum_{j=1}^{k+1} (-i)^{j+1} {k+1 \choose j}  \\ &\hspace{2cm} \times \int_0^\infty  dt \, \sqrt{t} \; \Big \langle \xi_N, \frac{(\cN_++1)^{(k+1)/2}}{t+(\cN_+ +1)^{k+1}} \, \operatorname{ad}_{i \cN_+}^{\, (j)} (\cG'_{N}) \, \frac{(\cN_++1)^{k+1-j}}{t+(\cN_+ +1)^{k+1}} \xi_N \Big\rangle 
        \end{split}\]
From (\ref{eq:adkN}) in Prop. \ref{prop:gene} we know that $ \cA_j: = (\cH_N+1)^{-1/2}\operatorname{ad}_{i\cN_+ }^{(j)}(\cG'_{N})(\cH_N+1)^{-1/2} $ is a self-adjoint operator on $\cF_+^{\leq N}$, with norm bounded uniformly in $N$. Hence, we find 
        \[\begin{split} |\text{II}| \leq \; &C \sum_{j=1}^{k+1}\int_0^\infty  \frac{dt}{(t+1)^{1/2}} \,\Big\|(\cH_N+1)^{1/2}(\cN_++1)^{(k+1)/2} \xi_N\Big\| \\ &\hspace{6cm} \times \Big\|(\cH_N+1)^{1/2} \frac{(\cN_++1)^{k+1-j}}{t+(\cN_+ +1)^{k+1}} \xi_N\Big\| \\
        \leq \; &C_k  \int_0^\infty  \frac{dt}{(t+1)^{\frac{2k+3}{2k+2}} } \, \Big\|(\cN_++1)^{(k+1)/2}(\cH_N+1)^{1/2} \xi_N\Big\| \Big\|(\cN_++1)^{k/2}(\cH_N+1)^{1/2} \xi_N \Big\|  \\ 
        \leq \; &C_k \Big\|(\cN_++1)^{(k+1)/2}(\cH_N+1)^{1/2} \xi_N\Big\| \Big\|(\cN_++1)^{k/2}(\cH_N+1)^{1/2} \xi_N \Big\| 
        \end{split}\]
With the induction assumption we conclude that 
        \begin{equation}\label{eq:bt2}
        \begin{split}
       |\text{II}| \leq \;& \frac{1}{2} \langle \xi_N, (\cN_++1)^{k+1}(\cH_N+1)\xi_N\rangle + C_k D_k (1+\zeta)^{k+1} \| \xi_N \|^2 
        \end{split}
        \end{equation}
with a constant $C_k$ (different from line to line), depending on $k$. Inserting (\ref{eq:I-apri}) and the last bound in (\ref{eq:bHG}), we obtain
\[ \langle \xi_N , (\cN_+ + 1)^{k+1} (\cH_N + 1) \xi_N \rangle \leq D_{k+1} (1+\zeta)^{k+2} \| \xi_N \|^2 \]
for all $\xi_N \in Q_\zeta$ and for an appropriate constant $D_{k+1}$ (one can take $D_{k+1} = 2 (C+C_k) D_k$, if $C$ and $C_k$ are as on the r.h.s. of (\ref{eq:I-apri}) and (\ref{eq:bt2})). We conclude that 
\[ \sup_{\xi_N \in Q_\zeta\backslash \{0\}} \frac{\langle \xi_N , (\cN_+ + 1)^{k+1} (\cH_N + 1) \xi_N \rangle}{\| \xi_N \|^2} \leq D_{k+1} (1+\zeta)^{k+2} \]
This completes the proof of (\ref{eq:claimk}).
\end{proof}

Next, we control the growth of powers of $\cN_+$ and of the product $(\cN_+ + 1)(\cH_N+1)$ under conjugation with the operator $\exp (A)$. These bounds are needed to apply Prop.~\ref{thm:tGNo1}. 
First, we focus on the growth of powers of the number of particles operator.
\begin{prop}\label{prop:cN} 
Suppose that $A$ is defined as in (\ref{eq:defA}). For any $k\in \bN$, there exists $C >0$ such that, on $\cF_+^{\leq N}$, we have the operator inequality  
\[ e^{-A} (\cN_++1)^k e^{A} \leq C (\cN_+ +1)^k   \]
\end{prop}
\begin{proof} Let $\xi\in\cF_+^{\leq N}$ and define $\varphi_{\xi}:\mathbb R\to \mathbb R$ by 
        \[\varphi_{\xi}(s):= \langle \xi, e^{-sA} (\cN_+ + 1)^k e^{sA} \xi \rangle \]
Then we have, using the decomposition $A = A_\s + A_\g - \text{h.c.}$ from (\ref{eq:defA}),   
        \[ \partial_s\varphi_{\xi}(s) = 2 \text{Re } \langle \xi, e^{-sA} \big[(\cN_+ + 1)^k, A_\sigma \big]  e^{sA} \xi \rangle + 2 \text{Re } \langle \xi, e^{-sA} \big[(\cN_+ + 1)^k, A_\gamma \big]  e^{sA} \xi \rangle \]
We start by controlling the commutator with $A_\sigma$. We find 
\[\begin{split}
\langle \xi, e^{-sA} &\big[(\cN_+ + 1)^k, A_\sigma \big]  e^{sA} \xi \rangle \\ &= \frac{1}{\sqrt{N}} \sum_{ r\in P_H, v\in P_L } \eta_r \sigma_v \langle e^{sA} \xi , b^*_{r+v} b^*_{-r}b^*_{-v} \big[(\cN_+ + 4)^k - (\cN_+ + 1)^k \big]  e^{sA} \xi  \rangle
        \end{split}\]
With the mean value theorem, we find a function $\theta:\mathbb N\to (0;3)$ such that 
        \[ (\cN_+ + 4)^k - (\cN_+ + 1)^k= k (\cN_+ +  \theta(\cN_+) +1)^{k-1} \]
Since $b_p \cN_+ =(\cN_+ + 1) b_p$ and $b_p^* \cN_+ = (\cN_+ - 1) b_p^*$, we obtain, using Cauchy-Schwarz and the boundedness of $\theta$,
        \begin{equation}\label{eq:N+k-As} \begin{split} 
        &\Big| \langle \xi, e^{-sA} \big[(\cN_+ + 1)^k, A_\sigma \big]  e^{sA} \xi \rangle  \Big|\\
        &\hspace{0.2cm} \leq   \frac{C}{\sqrt{N}}  \sum_{ r\in P_H, v\in P_L } |\eta_r| |\sigma_v| \big\|(\cN_+ +1)^{-3/4 +(k-1)/2}b_{r+v} b_{-r}b_{-v} e^{sA} \xi \big\| \\ &\hspace{8cm} \times \big\|  (\cN_+ +1)^{3/4 +(k-1)/2} e^{sA} \xi \big\| \\
        &\hspace{0.2cm} \leq \frac{C}{\sqrt{N}} \|\eta\|_2\|\sigma\|_2\big\|  (\cN_+ +1)^{3/4 +(k-1)/2} e^{sA} \xi \big\|^2 \\   &\hspace{0.2cm}\leq   \frac{C}{\sqrt{N}} \langle e^{sA} \xi, (\cN_+ +1)^{ k +1/2 } e^{sA} \xi \rangle\\
        &\hspace{0.2cm} \leq C \langle e^{sA} \xi, (\cN_+ +1)^k e^{sA} \rangle 
        \end{split}\end{equation}
        for a constant $C>0$ depending on $k$.  
        Similarly, the commutator with $A_\gamma$ is bounded by 
\begin{equation}\label{eq:N+k-Ag}\begin{split}
        &\Big| \langle \xi, e^{-sA} \big[(\cN_+ + 1)^k, A_\gamma \big]  e^{sA} \xi \rangle \Big|\\
        &\hspace{0.2cm} \leq  \frac{C}{\sqrt{N}} \sum_{ r\in P_H, v\in P_L } |\eta_r|  \big\|(\cN_+ +1)^{-1/4 +(k-1)/2}b_{r+v} b_{-r} e^{sA} \xi \big\| \\ &\hspace{7cm} \times \big\|  (\cN_+ +1)^{1/4 +(k-1)/2}b_{-v} e^{sA} \xi \big\| \\
        &\hspace{0.2cm} \leq  \frac{C}{\sqrt{N}} \|\eta\|_2 \big\|  (\cN_+ +1)^{3/4 +(k-1)/2} e^{sA} \xi \big\|^2 \\ &\hspace{0.2cm} \leq C  \langle \xi,e^{-sA} (\cN_+ +1)^k e^{sA}\xi \rangle
        \end{split}\end{equation}
This proves that \[ \partial_s\varphi_{\xi}(s) \leq C \varphi_{\xi}(s) \] so that, by Gronwall's lemma, we find a constant $C$ (depending on $k$) with
        \[\langle \xi, e^{-A} (\cN_+ + 1)^k e^{A}   \xi \rangle = C \langle \xi, (\cN_+ + 1)^k  \xi \rangle \, . \]
       \end{proof}

To control the growth of the product $(\cH_N +1)(\cN_++1)$ with respect to conjugation by $e^A$, we will use the following lemma.
\begin{lemma}\label{lm:commHA} Let $V \in L^3 (\bR^3)$ be non-negative, compactly supported and spherically symmetric. Let  $A$ and $\cH_N$ be defined as in \eqref{eq:defA} and, respectively, after (\ref{eq:KcVN}). Then,  
\be \label{eq:commHA}
[\cH_N, A] =  \sum_{j=0}^9 \Theta_j  + \emph{h.c.} 
\end{equation}
where
\[ \begin{split} \Theta_0 &= \Theta^{(1)}_{0} + \Theta^{(2)}_{0} = 
-\frac{1}{\sqrt{N}} \sum_{\substack{r \in P_H,\, v \in P_L }}   \widehat{V}(r/N)\, b^*_{r+v} b^*_{-r} \big( \g_v b_v + \s_v b^*_{-v} \big)  \\
\Theta_1 &= \Theta_1^{(1)} + \Theta_1^{(2)} = \frac{2}{\sqrt{N}} \sum_{r \in P_H, v \in P_L} \eta_r b_{r+v}^* b_{-r}^* \, \left[ r \cdot v \, \gamma_v \, b_v + (v^2 +r \cdot v) \, \sigma_v b^*_{-v}  \right] \\
\Theta_2 &= \Theta_2^{(1)} + \Theta_2^{(2)} =\frac{1}{N^{3/2}} \sum_{\substack{r \in P_H \\ v \in P_L}} 
\sum_{\substack{q \in \Lambda^*_+, u \in \Lambda^*: \\ u \not= -q, -r-v}} \widehat{V} (u/N) \eta_r b^*_{r+v+u} b^*_{-r} a_q^* a_{q+u} (\gamma_v b_v + \sigma_v b_{-v}^*)  \\
\Theta_3 &= \Theta_3^{(1)} + \Theta_3^{(2)} = \frac{1}{N^{3/2}} \sum_{\substack{r \in P_H \\ v \in P_L}} \sum_{\substack{q \in \Lambda^*_+ , u \in \Lambda^* :\\ u \not = -q ,r}} \widehat{V} (u/N) \eta_r b^*_{r+v} b^*_{-r+u} a^*_q a_{q+u} (\gamma_v b_v + \sigma_v b_{-v}^* ) \\
\Theta_4 & = \Theta_4^{(1)} + \Theta_4^{(2)} =  \frac{1}{N^{3/2}} \sum_{\substack{r \in P_H \\ v \in P_L}} \sum_{\substack{q \in \Lambda^*_+ , u \in \Lambda^* :\\ u \not = -q ,r}} \widehat{V} (u/N) \eta_r  b_{r+v}^* b_{-r}^* \\ & \hspace{7cm} \times \left( - \gamma_v a_q^* a_{q+u} b_{-u +v} + \sigma_v  b^*_{-v+u} a_q^* 
a_{q+u}\right)  \\
\end{split} \]
and 
\[ \begin{split} 
\Theta_5 &= \Theta_5^{(1)} + \Theta_5^{(2)} =- \frac{1}{N^{3/2}} \sum_{p \in P_H, v \in P_L} \sum_{r \in P_L}  \widehat{V} ((p-r)/N) \eta_r b^*_{p+v} b^*_{-p} (\gamma_v b_v + \sigma_v b_{-v}^*) \\ 
\Theta_6 &= \Theta_6^{(1)} + \Theta_6^{(2)} =- \frac{1}{N^{3/2}} \sum_{p \in P_H, v \in P_L}   \widehat{V} (p/N) \, \eta_0 \, b^*_{p+v} b^*_{-p} (\gamma_v b_v + \sigma_v b_{-v}^*) \\ 
\Theta_7 &= \Theta_7^{(1)} + \Theta_7^{(2)} =\frac{1}{N^{3/2}} \sum_{r \in P_H , v \in  P_L} \sum_{p \in P_L : p \not = -v}  \widehat{V} ((p-r)/N) \eta_r b^*_{p+v} b^*_{-p} (\gamma_v b_v + \sigma_v b_{-v}^*) \\
\Theta_8 &=\Theta_8^{(1)} + \Theta_8^{(2)} = 2N^2 \sqrt{N}  \l_{\ell} \sum_{r \in P_H , v \in P_L }  \widehat{\chi}_\ell(r)\, b^*_{r+v} b^*_{-r} \big( \g_v b_v + \s_v b^*_{-v} \big)  \\
\Theta_9 &=\Theta_9^{(1)} + \Theta_9^{(2)} = 2N \sqrt{N} \l_{\ell} \sum_{\substack{r \in P_H,\\ v \in P_L}}  \sum_{q \in \L^*} \widehat \chi_\ell(r-q) \eta_q\, b^*_{r+v} b^*_{-r} \big( \g_v b_v + \s_v b^*_{-v} \big)
\end{split} \]
We have
\begin{equation}\label{eq:EH1} |\langle \xi_1, \Theta^{(i)}_j \xi_2 \rangle| \leq C  \Big[\; \bmedia{\xi_1, \big( \cH_N + (\cN_++1)^2 \big)\xi_1 } +  \bmedia{\xi_2, \big( \cH_N + (\cN_++1)^2 \big)\xi_2 }\; \Big]  \end{equation}
for a constant $C > 0$, all $\xi_1, \xi_2 \in \cF_+^{\leq N}$, $i=1,2$ and all $j=0,1, \dots , 9$, and 
\begin{equation}\label{eq:EH2} \pm \big( \Theta^{(i)}_j + \text{h.c.} \big) \leq C N^{-1/4} \Big[  
(\cN_++1) (\cK+1)+   (\cN_++1)^{3} \Big]  \end{equation}
for $i=1,2$ and all $j=1,\dots , 9$ (but not for $j=0$). 
\end{lemma}

\begin{proof} 
We use the formulas 
\begin{equation}\label{eq:form-co} [ a_p^* a_q , b_r^* ] = \delta_{qr} b_p^*, \qquad [ a_p^* a_q , b_r ] = - \delta_{pr} b_q \end{equation}
to compute 
\begin{equation}\label{eq:cKA-comm} \begin{split} \big[ \cK , A \big] = \; &\frac{1}{\sqrt{N}} \sum_{p \in \Lambda_+^*} \sum_{r \in P_H, v \in P_L} \eta_r p^2  \left\{ \delta_{p, r+v} b_p^* b_{-r}^* (\gamma_v b_v + \sigma_v b_{-v}^*)  \right. \\ &\hspace{.2cm}  \left.  + \delta_{p,-r} b_{r+v}^* b_p^* (\gamma_v b_v + \sigma_v b_{-v}^*) -  \gamma_v \delta_{v,p} b_{r+v}^* b_{-r}^* b_p + \sigma_v \delta_{-v,p} b_{r+v}^* b_{-r}^* b_p^* \right\} + \text{h.c.}  \\
= \; & \frac{1}{\sqrt{N}} \sum_{r \in P_H, v \in P_L} 2r^2 \eta_r b_{r+v}^* b_{-r}^* (\gamma_v b_v + \sigma_v b_{-v}^*)  + \Theta_1 + \text{h.c.}  
\end{split} \end{equation}

Writing 
\begin{equation}\label{eq:quar-de} a_{p+u}^* a_q^* a_{q+u} a_p = a_{p+u}^* a_p a_q^* a_{q+u} - \delta_{p,q} a_{p+u}^* a_{p+u} \, , \end{equation}
using (\ref{eq:form-co}) to commute the r.h.s. of (\ref{eq:quar-de}) with $b_{r+v}^*, b_{-r}^*, b_v$ and, respectively, with $b_{-v}^*$, and normal ordering the operators appearing to the left of the factor $(\gamma_v b_v + \sigma_v b_{-v}^*)$ leads to 
\begin{equation}\label{eq:cVA-comm}
\begin{split}  \big[ \cV_N , A \big] =\; & \frac{1}{N^{3/2}} \sum_{\substack{r \in P_H, v \in P_L , u \in \Lambda^* :\\ u \not = -r, -r-v}} \widehat{V} (u/N) \eta_r b^*_{r+v+u} b^*_{-r-u} (\gamma_v b_v + \sigma_v b_{-v}^*)  \\ &+ \sum_{j=2}^4 \Theta_j + \text{h.c.} \end{split} \end{equation}
The first term on the r.h.s. of the last equation can be further decomposed as
\begin{equation}\label{eq:cV11A-comm} \begin{split} 
\frac{1}{N^{3/2}} &\sum_{\substack{r \in P_H, v \in P_L , u \in \Lambda^* :\\ u \not = -r, -r-v}} \widehat{V} (u/N) \eta_r b^*_{r+v+u} b^*_{-r-u} (\gamma_v b_v + \sigma_v b_{-v}^*) \\
= \; &\frac{1}{N^{3/2}} \sum_{r \in P_H, v \in  P_L} \sum_{p \in \Lambda^*_+ : p \not = -v} \widehat{V} ((p-r)/N) \eta_r b_{p+v}^* b_{-p}^* (\gamma_v b_v + \sigma_v b_{-v}^*) \\ = \; &\frac{1}{N^{3/2}} \sum_{r \in P_H , v \in  P_L} \sum_{p \in P_H}  \widehat{V} ((p-r)/N) \eta_r b^*_{p+v} b^*_{-p} (\gamma_v b_v + \sigma_v b_{-v}^*) \\ &+ \frac{1}{N^{3/2}} \sum_{r \in P_H , v \in  P_L} \sum_{p \in P_L : p \not = -v}  \widehat{V} ((p-r)/N) \eta_r b^*_{p+v} b^*_{-p} (\gamma_v b_v + \sigma_v b_{-v}^*) \\
=   \; &\frac{1}{N^{3/2}} \sum_{p \in P_H, v \in P_L} \sum_{r \in \Lambda^*}  \widehat{V} ((p-r)/N) \eta_r b^*_{p+v} b^*_{-p} (\gamma_v b_v + \sigma_v b_{-v}^*)  + \sum_{j=5}^7 \Theta_j 
\end{split} \end{equation}
The first term on the r.h.s. of (\ref{eq:cV11A-comm}) can be combined with the first term on the r.h.s. of (\ref{eq:cKA-comm}); with the relation (\ref{eq:eta-scat}), we obtain 
\[ \begin{split} \frac{1}{\sqrt{N}} &\sum_{r \in P_H, v \in P_L} 2r^2 \eta_r b_{r+v}^* b_{-r}^* (\gamma_v b_v + \sigma_v b_{-v}^*) \\ &+ \frac{1}{N^{3/2}} \sum_{p \in P_H, v \in P_L} \sum_{r \in \Lambda^*}  \widehat{V} ((p-r)/N) \eta_r b^*_{p+v} b^*_{-p} (\gamma_v b_v + \sigma_v b_{-v}^*) = \Theta_0 + \Theta_8 + \Theta_9 \end{split} \]
Combining (\ref{eq:cKA-comm}), (\ref{eq:cVA-comm}) and (\ref{eq:cV11A-comm}) with the last equation, we obtain the decomposition (\ref{eq:commHA}). Now, we prove the bounds (\ref{eq:EH1}), (\ref{eq:EH2}). 
First of all, using (\ref{eq:etap}), we observe that 
\[ \begin{split} | \langle \xi_1,  \Theta_1^{(1)} \, \xi_2 \rangle | \leq \; 
&\frac{2}{\sqrt{N}} \sum_{r \in P_H, v \in P_L} |\eta_r| |r| |v|  \| b_{-r} b_{r+v} \xi_1 \| \| b_v \xi_2 \| 
\\ \leq \; & C N^{-1/2} \| (\cK+1)^{1/2} (\cN_+ + 1)^{1/2} \xi_1 \| \| (\cK + 1)^{1/2} \xi_2 \| 
\end{split} \]
The term $\Theta_1^{(2)}$ can be estimated similarly as 
\[ \begin{split} | \langle \xi_1,  \Theta_1^{(2)} \, \xi_2 \rangle | \leq \; 
&\frac{2}{\sqrt{N}} \sum_{r \in P_H, v \in P_L}|\eta_r| |\sigma_v||v| |r+v|  \| b_{-r} b_{r+v}b_v (\cN_++1)^{-1}\xi_1 \| \| (\cN_++1)  \xi_2 \| 
\\ \leq \; & C N^{-1/2}\bigg(\sum_{r \in P_H, v \in P_L}|\eta_r|^2 |\sigma_v|^2|v|^2 \bigg)^{1/2} \| (\cK+1)^{1/2}  \xi_1 \| \| (\cN+ 1)\xi_2 \| \\
 \leq \; & C N^{-1/2} \| (\cK+1)^{1/2}  \xi_1 \| \| (\cN + 1)\xi_2 \|
\end{split} \]
This implies, on the one hand, that
\begin{equation*}%\label{eq:wtK1} 
| \langle \xi_1,\big( \Theta_1^{(1)}+\Theta_1^{(2)} \big) \xi_2 \rangle | \leq C \left[ \langle \xi_1 , (\cK + 1) \xi_1 \rangle + \langle \xi_2, (\cK + 1) \xi_2 \rangle \right] \end{equation*}
and, on the other hand, taking $\xi_1 = \xi_2$, that 
\begin{equation*}%\label{eq:wtK2} 
\pm \left( \Theta^{(1)}_1 +\Theta_1^{(2)}+ \text{h.c.} \right) \leq C N^{-1/2} (\cK + 1) (\cN_+ + 1)  \end{equation*}

Next, we consider the quintic terms $\Theta_2, \Theta_3, \Theta_4$. Switching to position space, we find
\begin{equation}\label{eq:wtV-bd} \langle \xi_1, \Theta_2^{(i)} \xi_2 \rangle  =  \int dx dy \, N^{3/2} V(N(x-y)) \langle \xi_1, \check{b}_x^* b^* (\check{\eta}_{H,x}) \check{a}_y^* \check{a}_y  b^{\sharp_i} (\check{\mu}_{L,x}) \xi_2 \rangle \end{equation}
where $\check{\eta}_{H,x} (z) = \check{\eta}_H (z-x)$ with $\check{\eta}_H$ being the function with Fourier coefficients $\eta_H (p) = \eta_p \chi (p \in P_H)$ and where $\mu = \gamma$ and $\sharp_i=\cdot$, if $i=1$, and $\mu = \sigma$ and $\sharp_i=*$ if $i=2$, with $\check{\gamma}_L$, $\check{\sigma}_L$ defined similarly as $\check{\eta}_H$ (but in this case, with the characteristic function of the set $P_L$). {F}rom (\ref{eq:wtV-bd}), and using that, by definition of the sets $P_H, P_L$, $\| \eta_H \|_2 \leq C N^{-1/4}$, $\| \gamma_L \|_2 \leq C N^{3/4}$, $\| \sigma_L \|_2 \leq \| \sigma \|_2 \leq C$, we obtain that 
\begin{equation*}%\label{eq:wtV-bd2} 
\begin{split} 
|\langle \xi_1, \Theta^{(i)}_2 \xi_2 \rangle| &\leq C \int dx dy N^2 V(N(x-y)) \| \check{a}_x \check{a}_y \xi_1 \| \| \check{a}_y (\cN_+ + 1) \xi_2 \| 
\\ &\leq C \delta \langle \xi_1, \cV_N \xi_1 \rangle + C \delta^{-1} N^{-1} \langle \xi_2 , (\cN_+ + 1)^3 \xi_2 \rangle \end{split} \end{equation*}
for all $\delta > 0$ and for $i=1,2$. Choosing $\delta = 1$ and $\delta = N^{-1/2}$ we obtain (\ref{eq:EH1}) and, respectively, (\ref{eq:EH2}), with $j=2$ and $i=1,2$. The bounds (\ref{eq:EH1}), (\ref{eq:EH2}) for $j=3,4$ can be proven analogously. 

As for the terms $\Theta_5, \Theta_6, \Theta_7$, we can proceed as follows:
\[ \begin{split} 
|\langle \xi_1, \Theta_5^{(1)} \, \xi_2 \rangle | \leq \; &\frac{1}{N^{3/2}}  \sum_{p \in P_H, r,v \in P_L} |\widehat{V} ((p-r)/N)| |\eta_r| \| b_{p+v} b_{-p} \xi_1 \|  \| b_v \xi_2 \| \\ 
\leq \; &\frac{1}{N^{3/2}} 
\Big( \sum_{p \in P_H , r,v \in P_L} |\eta_r| \, p^2 \| b_{-p} b_{p+v} \xi_1 \|^2 \Big)^{1/2}
\\ &\hspace{1cm} \times  \Big( \sum_{p \in P_H , r,v \in P_L} 
\frac{|\widehat{V} ((p-r)/N)|^2 |\eta_r|}{p^2}  \| b_v \xi_2 \|^2  \Big)^{1/2} \\
\leq \; & \frac{1}{\sqrt{N}} \| (\cK+1)^{1/2} (\cN_+ + 1)^{1/2} \xi_1 \| \,  \| (\cN_+ + 1)^{1/2} \xi_2 \|   
\end{split} \]
which immediately implies (\ref{eq:EH1}), (\ref{eq:EH2}) for $j=5$ and $i=1$. The contribution $\Theta_5^{(2)}$ can be bounded analogously, replacing $\|b_v \xi_2\|$ by $|\sigma_v| \| (\cN_+ + 1)^{1/2} \xi_2 \|$. 
%Also the terms $\Theta_j^{(i)}$ with $j=6,7$ can be bounded similarly. 
The term $\Theta_6^{(i)}$  can be bounded similarly. As for $\Theta_7^{(1)}$ (a similar bound holds for $\Theta_7^{(2)}$) we find:
\[ \begin{split} 
|\langle \xi_1, \Theta_7^{(1)} \, \xi_2 \rangle | 
%\leq \; &\frac{1}{N^{3/2}}  \sum_{r \in P_H, p, v \in P_L} |\widehat{V} ((p-r)/N)| |\eta_r| \| b_{p+v} b_{-p} \xi_1 \|  \| b_v \xi_2 \| \\ 
\leq \; &\frac{1}{N^{3/2}} 
\Big( \sum_{r \in P_H , p,v \in P_L} |\widehat{V} ((p-r)/N)||\eta_r| \, p^2 \| b_{-p} b_{p+v} \xi_1 \|^2 \Big)^{1/2}
\\ &\hspace{1cm} \times  \Big( \sum_{r \in P_H , p,v \in P_L} 
\frac{|\widehat{V} ((p-r)/N)| |\eta_r| }{p^2}  \| b_v \xi_2 \|^2  \Big)^{1/2} \\
\leq \; & N^{-1/4} \| (\cK+1)^{1/2} (\cN_+ + 1)^{1/2} \xi_1 \| \,  \| (\cN_+ + 1)^{1/2} \xi_2 \|   
\end{split} \]

Finally, let us consider the terms $\Theta_8, \Theta_9$. Since $\| \widehat{\chi}_\ell  \|_2 \leq C$ (for a constant $C$ depending only on $\ell$), we have 
\[ \begin{split} |\langle \xi_1, \Theta_8^{(1)} \xi_2 \rangle| \leq \; &\frac{1}{\sqrt{N}} \sum_{r \in P_H, v \in P_L} |\widehat{\chi}_\ell (r)| \| b_{r+v} b_{-r} \xi_1 \| \| b_v \xi_2 \|  \\ \leq \; &\frac{1}{\sqrt{N}} \| (\cN_+ + 1) \xi_1 \| \| (\cN_+ + 1)^{1/2} \xi_2 \| 
\end{split} \]
which implies (\ref{eq:EH1}) and (\ref{eq:EH2}) for $j=8$ and $i=1$. The bounds for $j=8$ and $i=2$ follow as usual replacing $\| b_v \xi_2 \|$ by $|\sigma_v| \| (\cN_+ + 1)^{1/2} \xi_2 \|$, and using the boundedness of $\| \sigma \|_2$. Also the estimates for $j=9$ can be proven analogously, since also $\| \widehat{\chi}_\ell * \eta \|_2 = \| \chi_\ell \check{\eta} \|_2 \leq \| \check{\eta} \|_2 = \| \eta \|_2$ is finite, uniformly in $N$.  

To conclude the proof of the lemma, we still have to show that $\Theta^{(i)}_0$ satisfies (\ref{eq:EH1}), for $i=1,2$. To this end, we observe that 
\[ \begin{split} |\langle \xi_1, \Theta_0^{(1)} \xi_2 \rangle | &\leq \frac{1}{\sqrt{N}} \sum_{r \in P_H, v \in P_L} |\widehat{V} (r/N)| \| b_{r+v} b_{-r} \xi_1 \| \| b_v \xi_2 \| \\ &\leq C \left[ \sum_{ r \in P_H, v \in P_L} r^2 \| b_{r+v} b_{-r} \xi_1 \|^2 \right]^{1/2} \left[ \frac{1}{N} \sum_{r \in P_H, v \in P_L} \frac{|\widehat{V} (r/N)|^2}{r^2} \| b_v \xi_2 \|^2 \right]^{1/2} \\ &\leq \| (\cK + 1)^{1/2} (\cN_+ + 1)^{1/2} \xi_1 \| \| (\cN_+ + 1)^{1/2} \xi_2 \| \end{split} \]
and that a similar estimate holds for $\Theta_0^{(2)}$. Here, we used the fact that 
\[ \frac{1}{N} \sum_{r \in P_H} \frac{|\widehat{V} (r/N)|^2}{r^2} \leq \frac{1}{N} \sum_{r \in \Lambda^*_+} \frac{|\widehat{V} (r/N)|^2}{r^2} \leq C \]
uniformly in $N$.
\end{proof}

With the bounds on the commutator $[\cH_N , A]$ established in Lemma \ref{lm:commHA}, we can now control the growth of $(\cH_N + 1) ( \cN_+ + 1)$ under the action of the $e^A$. 
\begin{prop}\label{prop:cNcH} Let $V \in L^3 (\bR^3)$ be non-negative, compactly supported and spherically symmetric. Let $A$ and $\cH_N$ be defined as in (\ref{eq:defA}) and, respectively, after (\ref{eq:KcVN}). Then there exists a constant $C>0$ such that for all $s\in[0;1]$ we have on $\cF_+^{\leq N}$ the operator inequality 
\[ e^{-sA} (\cN_++1)(\cH_N +1) e^{sA} \leq C  (\cN_++1)(\cH_N +1) +C (\cN_++1)^{3}  \]
\end{prop}

\begin{proof}
For a fixed $\xi \in \cF_+^{\leq N}$ we define $ \varphi_\xi:\mathbb R\to\mathbb R$ through 
        \[\varphi_\xi(s):= \langle \xi, e^{-sA} (\cN_++1)(\cH_N +1) e^{sA}\xi\rangle \]
Then, we have 
        \begin{equation}\label{eq:P1P2}\begin{split}\partial_s \varphi_\xi (s) &= \langle \xi, e^{-sA} \big[(\cN_++1)(\cH_N +1),A\big] e^{sA}\xi\rangle \\
        &=\langle \xi, e^{-sA}(\cN_++1) \big[\cH_N,A\big] e^{sA}\xi\rangle+\langle \xi, e^{-sA} \big[\cN_+,A\big](\cH_N +1) e^{sA}\xi\rangle \\ &=: \text{P}_1 +\text{P}_2
        \end{split}\end{equation}
        
We start by analysing $\text{P}_1$. {F}rom Lemma \ref{lm:commHA}, we have
\[ \begin{split} \text{P}_1 &= \sum_{j=0}^9 \sum_{i=1}^2 
\langle e^{sA} \xi, (\cN_+ + 1) \Theta_j^{(i)} e^{sA} \xi \rangle \\ &= \sum_{j=0}^9 \sum_{i=1}^2 
\langle e^{sA} \xi, (\cN_+ + 1)^{1/2} \Theta_j^{(i)} (\cN_+ + 1 + \ell_{ij})^{1/2} e^{sA} \xi \rangle 
\end{split} \]
for appropriate $\ell_{ij} \in \{ \pm 1, \pm 2, \pm 3 \}$. With (\ref{eq:EH1}) and with Proposition \ref{prop:cN}, we conclude that 
\begin{equation}\label{eq:P1}\begin{split}
        \big|\text{P}_{1}\big| &\leq  C\langle \xi, e^{-sA}(\cN_+ +1)(\cH_N + 1)e^{sA}\xi\rangle +C\langle \xi, e^{-sA}(\cN_+ +1)^3 e^{sA}\xi\rangle \\
        &\leq  C\langle \xi, e^{-sA}(\cN_+ +1)(\cH_N + 1)e^{sA}\xi\rangle +C\langle \xi, (\cN_+ +1)^3\xi \rangle  \end{split}
\end{equation}

Next we analyse $ \text{P}_2$. {F}rom (\ref{eq:N+k-As}) and (\ref{eq:N+k-Ag}), we have 
        \begin{equation}\label{eq:P2VK}\begin{split} |\text{P}_2 | &\leq C \big\langle \xi, (\cN_+ + 1) \, \xi\big\rangle + |\langle e^{sA} \xi, [ \cN_+ , A] \cH_N e^{sA} \xi \rangle| 
        \end{split}\end{equation}
With 
\[\begin{split}
        [\cN_+, A]  =\; & \frac 1{\sqrt N} \sum_{r \in P_H,\,v \in P_L} \eta_r \big(\,3 \sigma_v b^*_{r+v} b^*_{-r}b^*_{-v} +  \gamma_v b^*_{r+v} b^*_{-r}b_{v} +\text{h.c.} \big) \\
        =\; & 3 A_\sigma + A_\gamma + \text{h.c.} 
\end{split}\]
we write 
\begin{equation}\label{eq:comm-dec}\begin{split}\big[\cN_+,A\big]\cH_N &= 3 A_\sigma \cH_N+A_\gamma \cH_N+3A_\sigma^*\cH_N+A_\gamma^*\cH_N\\
        &= \big(3 A_\sigma \cH_N+\text{h.c.}\big) +\big(A_\gamma \cH_N+\text{h.c.}\big) + \big[A_\gamma^*, \cH_N\big] + 3\big[A_\sigma^*, \cH_N\big]\\
        &=:\text{P}_{21}+\text{P}_{22} +\big[A_\gamma^*, \cH_N\big] + 3\big[A_\sigma^*, \cH_N \big]
        \end{split}\end{equation}
Here, we introduced the normally ordered operators $\text{P}_{21}= \text{P}_{211} + \text{P}_{212}$, $\text{P}_{22} = \text{P}_{221} + \text{P}_{222}$, where 
        \begin{equation}\label{eq:P21P22K}
        \begin{split}
        &\text{P}_{211} := \frac1{\sqrt N} \sum_{\substack{p \in \L_+^* \\ r \in P_H,v \in P_L}} p^2\eta_r \sigma_v b^*_{r+v} b^*_{-r}b^*_{-v}a^*_pa_p +\text{h.c.};\\
        &\text{P}_{221} := \frac1{\sqrt N}\sum_{\substack{p \in \L_+^* \\ r \in P_H,v \in P_L}}p^2 \eta_r \gamma_v b^*_{r+v} b^*_{-r}a^*_p a_pb_{v} + \frac1{\sqrt N}\sum_{\substack{r \in P_H, \\ v \in P_L}} v^2 \eta_r \gamma_v b^*_{r+v} b^*_{-r}b_{v}+\text{h.c.}
        \end{split}
        \end{equation}
and, switching to position space, 
        \begin{equation}\label{eq:P21P22V}
        \begin{split}
        \text{P}_{212}:=\;& \frac1{2N^{3/2}} \sum_{\substack{r \in P_H, \\v \in P_L}}  \sum_{\substack{p,q,u \in \Lambda^*_+ , \\
         u \not = -p,-q}} \widehat{V}(u/N)\eta_r \sigma_v b^*_{r+v} b^*_{-r}b^*_{-v}a^*_{p+u}a^*_{q}a_{p}a_{q+u} +\text{h.c.} \\
         =& \frac12 \int_{\Lambda^3}dxdydz \, N^{3/2} V(N( x-y)) \check{b}^*_{x} \check{b}^*_{y} \check{b}^*_{z} a^*(\check{\eta}_{H,z}) a^*( \check{\sigma}_{L,z}) \check{a}_{x} \check{a}_{y} +\text{h.c.} ;\\
        \text{P}_{222}:=\;& \frac1{2N^{3/2}}  \sum_{\substack{r \in P_H, \\v \in P_L}}  \sum_{\substack{p,q,u \in \Lambda^*_+ , \\
         u \not = -p,-q}}  \widehat{V}(u/N) \eta_r \gamma_v b^*_{r+v} b^*_{-r}a^*_{p+u}a^*_{q}a_{p}a_{q+u}b_{v}\\
        & + \frac1{2N^{3/2}} \sum_{\substack{r \in P_H, \\v \in P_L}}  \sum_{\substack{p,u \in \Lambda^*_+ , \\
         u \not = -p,-v}}  \widehat{V}(u/N)\eta_r \g_v  b^*_{r+v} b^*_{-r}  \big( a^*_{p+u} a_p b_{v+u}+ a^*_{-p} a_{v+u}b_{-p-u} \big) +\text{h.c.}\\
         =&\frac{1}{N^{1/2}}\sum_{\substack{r \in P_H, \\v \in P_L}}   \eta_r \gamma_v b^*_{r+v} b^*_{-r}\cV_N b_{v}\\
        &+\int_{\Lambda^3}dxdydz \, N^{3/2}V(N(x- y))\check{\gamma}_{L} (x-z) \check{b}_{z}^* b^*( \check{\eta}_{H,z}) \check{a}^*_{y}\check{a}_{x}\check{b}_{y}  +  \text{h.c.} 
\end{split}
\end{equation}
where, as we did in (\ref{eq:wtV-bd}) in the proof of Lemma \ref{lm:commHA}, we introduced the 
notation $\check{\eta}_{H}, \check{\gamma}_L$ to indicate functions on $\Lambda$, with Fourier coefficients given by $\eta \chi_H$ and, respectively, by $\gamma \chi_L$, with $\chi_H$ and $\chi_L$ being characteristic functions of high ($|p| > N^{1/2}$) and low ($|p| < N^{1/2}$) momenta.
Since, with the notation introduced in Lemma \ref{lm:commHA} after (\ref{eq:commHA}), 
\[\begin{split} \big[ A_{\gamma}^*, \cH_N \big] = \sum_{j=0}^9 (\Theta_j^{(1)})^* , \qquad 
\big[ A_{\sigma}^*, \cH_N \big] = \sum_{j=0}^9 (\Theta_j^{(2)})^* \, , \end{split} \] 
it follows from (\ref{eq:EH1}) that
\begin{equation}\label{eq:commAgAsH} \begin{split} |\langle e^{sA} \xi , \big[ A_\gamma^*, \cH_N \big] e^{sA} \xi \rangle | &\leq C \langle e^{sA} \xi, \big( \cH_N + (\cN_+ + 1)^2 \big) e^{sA} \xi \rangle  \\
  |\langle e^{sA} \xi , \big[ A_\sigma^*, \cH_N \big] e^{sA} \xi \rangle | &\leq C \langle e^{sA} \xi, \big( \cH_N + (\cN_+ + 1)^2 \big) e^{sA} \xi \rangle \end{split} \end{equation}

Finally, we estimate the expectations of the operators (\ref{eq:P21P22K}), (\ref{eq:P21P22V}). The term $\text{P}_{211}$ defined in (\ref{eq:P21P22K}) is bounded by  
\begin{equation}\label{eq:P211-bd} \begin{split} \big| \langle e^{sA} \xi, 
\text{P}_{211}  e^{sA} \xi \rangle |   
        &\leq  \frac{1}{\sqrt{N}} \left[ \sum_{p, r,v \in \L_+^*} p^2 \, \|  b_{r+v} b_{-r} b_{-v}  a_{p} (\cN_+ + 1)^{-1} e^{sA} \xi \|^2 \right]^{1/2} \\ & \hspace{3cm} \times  \left[ \sum_{p,r,v \in \L_+^*} 
        \eta_r^2  \s_v^2 p^2 \|  a_{p}\, (\cN_+ + 1)  e^{sA} \xi \|^2 \right]^{1/2} 
         \\
        & \leq C \langle e^{sA} \xi, (\cN_+ + 1)(\cK + 1)e^{sA}\xi\rangle
        \end{split}\end{equation}
because $\| \eta \|_2 , \| \s \|_2$ are finite, uniformly in $N$. Similarly (using that $v^2 \leq 2 (r+v)^2 + 2 r^2$), we find  
\begin{equation}\label{eq:P221-bd} \begin{split} \big| \langle e^{sA} \xi, 
\text{P}_{221}  e^{sA} \xi \rangle |   
        &\leq  C \langle e^{sA} \xi, (\cN_+ + 1)(\cK + 1)e^{sA}\xi\rangle 
        \end{split} \end{equation}
The expectation of the operator $\text{P}_{212}$ in (\ref{eq:P21P22V}) can be bounded using its expression in position space by  
\begin{equation}\label{eq:P212-bd} \begin{split} \big| \langle &e^{sA} \xi, \text{P}_{212}  e^{sA} \xi \rangle |   
        \\ &\leq \int_{\Lambda^3}dx dy dz \, N^{3/2}V(N( x-y)) \big\|  a\big( \check{\eta}_{H,z} \big)\check{a}_z\check{a}_{x} \check{a}_{y} e^{sA} \xi\big\| \big\| a^*(\check{\sigma}_{L,z}) \check{a}_{x}\check{a}_{y} e^{sA} \xi\big\| \\ &\leq \int_{\Lambda^3}dx dy dz \, N^{5/4}V(N( x-y)) 
      \, \big\| \check{a}_z\check{a}_{x} \check{a}_{y} (\cN_+ + 1)^{1/2} e^{sA} \xi\big\| \big\| \check{a}_{x}\check{a}_{y} (\cN_+ + 1)^{1/2} 
        e^{sA} \xi\big\| \\
        &\leq \;C N^{-1/4} \langle e^{sA} \xi,(\cN_+ +1) \cV_N   e^{sA} \xi \rangle 
        \end{split}\end{equation}
where we used the estimates $\| \check{\eta}_{H,z} \|_2 \leq C N^{-1/4}$, $\| \check{\sigma}_{L,z} \|_2 \leq C$. As for the operator $\text{P}_{222}$ in (\ref{eq:P21P22V}), the expectation of the first term 
is controlled by  
 \begin{equation}\label{eq:P2221}\begin{split} \Big| \frac{1}{N^{1/2}}\sum_{\substack{r \in P_H, \\v \in P_L}}  & \eta_r \gamma_v \big \langle e^{sA} \xi, b^*_{r+v} b^*_{-r}\cV_N b_{v} e^{sA} \xi\big\rangle \Big| \\
        &\leq \;  \frac CN \sum_{r,v \in\Lambda_+^*} \big \| \cV_N^{1/2}  a_{r+v} a_{-r}e^{sA} \xi\big\|^2 +\sum_{r,v \in\Lambda_+^*} |\eta_r|^2 \big \| \cV_N^{1/2} a_{v}e^{sA} \xi\big\|^2 \\
        &\leq \;  \frac CN \int_{\Lambda^2}dxdy\;  \sum_{r,v \in\Lambda_+^*}N^{2}V(N(x-y)) \big \|   a_{v} a_{r}\check{a}_x\check{a}_y e^{sA} \xi\big\|^2 \\
        &\hspace{0.5cm}+C\int_{\Lambda^2}dxdy\;\sum_{v \in\Lambda_+^*}N^{2}V(N(x-y))   \big \|  a_{v}\check{a}_x\check{a}_y e^{sA} \xi\big\|^2 \\
        &\leq \; C\langle \xi, e^{-sA}(\cN_+ + 1)(\cV_N + 1)e^{sA}\xi\rangle
        \end{split}\end{equation}
while the expectation of the second term is bounded in position space by 
        \begin{equation}\label{eq:P2222} \begin{split} 
        \Big | \int_{\Lambda^3} dx &dy dz \, N^{3/2}V(N(x- y)) \check{\gamma}_{L}(x-z) 
        \big\langle e^{sA}\xi,\check{b}_z^* b^*(\check{\eta}_{H,z})\check{a}^*_{y}\check{a}_{x}\check{b}_{y} e^{sA}\xi\big\rangle\Big |\\
        & \leq\;\Big( \int_{\Lambda^3}dxdydz \, N^{3/2} V(N(x-y)) \| \check{\eta}_{H,z} \|_2^2 \, \big\| \check{a}_{y} \check{a}_z  (\cN_+ + 1)^{1/2} e^{sA} \xi\big\|^2 \Big)^{1/2}\\
        &\hspace{2cm}\times \Big(\int_{\Lambda^3} dx dy dz \, N^{3/2} V(N(x-y)) \, |\check{\gamma}_{L} (x-z)|^2 \, \big\|  \check{a}_{x}\check{a}_y  e^{sA} \xi\big\|^2\Big)^{1/2} \\
        &\leq \; C N^{-1}\|\check{\gamma}_{L} \|_2\|\check{\eta}_{H} \|_2 \langle \xi, e^{-sA}(\cN_++1)^3 e^{sA}\xi\rangle ^{1/2} \langle \xi, e^{-sA}\cV_N e^{sA}\xi\rangle ^{1/2} \\
        &\leq \; C \langle \xi, e^{-sA}\cV_N e^{sA}\xi\rangle + C \langle \xi, (\cN_++1)^2 \xi\rangle
        \end{split}\end{equation}
        because $\| \check{\gamma}_L \|_2 \leq C N^{3/4}$ and $\| \check{\eta}_{H,z}\|_2 = \| \check{\eta}_H \|_2 \leq C N^{-1/4}$ for all $z \in \Lambda$. 
   {F}rom (\ref{eq:P2221}) and (\ref{eq:P2222}), we obtain that
   \begin{equation}\label{eq:P222-bd}   \langle e^{sA} \xi , \text{P}_{222} e^{sA} \xi \rangle \leq C \langle e^{sA} \xi , (\cH_N + 1) ( \cN_+ + 1) e^{sA} \xi \rangle \end{equation}
   
Combining (\ref{eq:P2VK}) with (\ref{eq:comm-dec}), (\ref{eq:P21P22K}), (\ref{eq:P21P22V}), (\ref{eq:commAgAsH}), (\ref{eq:P211-bd}), (\ref{eq:P221-bd}), (\ref{eq:P212-bd}) and (\ref{eq:P222-bd}), we conclude that 
        \begin{equation*}\begin{split}%\label{eq:sumP2}
        \big|\text{P}_2\big|
        &\leq C\langle \xi, e^{-sA}(\cN_++  1)(\cH_N + 1)e^{sA}\xi\rangle +C\langle \xi, (\cN_+ + 1)^{3}\xi\rangle
        \end{split}
        \end{equation*}
Applying (\ref{eq:P1}) and the last bound on the r.h.s. of (\ref{eq:P1P2}), we arrive at
        \[\partial_s\varphi_\xi(s) \leq C\varphi_\xi(s)+C\langle \xi, (\cN_+ + 1)^{3}\xi\rangle \]
for some constant $C>0$, independent of $\xi\in\cF_+^{\leq N}$. By Gronwall's lemma, we conclude that there exists another constant $C > 0$ such that, for all $s\in [0;1]$, 
        \[ \begin{split} \langle e^{sA} \xi, (\cN_+ + 1) ( \cH_N + 1) e^{sA} \xi \rangle = \; &\varphi_\xi(s) \\ \leq \; &C \varphi_\xi(0)+ C \langle \xi, (\cN_+ + 1)^{3}\xi\rangle \\ = \; &C \langle \xi, (\cN_++1)(\cH_N +1)\xi\rangle +  C\langle \xi, (\cN_+ + 1)^{3}\xi\rangle \, . \end{split}\]
This concludes the proof of the proposition.
\end{proof}

We summarize the results of this section in the following corollary, which is a simple consequence of Prop. \ref{prop:hpN}, Prop. \ref{prop:cN} and Prop. \ref{prop:cNcH}. 
\begin{cor}\label{cor:apri}
Let $V \in L^3 (\bR^3)$ be non-negative, compactly supported and spherically symmetric. 
Let $E_N$ be the ground state energy of $H_N$, defined in (\ref{eq:Ham0}). Let $\psi_N \in L^2_s (\Lambda^N)$ with $\| \psi_N \| = 1$ belong to the spectral subspace of $H_N$ with energies below $E_N + \zeta$, for some $\zeta > 0$, i.e. 
\begin{equation*}
\psi_N = {\bf 1}_{(-\infty ; E_N + \zeta]} (H_N) \psi_N 
\end{equation*}
Let $\xi_N = e^{-A} e^{-B(\eta)} U_N \psi_N$ be the cubically renormalized excitation vector associated with $\psi_N$. Then, there exists a constant $C > 0$ such that 
\begin{equation*} \left\langle \xi_N, \left[  (\cN_+ +1)(\cH_N+1) + (\cN_+ +1)^3 \right] \xi_N \right\rangle  \leq C (1 + \zeta^{3}) \, . \end{equation*}
\end{cor}

\section{Diagonalization of the Quadratic Hamiltonian} \label{sec:diag}

{F}rom Proposition \ref{thm:tGNo1} we can decompose the cubically renormalized excitation Hamiltonian $\cJ_N$ defined in (\ref{eq:cJN-def}) as 
\begin{equation}\label{eq:deftGN}
        \begin{split}
        \cJ_N = C_{\cJ_N} + \cQ_{\cJ_N}  + \cV_N + \cE_{\cJ_N}
        \end{split}
        \end{equation}
with the constant $C_{\cJ_N}$ given in \eqref{eq:tildeCN}, the quadratic part 
 \begin{equation}\label{def:CNQN}
        \begin{split} 
       \cQ_{\cJ_N}=&\;\sum_{p\in\Lambda^*_+}\Big[F_pb^*_pb_p+\frac{1}{2}G_p\big(b^*_pb^*_{-p}+b_pb_{-p}\big)\Big]
        \end{split}
        \end{equation}
with the coefficients $F_p$, $G_p$ as in \eqref{eq:defFpGp} and the error term  $\cE_{\cJ_N}$ satisfying
         \[\pm {\mathcal{E}}_{\cJ_N}  \leq C N^{-1/4} \big[\cH_N+ (\cN_++1)^2\big](\cN_++1)\]
as an operator inequality on $\cF_+^{\leq N}$. 

Our goal in this section is to diagonalise the quadratic operator $\cQ_{\cJ_N}$. To reach this goal, we need first to establish some bounds for the coefficients $F_p, G_p$ in (\ref{def:CNQN}).
\begin{lemma}\label{lm:FpGp}
Let $V \in L^3 (\bR^3)$ be non-negative, compactly supported and spherically symmetric. Let $F_p, G_p$ be defined as in  \eqref{eq:defFpGp}. Then there exists a constant $C > 0$ such that 
\begin{equation*}
i) \quad p^2 /2 \leq F_p \leq C (1+p^2)\;; \hspace{0.5cm}ii) \quad |G_p| \leq \frac{C}{p^2} \; ;  \hspace{0.5cm} iii) \quad |G_p| < F_p 
\end{equation*}  
for all $p \in \Lambda^*_+$. 
\end{lemma}
\begin{proof} We first show the lower bound in i). For $p \in \L^*_+$ with $|p| \leq N^{1/2}$ we use
\[
\Big| \big(\widehat V(\cdot/N) \ast \widehat f_{\ell,N}\big)(p) - \big(\widehat V(\cdot/N) \ast \widehat f_{\ell,N}\big)(0) \Big| \leq \frac C N |p|\,,
\]
Since $(\g_{p}^2 + \s_{p}^2)\geq 1$, we have
\[
F_{p} \geq p^2 + \big(\widehat V(\cdot/N) \ast \widehat f_{\ell,N}\big)(0)  (\g_{p}+ \s_{p})^2 - C N^{-1/2} \geq p^2 - C N^{-1/2} \geq  p^2 /2
\]
for all $p \in \L^*_+$  such that  $|p| \leq N^{1/2}$, if $N$ is large enough. On the other side, for $|p| \geq N^{1/2}$ the inequality is clear, being  $|(\widehat V(\cdot/N) \ast \widehat f_{\ell,N})_p|\leq C$. The upper bound $F_{p} \leq C (1 +p^2)$ follows easily from the definition, from the boundness of $(\widehat V(\cdot/N) \ast \widehat f_{\ell,N})_p$ and from the fact that $|\s_{p}|, \g_{p} \leq C$ for all $p \in \L_+^*$. 

The proof of part ii) makes use of the relation (\ref{eq:eta-scat}) for the coefficients $\eta_{p}$.  For any $p \in \L^*_+$ we have
\begin{equation}\label{eq:GP-scat} 
G_{p} = 2 p^2 \eta_p + \widehat{V} (p/N) + \frac{1}{N} \sum_{q \in \Lambda^*} \widehat{V} ((p-q)/N) {\eta}_q + \wt{G}_{p} 
\end{equation}
where $|\wt{G}_{p}| \leq C p^{-2}$ for all $p \in \L^*_+$. Here we used the fact that $|\eta_p| \leq C p^{-2}$, which implies 
\[ | \sigma_{p} \gamma_{p} - \eta_{p} | \leq C |p|^{-6}, \qquad |(\sigma_{p} + \gamma_{p})^2 - 1| \leq C |p|^{-2} \]
With the relation (\ref{eq:eta-scat}) we obtain 
\begin{equation}\label{eq:GNp2} G_{p} = 2N^3 \lambda_{\ell} \widehat{\chi}_\ell(p) + 2N^2 \lambda_{\ell} \sum_{q \in \Lambda^*} \widehat{\chi}_\ell (p-q) {\eta}_q + \wt{G}_{p} \end{equation}
From Lemma \ref{3.0.sceqlemma}, part i), we have $N^3 \lambda_{\ell} \leq C$. A simple computation shows that 
\begin{equation}\label{eq:chip} \widehat{\chi}_\ell (p) = \int_{|x| \leq \ell}  e^{-ip \cdot x} dx = \frac{4\pi}{|p|^2} \left( \frac{\sin (\ell |p|)}{|p|} - \ell \cos (\ell |p|) \right) \end{equation}
which, in particular, implies that $|\widehat{\chi}_\ell (p)| \leq C |p|^{-2}$. Similarly, we find
\[ N^2 \lambda_\ell \sum_{q \in \Lambda^*} \widehat{\chi}_\ell (p-q) {\eta}_q = -N^3 \lambda_{\ell} \int_\Lambda \chi_\ell (x) w_{\ell} (N x) e^{-ip\cdot x} dx = -N^3 \lambda_{\ell} \int_{|x| \leq \ell} w_{\ell} (Nx) e^{-ip \cdot x} dx \]
Switching to spherical coordinates and integrating by parts, we find (abusing slightly the notation by writing $w_{\ell} (N r)$ to indicate $w_{\ell} (N x)$ for $|x| = r$),   
\[  \begin{split} 
\int_{|x| \leq \ell} w_{\ell} (N x) e^{-ip \cdot x} dx &= 2\pi \int_0^\ell dr \, r^2 w_{\ell} (N r) \int_0^\pi d\theta \, \sin \theta \, e^{-i |p| r \cos \theta} \\ &= \frac{4\pi}{|p|} \int_0^\ell dr \, r w_{\ell} (\ell r) \sin (|p|r) \\ &= - \frac{4\pi}{|p|^2} \lim_{r \to 0} r w_{\ell} (Nr) + \frac{4\pi}{|p|^2} \int_0^\ell dr \,  \frac{d}{dr} (r w_{\ell} (Nr)) \cos (|p| r) \end{split}\]
With (\ref{3.0.scbounds1}) and using again the bound $N^3\lambda_{\ell} \leq C$, we conclude that there is a constant $C > 0$ such that 
\begin{equation}\label{eq:bd-convchi}  \Big|  N^2 \lambda_\ell \sum_{q \in \Lambda^*} \widehat{\chi}_\ell (p-q) {\eta}_q \Big| \leq C |p|^{-2} \end{equation}
for all $p \in \L^*_+$.  

Finally, we show iii). To this end, we notice that 
\begin{equation}\label{eq:iii-step1} F_p - G_p = p^2 (\g_p - \s_p)^2 > 0 \end{equation}
because $\g_p \not = \s_p$ for all $p \in \L^*_+$. Furthermore, arguing as we did in the proof of part i) to show that $F_p \geq p^2/2$, we find  
\begin{equation}\label{eq:iii-step2}  F_p + G_p = (\g_p + \s_p)^2 \left[ p^2 + (\widehat{V}(./N) * \widehat{f}_{\ell,N})_p \right] \geq (\g_p+\s_p)^2 p^2 /2 \end{equation}
for all $p \in \L^*_+$. Since $\g_p \not = -\s_p$, we conclude that $F_p + G_p > 0$ for all $p \in \L^*_+$; (\ref{eq:iii-step1}) and (\ref{eq:iii-step2}) give $|G_p| < F_p$ for all $p \in \L^*_+$, as claimed. 
\end{proof}

Lemma \ref{lm:FpGp} shows that, $|G_p|/ F_p < 1$ for all $p \in \Lambda^*_+$. Hence, we can introduce coefficients $\tau_p \in \bR$ such that 
\begin{equation}\label{eq:taup} \tanh (2\tau_p) = - \frac{G_p}{F_p}  \, \end{equation}
for all $p \in \Lambda^*_+$. Equivalently, \[    \tau_p = \frac{1}{4}  \log \frac{1-G_p/F_p}{1+G_p/F_p}. \] 
Using these coefficients, we define the generalized Bogoliubov transformation $e^{B(\tau)}: \cF_+^{\leq N}\to \cF_+^{\leq N}$ with
        \begin{equation*}%\label{def:Btau}
        B(\tau):=\frac{1}{2}\sum_{p\in\Lambda^*_+}\tau_{p}\big(b^*_{-p}b^*_p-b_{-p}b_p\big) 
 \end{equation*}
We are going to conjugate the excitation Hamiltonian $\cJ_N$ defined in (\ref{eq:deftGN}) with 
$e^{B(\tau)}$ to diagonalize its quadratic component $\cQ_{\cJ_N}$. In the next lemma we 
show that, up to small errors, the other terms in $\cJ_N$ are left unchanged by this transformation. Here, we use the fact that, from Lemma \ref{lm:FpGp}, $|\tau_p | \leq C |p|^{-4}$ for some constant $C > 0$ and all $p \in \L^*_+$. 

\begin{lemma}\label{lm:err-con}
Let $V \in L^3 (\bR^3)$ be non-negative, compactly supported and spherically symmetric. Let $\tau_p$ be defined as in (\ref{eq:taup}), with $F_p, G_p$ as in \eqref{eq:defFpGp} and $\cV_N, \cH_N$ be as defined in (\ref{eq:KcVN}). Then, there exists a constant $C > 0$ such that 
\begin{equation}\label{eq:err-con} e^{-B(\tau)} (\cN_+ +1)(\cH_N+1) e^{B(\tau)} \leq C  (\cN_+ +1)(\cH_N+1) \end{equation}
and 
\begin{equation}\label{eq:err-VN} \pm \big[e^{-B(\tau)} \cV_N e^{B(\tau)}-\cV_N \big]\leq C  N^{-1/2}(\cH_N+1)(\cN_+ +1) \end{equation}
\end{lemma} 
\begin{proof}
The proof of (\ref{eq:err-con}) is similar to the one of \cite[Lemma 5.4]{BBCS2}; the only difference is the fact that, here, the potential energy $\cV_N$ scales differently with $N$. We review therefore the main steps of the proof, focussing on terms involving $\cV_N$.

We are going to apply Gronwall's lemma. For $ \xi \in\cF_+^{\leq N}$ and $s \in \bR$, we compute
        \[\begin{split} &\partial_s  \langle \xi, e^{-sB(\tau)}  (\cH_N +1 )(\cN_+ +1) e^{sB(\tau)} \xi \rangle = -\langle \xi, e^{-s B(\tau)} [B(\tau),  (\cH_N+1)(\cN_+ +1) ] e^{sB(\tau)} \xi \rangle \end{split} \]
By the product rule, we have
        \begin{equation}\label{eq:commBHN} \begin{split} 
        [ B(\tau),  &(\cH_N+1)(\cN_+ +1) ] \\ = \; &(\cH_N +1)  
        [B(\tau), \cN_+  ] + [ B(\tau),  \cK](\cN_+ +1) + [B(\tau), \cV_N] (\cN_+ +1)  \end{split}
        \end{equation}
The first term on the r.h.s. of (\ref{eq:commBHN}) can be written as  
        \begin{equation}\label{eq:commHBN} \begin{split}
        \langle \xi, e^{-sB(\tau)} (\cH_N + 1) &[B(\tau), \cN_+ ] e^{s B(\tau)} \xi \rangle \\ = \; &\sum_{p,q \in \Lambda^*_+} \tau_p q^2 \langle \xi,  e^{-sB(\tau)} a_q^* a_q (b_p b_{-p}+b_p^*b_{-p}^*) e^{sB(\tau)} \xi \rangle \\ &+ \sum_{p \in \Lambda^*_+} \tau_p \langle \xi,  e^{-sB(\tau)} \cV_N (b_p b_{-p}+b_p^*b_{-p}^*) e^{sB(\tau)} \xi \rangle \\ =: \; &\text{I} + \text{II} 
        \end{split} \end{equation}
From the proof of \cite[Lemma 5.4]{BBCS2}, we have
        \[ \begin{split} 
        | \text{I} |  &\leq C \langle e^{sB(\tau)} \xi , (\cN_+ +1) ( \cK+1) e^{sB(\tau)} \xi \rangle \end{split} \]
To estimate $\text{II}$, we switch to position space. We find 
        \[\begin{split} | \text{II} | &\leq \sum_{p \in \Lambda^*_+}  |\tau_p|  \int dx dy \, N^2 V(N(x-y)) \Big| \langle \check{a}_x \check{a}_y e^{sB(\tau)} \xi , \check{a}_x \check{a}_y (b_p b_{-p} + b_p^* b_{-p}^*) e^{sB(\tau)} \xi \rangle \Big| 
        \\ &\leq  \sum_{p \in \Lambda^*_+} |\tau_p | \int dx dy \, N^2 V(N(x-y)) \| \check{a}_x \check{a}_y (\cN_+ +1)^{1/2} e^{sB(\tau)} \xi \| \\ &\hspace{.3cm} \times \Big[ \| (b_p b_{-p} + b_p^* b_{-p}^*) (\cN_+ +1)^{-1/2} \check{a}_x \check{a}_y   e^{sB(\tau)} \xi \| + \| \check{a}_y e^{s B(\tau)} \xi \| + \| \check{a}_x e^{ sB(\tau)} \xi \| + \| \xi \| \Big] \\  &\leq C \langle \xi, e^{-sB(\tau)} (\cV_N + 1) (\cN_+ +1) e^{sB(\tau)} \xi \rangle \end{split} \]
since $\tau \in\ell^1(\Lambda_+^*)$, uniformly in $N$. {F}rom (\ref{eq:commHBN}), we obtain that 
        \begin{equation}\label{eq:BN-fin} \Big| \langle \xi, e^{-sB(\tau)} (\cH_N + 1) [B(\tau), \cN_+ ] e^{s B(\tau)} \xi \rangle \Big| \leq C 
        \langle \xi, e^{-sB(\tau)}  (\cH_N +1 )(\cN_+ +1) e^{sB(\tau)} \xi \rangle 
        \end{equation}

The second term on the r.h.s. of (\ref{eq:commBHN}) can be bounded as in \cite{BBCS2} by 
        \begin{equation}\label{eq:Kcom}  \begin{split} \Big|\langle \xi, &e^{-sB(\tau)} [B(\tau) ,  \cK ] (\cN_+ +1) e^{sB(\tau)} \xi \rangle \Big| \leq C \langle \xi, e^{-s B(\tau)}  (\cH_N+1)(\cN_+ +1) e^{sB(\tau)} \xi \rangle \end{split} \end{equation}
       
Finally, we analyse the third term on the r.h.s. of (\ref{eq:commBHN}). Again, it is convenient to switch to position space. We find
        \begin{equation}\label{eq:comVNBtau} \begin{split} 
        [ B(\tau), \cV_N ] (\cN_+ +1) =\; &\frac{1}{2} \int_{\Lambda \times \Lambda} dx dy \, N^{2} V(N (x-y)) \check{\tau} (x-y) (\check{b}_x^* \check{b}_y^* +  \check{b}_x \check{b}_y) (\cN_+ +1)   \\ &+  \int_{\Lambda \times \Lambda} dx dy\; N^{2} V(N (x-y))  \big[ b_x^* b_y^* a^* (\check{\tau}_y) \check{a}_x + \text{h.c.} \big] (\cN_+ +1)  \end{split} \end{equation}
where $\check{\tau} (x) = \sum_{p \in \Lambda^*_+} \tau_p e^{ip \cdot x}$. Using $\| \check{\tau} \|_\infty \leq \| \tau \|_1 \leq C < \infty$ as well as $\| \check{\tau}_y \| = \| \check{\tau} \| = \| \tau \| \leq C < \infty$ independently of $y \in \Lambda$ and of $N$, it is then simple to check that
       \[\Big|   \langle \xi, e^{-sB(\tau)}  [ B(\tau), \cV_N ] (\cN_+ +1)e^{sB(\tau)}\xi \rangle \Big| \leq C \langle \xi, e^{-sB(\tau)} (\cV_N +\cN_+ +1) (\cN_+ +1)e^{sB(\tau)} \xi \rangle   \]
Combining this bound with (\ref{eq:BN-fin}) and (\ref{eq:Kcom}), we obtain 
\[ \Big| \partial_s  \langle \xi, e^{-sB(\tau)}  (\cH_N +1 )(\cN_+ +1) e^{sB(\tau)} \xi \rangle \Big| \leq C   \langle \xi, e^{-sB(\tau)}  (\cH_N +1 )(\cN_+ +1) e^{sB(\tau)} \xi \rangle \]
By Gronwall's inequality and integrating over $s \in [0;1]$ we conclude (\ref{eq:err-con}).

To prove (\ref{eq:err-VN}), on the other hand, we write 
  \[\begin{split}
        e^{-B(\tau)} \cV_N e^{B(\tau)}-\cV_N = \int_0^1 ds\, e^{-sB(\tau)} [  \cV_N,B(\tau) ]e^{sB(\tau)} 
        \end{split}\]
With (\ref{eq:comVNBtau}), it is simple to check that
        \[\pm[ B(\tau), \cV_N ] \leq N^{-1/2} (\cV_N + \cN_+ +1)(\cN_++1) \]
By (\ref{eq:err-con}) and $\cN_+\leq \cK$, the last bound immediately implies
        \[ \pm \big[e^{-B(\tau)} \cV_N e^{B(\tau)}-\cV_N \big]\leq C  N^{-1/2}(\cH_N+1)(\cN_+ +1) \, .\] \end{proof} 
 
The next lemma shows that the generalized Bogoliubov transformation $e^{B(\tau)}$ diagonalizes the quadratic operator $\cQ_{\cJ_N}$, up to errors that are negligible in the limit of large $N$. A similar result was established in \cite[Lemma 5.2]{BBCS2}, but only under the additional assumption of small interaction potential which guaranteed $\| \tau \|$ to be sufficiently small and therefore allowed us to use the identity (\ref{eq:ebe}) and the bounds in Lemma \ref{lm:dp} to control the action of $e^{B(\tau)}$. Below, we provide a new proof which does not require smallness of $\| \tau \|$. 
\begin{lemma}\label{lm:diago}
Let $V \in L^3 (\bR^3)$ be non-negative, compactly supported and spherically symmetric. Let $\cQ_{\cJ_N}$ be defined as in (\ref{def:CNQN}) and $\tau_p$ as in (\ref{eq:taup}) with the coefficients $F_p, G_p$ as in  \eqref{eq:defFpGp}. Then   
\[ e^{-B(\tau)} \cQ_{\cJ_N} e^{B(\tau)} = \frac{1}{2} \sum_{p \in \Lambda^*_+} \left[ -F_p + \sqrt{F_p^2 - G_p^2} \right] + \sum_{p \in \Lambda^*_+} \sqrt{F_p^2 - G_p^2} \; a_p^* a_p + \delta_{N} \]
where the operator $\delta_{N}$ is such that, on $\cF_+^{\leq N}$, 
\begin{equation} \label{eq:pmdelta} 
\pm \delta_N \leq C N^{-1} (\cK + 1)(\cN_+ +1)  \end{equation}
\end{lemma}
\begin{proof}
Using the commutation relations (\ref{eq:comm-bp}), we expand 
\[ \begin{split}
e^{-B(\tau)} b_p e^{B(\tau)} =\; &  b_p + \int_0^1 ds \, e^{-s B(\t)} [b_p, B(\t)]   e^{s B(\t)}   \\
  = \; &b_p + \int_0^1 ds \, e^{-s B(\t)}  \t_{p} b^*_{-p}  e^{s B(\t)} \\
&- \int_0^1 ds \, e^{-s B(\t)} \Big( \t_{p} \frac{\cN_+} N b^*_{-p} + \frac 1 N \sum_{q \in \L^*_+} b^*_q a^*_{-q} a_p \t_{q} \Big)   e^{s B(\t)}  \\
=\; & b_p + \t_{p} b^*_{-p} +  \int_0^1 ds_1 \int_0^{s_1} ds_2 \, e^{-s_2 B(\t)} \t^2_p b_p  e^{s_2 B(\t)} \\
&- \int_0^1 ds \, e^{-s B(\t)} \Big( \t_{p} \frac{\cN_+} N b^*_{-p} + \frac 1 N \sum_{q \in \L^*_+} b^*_q a^*_{-q} a_p \t_{q} \Big)   e^{s B(\t)} \\
&- \int_0^1 ds_1 \int_0^{s_1} ds_2 \, e^{-s_2 B(\tau)} \left( \tau_p^2 \frac{\cN_+}{N} b_p  + \frac{\tau_p}{N} \sum_{q \in \L^*_+} \tau_q a^*_{-p} a_{-q} b_q \right) e^{s_2 B(\tau)}   
\end{split}\]
Iterating the expansion, and using Lemma \ref{lm:Ngrow} to control the error term, we get
\be \begin{split} \label{eq:deco-tau}
e^{-B(\tau)} b_p e^{B(\tau)} =\; & \cosh(\t_{p}) b_p + \sinh(\t_{p}) b^*_{-p} + D_{p}
\end{split}\ee
with the remainder operator 
\[ \begin{split}
D_{p} = \; & \sum_{n \geq 0}  \int_0^1 d s_1 \cdots \int_0^{s_{2n}} d s_{2n+1}\\
& \hskip 1.5cm \times e^{-s_{2n+1}B(\t)} \Big(- \t_{p}^{2n+1} \, \frac{\cN_+}{N} b^*_{-p} - \frac 1 N \,\t_{p}^{2n} \sum_{q \in \L^*_+} b^*_q a^*_{-q} a_p \t_{q} \Big) e^{s_{2n+1}B(\t)} \\
& + \; \sum_{n \geq 1}  \int_0^1 d s_1 \cdots \int_0^{s_{2n-1}} d s_{2n}\\
& \hskip 1.5cm \times e^{-s_{2n}B(\t)} \Big(- \t_{p}^{2n} b_p\, \frac{\cN_+}{N}  - \frac 1 N \,\t_{p}^{2n-1} \sum_{q \in \L^*_+} a^*_{-p} a_{-q} b_q \t_{q} \Big) e^{s_{2n}B(\t)} 
\end{split}\]
From $\|\t \|_1 \leq C$ and Lemma \ref{lm:Ngrow}, it follows that
\be \label{eq:DpBound}
\| (\cN+1)^{n/2} D_{p} \xi \| \leq  \frac{C}{N} |\t_{p}| \| \|(\cN_++1)^{(n+3)/2} \xi \| + \frac{C}{N} \int_0^1ds\; \| a_p (\cN_+ +1)^{(n+2)/2} e^{s B(\tau)} \xi \|
\ee

With (\ref{eq:deco-tau}) and using the shorthand notation $\wt{\gamma}_{p} = \cosh \tau_{p}, \wt{\sigma}_{p} = \sinh \tau_{p}$, we can write
\begin{equation}\label{eq:diag1} \begin{split}  
e^{-B(\tau)}& \cQ_{\cJ_{N}} e^{B(\tau)} \\
= \; & \sum_{p \in \Lambda^*_+} \big(F_{p} \wt{\sigma}_{p}^2 + G_{p} \wt{\gamma}_{p} \wt{\sigma}_{p} \big) + \sum_{p \in \Lambda^*_+} \Big[ F_{p} (\wt{\gamma}^2_{p} + \wt{\sigma}^2_{p}) + 2 G_{p} \wt{\sigma}_{p} \wt{\gamma}_{p} \Big] b_p^* b_p \\ &+ \frac{1}{2} \sum_{p \in \Lambda^*_+} \Big[ 2 F_p \wt{\gamma}_p \wt{\sigma}_p + G_p (\wt{\gamma}_p^2 + \wt{\sigma}^2_p) \Big] (b_p b_{-p} + b_p^* b_{-p}^* ) + {\delta}_{N} \end{split}  \end{equation}
where 
\begin{equation}\label{eq:delta1} \begin{split} {\delta}_{N} = \; &\sum_{p \in \Lambda_+^*} F_{p} D_{p}^* e^{-B(\tau)} b_{p} e^{B(\tau)} + \sum_{p \in \Lambda^*_+} F_{p} (\wt{\gamma}_{p} b_p^* + \wt{\sigma}_{p} b_p) D_{p} \\ &+ \frac{1}{2} \sum_{p \in \Lambda^*_+} G_{p} \Big[ D_{p}^* e^{-B(\tau)} b_{-p}^* e^{B(\tau)} + \text{h.c.} \Big] + \frac{1}{2} \sum_{p \in \Lambda^*_+} G_{p} \Big[ (\wt{\gamma}_{p} b_p^* + \wt{\sigma}_{p} b_{-p}) D_{-p}^* + \text{h.c.} \Big]  \end{split} \end{equation}
With (\ref{eq:taup}), a lengthy but straightforward computation leads to 
\[ e^{-B(\tau)} \cQ_{\cJ_{N}} e^{B(\tau)} = \; 
\frac{1}{2} \sum_{p \in \Lambda^*_+} \Big[ - F_{p} + \sqrt{F_{p}^2 - G_{p}^2} \Big] + \sum_{p \in \Lambda^*_+} \sqrt{F_{p}^2 - G_{p}^2} \; b_p^* b_p +{\delta}_{N} \]
{F}rom the bound $F_{p} \leq C (1+p^2)$ in Lemma \ref{lm:FpGp} we obtain  
\[ \begin{split} 
\Big| \sum_{p \in \Lambda^*_+} \sqrt{F_{p}^2 - G_{p}^2} \, \Big[ \langle \xi , b_p^* b_p \, \xi \rangle - \langle \xi , a_p^* a_p \xi \rangle \Big]  \Big| = &\; \Big| \frac{1}{N} \sum_{p \in \Lambda^*_+} \sqrt{F_{p}^2 - G_{p}^2} \, \langle \xi, a_p^* \cN_+  a_p \xi \rangle \Big| \\ \leq \; &\frac{C}{N} \sum_{p \in \Lambda^*_+} (p^2 + 1) \| a_p (\cN_+ +1)^{1/2} \xi \|^2 \\ = \; &\frac{C}{N} \langle \xi, (\cN_+ +1)(\cK+1) \xi \rangle \end{split} \]
for all $\xi \in \cF_+^{\leq N}$. Hence, the claim follows if we can show that the operator ${\delta}_{N}$ defined in (\ref{eq:delta1}) satisfies (\ref{eq:pmdelta}). Consider first the expectation of the first term on the r.h.s. of (\ref{eq:delta1}). Using \eqref{eq:DpBound}, the bounds $|F_{p}| \leq C (1 + p^2)$ and $|\tau_p| \leq C |p|^{-4}$, Lemma \ref{lm:Ngrow} and then also Lemma \ref{lm:err-con}, we arrive at  
\begin{equation*}%\label{eq:delta11} 
\begin{split} 
 \Big| &\sum_{p \in \Lambda^*_+} F_{p} \langle D_{p} \xi , e^{-B(\tau)} b_p e^{B(\tau)} \xi \rangle \Big| \\ 
&\leq  \frac C N  \sum_{p \in \Lambda^*_+} |F_p| \, \| (\cN_+ +1)^{-1/2} D_{p} \xi \| \| (\cN_+ +1)^{1/2}   e^{-B(\tau)} b_p e^{B(\tau)}   \xi \|  \\
& \leq  \frac C N  \sum_{p \in \Lambda^*_+} (1+ p^2 ) |\t_p|  \| (\cN_++1) \xi \|  \| a_p  (\cN_++1)^{1/2} e^{B(\t)} \xi \| \\
& \hspace{0.5cm} +  \frac C N  \sum_{p \in \Lambda^*_+} (1+ p^2 ) \int_0^1ds\; \| a_p (\cN_+ +1)^{1/2} e^{s B(\tau)} \xi \|   \| a_p  (\cN_++1)^{1/2} e^{B(\t)} \xi \| \\
& \leq \frac C N \bmedia{\xi, (\cH_N + 1)(\cN_++1) \xi }
\end{split} \end{equation*}
The expectation of the second term on the r.h.s. of (\ref{eq:delta1}) can be bounded similarly. As for the third term on the r.h.s. of (\ref{eq:delta1}), we estimate 
\[ \begin{split} 
\Big| \sum_{p \in \Lambda^*_+} G_{p} &\langle  \xi , D_{p}^* e^{-B(\tau)} b^*_{-p} e^{B(\tau)} \xi \rangle \Big| \\ 
=\; & \sum_{p \in \Lambda^*_+} |G_{p}| \| (\cN_+ +1)^{-1/2}  D_{p} \xi \| \| (\cN_+ +1)^{1/2} e^{-B(\tau)} b_{-p}^* e^{B(\tau)} \xi \| \\ 
\leq \; & \frac C N  \sum_{p \in \Lambda^*_+} |G_{p}|  \Big(  |\t_{p}|  \| (\cN_++1) \xi \| +  \int_0^1ds\; \| b_p (\cN_++1)^{1/2} e^{s B(\tau)} \xi \| \Big)  \| (\cN_++1) \xi \| \\
\leq \; & \frac C N   \| (\cN_++1) \xi \|^2
 \end{split} \]
where we used Lemma \ref{lm:Ngrow} and the bound $|G(p)| \leq C |p|^{-2}$ from Lemma \ref{lm:FpGp}. The last term on the r.h.s. of (\ref{eq:delta1}) can be controlled similarly. 
\end{proof}

It follows from Lemma \ref{lm:err-con} and Lemma \ref{lm:diago} that the new excitation Hamiltonian $\cM_N : \cF_+^{\leq N} \to \cF_+^{\leq N}$ defined by 
\begin{equation*} %\label{def:cJN} 
\cM_N =  e^{-B(\tau)} \cJ_N e^{B(\tau)} = e^{-B(\tau)} e^{-A} e^{-B(\eta)}U_N H_N U_N^* e^{B(\eta)} e^{A} e^{B(\tau)} \end{equation*} 
can be decomposed as \[ \cM_N = C_{\cM_N}  + \cQ_{\cM_N} +\cV_N + \cE_{\cM_N} \]
where 
        \begin{equation}\begin{split}\label{def:CNQNJN} C_{\cM_N}:=&\; C_{\cJ_N} +\frac{1}{2} \sum_{p \in \Lambda^*_+} \left[ -F_p + \sqrt{F_p^2 - G_p^2} \right];\hspace{0.5cm} \cQ_{\cM_N}:=\;\sum_{p \in \Lambda^*_+} \sqrt{F_p^2 - G_p^2} \; a_p^* a_p
        \end{split}
        \end{equation}
with $C_{\cJ_N}$ as in (\ref{eq:tildeCN}) and $F_p, G_p$ as in (\ref{eq:defFpGp}) 
and where  the error $ \cE_{\cM_N}$ is such that 
        \[ \pm \cE_{\cM_N}\leq CN^{-1/4} \big[(\cH_N+1) (\cN_++1)  + (\cN_++1)^3 \big] \]

To conclude this section, we are going to compute the constant $C_{\cM_{N}}$ and the diagonal coefficients $(F_{p}^2 - G_{p}^2)^{1/2}$ appearing in the quadratic operator $\cQ_{\cM_{N}}$, up to errors that are negligible in the limit $N \to \infty$. To this end, we introduce the notation 
\begin{equation}\label{def:Ebog} 
E_{\text{Bog}}:= \frac12\sum_{p\in\Lambda_+^*}\bigg[ \sqrt{p^4+16\pi \frak{a}_0   p^2} -p^2-8\pi \frak{a}_0  + \frac{(8\pi \frak{a}_0 )^2}{2p^2}   \bigg]
\end{equation}

\begin{lemma}\label{lm:CQJNorder1}
Let $V \in L^3 (\bR^3)$ be non-negative, compactly supported and spherically symmetric. \begin{enumerate}
\item[i)] The constant $C_{\cM_N}$ in (\ref{def:CNQNJN}) is given by 
        \[ C_{\cM_N} =  4\pi (N-1) \frak{a}_0 + e_\Lambda \mathfrak{a}_0^2 + E_{\emph{Bog}} + \cO (N^{-1}\log N) \] 
\begin{equation}\label{eq:eLambda} e_\Lambda = 2 -  \lim_{M\to\infty} \sum_{\substack{p\in \bZ^3 \backslash \{0 \}: \\ |p_1|,|p_2|,|p_3| \leq  M}}  \frac{4\cos(|p|)}{p^2}  \end{equation}
In particular, we will show that the limit exists.
\item[ii)] The quadratic operator $ \cQ_{\cM_N}$ in (\ref{def:CNQNJN}) is given by  
        \[ \cQ_{\cM_N} = \sum_{p \in \Lambda^*_+}\sqrt{p^4+16 \pi \frak{a}_0 p^2} \;a_p^* a_p + \widetilde{\delta}_N \]
where the error $\wt{\delta}_N$ is bounded by
        \[\pm \widetilde{\delta}_N \leq C N^{-1} (\cK +1)  \]
\end{enumerate} 
\end{lemma}
\begin{proof} To show $i)$ we recall from (\ref{def:CNQNJN}) that $C_{\cM_{N}} $ is explicitly given by 
		\begin{equation}\label{eq:lmCQJN1}
		\begin{split}
		C_{\cM_{N}} =&\; \frac{N-1}{2}\widehat V(0) - \frac1N\sum_{p \in\L_+^*} \big( \widehat V(./N)\ast \eta\big)_p \gamma_p \sigma_p -  \sum_{p\in \Lambda_+^*} \frac{\big(\widehat{V} (\cdot/N)\ast\widehat{f}_{\ell, N}\big)_p^2}{4p^2}\\
		&\;  +\frac{1}{2N} \sum_{p, q\in \L_+^*}\widehat V((p-q)/N)\sigma_p\gamma_p\sigma_q\gamma_q+\frac1N\sum_{p\in \L_+^*}\bigg[p^2\eta_p^2 +\frac1 {2N} (\widehat V (\cdot/N) \ast \eta)_p \eta_p \bigg] \\
		& \; +E_{ \text{Bog},N }
		\end{split}
		\end{equation}
		with 
        \begin{equation}\label{eq:EbogN} \begin{split} E_{ \text{Bog},N } :=&\; \frac{1}{2}\sum_{p \in \Lambda^*_+}\Bigg[\sqrt{ p^{4} + 2p^2 \big(\widehat{V} (\cdot/N)\ast\widehat{f}_{\ell,N}\big)_p } -p^2  \\
        &\;\hspace{3cm} - \big(\widehat{V} (\cdot/N)\ast\widehat{f}_{\ell,N}\big)_p + \frac{\big(\widehat{V} (\cdot/N)\ast\widehat{f}_{\ell,N}\big)_p^2}{2p^2} \Bigg] \end{split}\end{equation}
  First, we compare $E_{\text{Bog},N}$ with its limiting value  (\ref{def:Ebog}). From (\ref{eq:EbogN}), we write $E_{\text{Bog},N} = - (1/2) \sum_{p \in \Lambda^*_+} e_{N,p}$, with 
\[ e_{N,p} =  p^2   + \big(\widehat{V} (\cdot/N)\ast\widehat{f}_{\ell,N}\big)_p - \sqrt{ p^{4} + 2p^2 \big(\widehat{V} (\cdot/N)\ast\widehat{f}_{\ell,N}\big)_p } - \frac{\big(\widehat{V} (\cdot/N)\ast\widehat{f}_{\ell,N}\big)_p^2}{2p^2} \]
Taylor expanding the square root, we easily check that $|e_{N,P}| \leq C / |p|^4$, for a constant $C > 0$, independent of $N$ and of $p$, if $|p| $ is sufficiently large w.r.t. $\big(\widehat{V} (\cdot/N)\ast\widehat{f}_{\ell,N}\big)_p$. On the other side, replacing $(\widehat{V}(./N) \ast \widehat{f}_{\ell,N})_p$ by $(\widehat{V}(./N) \ast \widehat{f}_{\ell,N})_0$ and then, using Lemma \ref{3.0.sceqlemma}, part ii), by $8\pi \frak{a}_0 $, we produce an error that can be estimated by
\[ \left| e_{N,p} - \left[ p^2 + 8\pi \frak{a}_0  - \sqrt{|p|^4 + 16 \pi \frak{a}_0  p^2} - \frac{(8\pi \frak{a}_0 )^2}{2p^2} \right] \right| \leq  \frac{C}{N|p|^3} \]
 Using this bound for  $|p| < N$ and $|e_{N,p}| \leq C / |p|^4$ for $|p| > N$, we obtain 
 \begin{equation}\label{eq:EBogconv} |E_{\text{Bog},N} - E_\text{Bog}| \leq C N^{-1} \log N \end{equation}	

Hence, let us analyse the remaining terms on the r.h.s. of \eqref{eq:lmCQJN1}. First, using the scattering equation (\ref{eq:eta-scat}) and the approximation (\ref{eq:Vfa0}), we find that
		\[\begin{split}
		&-\frac12\widehat V(0)+\frac1N\sum_{p\in \L^*_+} \bigg[p^2\eta_p^2 + \frac1 {2N} (\widehat{V}(\cdot/N) \ast \eta_\ell)_p \eta_p \bigg]\\
		&\hspace{3cm} = -\frac12\widehat V(0)-\frac1N\sum_{p\in \L_+^*} \widehat{V } (p/N) \eta_p +\mathcal{O}(N^{-1}) = -4\pi \mathfrak{a}_0 +\mathcal{O}(N^{-1})
		\end{split}\]
Moreover, using that $ |\sigma_p\gamma_p -\eta_p|\leq C/p^6$, we have 
        \[\begin{split} &-\frac{1}N\sum_{p\in\L_+^*}\big(\widehat{V} (\cdot/N)\ast{\eta}\big)_p \sigma_p\gamma_p \\
        &\hspace{0.5cm}=\;-\frac{1}{N} \sum_{p,q\in \L_+^*} \widehat V((p-q)/N) \eta_q \sigma_p \gamma_p  - \frac{1}{N} \sum_{p\in \L_+^*} \widehat V(p/N) \eta_p \eta_0 + \cO(N^{-1})
        \end{split}\]
and, writing $\sigma_q \gamma_q \sigma_p \gamma_p = (\sigma_q \gamma_q  - \eta_q + \eta_q) (\sigma_p \gamma_p - \eta_p + \eta_p)$ and expanding the product, 
        \[\begin{split} &\frac1{2N}\sum_{p,q\in\L_+^*}\widehat{V} ((p-q)/N)\sigma_q\gamma_q\sigma_p\gamma_p \\
        &\hspace{1cm}=\; \frac1{N}\sum_{p,q\in\L_+^*}\widehat{V} ((p-q)/N)\eta_q \sigma_p \gamma_p -\frac1{2N}\sum_{p,q\in\L_+^*}\widehat{V} ((p-q)/N)\eta_q\eta_p +\cO (N^{-1})
        \end{split}\]
Summing up the different contributions from above, we arrive at  
		\[%\label{eq:lmCQJN2} 
		\begin{split}
		C_{\cM_{N}} = &\;  \frac{N}2\widehat V(0) -4\pi\frak{a}_0 + E_{\text{Bog}} + \cO (N^{-1}\log N)\\
		& -\frac1{2N}\sum_{p\in\L_+^*}(\widehat{V} (./N)\ast \eta)_p \eta_p   -  \sum_{p\in \Lambda_+^*} \frac{(\widehat{V}(\cdot/N) \ast \widehat{f}_{\ell,N})_p^2}{4p^2} -\frac1{2N}\sum_{p\in\L_+^*} \widehat{V} (p/N)\eta_p \eta_0 
		\end{split}\]
Writing (recall from (\ref{eq:fellN}) that $\widehat{f}_{\ell,N} (p) = \delta_{p,0} + N^{-1} \eta_p$) 
\[ -\frac{1}{2N} \sum_{p \in \L_+^*} (\widehat{V} (./N)*\eta)_p \eta_p = -\frac{1}{2N} \sum_{p \in \L_+^*} (\widehat{V} (./N) *\widehat{f}_{\ell,N})_p \eta_p	 + \frac{1}{2} \sum_{p \in \L^*} \widehat{V} (p/N) \eta_p - \frac{1}{2} \widehat{V} (0) \eta_0 \]
and noticing that 
\[ \frac{N}{2} \widehat{V} (0) + \frac{1}{2} \sum_{p \in \L^*} \widehat{V} (p/N) \eta_p = \frac{N}{2} (\widehat{V} (./N) * \widehat{f}_{\ell,N})_0 \]
we arrive at 
\[ \begin{split} 
C_{\cM_{N}}  = \; &\frac{N}{2} (\widehat{V} (./N) * \widehat{f}_{\ell,N})_0 (1- \eta_0/N) - 4\pi \frak{a}_0 + E_\text{Bog} + \cO (N^{-1} \log N) \\ &- \sum_{p \in \L^*_+} \frac{(\widehat{V} (./N) * \widehat{f}_{\ell,N})_p}{2p^2} \left[ p^2 \eta_p + \frac{1}{2} (\widehat{V} (./N) *\widehat{f}_{\ell,N})_p \right] \end{split} 
\]
With Lemma \ref{3.0.sceqlemma}, we compute 
\[ \begin{split}  (\widehat{V} (./N) * \widehat{f}_{\ell,N})_0 &= 8\pi \frak{a}_0 \left( 1 + \frac{3\frak{a}_0}{2\ell N} \right) + \cO (N^{-2}) \\
\eta_0 &= - N^{-2} \widehat{w}_\ell (0) = - N^{-2} \int dx w_\ell (x) dx = - \frac{2}{5} \pi \frak{a}_0 \ell^2 + \cO (N^{-1}) \end{split} 
\]
Hence, with the scattering equation (\ref{eq:eta-scat0}), we obtain 
\begin{equation}\label{eq:CMN} \begin{split} 
C_{\cM_{N}}  = \; &4 \pi\frak{a}_0 (N-1) + E_\text{Bog} + 6 \pi \frak{a}_0^2 \ell^{-1} \left[ 1 + \frac{4\pi}{15} \ell^3 \right] + \cO (N^{-1} \log N) \\ &- \sum_{p \in \L^*_+} N^3 \l_\ell \widehat{\chi}_\ell (p) \frac{(\widehat{V} (./N)* \widehat{f}_{\ell,N})_p}{2p^2} \\ = \; &4 \pi\frak{a}_0 (N-1) + E_\text{Bog} + 6 \pi \frak{a}_0^2 \left[ \frac{1}{\ell} + \frac{4\pi}{15} \ell^2 - \frac{2}{\ell^{3}} \sum_{p\in\L_+^*} \frac{\widehat{\chi}_\ell (p)}{p^2} \right] + \cO (N^{-1} \log N) \end{split} \end{equation}
because $|(\widehat{V} (./N)*\widehat{f}_{\ell,N})_p - 8\pi \frak{a}_0| \leq C |p|/N$ and $|\widehat{\chi}_\ell (p)| \leq C|p|^{-2}$ by \eqref{eq:chip}. To finish the proof of i), we evaluate the expression in the parenthesis, showing that it is actually independent of $\ell$, for $\ell \in (0;1/2)$. Again by \eqref{eq:chip}, we find 
\[ -\frac{2}{\ell^{3}} \sum_{p\in \Lambda_+^*}\frac{\widehat {\chi_\ell}(p)}{p^2} = -8\pi \ell^2 \lim_{M\to \infty} \sum_{\substack{p\in \Lambda_+^*: |p_i|\leq 2\pi M  }}  \frac{ \sin(\ell|p|)- \ell|p|\cos(\ell|p|)}{(|p|\ell)^5} \]
because the sum converges absolutely, since $|\widehat{\chi}_\ell(p)| / |p|^2 \leq C/|p|^4$ for all $ p\in\Lambda^*_+$. With 
\begin{equation}\label{eq:chi2ellp} \begin{split}
%\widehat{\chi_\ell}(q)=& \int_{\bR^3} \chi_\ell (x) e^{-iq x}\; dx = 4\pi \ell^3 \bigg( \frac{\sin(\ell |q|)}{(\ell |%q|)^3} -\frac{\cos(\ell |q|)}{(\ell q)^2} \bigg),\\
\widehat{\big(\chi_\ell |\cdot|^2 \big)}(q) =& \int_{\bR^3} \chi_\ell (x) x^2 e^{-iq x}\; dx \\ = & 4\pi \ell^5 \bigg(  -\frac{6\sin(\ell |q|)}{(\ell |q|)^5}+\frac{6\cos(\ell |q|)}{(\ell q)^4} + \frac{3\sin(\ell |q|)}{(\ell |q|)^3} -  \frac{\cos(\ell |q|)}{(\ell q)^2} \bigg)\end{split}\end{equation}
and (\ref{eq:chip}), we conclude that 
\begin{equation}\label{eq:lmCQJN3}
\begin{split}
-\frac{2}{\ell^{3}} &\sum_{p\in \Lambda_+^*}\frac{\widehat {\chi_\ell}(p)}{p^2} \\ &= -8\pi \ell^2 \lim_{M\to \infty} \sum_{\substack{p\in \Lambda_+^*: |p_i|\leq 2\pi M  }}  \frac{ \sin(\ell|p|)- \ell|p|\cos(\ell|p|)}{(|p|\ell)^5} \\ &= -8\pi \ell^2 \lim_{M\to \infty} \sum_{p \in \L^*_+ : |p_i | \leq 2\pi M} \left[ \frac{1}{3} \frac{\cos (\ell |p|)}{(\ell p)^2} + \frac{1}{2} \left( \frac{\sin (\ell |p|)}{(\ell |p|)^3} - \frac{\cos (\ell |p| )}{(\ell p)^2} \right) \right. \\ & \hspace{2.5cm} \left.- \frac{1}{6} \left( - 6 \frac{\sin (\ell |p|)}{(\ell |p|)^5} + 6 \frac{\cos (\ell |p|)}{ (\ell |p|)^4} + 3 \frac{\sin (\ell |p|)}{(\ell |p|)^3} - \frac{\cos (\ell |p|)}{(\ell p)^2} \right) \right] \\ 
&= -\frac{8\pi \ell^2}{3} \lim_{M \to \infty} \sum_{p \in \L^*_+ : |p_i | \leq 2\pi M}  \frac{\cos (\ell |p|)}{(\ell p)^2}  \\ &\hspace{.5cm} -2 \ell^2 \lim_{M \to \infty} \sum_{p \in \L^*_+ : |p_i | \leq 2\pi M}  \left[ \frac{1}{2\ell^3} 
\widehat{\chi}_\ell (p) - \frac{1}{6\ell^5} \widehat{ \chi_\ell |.|^2} (p) \right] \\ 
&=  -\frac{8\pi}{3} \lim_{M \to \infty} \sum_{p \in \L^*_+ : |p_i | \leq 2\pi M}  \frac{\cos (\ell |p|)}{p^2}  - \frac{1}{\ell}  \left[ \chi_\ell (0) - \widehat{\chi}_\ell (0) \right] - \frac{1}{3\ell^3} \widehat{\chi_\ell |.|^2} (0)
\end{split} \end{equation}
because square partial sums for the Fourier series of $x \to \chi_\ell (x)$ and of $x \to x^2 \chi_\ell (x)$ converge at $x = 0$; see \cite{Ces}. With (\ref{eq:chip}) and (\ref{eq:chi2ellp}), we obtain
\[ \widehat{\chi}_\ell (0) = \frac{4\pi \ell^3}{3}  , \qquad \widehat{\chi_\ell |.|^2} (0) = \frac{4\pi  \ell^5}{5}\] 
Thus
\begin{equation}\label{eq:sumIell} - \frac{2}{\ell^{3}} \sum_{p \in \L^*_+} \frac{\widehat{\chi}_\ell (p)}{p^2} = I_\ell - \frac{1}{\ell} -\frac{4\pi \ell^2}{15} \end{equation}
where
\begin{equation}\label{eq:Ielldef}  I_\ell = \frac{4\pi \ell^2}{3} - \frac{8\pi}{3} \lim_{M \to \infty} \sum_{p \in \L^*_+ : |p_i| \leq 2\pi M} \frac{\cos (\ell |p|)}{p^2} \end{equation}
We claim now that $I_\ell$ is independent of the choice of $\ell \in (0;1/2)$. This implies that, for example,  
\[ I_\ell = I_{1/(2\pi)} = \frac{1}{3\pi} - \frac{2}{3\pi} \lim_{M \to \infty} \sum_{p \in \bZ^3 : |p_i| \leq M} \frac{\cos (|p|)}{p^2} \]
Inserting in (\ref{eq:sumIell}) and then in (\ref{eq:CMN}), we obtain  
\[ C_{\cM_N} = 4\pi \frak{a}_0 (N-1) + E_\text{Bog} + 2 \frak{a}_0^2 \left[ 1 - 2 \lim_{M \to \infty} \sum_{p \in \bZ^3 : |p_i| \leq M} \frac{\cos (|p|)}{p^2} \right] + \cO (N^{-1} \log N) \]
which concludes the proof of part i) (notice that, in particular, our analysis shows the existence of the limit $M \to \infty$ in the parenthesis). It remains to show that $I_\ell$, defined as in (\ref{eq:Ielldef}) is independent of $\ell$. To this end, we observe that 
\begin{equation}\label{eq:chi/x} \widehat{\Big(\frac{\chi_\ell}{|\cdot|}\Big)} (q) = \int_{\bR^3} \chi_\ell(x) \frac{1}{|x|}e^{-iq x}\; dx = 4\pi\bigg( \frac{1}{q^2} -\frac{\cos(\ell |q|)}{q^2} \bigg) \, .
\end{equation}
Hence, for $\ell_1, \ell_2 \in (0;1/2)$, we find 
\[ \begin{split} I_{\ell_1} - I_{\ell_2} &= \frac{4\pi}{3} (\ell_1^2 - \ell_2^2) - \frac{8 \pi}{3} \lim_{M \to \infty} \sum_{p\in \L_+^* : |p_i| \leq 2\pi M} \frac{\cos (\ell_1 |p|) - \cos (\ell_2 |p|)}{p^2} \\ &= \frac{4\pi}{3} (\ell_1^2 - \ell_2^2) -\frac{2}{3} \lim_{M \to \infty} \sum_{p\in \L_+^* : |p_i| \leq 2\pi M} \widehat{h} (p)  \end{split} \]
where $h (x) = (\chi_{\ell_2} (x) - \chi_{\ell_1} (x))/|x|$. By \cite{Ces}, we find 
\[ I_{\ell_1} - I_{\ell_2} = \frac{4\pi}{3} (\ell_1^2 - \ell_2^2) - \frac{2}{3} (h(0) - \widehat{h} (0)) =0  \]
because $h(0) = 0$ and, with (\ref{eq:chi/x}), $\widehat{h} (0) = - 2\pi (\ell_1^2 - \ell_2^2)$; hence $I_{\ell_1} = I_{\ell_2}$, as claimed. 

Finally, we prove part ii). Here, we use the two bounds
        \[\begin{split}\Big| \sqrt{p^4+2p^2\big(\widehat{V} (\cdot/N)\ast\widehat{f}_{\ell,N}\big)_p}- \sqrt{p^4+2p^2\big(\widehat{V} (\cdot/N)\ast\widehat{f}_{\ell,N}\big)_0}\Big|\leq C N^{-1}|p|
        \end{split}\]
as well as
        \[\Big|\sqrt{p^4+2p^2\big(\widehat{V} (\cdot/N)\ast\widehat{f}_{\ell,N}\big)_0}-\sqrt{p^4+16\pi \frak{a}_0  p^2}\Big|\leq C N^{-1}\]
It follows immediately that
         \[ \cQ_{\cM_N} =\sum_{p \in \Lambda^*_+}\sqrt{p^4+2p^2\big(\widehat{V} (\cdot/N)\ast\widehat{f}_{\ell,N}\big)_p} \;a_p^* a_p =\sum_{p \in \Lambda^*_+}\sqrt{p^4+16 \pi \frak{a}_0  p^2}  \;a_p^* a_p + \widetilde{\delta}_N \]
where the operator $\widetilde{\delta}_N$ is bounded by $\pm \widetilde{\delta}_N \leq C N^{-1}(\cK +1)$. 
This concludes the proof of the  lemma.
\end{proof}

Combining Proposition \ref{thm:tGNo1} with the results of the last two sections, we obtain the following corollary, which will be used in the next section to show Theorem \ref{thm:main}.
\begin{cor} 
Let $V \in L^3 (\bR^3)$ be non-negative, compactly supported and spherically symmetric. Then there exists a constant $C > 0$ such that the excitation Hamiltonian 
        \[\cM_{N} = e^{-B(\tau)} e^{-A} e^{-B(\eta)} U H_N U^* e^{B(\eta)}e^{A} e^{B(\tau)} : \cF_+^{\leq N} \to \cF_+^{\leq N}\]
can be written as 
\begin{equation}\label{eq:cor} \begin{split} \cM_{N} = &\;4\pi (N-1) \frak{a}_0+ e_\L \mathfrak{a}_0^2 + \frac{1}{2} \sum_{p \in \Lambda^*_+} \left[ - p^2  - 8\pi \frak{a}_0  + \sqrt{p^4 + 16 \pi \frak{a}_0   p^2} + \frac{(8\pi \frak{a}_0 )^2}{2p^2} \right] \\ &+ \sum_{p \in \Lambda^*_+} \sqrt{p^4 + 16 \pi \frak{a}_0  p^2} \; a_p^* a_p + \cV_N + \cE_{\cM_{N}} \end{split} \end{equation}
with $e_\Lambda$ as in (\ref{eq:eLambda}) and where 
\[ \pm \cE_{\cM_{N}} \leq C N^{-1/4}  [ (\cH_N +1) (\cN_+ + 1)  + (\cN_+ +1)^3]  \]
Furthermore, let $\psi_N \in L^2_s (\bR^{3N})$ with $\| \psi_N \|  = 1$ belong to the spectral subspace of $H_N$ with energies below $E_N + \zeta$, where $E_N$ is the ground state energy of $H_N$ and $\zeta > 0$. In other words, assume that 
\[ \psi_N = {\bf 1}_{(-\infty; E_N + \zeta]} (H_N) \psi_N \]
Let $\xi_N = e^{-B(\tau)}e^{-A} e^{-B(\eta)} U \psi_N \in \cF_+^{\leq N}$ be the excitation vector associated with $\psi_N$. Then there exists a constant $C > 0$ such that 
\begin{equation}\label{eq:HNN-bd} \langle \xi_N , [ (\cH_N + 1) (\cN_+ + 1) + (\cN_+ +1)^3 ]  \xi_N \rangle \leq C (1+\zeta^3) \end{equation}
\end{cor}        
\begin{proof}
Eq. (\ref{eq:cor}) follows from Prop. \ref{thm:tGNo1}, Lemma \ref{lm:err-con}, Lemma \ref{lm:diago} and Lemma \ref{lm:CQJNorder1}. Eq. (\ref{eq:HNN-bd}) is, on the other hand, a consequence of Corollary \ref{cor:apri}.
\end{proof}
        
%%%%%%%%%%%%%%%%%%%%%%%%%%%%%%%%%%%%%%%%%%%%%%%%%%%%%%%%%%%%%%%%%%%%%%%%%%%%%%%%%%%%%%%%%%%%%%%%%%%%%%%%%%%%%%%%%%%%%%%%%%%%%%%%%%%%%%%%%%%%%%%%%%%%%%%%%%%%%%%%%%%%%%%%%%%%%
%%%%%%%%%%%%%%%%%%%%%%%%%%%%%%%%%%%%%%%%%%%%%%%%%%%%%%%%%%%%%%%%%%%%%%%%%%%%%%%%%%%%%%%%%%%%%%%%%%%%%%%%%%%%%%%%%%%%%%%%%%%%%%%%%%%%%%%%%%%%%%%%%%%%%%%%%%%%%%%%%%%%%%%%%%%%%
%%%%%%%%%%%%%%%%%%%%%%%%%%%%%%%%%%%%%%%%%%%%%%%%%%%%%%%%%%%%%%%%%%%%%%%%%%%%%%%%%%%%%%%%%%%%%%%%%%%%%%%%%%%%%%%%%%%%%%%%%%%%%%%%%%%%%%%%%%%%%%%%%%%%%%%%%%%%%%%%%%%%%%%%%%%%%
%%%%%%%%%%%%%%%%%%%%%%%%%%%%%%%%%%%%%%%%%%%%%%%%%%%%%%%%%%%%%%%%%%%%%%%%%%%%%%%%%%%%%%%%%%%%%%%%%%%%%%%%%%%%%%%%%%%%%%%%%%%%%%%%%%%%%%%%%%%%%%%%%%%%%%%%%%%%%%%%%%%%%%%%%%%%%
%%%%%%%%%%%%%%%%%%%%%%%%%%%%%%%%%%%%%%%%%%%%%%%%%%%%%%%%%%%%%%%%%%%%%%%%%%%%%%%%%%%%%%%%%%%%%%%%%%%%%%%%%%%%%%%%%%%%%%%%%%%%%%%%%%%%%%%%%%%%%%%%%%%%%%%%%%%%%%%%%%%%%%%%%%%%%
%%%%%%%%%%%%%%%%%%%%%%%%%%%%%%%%%%%%%%%%%%%%%%%%%%%%%%%%%%%%%%%%%%%%%%%%%%%%%%%%%%%%%%%%%%%%%%%%%%%%%%%%%%%%%%%%%%%%%%%%%%%%%%%%%%%%%%%%%%%%%%%%%%%%%%%%%%%%%%%%%%%%%%%%%%%%%

\section{Proof of Theorem \ref{thm:main}}
\label{sec:proof}

We define 
\[ E_{\cM_N} := 4\pi (N-1) \frak{a}_0 + e_\Lambda \frak{a}_0^2+  \frac{1}{2} \sum_{p \in \Lambda^*_+} \left[ - p^2  - 8\pi \frak{a}_0  + \sqrt{p^4 + 16 \pi \frak{a}_0  p^2} + \frac{(8\pi \frak{a}_0 )^2}{2p^2}  \right] \]
To prove Theorem \ref{thm:main}, we compare the eigenvalues of $\cM_N - E_{\cM_N}$ below a threshold $\zeta> 0$ with those of the diagonal quadratic operator 
\begin{equation}\label{eq:cDdef} \cD := \sum_{p \in \Lambda^*_+} \eps_p  a_p^* a_p \end{equation}
with the dispersion $\eps_p = (|p|^4 +16 \pi \frak{a}_0  p^2)^{1/2}$ for all $p \in \Lambda^*_+$. For $m \in \bN$, we denote by $\lambda_m$ the $m$-th eigenvalue of $\cM_N - E_{\cM_N}$ and by $\nu_m$ the $m$-th eigenvalue of $\cD$ (in both cases, eigenvalues are counted with multiplicity). To show Theorem \ref{thm:main}, we prove that
\begin{equation}\label{eq:lambdas} |\lambda_m - \nu_m | \leq C N^{-1/4} (1+\zeta^3)  \end{equation}
for all $m \in \bN \backslash \{ 0 \}$ such that $\lambda_m < \zeta$. Using (\ref{eq:lambdas}), Theorem \ref{thm:main} can be proven as follows. Taking the expectation of (\ref{eq:cor}) in the vacuum, we conclude that $\lambda_1 \leq C N^{-1/4}$. Hence, for $N$ large enough, we have $\lambda_1 \leq \zeta$ and we can apply (\ref{eq:lambdas}) to show that $|\lambda_1 - \nu_1| \leq C N^{-1/4}$. Since $\nu_1 = 0$, we conclude that $|\lambda_1| \leq C N^{-1/4}$ and therefore that \begin{equation} \label{eq:ENENJ} |E_N - E_{\cM_N}| \leq C N^{-1/4} \end{equation} where $E_N$ is the ground state energy of $H_N$, as defined in (\ref{eq:Ham0}). This proves (\ref{1.groundstate}). Eq.~(\ref{1.excitationSpectrum}), on the other hand, follows from (\ref{eq:lambdas}), from (\ref{eq:ENENJ}) and from the observation that the eigenvalues of $\cD$ have the form  
\[ \nu_j = \sum_{p \in \Lambda^*_+}  n_p^{(j)} \eps_p  \]
for every $j \in \bN \backslash \{ 0 \}$. Here the coefficients $n_p^{(j)} \in \bN$, for all $j \in \bN$ and all $p \in \Lambda^*_+$. Notice that the eigenvector of $\cD$ associated with the eigenvalue $\nu_j$ is given 
by 
\begin{equation}\label{eq:xij} \xi_j = C_j \prod_{p \in \Lambda^*_+} (a^*_p)^{n_p^{(j)}} \Omega \end{equation}  
for an appropriate normalization constant $C_j > 0$ (if $\nu_j$ is degenerate, the choice of $\xi_j$ is not unique; we will always use eigenvectors of the form (\ref{eq:xij})). 

To show (\ref{eq:lambdas}), we will combine a lower and an upper bound for $\lambda_m$ in terms of $\nu_m$. Since $\cV_N \geq 0$, we can ignore the potential energy operator appearing on the r.h.s. of (\ref{eq:cor}) when proving the lower bound. For the upper bound, on the other hand, we make use of the following lemma, where we control the expectation of $\cV_N$ on low-energy eigenspaces of the quadratic operator $\cD$. 
\begin{lemma}\label{lm:potbnd}
Let $V \in L^3 (\bR^3)$ be non-negative, compactly supported and spherically symmetric and let $\cV_N$ be defined as in (\ref{eq:KcVN}). Let $\zeta >0$ and $m \in \bN$ such that $\nu_m < \zeta$. Let $\xi_1, \dots , \xi_m$ be defined as in (\ref{eq:xij}) ($\xi_j$ is an eigenvector of $\cD$ associated with the eigenvalue $\nu_j$) and $Y^m_{\cD}$ be the subspace spanned by $\xi_1, \dots , \xi_m$. Then there exists $C>0$ such that 
\[ \langle \xi, \cV_N \xi \rangle \leq  \frac{C (\zeta + 1)^{7/2}}{N}\]
for all normalized $\xi \in Y_{\cD}^m$.
\end{lemma} 
\begin{proof}
The bounds $\eps_p \geq p^2$ and $\nu_1 \leq \dots \leq \nu_m  \leq \zeta$ imply that $a_q \xi_j = 0$ for all $q \in \Lambda^*_+$ with $|q| > \zeta^{1/2}$. This also implies that $a_q \xi =0$ for all $\xi \in Y^m_{\cD}$.  Hence
\[ \begin{split} \langle \xi, \cV_N \xi \rangle &\leq \frac{1}{N} \sum_{p,q,u \in \Lambda_+^* }  |\widehat{V} (u/N)| \| a_{q+u} a_p \xi \| \| a_{p+u} a_q \xi \| \\ &\leq \frac{C}{N}  \sum_{p,q,u \in \Lambda_+^* : |p| , |q| , |u| \leq C \zeta^{1/2} }  \| a_{q+u} a_p \xi \| \| a_{p+u} a_q \xi \|  \leq \frac{C \zeta^{3/2}}{N} \| (\cN_+ + 1) \xi \|^2   \end{split} \]
Since $\cN_+ \leq C \cD$, we find 
\[   \langle \xi, \cV_N \xi \rangle \leq 
 \frac{C \zeta^{3/2}}{N} \| (\cD+1) \xi \|^2 \leq \frac{C (\zeta + 1)^{7/2}}{N} \]
 \end{proof}
 
In addition to Lemma \ref{lm:potbnd}, we will need the following result which is an extension of Lemma 7.3 in \cite{BBCS2} to the Gross-Pitaevskii regime.
\begin{lemma}\label{lm:cVN}
Let $V \in L^3 (\bR^3)$ be non-negative, compactly supported and spherically symmetric, let $\cK, \cV_N$ be defined as in (\ref{eq:KcVN}). Then there exists $C > 0$ such that, on $\cF^{\leq N}_+$, 
\[ \cV_N \leq C \cN_+  \cK  \, . \]
\end{lemma}
\begin{proof}
We bound
\[ \begin{split} 
\langle \xi, \cV_N \xi \rangle 
\leq \; & \frac{1}{N} \sum_{p,q \in \Lambda^*_+ , u \in \Lambda^* : u \not=-p,-q} |\widehat{V} (u/N)| \|  a_{p+u} a_q \xi \| \|  a_{q+u} a_p \xi \| \\
\leq \; &\frac{1}{N} \sum_{p,q \in \Lambda^*_+ , u \in \Lambda^* : u \not=-p,-q} \frac{|\widehat{V} (u/N)|}{(q+u)^2}  (p+u)^2 \|  a_{p+u} a_q \xi \|^2 \\ \leq \; & \left[ \sup_{q \in \Lambda^*_+} \frac{1}{N} \sum_{u \in \Lambda^* : u \not = -q} \frac{|\widehat{V} (u/N)|}{(u+q)^2}  \right] \| \cK^{1/2} \cN^{1/2} \xi \|^2  \leq C \| \cK^{1/2} \cN^{1/2} \xi \|^2
\end{split}
\]
\end{proof}

With the help of Lemma \ref{lm:potbnd} and Lemma \ref{lm:cVN}, we are now ready to prove (\ref{eq:lambdas}). 

Let us first prove a lower bound for $\lambda_m$, under the assumption that $\lambda_m < \zeta$. 
From the min-max principle, we have 
\[  \lambda_m  = \inf_{\substack{Y \subset 
\cF_+^{\leq N} :\\ \text{dim } Y = m}} \, \sup_{\substack{\xi \in Y : \\ \| \xi \| =1}} \, \langle \xi, (\cM_N - E_{\cM_N}) \xi 
\rangle   \]
From the assumption $\lambda_m < \zeta$ we obtain 
\[  \lambda_m  = \inf_{\substack{Y \subset 
P_\zeta (\cF_+^{\leq N}) :\\ \text{dim } Y = m}} \, \sup_{\substack{\xi \in Y : \\ \| \xi \| =1}} \, \langle \xi, (\cM_N - E_{\cM_N}) \xi 
\rangle   \]
where $P_\zeta$ is the spectral projection of $\cM_N - E_{\cM_N}$ associated with the interval $(-\infty; \zeta]$. Hence, with (\ref{eq:cor}), $\cV_N \geq 0$ and (\ref{eq:HNN-bd}) we find  
\[ \begin{split} \lambda_m  &\geq \inf_{\substack{Y \subset 
P_\zeta (\cF_+^{\leq N}) :\\ \text{dim } Y = m}} 
\, \sup_{\substack{\xi \in Y : \\ \| \xi \| =1}} \langle \xi, \cD \xi \rangle  - C N^{-1/4} (1+\zeta^3) \\ &\geq \inf_{\substack{Y \subset 
\cF_+^{\leq N} :\\ \text{dim } Y = m}} \,
 \sup_{\substack{\xi \in Y : \\ \| \xi \| =1}} \langle \xi, \cD \xi \rangle  - C N^{-1/4} (1+\zeta^3) = \nu_m - C N^{-1/4} (1+\zeta^3) 
\end{split} \]

Let us now prove an upper bound for $\lambda_m$. From the assumption $\lambda_m < \zeta$ and from the lower bound proven above, it follows that $\nu_m \leq \zeta + 1$ (without loss of generality, we can assume $N^{-1/4} \zeta^3 \leq 1$, since otherwise the statement of the theorem is trivially satisfied). The min-max principle implies that 
\begin{equation}\label{eq:minmax-up} \lambda_m = \inf_{\substack{Y \subset \cF_+^{\leq N} : \\ \text{dim }Y=m} } \; \sup_{\substack{\xi \in Y : \\ \| \xi \| = 1}} \langle \xi , (\cM_N - E_{\cM_N}) \xi \rangle \leq \sup_{\substack{\xi \in Y_{\cD}^m :\\ \| \xi \| = 1}} \langle \xi, (\cM_N - E_{\cM_N}) \xi \rangle  
\end{equation}
where $Y_{\cD}^m$ denotes the subspace spanned by the $m$ vectors $\xi_1, \dots , \xi_m$ defined in (\ref{eq:xij}). From Lemma \ref{lm:cVN} and the inequalities $\cN \leq C \cK \leq C \cD \leq C \nu_m \leq C (\zeta + 1)$ on $Y_{\cD}^m$, we find that
\[ \begin{split} 
\langle \xi , [ (\cH_N +1)(\cN_+ + 1) + (\cN_+ +1)^3] \xi \rangle &\leq C \langle \xi, (\cN_+ +1)^2 (\cK+1) \xi \rangle \leq C (1 + \zeta^3) \end{split} \]
for all normalized $\xi \in Y_\cD^m$. 
Inserting the last inequality and the bound from Lemma \ref{lm:potbnd} in (\ref{eq:cor}), we obtain that
\[ \langle \xi, (\cM_N -  E_{\cM_N}) \xi \rangle \leq  \langle \xi, \cD \xi \rangle + C N^{-1/4} (1+\zeta^3) \]
for all $\xi \in Y_{\cD}^m$. From (\ref{eq:minmax-up}), we conclude that 
\[ \lambda_m \leq  \sup_{\xi \in Y_\cD^m : \| \xi \| =1} \langle \xi , \cD \xi \rangle + C N^{-1/4} (1+\zeta^3) \leq \nu_m + C N^{-1/4} (1+\zeta^3) \]

Combining lower and upper bound, we showed that $|\lambda_m - \nu_m| \leq CN^{-1/4} (1+\zeta^3)$, for all $m \in \bN$ such that $\lambda_m < \zeta$. This completes the proof of Theorem \ref{thm:main}. 

\medskip

To conclude this section, we come back to the remark after Theorem \ref{thm:main}, concerning the eigenvectors of the Hamilton operator $H_N$ introduced in (\ref{eq:Ham0}). Theorem \ref{thm:main} shows that the eigenvalues of $H_N$ can be approximated in terms of the eigenvalues of the diagonal quadratic operator $\cD$ defined in (\ref{eq:cDdef}). Following standard arguments one can also approximate the eigenvectors of $H_N$ through the (appropriately transformed) eigenvectors of $\cD$. More precisely, let $\theta_1 \leq \theta_2 \leq \dots $ denote the ordered eigenvalues of $H_N$ (i.e. $\theta_j = \lambda_j + E_{\cM_N}$, with the notation introduced after (\ref{eq:cDdef})) and let 
$0 = \nu_1 \leq \nu_2 \leq \dots$ denote the eigenvalues of the diagonal quadratic operator $\cD$ defined in (\ref{eq:cDdef}). Fix $j \in \bN \backslash \{0 \}$ with $\nu_j < \nu_{j+1}$. From (\ref{eq:lambdas}) we obtain that also $\theta_j < \theta_{j+1}$, if $N$ is large enough. We denote by $P_j$ the spectral projection onto the eigenspace of $H_N$ associated with the eigenvalues $\theta_1 \leq \dots \leq \theta_j$ and by $Q_j$ the orthogonal projection onto the eigenspace of $\cD$ associated with the eigenvalues $0 = \nu_1 \leq \dots \leq \nu_j$. Then, we find 
 \begin{equation}\label{eq:eigvec} \Big\| e^{-B (\tau)} e^{-A} e^{-B(\eta)} U_N P_j U_N^* e^{B(\eta)} e^{A}e^{B(\tau)} - Q_j \Big\|_\text{HS}^2 \leq \frac{C (j+1) (1 + \nu^3_j)}{\nu_{j+1} - \nu_j} N^{-1/4} \end{equation}

In particular, if $\psi_N$ denotes a ground state vector of the Hamiltonian $H_N$, there exists a phase $\omega \in [0;2\pi)$ such that 
\begin{equation}\label{eq:gs-appro} \big\| \psi_N - e^{i\omega} U_N^* e^{B(\eta)}e^{A} e^{B(\tau)} \Omega \big\|^2 \leq \frac{C}{\theta_1 - \theta_0} N^{-1/4} \end{equation}
The proof of (\ref{eq:eigvec}) and (\ref{eq:gs-appro}) can be obtained, using the results of Theorem \ref{thm:main}, analogously as in \cite[Section7]{GS}. We omit the details.

%%%%%%%%%%%%%%%%%%%%%%%%%%%%%%%%%%%%%%%%%%%%%%%%%%%%%%%%%%%%%%%%%%%%

\section{Analysis of $ \cG_N$} \label{sec:GN}

In this section, we prove Proposition \ref{prop:gene}, devoted to the properties of the excitation Hamiltonian $\cG_N$ defined in (\ref{eq:GN}). In particular, we will show part b) of Prop.
\ref{prop:gene}, since part a) was proven already in \cite[Prop. 3.2]{BBCS1}. In fact, the bound 
(\ref{eq:Delta-bd}) is a bit more precise than the estimate appearing in \cite[Prop. 3.2]{BBCS1} but it can be easily obtained, combining the results of Prop. 4.2, Prop. 4.3, Prop. 4.4, 
Prop. 4.5 and Prop.4.7 in \cite{BBCS1}, taking always $\kappa =1$. Notice that, in these propositions from \cite{BBCS1}, the assumption that $\kappa > 0$ is sufficiently small is only used to guarantee that $\| \eta \|$ is small enough, so that we can apply the expansion from \cite[Lemma 2.6]{BBCS1} (which corresponds to (\ref{eq:conv-serie}) above). Here, we do not assume that the size of the potential is small, but nevertheless we make sure that $\| \eta \|$ is small enough by requiring that $\ell \in (0;1/2)$ is sufficiently small (because of the bound (\ref{eq:ellbdeta})).  

As for the bound (\ref{eq:adkN}), it was not explicitly shown in \cite{BBCS1}; however, it follows from the analysis in \cite{BBCS1} by noticing that the commutator $[i\cN_+ , \Delta_N]$ is given by the sum of the same monomials in creation and annihilation operators contributing to $\Delta_N$, multiplied with a constant $\lambda $ (given by the difference between the number of creation and the number of annihilation operators in the monomial). To be more precise, it follows from \cite{BBCS1} that $\Delta_N$ can be written as a sum $\Delta_N=\sum_{k=0}^\infty \Delta_N^{(k)}$ where the errors $\Delta_N^{(k)}, k\in\mathbb N$, are sums of monomials of creation and annihilation operators that satisfy $\pm\Delta_N^{(k)}\leq (C\| \eta \|)^k (\cH_N+1)$ for some constant $C>0$, independent of $N$. Moreover, the commutator of a given monomial in $\Delta_N^{(k)}$ with $\cN_+$ is given by the same monomial, multiplied by some constant $\lambda^{(k)}$ which is bounded by $ |\lambda^{(k)}|\leq (2k+1)\leq C^k$ (if, w.l.o.g., the constant $C$ is sufficiently large). Hence, terms in $[i \cN_+ , \Delta_N]$, and analogously in higher commutators of $\Delta_N$ with $i\cN_+$, can be estimated exactly like terms in $\Delta_N$ (up to an unimportant additional constant), leading to (\ref{eq:adkN}). From now on, we will therefore focus on part b) of Proposition \ref{prop:gene}.

Using (\ref{eq:cLN}), we write 
\begin{equation}\label{eq:cGN-dec} \cG_N = \cG_N^{(0)} + \cG_N^{(2)} + \cG_N^{(3)} + \cG_N^{(4)},\end{equation} with $\cG_N^{(j)} = e^{-B(\eta)} \cL_N^{(j)} e^{B(\eta)}$, for $j=0,2,3,4$. In the rest of this section, we will compute the operators on the r.h.s. of (\ref{eq:cGN-dec}), up to errors that are negligible in the limit of large $N$ (on low-energy states). To quickly discard some of the error terms, it will be useful to have a rough estimate on the action of the Bogoliubov transformation $e^{-B(\eta)}$; this is the goal of the next lemma. 
\begin{lemma}\label{lem:prelmG21} 
Let $B(\eta)$ be defined as in (\ref{eq:defB}), with $\eta$ as in (\ref{eq:ktdef}). 
Let $V \in L^3 (\bR^3)$ be non-negative, compactly supported and spherically symmetric. Let $\cK, \cV_N$ be defined as in (\ref{eq:KcVN}). Then for every $j \in \bN$ there exists a constant $C>0$ such that 
\[\begin{split}  e^{-B(\eta)}\cK(\cN_++1)^je^{B(\eta)} &\leq C \cK (\cN_++1)^j + C N (\cN_++1)^{j+1} \\ 
e^{-B(\eta)}\cV_N(\cN_++1)^je^{B(\eta)} &\leq C \cV_N (\cN_++1)^j + C N(\cN_++1)^j
\end{split}\]
\end{lemma}
\begin{proof}
To prove the bound for the kinetic energy operator, we apply Gronwall's inequality. For $\xi\in\cF_+^{\leq N}$ and $s\in\bR$ we define $\Phi_s = e^{s B(\eta)} \xi$ and we consider  
\begin{equation}\label{eq:GronwallKNj}
\begin{split} 
&\partial_s \langle\Phi_s, \cK(\cN_++1)^j \Phi_s \rangle\\ &=\langle\Phi_s , [\cK(\cN_++1)^j,B(\eta)] \Phi_s \rangle\\
&=\langle \Phi_s , [\cK,B(\eta)](\cN_++1)^j \Phi_s \rangle+\langle \Phi_s , \cK[(\cN_++1)^j,B(\eta)] \Phi_s \rangle 
\end{split}
\end{equation}
With 
\[
  [\cK, B(\eta)]=\sum_{p\in\L^*_+}p^2\eta_p\left(b_pb_{-p}+b^*_pb^*_{-p}\right)
 \]
the first term on the r.h.s. of (\ref{eq:GronwallKNj}) can be bounded by 
\begin{equation}\label{eq:KB}
\begin{split} 
\Big|  \langle\Phi_s , & [\cK,B(\eta)](\cN_++1)^j \Phi_s \rangle \Big| \\ \leq  & \; C \sum_{p\in\L^*_+} p^2|\eta_p| \| b_p (\cN_+ + 1)^{j/2} \Phi_s \| \| (\cN_+ + 1)^{(j+1)/2} \Phi_s \| \\ \leq \; &C \langle \Phi_s , \cK (\cN_+ +1)^j \Phi_s \rangle + C N \langle \xi, (\cN_+ + 1)^{j+1}  \xi \rangle 
\end{split} \end{equation}
Here, we used Cauchy-Schwarz, the estimate (\ref{eq:etapN}) and Lemma \ref{lm:Ngrow} to replace, in the second term on the r.h.s., $\Phi_s$ by $\xi$.  As for the second term on the r.h.s. of (\ref{eq:GronwallKNj}), we have 
\begin{equation*}
\begin{split} 
\langle\Phi_s , &\cK[(\cN_++1)^j,B(\eta)] \Phi_s \rangle\\
&=\sum_{k=0}^{j-1}\sum_{p\in\L^*_+}\eta_p\langle\Phi_s , \cK(\cN_++1)^{j-k-1}(b^*_pb^*_{-p}+b_pb_{-p})(\cN_++1)^k \Phi_s \rangle
\end{split}
\end{equation*}
Writing $\cK = \sum_{q \in \L_+^*} q^2 a_q^* a_q$ and normal ordering field operators, we arrive at
\begin{equation*}
\begin{split} 
\big| &\langle\Phi_s , \cK[(\cN_++1)^j,B(\eta)] \Phi_s \rangle \big |\\
&\leq C \sum_{p,q\in\L^*_+}q^2|\eta_p|\| a_qa_p(\cN_++1)^{(j-1)/2} \Phi_s \|\| a_q(\cN_++1)^{j/2} \Phi_s \|\\
&\hspace{.4cm} +C \sum_{p\in\L^*_+}p^2|\eta_p|\| a_p(\cN_++1)^{(j-1)/2} \Phi_s \|\| (\cN_++1)^{j/2} \Phi_s \|\\
 &\leq C \langle\Phi_s , \cK(\cN_++1)^{j} \Phi_s \rangle 
 + C N\langle\xi,(\cN_++1)^{j} \xi\rangle. \end{split}
\end{equation*}
Inserting the last bound and (\ref{eq:KB}) into the r.h.s. of (\ref{eq:GronwallKNj}) and applying Gronwall, we obtain the bound for the kinetic energy operator. 

To show the estimate for the potential energy operator, we proceed similarly. Using again the notation $\Phi_s = e^{sB(\eta)} \xi$, we compute 
\begin{equation}\label{def:prlmG4D1D2}
\begin{split} 
\partial_s \langle \langle \Phi_s , \cV_N &(\cN_+ +1)^j \Phi_s \rangle \\ &= 
 \langle \Phi_s, \big[\cV_N , B(\eta)\big](\cN_+ +1)^j \Phi_s \rangle +\langle \Phi_s , \cV_N \big[(\cN_+ +1)^j, B(\eta)\big] \Phi_s \rangle \end{split}\end{equation}
Using the identity 
  \[\begin{split} \big[\cV_N , B(\eta)\big] = \; &\frac{1}{2N} \sum_{q \in \Lambda_+^*, r \in \Lambda^* : r \not = -q} \widehat{V} (r/N) \eta_{q+r}  b^*_q b^*_{-q}  \\
        & +\frac{1}{N} \sum_{p,q \in \Lambda_+^* , r \in \Lambda^* : r \not = p,-q} \widehat{V} (r/N) \eta_{q+r}  b_{p+r}^* b_q^* a^*_{-q-r} a_p  + \text{h.c.}   
        \end{split}\]
and switching to position space, we can bound the expectation of the first term on the r.h.s. of (\ref{def:prlmG4D1D2}) by 
  \begin{equation*} \begin{split} \big|  \langle \Phi_s, &\big[\cV_N , B(\eta)\big](\cN_+ +1)^j \Phi_s \rangle \big| \\ \leq  \; & 
   \bigg| \frac{1}{2}\int_{\Lambda^2}dx dy\, N^2V(N(x-y))\check{\eta}(x-y) \big\langle \Phi_s ,  \check{b}^*_x \check{b}^*_y (\cN_+ +1)^j \Phi_s \big\rangle \bigg| \\ &+ \bigg| \int_{\Lambda^2}dx dy\, N^2V(N(x-y))\big\langle \Phi_s , \check{b}^*_x \check{b}^*_y a^*(\check{\eta}_y) \check{a}_x (\cN_+ +1)^j \Phi_s \big\rangle \bigg| \\ \leq \; &     C \int_{\Lambda^2}dx dy\, N^3V(N(x-y))\big\|(\cN_++1)^{j/2}\check{b}_x\check{b}_y\Phi_s \big\| \big\|(\cN_++1)^{j/2}\Phi_s \big\| \\ 
     &+ C \int_{\Lambda^2}dx dy\, N^2V(N(x-y)) \| \check{\eta}_y \|_2 \, \big\|(\cN_++1)^{j/2}\check{b}_x\check{b}_y\Phi_s \big\| \big\| \check{a}_x(\cN_++1)^{(j+1)/2}\Phi_s \big\|\\ 
    \leq \; & C \langle\Phi_s , \cV_N(\cN_++1)^j\Phi_s \rangle + C N \langle\xi, (\cN_++1)^j\xi \rangle
\end{split}\end{equation*}
where we used Cauchy-Schwarz, the bound $\| \check{\eta}_y \|_2 \leq C$, the fact that $\cN_+ \leq N$ on $\cF^{\leq N}_+$ and, in the last step, Lemma \ref{lm:Ngrow} to replace $\Phi_s$ with $\xi$ in the second term. As for the second term on the r.h.s. of (\ref{def:prlmG4D1D2}), it can be controlled similarly, using the identity 
\[\begin{split}\cV_N \big[(\cN_+ +1)^j, B(\eta)\big]  =&    \sum_{k=1}^j  \sum_{p \in \L^*_+} \eta_p (\cN_+ +1)^{j-k-1} \cV_N (b_p b_{-p} + b_p^* b_{-p}^*) (\cN_+ +1)^k\end{split} \]
and expressing $\cV_N$ in position space. We conclude that 
\[  \big| \partial_s \langle \langle \Phi_s , \cV_N (\cN_+ +1)^j \Phi_s \rangle \big|  \leq C \langle \Phi_s , \cV_N(\cN_++1)^j\Phi_s \rangle + C N \langle\xi, (\cN_++1)^j\xi \rangle \]
Gronwall's lemma gives the desired bound.
\end{proof}

%%%%%%%%%%%%%%%%%%%%%%%%%%%%%%%%%%%%%%%%%%%%%%%%%%%%%%%%%%%%%%%%%%%%%%%%%%%%%%%%
%%%%%%%%%%%%%%%%%%%%%%%%%%%%%%%%%%%%%%%%%%%%%%%%%%%%%%%%%%%%%%%%%%%%%%%%%%%%%%%%%%%%%%%%%%%%%%%%%%%%%%%%%%%%%%%%%%%%%%%%%%%%%%%%%%%%%%%%%%%%%%%%%%%%%%%%%%%%%%%%%%%%%%%%%%
%%%%%%%%%%%%%%%%%%%%%%%%%%%%%%%%%%%%%%%%%%%%%%%%%%%%%%%%%%%%%%%%%%%%%%%%%%%%%%%%%%%%%%%%%%%%%%%%%%%%%%%%%%%%%%%%%%%%%%%%%%%%%%%%%%%%%%%%%%%%%%%%%%%%%%%%%%%%%%%%%%%%%%%%%%
\subsection{Analysis of $ \cG_N^{(0)} = e^{-B(\eta)} \cL_{N}^{(0)} e^{B(\eta)}$}

From (\ref{eq:cLNj}), recall that 
\[\cL_{N}^{(0)} =\; \frac{N-1}{2N} \widehat{V} (0) (N-\cN_+ ) + \frac{\widehat{V} (0)}{2N} \cN_+  (N-\cN_+ ) \]
With Lemma \ref{lm:Ngrow}, we immediately obtain that 
\[  \cG_{N}^{(0)} = \frac{(N-1)}{2} \widehat{V} (0) + \cE_{N}^{(0)} \]
where the error operator $\cE_N^{(0)}$ is such that, on $\cF_+^{\leq N}$, 
\begin{equation}\label{eq:GN0} \begin{split} \pm \cE_{N}^{(0)} &\leq \frac{C}{N} (\cN_+ +1)^2  \end{split} \end{equation}

%%%%%%%%%%%%%%%%%%%%%%%%%%%%%%%%%%%%%%%%%%%%%%%%%%%%%%%%%%%%%%%%%%%%%%%%%%%%%%%%%%%%%%%%%%%%%%%%%%%%%%%%%%%%%%%%%%%%%%%%%%%%%%%%%%%%%%%%%%%%%%%%%%%%%%%%%%%%%%%%%%%%%%%%%%
%%%%%%%%%%%%%%%%%%%%%%%%%%%%%%%%%%%%%%%%%%%%%%%%%%%%%%%%%%%%%%%%%%%%%%%%%%%%%%%%%%%%%%%%%%%%%%%%%%%%%%%%%%%%%%%%%%%%%%%%%%%%%%%%%%%%%%%%%%%%%%%%%%%%%%%%%%%%%%%%%%%%%%%%%%
%%%%%%%%%%%%%%%%%%%%%%%%%%%%%%%%%%%%%%%%%%%%%%%%%%%%%%%%%%%%%%%%%%%%%%%%%%%%%%%%%%%%%%%%%%%%%%%%%%%%%%%%%%%%%%%%%%%%%%%%%%%%%%%%%%%%%%%%%%%%%%%%%%%%%%%%%%%%%%%%%%%%%%%%%%
%%%%%%%%%%%%%%%%%%%%%%%%%%%%%%%%%%%%%%%%%%%%%%%%%%%%%%%%%%%%%%%%%%%%%%%%%%%%%%%%%%%%%%%%%%%%%%%%%%%%%%%%%%%%%%%%%%%%%%%%%%%%%%%%%%%%%%%%%%x%%%%%%%%%%%%%%%%%%%%%%%%%%%%%%%%
%%%%%%%%%%%%%%%%%%%%%%%%%%%%%%%%%%%%%%%%%%%%%%%%%%%%%%%%%%%%%%%%%%%%%%%%%%%%%%%%%%%%%%%%%%%%%%%%%%%%%%%%%%%%%%%%%%%%%%%%%%%%%%%%%%%%%%%%%%%%%%%%%%%%%%%%%%%%%%%%%%%%%%%%%%
%%%%%%%%%%%%%%%%%%%%%%%%%%%%%%%%%%%%%%%%%%%%%%%%%%%%%%%%%%%%%%%%%%%%%%%%%%%%%%%%%%%%%%%%%%%%%%%%%%%%%%%%%%%%%%%%%%%%%%%%%%%%%%%%%%%%%%%%%%%%%%%%%%%%%%%%%%%%%%%%%%%%%%%%%%

\subsection{Analysis of $\cG_N^{(2)}= e^{-B(\eta)} \cL^{(2)}_N e^{B(\eta)} $}
\label{sec:GN2}

We define the error operator $\cE_N^{(2)}$ by the identity  
\begin{equation}\label{eq:G2-sec} \cG_N^{(2)}  = \cG_N^{(2,K)} + \cG_N^{(2,V)} + \cE_N^{(2)}  \end{equation}
where we set 
\begin{equation}\label{eq:G_N2K} \begin{split} 
\cG_N^{(2,K)} = \; &\cK + \sum_{p \in \Lambda^*_+} \Big[ p^2 \sigma_p^2\left(1+\frac{1}{N}-\frac{\cN_+}{N}\right) + p^2 \sigma_p \gamma_p  \big( b_p b_{-p} + b_p^* b_{-p}^* \big) + 2 p^2 \sigma_p^2 b_p^* b_p \Big]\\
&+\sum_{p\in\L^*_+}\frac{1}{N}p^2\s_p^2\sum_{q\in\L^*_+}\left[\left(\gamma_q^2+\s_q^2\right)b^*_qb_q+\s_q^2\right]\\
&+\sum_{p\in\L^*_+}\frac{1}{N}p^2\s_p^2\sum_{q\in\L^*_+}\left(\gamma_q\s_qb_{-q}b_q+\text{h.c.}\right)\\
&+\sum_{p\in\L^*_+}\left[p^2\eta_p b_{-p}d_p+\text{h.c.}\right]
\end{split} \end{equation}
and $\cG_N^{(2,V)}$ is defined as in
\begin{equation}\label{eq:G_N2V} \begin{split}
\cG_N^{(2,V)}= \; &\sum_{p \in \Lambda^*_+} \left[ \widehat{V} (p/N) \sigma_p^2 +  \widehat{V} (p/N) \sigma_p \gamma_p\left(1-\frac{\cN_+}{N}\right) \right]  \\
&+ \sum_{p \in \Lambda^*_+}  \widehat{V} (p/N) (\gamma_p + \sigma_p)^2 b_p^* b_p \\ &+ \frac{1}{2} \sum_{p \in \Lambda^*_+}  \widehat{V} (p/N) (\gamma_p+\sigma_p)^2 (b_pb_{-p} + b_p^* b_{-p}^*)\\
&+\sum_{p \in \Lambda^*_+}\left[ \frac{1}{2}\widehat{V} (p/N)\left(\g_pb_p+\s_pb^*_{-p}\right)d_p+\frac{1}{2}\widehat{V} (p/N)d_p\left(\g_pb_p+\s_pb^*_{-p}\right)\right]+\text{h.c.}
\end{split} \end{equation}
The goal of this subsection consists in proving the following lemma, where we bound the error 
term ${\cE}_{N}^{(2)}$. 
\begin{lemma} \label{prop:G2} 
Let $\cE_N^{(2)}$ be as defined in (\ref{eq:G2-sec}). Then, under the same assumptions as in Proposition \ref{prop:gene}, we find $C > 0$ such that  
\begin{equation}\label{eq:propG2} \pm \cE_N^{(2)} \leq C N^{-1/2} (\cK+\cN_+^2 + 1) (\cN_+ + 1) \end{equation}
\end{lemma}

\begin{proof}
From (\ref{eq:cLNj}), we have $\cL^{(2)}_N = \cK + \cL_N^{(2,V)}$, with 
\begin{equation}\label{eq:L2VN} \cL^{(2,V)}_N =  \sum_{p \in \L^*_+} \widehat{V} (p/N) \left[ b_p^* b_p - \frac{1}{N} a_p^* a_p \right] + \frac{1}{2} \sum_{p \in \L^*_+} \widehat{V} (p/N) \left[ b_p^* b_{-p}^* + b_p b_{-p} \right] \end{equation}
We consider first the contribution of the kinetic energy operator $\cK$. We write 
\[ \begin{split} \cK = \; &\frac{N-1}{N} \sum_{p \in \L^*_+} p^2 b^*_p b_p +\sum_{p \in \L^*_+} p^2 b^*_p b_p \frac{\cN_+}{N} + \cK \, \frac{(\cN_+-1)^2}{N^2} \end{split} \]
Writing $\cN_+ = \sum_{p \in \L^*_+}  \big( b_p^* b_p + N^{-1}\, a^*_p \cN_+ a_p\big)$ in the second term, we find 
\begin{equation}\label{eq:dec-BKB} e^{-B(\eta)}  \cK e^{B(\eta)} = \sum_{p \in \L^*_+}  p^2 e^{-B(\eta)} b_p^* b_p e^{ B(\eta)} + \frac{1}{N} \sum_{p,q \in \L^*_+} p^2 e^{-B(\eta)} b_p^* b_q^* b_q b_p e^{B(\eta)} +  \wt{\cE}_{1} \end{equation}
where, with Lemma \ref{lem:prelmG21},  
\begin{equation}\label{eq:wtE10}\pm \wt{\cE}_{1} \leq CN^{-2}e^{-B(\eta)} \cK(\cN_++1)^2 e^{B(\eta)} \leq C N^{-1}  \cK (\cN_+ + 1)  +  C N^{-1} (\cN_+ + 1)^3  \end{equation}

Next, we study the first term on the r.h.s. of (\ref{eq:dec-BKB}). We claim that
\begin{equation}\label{eq:splitKfinal1}
\begin{split} 
\sum_{p\in \Lambda^*_+}p^2 e^{-B(\eta)}b^*_pb_p e^{B(\eta)}=&\,\cK + \sum_{p \in \Lambda^*_+} p^2 \sigma_p^2\left(1-\frac{\cN_+}{N}\right) \\
&+ \sum_{p \in \Lambda^*_+} \Big[ p^2 \sigma_p \gamma_p  \big( b_p b_{-p} + b_p^* b_{-p}^* \big) + 2 p^2 \sigma_p^2 b_p^* b_p \Big]\\
&+\sum_{p\in\L^*_+}\left[p^2\eta_p b_{-p}d_p+\text{h.c.}\right] + \wt{\cE}_2 
 \end{split} 
\end{equation}
with the error operator $\wt{\cE}_2$ such that 
\begin{equation}\label{eq:cEN2K} \pm \wt{\cE}_{2} \leq C N^{-1/2} (\cK+1) ( \cN_+ + 1)  \end{equation}
To prove (\ref{eq:cEN2K}), we use (\ref{eq:ebe}) to decompose  
\begin{equation*}%\label{eq:termsK2}
\begin{split} 
\sum_{p\in \Lambda^*_+}&p^2 e^{-B(\eta)}b^*_pb_p e^{B(\eta)}=\text{E}_{1}+\text{E}_{2}+\text{E}_{3},
\end{split} 
\end{equation*}
with
\begin{equation*}%\label{eq:E123}
\begin{split} 
\text{E}_{1}&:=\sum_{p\in \Lambda^*_+}p^2(\gamma_pb^*_p +\sigma_pb_{-p} )(\gamma_pb_p +\sigma_pb^*_{-p})\\
\text{E}_{2}&:=\sum_{p\in \Lambda^*_+}p^2\Big[(\gamma_pb^*_p +\sigma_pb_{-p} )d_p+d^*_p(\gamma_pb_p +\sigma_pb^*_{-p})\Big]\\
\text{E}_{3}&:=\sum_{p\in \Lambda^*_+}p^2d^*_pd_p
\end{split} 
\end{equation*}
The term $\text{E}_{1}$ can be rewritten as
\begin{equation*}
\begin{split} 
\text{E}_{1} =\,&\cK + \sum_{p \in \Lambda^*_+} p^2 \sigma_p^2\left(1-\frac{\cN_+}{N}\right) + \sum_{p \in \Lambda^*_+} \Big[ p^2 \sigma_p \gamma_p  \big( b_p b_{-p} + b_p^* b_{-p}^* \big) + 2 p^2 \sigma_p^2 b_p^* b_p \Big] +\wt{\cE}_{3},
\end{split} 
\end{equation*}
where
\begin{equation*}
%\begin{split} 
\wt{\cE}_{3} =\,\frac{1}{N}\sum_{p \in \Lambda^*_+} p^2\Big[a^*_p\cN_+a_p+\s_p^2a^*_pa_p\Big]
%\end{split} 
\end{equation*}
is such that, for any $\xi\in\cF_+^{\leq N}$,
\begin{equation}\label{eq:wtE11}
\begin{split} 
|\langle\xi,\wt{\cE}_{3} \xi\rangle|\leq&\frac{1}{N}\sum_{p \in \Lambda^*_+}\Big[ p^2\|a_p\cN_+^{1/2}\xi\|^2+p^2\s_p^2\|a_p\xi\|^2\Big]\leq CN^{-1}\langle\xi,\cK(\cN_++1)\xi\rangle
\end{split} 
\end{equation}
The term $\text{E}_{2}$ can be split as 
\begin{equation*}
\text{E}_{2}=\sum_{p\in \Lambda^*_+}\left[p^2\eta_p b_{-p}d_p+\text{h.c.}\right] +\wt{\cE}_{4}
\end{equation*}
where
\[ \begin{split} | \langle \xi , \wt{\cE}_4 \xi \rangle | \leq \; &\sum_{p \in \L^*_+}  p^2 |\sigma_p - \eta_p| |\langle \xi , b_{-p} d_p \xi \rangle| +  \sum_{p \in \L^*_+}  p^2 |\gamma_p| | \langle \xi, b_p^* d_p \xi \rangle | \\ \leq \; &\frac{1}{N} \sum_{p \in \L^*_+} p^2 |\eta_p|^3 \| (\cN_+ + 1) \xi \|^2 \\ &+ \frac{1}{N} \sum_{p \in \L^*_+} p^2 \| b_p (\cN_+ + 1)^{1/2} \xi \| \left[ |\eta_p| \| (\cN_+ + 1) \xi \| + \| b_p (\cN_+ + 1)^{1/2} \xi \| \right] \\ \leq \; & C N^{-1/2}  \| \cK^{1/2} (\cN_+ + 1)^{1/2}  \xi \|^2
\end{split} \]  
As for the term $\text{E}_3$, we estimate
\[\begin{split}  |\langle \xi , \text{E}_3 \xi \rangle | &\leq \sum_{p \in \L^*_+} p^2 \| d_p \xi \|^2 \\ &\leq  \frac{C}{N^2}  \sum_{p \in \L^*_+} p^2  \left[ |\eta_p|^2 \| (\cN_+ + 1)^{3/2} \xi \|^2  + \|  b_p (\cN_+ + 1) \xi \|^2 \right] 
\\&\leq C N^{-1} \|  (\cN_+ + 1)^{3/2} \xi \|^2 + C N^{-1} \| (\cK+1)^{1/2} (\cN_+ + 1)^{1/2} \xi \|^2  \end{split} \]
for any $\xi \in \cF^{\leq N}_+$. This concludes the proof of (\ref{eq:splitKfinal1}), with the estimate (\ref{eq:cEN2K}).

Next, we consider the second term on the r.h.s. of (\ref{eq:dec-BKB}). We claim that
\begin{equation}\label{eq:splitKfinal}
\begin{split} 
\frac{1}{N}\sum_{p,q \in \Lambda^*_+} &p^2 e^{-B(\eta)}b^*_pb^*_qb_qb_p e^{B(\eta)}\\ = \; &\frac{1}{N}\sum_{p \in \Lambda^*_+}p^2\s_p^2 +\frac{1}{N}\sum_{p ,q\in \Lambda^*_+}p^2\s_p^2\s_q^2 +\frac{1}{N}\sum_{p,q\in \Lambda^*_+}p^2\s_p^2(\g_q^2+\s_q^2)b^*_qb_q\\
&+\frac{1}{N}\sum_{p,q\in \Lambda^*_+}p^2\s_p^2\g_q\s_q\left(b^*_qb^*_{-q}+\text{h.c.}\right) +\wt{\cE}_5
 \end{split} 
\end{equation}
with an error term $\wt\cE_5$ such that 
\begin{equation}\label{eq:cEK3N}
\pm \wt{\cE}_5 \leq C N^{-1/2}(\cK+\cN_+^2 + 1)(\cN_++1) \, .
\end{equation}
To prove (\ref{eq:cEK3N}), we consider first the operator 
\begin{equation*}\begin{split} 
\text{D} = \; & \sum_{q \in \Lambda^*_+} e^{-B(\eta)}b^*_qb_q e^{B(\eta)}  \\ = \; &\sum_{q \in \L^*_+} (\g_q b_q^* + \s_q b_{-q} +d_q^*) ( \g_q b_q + \s_q b_{-q}^* + d_q) \\  = \; &\sum_{q \in \L^*_+} \left[ (\g_q^2 + \s^2_q) b_q^* b_q + \s_q \g_q (b_q^* b_{-q}^* + b_q b_{-q}) + \sigma_q^2 \right] + \wt{\cE}_6
\end{split} \end{equation*}
where the error $\wt{\cE}_6$ is such that 
\begin{equation}\label{eq:cEK3} \pm \wt{\cE}_6 \leq C N^{-1} (\cN_+ + 1)^2 \end{equation}
as can be easily checked using the commutation relations (\ref{eq:comm-bp}) and the bound (\ref{eq:d-bds}). We go back to (\ref{eq:splitKfinal}), and we 
observe that
\begin{equation}\label{eq:II-deco} \begin{split}   \frac{1}{N}\sum_{p,q \in \Lambda^*_+}p^2 &e^{-B(\eta)}b^*_pb^*_qb_qb_p e^{B(\eta)} \\ = \; &\frac{1}{N} \sum_{p \in \L^*_+} p^2 e^{-B(\eta)} b_p^* e^{B(\eta)} \text{D} e^{-B(\eta)} b_p e^{B(\eta)} \\ = \; & \frac{1}{N} \sum_{p \in \L^*_+} p^2 (\g_p b_p^* + \s_p b_{-p} + d_p^*) \, \text{D} \, (\g_p b_p + \s_p b^*_{-p} + d_p) \\ = \; & \frac{1}{N} \sum_{p \in \L^*_+} p^2 \s_p^2 b_p \text{D} b_p^* + \wt{\cE}_{7} 
\end{split} 
\end{equation}
where 
\[ \begin{split} 
\wt{\cE}_7 = \; &\frac{1}{N} \sum_{p \in \L^*_+} p^2 \s_p (\g_p b_p^* + d_p^*) \text{D} b_{-p}^* + \frac{1}{N}  \sum_{p \in \L^*_+} p^2 \s_p b_{-p} \text{D} (\g_p b_{p} + d_p) 
\\ &+ \frac{1}{N} \sum_{p \in \L_+^*} p^2 (\g_p b_p^* + d_p^*) \text{D} (\g_p b_p + d_p) 
\end{split} \]
can be bounded using (\ref{eq:d-bds}) and the fact that, by Lemma \ref{lm:Ngrow}, 
$D \leq C (\cN_+ + 1)$, by
\[ \begin{split} |\langle \xi , \wt{\cE}_7 \, \xi \rangle | \leq \; & \frac{2}{N} \sum_{p \in \L^*_+} p^2 |\sigma_p | \left[ \| \text{D}^{1/2} b_p \xi \| + \| \text{D}^{1/2} d_p \xi \| \right] \| \text{D}^{1/2} b_p^* \xi \| \\ &+ \frac{1}{N} \sum_{p \in \L^*_+} p^2  \left[ \| \text{D}^{1/2} b_p \xi \| + \| \text{D}^{1/2} d_p \xi \| \right] \left[ \| \text{D}^{1/2} b_p \xi \| + \| \text{D}^{1/2} d_p \xi \| \right] \\ \leq \; & \frac{C}{N} \sum_{p \in \L^*_+} p^2 |\eta_p| \left[ \| b_p (\cN_+ + 1)^{1/2} \xi \| + N^{-1/2} |\eta_p| \| (\cN_+ + 1)^{3/2}  \xi \| \right] \| (\cN_+ + 1) \xi \| \\ &+ \frac{C}{N} \sum_{p \in \L_+*} p^2 \left[ \| b_p (\cN_+ + 1)^{1/2} \xi \|^2 + N^{-1} |\eta_p|^2 \| (\cN_+ + 1)^{3/2} \xi \|^2 \right]  \\ \leq \; &C N^{-1/2} \| (\cK + \cN_+^2 + 1)^{1/2} (\cN_+ + 1)^{1/2} \xi \|^2 \end{split} \]
As for the other term on the r.h.s. of (\ref{eq:II-deco}), we have, by (\ref{eq:cEK3}), 
\begin{equation}\label{eq:last-step2} \begin{split} 
\frac{1}{N} \sum_{p\in \L^*_+} p^2 \s_p^2 b_p \text{D} b_p^* = \; &\frac{1}{N} \sum_{p,q \in \L^*_+} p^2 \s_p^2 (\g_q^2 + \s_q^2) b_p b_q^* b_q b_p^*  + \frac{1}{N} \sum_{p,q \in \L^*_+} p^2 \s_p^2 \s_q^2 b_p b_p^*  \\ &+ \frac{1}{N} \sum_{p,q \in \L^*_+} p^2 \s_p^2 \g_q \s_q b_p (b_q^* b_{-q}^* + \hc) b_p^*+ \wt{\cE}_8 \end{split} \end{equation}
where, using (\ref{eq:cEK3}), it is easy to check that $\pm \wt{\cE}_8 \leq C N^{-1} (\cN_+ + 1)^2$. 
Rearranging the other terms on the r.h.s. of (\ref{eq:last-step2}) in normal order and using the commutator relations (\ref{eq:comm-bp}), we obtain (\ref{eq:splitKfinal}) with an error term satisfying (\ref{eq:cEK3N}). 

%%%%%%%%%%%%%%%%%%%%%%%%%%%%%%%%%%%%%%%%%%%%%%%%%%%%%%%%%%%%%%%%%%%%
%%%%%%%%%%%%%%%%%%%%%%%%%%%%%%%%%%%%%%%%%%%%%%%%%%%%%%%%%%%%%%%%%%%%
%%%%%%%%%%%%%%%%%%%%%%%%%%%%%%%%%%%%%%%%%%%%%%%%%%%%%%%%%%%%%%%%%%%%
Finally, we focus on the contribution of (\ref{eq:L2VN}). We claim that 
\begin{equation}\label{eq:cEV-def} \begin{split}
& \hspace{-.5cm} e^{-B(\eta)} \cL^{(2,V)}_{N}  e^{B(\eta)} \\ = \; &\sum_{p \in \Lambda^*_+} \left[  \widehat{V} (p/N) \sigma_p^2 +  \widehat{V} (p/N) \sigma_p \gamma_p\left(1-\frac{\cN_+}{N}\right) \right]  \\
&+ \sum_{p \in \Lambda^*_+}  \widehat{V} (p/N) (\gamma_p + \sigma_p)^2 b_p^* b_p + \frac{1}{2} \sum_{p \in \Lambda^*_+}  \widehat{V} (p/N) (\gamma_p+\sigma_p)^2 (b_pb_{-p} + b_p^* b_{-p}^*)\\
&+\sum_{p \in \Lambda^*_+}\left[ \frac{1}{2}\widehat{V} (p/N)\left(\g_pb_{-p}+\s_pb^*_p\right)d_p+\frac{1}{2}\widehat{V} (p/N)d_p\left(\g_pb_{-p}+\s_pb^*_p\right)\right]+\text{h.c.}\\
&+ \wt{\cE}_9
\end{split} \end{equation}
where 
\begin{equation}\label{eq:errorV}
  \pm \wt\cE_9 \leq CN^{-1/2} (\cK + \cN_+^2 + 1) (\cN_++1)
\end{equation}
To prove (\ref{eq:cEV-def}), (\ref{eq:errorV}), we start from (\ref{eq:L2VN}) and decompose
\begin{equation}\label{eq:G2-deco} \begin{split} 
e^{-B(\eta)} \cL^{(2,V)}_{N}  e^{B(\eta)}  =\; & \sum_{p \in \Lambda^*_+}  \widehat{V} (p/N) e^{-B(\eta)} b_p^* b_p e^{B(\eta)} - \frac{1}{N} \sum_{p \in \Lambda^*_+} \widehat{V} (p/N) e^{B(\eta)} a_p^* a_p e^{-B(\eta)} \\ &+ \frac{1}{2} \sum_{p \in \Lambda^*_+} \widehat{V} (p/N) e^{-B(\eta)} \big[ b_p b_{-p} + b_p^* b_{-p}^* \big] e^{B(\eta)} \\ =: \; &\text{F}_1 + \text{F}_2 + 
\text{F}_3 \end{split} \end{equation} 
The operators $\text{F}_1$ and $\text{F}_2$ can be handled exactly as in the proof \cite[Prop. 7.6]{BBCS2} (notice that the bounds are independent of $\beta \in (0;1)$ and they can be readily extended to the case Gross-Pitaevskii case $\beta = 1$). We obtain that 
\begin{equation*}%\label{eq:GN21}
 \text{F}_1 =  \sum_{p \in \L^*_+} \widehat{V} (p/N) [\gamma_p b_p^* + \sigma_p b_{-p} ] [ \gamma_p b_p + \sigma_p b_{-p}^*] + \wt{\cE}_{10} \end{equation*}
where  
\begin{equation*}%\label{eq:pmEV1} 
\pm \wt{\cE}_{10} \leq C N^{-1} (\cN_+ +1)^2 \end{equation*}
and that 
\begin{equation*}%\label{eq:pmGN22} 
\pm \text{F}_2 \leq C N^{-1} (\cN_+ + 1) \end{equation*}
 
Let us consider the last term on the r.h.s. of (\ref{eq:G2-deco}). With (\ref{eq:ebe}), we obtain 
\begin{equation*}%\label{eq:G23-deco1} 
\begin{split}
\text{F}_3 = \; &\frac{1}{2} \sum_{p \in \Lambda^*_+} \widehat{V} (p/N) \left[ \gamma_p b_p + \sigma_p b_{-p}^* \right] \left[ \gamma_p b_{-p} + \sigma_p b_p^* \right]   \\ 
&+ \frac{1}{2}  \sum_{p \in \Lambda^*_+} \widehat{V} (p/N) \, \left[ (\g_pb_p+\s_pb^*_{-p}) \, d_{-p} + d_p\, (\g_pb_{-p}+\s_pb^*_{p}) \right]  \\ &+ \wt\cE_{11}  + \hc \end{split} \end{equation*}
where the error term $\wt{\cE}_{11}  = (1/2)  \sum_{p \in \Lambda^*_+} \widehat{V} (p/N)  d_p d_{-p}$  can be bounded, using (\ref{eq:d-bds}), by
\[ \begin{split} |\langle \xi , \wt{\cE}_{11} \xi \rangle | &\leq C \sum_{p \in \L^*_+} |\widehat{V} (p/N)|  \| d_p^* \xi \| \| d_{-p} \xi \| \\ &\leq \frac{C}{N^2} \sum_{p \in \L^*_+}  |\widehat{V} (p/N)|  \| (\cN_+ + 1)^{3/2} \xi \| \left[ |\eta_p| 
\| (\cN_+ + 1)^{3/2} \xi \| + \| b_p (\cN_+ + 1) \xi \| \right] \\ &\leq 
C N^{-1/2}  \|  (\cN_+ + 1)^{3/2} \xi \|^2 \end{split} \] 
since $\| \widehat{V} (./N) \|_2 \leq C N^{3/2}$. This concludes the proof of (\ref{eq:cEV-def}) and (\ref{eq:errorV}). Comparing (\ref{eq:G2-sec}), (\ref{eq:G_N2K}) and (\ref{eq:G_N2V}) with (\ref{eq:dec-BKB}). (\ref{eq:splitKfinal1}), (\ref{eq:splitKfinal}) and (\ref{eq:cEV-def}), we conclude that the bounds (\ref{eq:wtE10}), (\ref{eq:cEN2K}), (\ref{eq:wtE11}), (\ref{eq:cEK3N}) and (\ref{eq:errorV}) imply the desired estimate (\ref{eq:propG2}). 
\end{proof}

%%%%%%%%%%%%%%%%%%%%%%%%%%%%%%%%%%%%%%%%%%%%%%%%%%%%%%%%%%%%%%%%%%%%%%%%%%%%%%%%%%%%%%%%%%%%%%%%%%%%%%%%%%%%%%%%%%%%%%%%%%%%%%%%%%%%%%%%%%%%%%%%%%%%%%%%%%%%%%%%%%%%%%%%%%
\subsection{Analysis of $ \cG_N^{(3)}=e^{-B(\eta)}\cL^{(3)}_N e^{B(\eta)}$}

From (\ref{eq:cLNj}), we have
\[ \cG_N^{(3)} = \frac{1}{\sqrt{N}} \sum_{p,q \in \L^*_+ : p + q \not = 0} \widehat{V} (p/N) e^{-B(\eta)} b^*_{p+q} a^*_{-p} a_q e^{B(\eta)} + \hc  \]
We define the error operator $\cE^{(3)}_N$ through the identity
\begin{equation}\label{eq:GN3-deco} 
\cG_N^{(3)} = \frac{1}{\sqrt{N}} \sum_{p,q \in \L^*_+ : p + q \not = 0} \widehat{V} (p/N) \left[ b_{p+1}^* b_{-p}^* (\gamma_q b_q + \s_q b^*_{-q})  + \hc \right] + \cE^{(3)}_N \end{equation}
The goal of this subsection is to prove the next lemma, where we estimate $\cE^{(3)}_N$. 
\begin{lemma}\label{lm:GN-3}
Let $\cE_N^{(3)}$ be as defined in (\ref{eq:GN3-deco}). Then, under the same assumptions as in Proposition \ref{prop:gene}, we find $C > 0$ such that  
\begin{equation}\label{eq:lm-GN3} \pm \cE_N^{(3)} \leq CN^{-1/2} (\cV_N + \cN_+ + 1) ( \cN_+ + 1) \end{equation}
\end{lemma}
\begin{proof} 
With 
\[ a_{-p}^* a_q =  b^*_{-p} b_q + N^{-1} a^*_{-p} \cN_+ a_q \]
we obtain 
\begin{equation}\label{eq:cGN3-wt} \cG_N^{(3)} =  \frac{1}{\sqrt{N}} \sum_{p,q \in \L^*_+ : p + q \not = 0} \widehat{V} (p/N) e^{-B(\eta)} b^*_{p+q} b^*_{-p} b_q e^{B(\eta)} + \wt{\cE}_1 + \hc \end{equation}
where 
\[ \wt{\cE}_1 = \frac{1}{N^{3/2}} \sum_{p,q \in \L^*_+ : p + q \not = 0} 
\widehat{V} (p/N) e^{-B(\eta)} b^*_{p+q} a^*_{-p} a_q e^{B(\eta)} \]
can be bounded, switching to position space, by 
  \[\begin{split}\big|\langle \xi, &\wt{\cE}_1 \xi\rangle  \big| \\ &\leq \; \int_{\Lambda^2}dxdy\;N^{3/2}V(N(x-y))\big\|  \check{a}_x \check{a}_y (\cN_+ + 1)^{1/2} e^{B(\eta)} \xi \big\| \big\| \check{a}_x (\cN_+ + 1)^{1/2} e^{B(\eta)}\xi\big\| \\  
        &\leq \; C N^{-3/2} \langle\xi, e^{-B(\eta)}\cV_N(\cN_++1)e^{B(\eta)}\xi\rangle + CN^{-1/2}\langle\xi, e^{-B(\eta)} (\cN_++1)^2 e^{B(\eta)} \xi\rangle
        \end{split}\]
With Lemma \ref{lm:Ngrow} and Lemma \ref{lem:prelmG21} we conclude that 
        \begin{equation}\label{eq:wtE1} \big|\langle \xi, \wt{\cE}_1 \xi\rangle  \big|\leq CN^{-1/2} \langle \xi,(\cV_N+\cN_+ +1)(\cN_++1)\xi\rangle \end{equation}
        
To control the first term on the r.h.s. of (\ref{eq:cGN3-wt}), we use (\ref{eq:ebe}) to decompose 
        \begin{equation}\label{eq:deco2}   \frac{1}{\sqrt{N}} \sum_{p,q \in \L^*_+ : p + q \not = 0} \widehat{V} (p/N) e^{-B(\eta)} b^*_{p+q} b^*_{-p} b_q e^{B(\eta)}  = \text{M}_0 + \text{M}_1+\text{M}_2+\text{M}_3\end{equation}
where     
         \begin{equation}\label{eq:def-G3M0}\begin{split}
        \text{M}_0:=\;& \frac{1}{\sqrt{N}} \sum_{p,q }^* \widehat{V} (p/N) \big[ \gamma_{p+q}\gamma_pb^*_{p+q}  b^*_{-p} +  \gamma_{p+q}\sigma_pb^*_{p+q}b_{p}+\sigma_{p+q}\sigma_pb_{-p-q}b_{p}\\
        &\hspace{1.5cm}+\sigma_{p+q}\gamma_p b^*_{-p}b_{-p-q}-N^{-1}\sigma_{p+q}\gamma_pa^*_{-p}a_{-p-q}\big]\, \big[\gamma_qb_{q} +\sigma_qb^*_{-q}\big]   \\
        \end{split}\end{equation}
and
\begin{equation*}\begin{split}
        \text{M}_1:=\;& \frac{1}{\sqrt{N}} \sum_{p,q }^* \widehat{V} (p/N) \big[\gamma_{p+q}b^*_{p+q}d^*_{-p} +\sigma_{p+q}b_{-p-q} d^*_{-p}+  \gamma_pd^*_{p+q}b^*_{-p} +\sigma_pd^*_{p+q}b_{p}d^*_{p+q}\big]\\
        &\hspace{3cm}\times \big[\gamma_qb_{q} +\sigma_qb^*_{-q}\big] ; \\
        &\; +\frac{1}{\sqrt{N}} \sum_{p,q }^* \widehat{V} (p/N)  \big[ \gamma_{p+q}\gamma_pb^*_{p+q}  b^*_{-p} +  \gamma_{p+q}\sigma_pb^*_{p+q}b_{p}+\sigma_{p+q}\sigma_pb_{-p-q}b_{p}\\
        &\hspace{3.5cm}+\sigma_{p+q}\gamma_p b^*_{-p}b_{-p-q}-N^{-1}\sigma_{p+q}\gamma_pa^*_{-p}a_{-p-q}\big]d_q ;\\        
        \text{M}_2:=\;& \frac{1}{\sqrt{N}} \sum_{p,q }^* \widehat{V} (p/N) \big[\gamma_{p+q}b^*_{p+q}d^*_{-p} +\sigma_{p+q}b_{-p-q} d^*_{-p}+  \gamma_pd^*_{p+q}b^*_{-p} +\sigma_pd^*_{p+q}b_{p}\big]d_q \\
        &+\frac{1}{\sqrt{N}} \sum_{p,q }^* \widehat{V} (p/N)d^*_{p+q}d^*_{-p} \big[\gamma_qb_{q} +\sigma_qb^*_{-q}\big]  ;\\
        \text{M}_3:=\;& \frac{1}{\sqrt{N}} \sum_{p,q }^* \widehat{V} (p/N) d^*_{p+q}d^*_{-p}d_q 
        \end{split}\end{equation*}
Here, we introduced the shorthand notation $ \sum_{p,q }^*\equiv \sum_{p,q \in \Lambda_+^* , p+q \not = 0}$ and we used the identity $b_{-p-q} b^*_{-p} = b^*_{-p} b_{-p-q}  - N^{-1} a_{-p}^* a_{-p-q}$, for all $q \in \Lambda^*_+$. Notice that the index $i$ in $M_i$ counts the number of $d$-operators it contains. 

Let us start by analysing $\text{M}_3$. Switching to position space we find, using (\ref{eq:dxy-bds}) and the bound $\| \check{\eta} \|_\infty \leq CN$ (as follows from  (\ref{3.0.scbounds1}) since, by the definition (\ref{eq:ktdef}), we have $\check{\eta} (x) = -N w_\ell (Nx)$),    
        \begin{equation} \label{eq:M3} \begin{split} |\langle \xi , \text{M}_3 \xi \rangle | \leq &\; \int dxdy\, N^{5/2} V(N(x-y)) \| (\cN_+ + 1)^{-1} \check{d}_x \check{d}_y \xi \| \| (\cN_+ + 1) \check{d}_x \xi \|  \\ 
        \leq \; & \frac{C}{N^3}  \int dxdy\, N^{5/2} V(N(x-y)) \left[ \| (\cN_+ + 1)^{5/2} \xi \| + \| \check{a}_x (\cN_+ + 1)^2 \xi \| \right] \\  &\hspace{0cm} \times  \left[ N \| (\cN_+ + 1) \xi \| + \| \check{a}_x (\cN_+ +1)^{3/2} \xi \| \right. \\  &\hspace{3cm} \left. + \| \check{a}_y (\cN_+ +1)^{3/2} \xi \| + \| \check{a}_x \check{a}_y (\cN_+ + 1) \xi \| \right] \\
        \leq \; & C N^{-1/2} \langle \xi , (\cV_N + \cN_+^2 + 1) \xi \rangle         
        \end{split} \end{equation}
        
As for $ \text{M}_2$, it reads in position space
        \[\begin{split}
         \text{M}_2=&\; \int dxdy\, N^{5/2} V(N(x-y)) \big[ b^*(\check{\gamma}_x) \check{d}^*_y + b(\check{\sigma}_x) \check{d}^*_y+ \check{d}^*_x b^*(\check{\gamma}_y)+\check{d}^*_x b(\check{\sigma}_y) \big] \check{d}_x \\
         &\; +  \int dxdy\, N^{5/2} V(N(x-y)) \check{d}^*_x \check{d}^*_y \big[ b (\check{\gamma}_x)+b^*(\check{\sigma}_x) \big]  \\ =:&\;  \text{M}_{21} + \text{M}_{22}
        \end{split}\]
To control $ \text{M}_{22}$, we use the bound (\ref{eq:dxy-bds}) to estimate
       \[\begin{split}\big| \langle\xi, \text{M}_{22}\xi\rangle\big| \leq &\;  \int dxdy\, N^{5/2} V(N(x-y))  \big \|(\cN_+ +1)^{-1} \check{d}_y \check{d}_x \xi\Big\|\\
        &\hspace{4cm}\times\Big\|(\cN_++1) \big[ b (\check{\gamma}_x)+b^*(\check{\sigma}_x) \big] \xi\big\|\\
        \leq &\; C N^{-2}  \int dx dy N^{5/2} V(N(x-y)) \left[ \| (\cN_+ + 1)^{3/2} \xi \| + \| \check{a}_x (\cN_+ + 1) \xi \| \right] \\  &\hspace{0cm} \times  \left[ N \| (\cN_+ + 1) \xi \| + \| \check{a}_x (\cN_+ +1)^{3/2} \xi \| \right. \\  &\hspace{3cm} \left. + \| \check{a}_y (\cN_+ +1)^{3/2} \xi \| + \| \check{a}_x \check{a}_y (\cN_+ + 1) \xi \| \right] \\
         \leq \; & C N^{-1} \langle \xi , (\cV_N + \cN_+ + 1) (\cN_+ + 1) \xi \rangle         
        \end{split} \]
With the first and the second bounds in (\ref{eq:dxy-bds}), we can also control 
$\text{M}_{21}$. We find 
\begin{equation}\label{eq:bd-M21} \begin{split} 
|\langle \xi , \text{M}_{21} \xi \rangle| \leq \; &C \int dx dy \, N^{5/2} V(N(x-y))  \|\check{d}_x \xi \|  \\ &\hspace{1cm} \times  \left[ \| \check{d}_y b (\check{\gamma}_x) \xi \| + \| \check{d}_y b^* (\check{\s}_x) \xi \|  +  \| b (\check{\gamma}_y) \check{d}_x \xi \| + \| b^* (\check{\s}_y) \check{d}_x \xi \| \right] \\ 
 \leq \; &C \int dx dy \, N^{1/2} V(N(x-y))  \left[ \| (\cN_+ + 1)^{3/2} \xi \| + \| \check{a}_x (\cN_+ + 1) \xi \| \right]   \\ &\hspace{1cm}\times  \Big[ N \| (\cN_+ + 1) \xi \| +  \| \check{a}_x (\cN_+ + 1)^{3/2} \xi \|+  \| \check{a}_y (\cN_+ + 1)^{3/2} \xi \|  \\ &\hspace{7.6cm} + \| \check{a}_x \check{a}_y (\cN_+ + 1) \xi \| \Big]  \\ 
 \leq \; &C N^{-1} \langle \xi , (\cV_N + \cN_+ + 1) ( \cN_+ + 1) \xi \rangle 
 \end{split} \end{equation}
where we used that $\| \check{\eta} \|_\infty \leq C N$. Hence, we proved that
\begin{equation}\label{eq:M2} |\langle \xi , \text{M}_2 \xi \rangle |  \leq C N^{-1} \langle \xi , (\cV_N + \cN_+ + 1) ( \cN_+ + 1) \xi \rangle \end{equation}

Next, let us consider the operator $ \text{M}_1$. In position space, we find
        \[\begin{split} \text{M}_1= &\;  \int dxdy\, N^{5/2} V(N(x-y)) \big[ b^*(\check{\gamma}_x)\check{d}^{*}_y + b(\check{\sigma}_x) \check{d}^*_y+\check{d}^*_x b^*(\check{\gamma}_y)+\check{d}^*_x b(\check{\sigma}_y) \big]\\
        &\hspace{5cm}\times\big[ b(\check{\gamma}_x) + b^*(\check{\sigma}_x)\big]\\
        &\; +  \int dxdy\, N^{5/2} V(N(x-y))\big[b^*(\check{\gamma}_x)b^*(\check{\gamma}_y)+b^*(\check{\gamma}_x)b(\check{\sigma}_y)+b(\check{\sigma}_x)b(\check{\sigma}_y)  \\
        & \hspace{3cm}+b^*(\check{\gamma}_x)b(\check{\sigma}_y)-N^{-1}a^*(\check{\gamma}_x)a(\check{\sigma}_y)\big] \check{d}_x  \\ =: &\;  \text{M}_{11} + \text{M}_{12}
        \end{split}\]
To estimate $\text{M}_{11}$, we proceed as in (\ref{eq:bd-M21}). With (\ref{eq:dxy-bds}), using again $\| \check{\eta} \|_\infty \leq CN$, we find 
           \[\begin{split}\big| \langle\xi, \text{M}_{11}\xi\rangle\big| \leq &  \int dxdy\, N^{5/2} V(N(x-y))  \left[ \| b (\check{\g}_x) \xi \| + \| b^* (\check{\s}_x) \xi \| \right] 
 \\ &\hspace{1cm} \times \left[ \| \check{d}_y b(\check{\gamma}_x) \xi \| + \|  \check{d}_y b^* (\check{\s}_x) \xi \| + \| b (\check{\gamma}_y)  \check{d}_x \xi \| + \| b^* (\check{\s}_y)  \check{d}_x \xi \| \right]           \\  &\leq C \int dx dy N^{3/2} V (N (x-y)) \Big[ \| (\cN_+ + 1)^{1/2} \xi \| + \|  \check{a}_x \xi \| \Big] \\ &\hspace{1cm} \times  \Big[ N  \| (\cN_+ + 1) \xi \| +  \| \check{a}_x (\cN_+ + 1)^{3/2} \xi \|+  \| \check{a}_y (\cN_+ + 1)^{3/2} \xi \|  \\ &\hspace{7.6cm} + \| \check{a}_x \check{a}_y (\cN_+ + 1) \xi \| \Big]  
           \\ & \leq C N^{-1/2}   \langle \xi , (\cV_N + \cN_+ + 1) ( \cN_+ + 1) \xi \rangle 
\end{split} \] 
As for the term $M_{12}$, we use the bound
\[ \begin{split} \Big\| &(\cN_+ + 1)^{1/2}  \big[ b^*(\check{\gamma}_x) b^*(\check{\gamma}_y)+b^*(\check{\gamma}_x) b(\check{\sigma}_y)+b(\check{\sigma}_x)b(\check{\sigma}_y) 
\\& \hspace{6cm} + b^*(\check{\gamma}_x) b(\check{\sigma}_y)  -N^{-1}a^*(\check{\gamma}_x) a(\check{\sigma}_y)\big] \xi\Big\| \\ &\leq C 
 \Big[ \big\|(\cN_++1)^{3/2} \xi\big\| +\big\| \check{a}_x (\cN_++1)  \xi\big\|+\big\| \check{a}_y (\cN_++1)  \xi\big\|+\big\|\check{a}_x \check{a}_y (\cN_++1)^{1/2} \xi\big\| \Big]
        \end{split}\]
to conclude that 
\[ \begin{split} 
|\langle \xi, & \text{M}_{12} \xi \rangle |\\  \leq \; &C \int dx dy \, N^{3/2} V(N(x-y)) \left[ \| (\cN_+ + 1) \xi \| + \| \check{a}_x (\cN_+ + 1)^{1/2} \xi \| \right] \\ &\hspace{0cm} \times 
 \Big[ \big\|(\cN_++1)^{3/2} \xi\big\| +\big\| \check{a}_x (\cN_++1)  \xi\big\|+\big\| \check{a}_y (\cN_++1)  \xi\big\|+\big\|\check{a}_x \check{a}_y (\cN_++1)^{1/2} \xi\big\| \Big] \\
 \leq \; &C N^{-1/2} \langle \xi , (\cV_N + \cN_+ + 1) ( \cN_+ + 1) \xi \rangle \end{split} \]
Thus, 
\begin{equation}\label{eq:M1} |\langle \xi , \text{M}_1 \xi \rangle | \leq C N^{-1/2}  \langle \xi , (\cV_N + \cN_+ + 1) ( \cN_+ + 1) \xi \rangle  \end{equation}

Finally, we consider (\ref{eq:def-G3M0}). We split $\text{M}_0 = \text{M}_{01} + \text{M}_{02}$, with 
        \[\begin{split}\text{M}_{01}:=\;& \frac{1}{\sqrt{N}} \sum_{p,q }^* \widehat{V} (p/N) \gamma_{p+q}\gamma_pb^*_{p+q}  b^*_{-p}\big[\gamma_qb_{q} +\sigma_qb^*_{-q}\big] ;  \\
        \text{M}_{02}:=\;& \frac{1}{\sqrt{N}} \sum_{p,q }^* \widehat{V} (p/N) \big[  \gamma_{p+q}\sigma_pb^*_{p+q}b_{p}+\sigma_{p+q}\sigma_pb_{-p-q}b_{p}+\sigma_{p+q}\gamma_p b^*_{-p}b_{-p-q}\\
        &\hspace{5cm}-N^{-1}\sigma_{p+q}\gamma_pa^*_{-p}a_{-p-q}\big]  \big[\gamma_qb_{q} +\sigma_qb^*_{-q}\big] 
        \end{split}\]
        Switching to position space, we find 
        \[\begin{split} \big | \langle\xi,\text{M}_{02}\xi \rangle \big|  \leq \;& C \int dxdy\, N^{5/2} V(N(x-y)) \Big[ \| \check{a}_x \xi \| + \| (\cN_+ + 1)^{1/2} \xi \| \Big] \\ & \hspace{1cm} \times 
        \Big[ \| (\cN_+ + 1) \xi \| + \| \check{a}_x (\cN_+ + 1)^{1/2} \xi \| + \| \check{a}_y (\cN_+ + 1)^{1/2} \xi \| \Big] \\ \leq \; &C N^{-1/2}\langle \xi, (\cN_++1)^2\xi\rangle \end{split} \]
        As for $\text{M}_{01}$, we write $\g_p = 1 + (\g_p -1)$ and $\g_{p+q} = 1 + (\g_{p+q}-1)$. Using that $|\g_p - 1| \leq C / p^4$ and $\s_q$ are square summable, it is easy to check that 
        \[ \text{M}_{01} = \frac{1}{\sqrt{N}} \sum^*_{p,q} \widehat{V} (p/N) b_{p+q}^* b^*_{-p} \left[ \g_q b_q + \s_q b^*_{-q} \right] + \wt{\cE}_{2} \]
        where $\wt{\cE}_{2}$ is such that 
        \[ |\langle \xi, \wt{\cE}_{2} \xi \rangle | \leq C N^{-1/2} \langle \xi , (\cN_+ + 1)^2 \xi \rangle \]
        Combining the last bound, with the bounds (\ref{eq:wtE1}), (\ref{eq:M3}), (\ref{eq:M2}), (\ref{eq:M1}) and the decompositions (\ref{eq:cGN3-wt}) and (\ref{eq:deco2}), we obtain (\ref{eq:lm-GN3}). 
        \end{proof}

%%%%%%%%%%%%%%%%%%%%%%%%%%%%%%%%%%%%%%%%%%%%%%%%%%%%%%%%%%%%%%%%%%%%%%%%%%%%%%%%%%%%%%%%%%%%%%%%%%%%%%%%%%%%%%%%%%%%%%%%%%%%%%%%%%%%%%%%%%%%%%%%%%%%%%%%%%%%%%%%%%%%%%%%%%
%%%%%%%%%%%%%%%%%%%%%%%%%%%%%%%%%%%%%%%%%%%%%%%%%%%%%%%%%%%%%%%%%%%%%%%%%%%%%%%%%%%%%%%%%%%%%%%%%%%%%%%%%%%%%%%%%%%%%%%%%%%%%%%%%%%%%%%%%%%%%%%%%%%%%%%%%%%%%%%%%%%%%%%%%%
%%%%%%%%%%%%%%%%%%%%%%%%%%%%%%%%%%%%%%%%%%%%%%%%%%%%%%%%%%%%%%%%%%%%%%%%%%%%%%%%%%%%%%%%%%%%%%%%%%%%%%%%%%%%%%%%%%%%%%%%%%%%%%%%%%%%%%%%%%%%%%%%%%%%%%%%%%%%%%%%%%%%%%%%%%
%%%%%%%%%%%%%%%%%%%%%%%%%%%%%%%%%%%%%%%%%%%%%%%%%%%%%%%%%%%%%%%%%%%%%%%%%%%%%%%%%%%%%%%%%%%%%%%%%%%%%%%%%%%%%%%%%%%%%%%%%%%%%%%%%%%%%%%%%%%%%%%%%%%%%%%%%%%%%%%%%%%%%%%%%%
%%%%%%%%%%%%%%%%%%%%%%%%%%%%%%%%%%%%%%%%%%%%%%%%%%%%%%%%%%%%%%%%%%%%%%%%%%%%%%%%%%%%%%%%%%%%%%%%%%%%%%%%%%%%%%%%%%%%%%%%%%%%%%%%%%%%%%%%%%%%%%%%%%%%%%%%%%%%%%%%%%%%%%%%%%
%%%%%%%%%%%%%%%%%%%%%%%%%%%%%%%%%%%%%%%%%%%%%%%%%%%%%%%%%%%%%%%%%%%%%%%%%%%%%%%%%%%%%%%%%%%%%%%%%%%%%%%%%%%%%%%%%%%%%%%%%%%%%%%%%%%%%%%%%%%%%%%%%%%%%%%%%%%%%%%%%%%%%%%%%%
\subsection{Analysis of $ \cG_N^{(4)}=e^{-B(\eta)}\cL^{(4)}_N e^{B(\eta)}$}

From (\ref{eq:cLNj}) we have
\begin{equation}\label{eq:GN4-b} \cG_N^{(4)} = \frac{1}{2N} \sum_{p,q \in \L^*_+, r \in \L^* : r \not = -p,-q}  \widehat{V} (r/N)  e^{-B(\eta)} a^*_{p+r} a_q^* a_p a_{q+r} e^{B(\eta)}  \end{equation}
We define the error operator $\cE^{(4)}_N$ through the identity 
\begin{equation}\label{eq:GN4-deco}
        \begin{split}  \cG_N^{(4)} = &\;\cV_N+ \frac{1}{2N} \sum_{p,q\in\Lambda_+^*}\widehat{V} ((p-q)/N)\sigma_q\gamma_q\sigma_p\gamma_p\,\big(1+1/N-2\;\cN_+/N\big)\\
        &\; + \frac{1}{2N} \sum_{p\in\Lambda_+^*, q\in\Lambda^*}\widehat{V} ((p-q)/N) \eta_q \\ &\hspace{1cm} \times \Big[ \gamma_p^2 b^*_{p}b^*_{-p}+2 \gamma_p\sigma_{p} b^*_{p}b_p+\sigma_p^2b_pb_{-p}+d_p\big( \gamma_p  b_{-p}+\sigma_p b^*_{p}\big)\\
        &\;\hspace{7cm}+\big(\gamma_p b_p +\sigma_p b^*_{-p}\big)d_{-p} + \text{h.c.}\Big]\\
        &\;+\frac{1}{N^2} \sum_{p,q,u\in\Lambda_+^*}\widehat{V} ((p-q)/N)\eta_p\eta_q \\ &\hspace{1cm} \times \Big[ \gamma_u^2 b^*_u b_u + \sigma_u^2 b^*_u b_u + \gamma_u\sigma_u b^*_u b^*_{-u}+\gamma_u\sigma_u b_u b_{-u}  +\sigma_u^2\Big]\\
        &\;+ \cE_N^{(4)}
        \end{split}
        \end{equation}
The goal of this subsection is to bound the error term $\cE^{(4)}_N$.
\begin{lemma}\label{lm:GN-4}
Let $\cE_N^{(4)}$ be as defined in (\ref{eq:GN4-deco}). Then, under the same assumptions as in Proposition \ref{prop:gene}, we find $C > 0$ such that  
\begin{equation}\label{eq:lmGN4} \pm \cE^{(4)}_N \leq C N^{-1/2}  (\cV_N+\cN_++1)(\cN_++1)\end{equation}
\end{lemma}

\begin{proof}
First of all, we replace, on the r.h.s. of (\ref{eq:GN4-b}), all $a$-operators by $b$-operators. To this end, we notice that
 \begin{equation*} a^*_{p+r}a^*_{q}a_{p}a_{q+r} =b^*_{p+r}b^*_{q}b_{p}b_{q+r}\bigg(1-\frac3N+\frac{2\,\cN_+}N\bigg)+a^*_{p+r}a^*_{q}a_{p}a_{q+r}\,\Theta_{\cN_+} \end{equation*}
where
        \[\begin{split}\Theta_{\cN_+}:=& \bigg[\frac{\big(N-\cN_++2\big)}N\frac{\big(\cN_+-1\big)}N +\frac{\big(\cN_+-2\big)}N\bigg]^2\\
        &\hspace{2cm} + \bigg[-\frac{\cN_+^2}{N^2}+\frac{3\cN_+}{N^2}-\frac{2}{N^2}\bigg]\bigg[\frac{\big(N-\cN_++2\big)}N\frac{\big(N-\cN_++1\big)}N\bigg]
        \end{split}\]
is such that $ \pm\Theta_{\cN_+} \leq C (\cN_++1)^2/N^2 $ on $\cF_+^{\leq N}$. With Lemma \ref{lem:prelmG21} we conclude that 
 \begin{equation}\label{eq:def-G40G41}\begin{split} \cG_N^{(4)} =&\frac{( N+1)}{2N^2} \sum_{\substack{p,q \in \Lambda_+^*, r \in \Lambda^*:\\ r \not = -p,-q}} \widehat{V} (r/N) e^{-B(\eta)} b_{p+r}^*b_q^*  b_{p} b_{q+r} e^{B(\eta)}\\
         & +\frac{1}{N^2} \sum_{\substack{p,q,u \in \Lambda_+^*, r \in \Lambda^*:\\ r \not = -p,-q}} \widehat{V} (r/N) e^{-B(\eta)} b_{p+r}^*b_q^* b^*_ub_u  b_{p} b_{q+r} e^{B(\eta)} +\wt{\cE}_1          \end{split}\end{equation}
with the error $\wt{\cE}_{1}$ satisfying 
         \[ \pm \wt{\cE}_{1}  \leq C N^{-1} (\cV_N + \cN_+ +1)( \cN_+ +1)  \] 
We split the rest of the proof in two steps, where we analyze separately the two terms on the r.h.s. of (\ref{eq:def-G40G41}). 

{\it Step 1.} The first term on the r.h.s. of (\ref{eq:def-G40G41}) can be written as 
  \begin{equation}\label{eq:G40order1}
        \begin{split} 
        \frac{( N+1)}{2N^2} &\sum_{\substack{p,q \in \Lambda_+^*, r \in \Lambda^*:\\ r \not = -p,-q}} \widehat{V} (r/N) e^{-B(\eta)} b_{p+r}^*b_q^*  b_{p} b_{q+r} e^{B(\eta)}\\ 
        = &\;\cV_N+ \frac{1}{2N} \sum_{p,q\in\Lambda_+^*}\widehat{V} ((p-q)/N)\sigma_q\gamma_q\sigma_p\gamma_p\,\big(1+1/N-2\cN_+/N\big)\\
        &\; + \frac{1}{2N} \sum_{p,q\in\Lambda_+^*}\widehat{V} ((p-q)/N)\sigma_q\gamma_q \Big[ \gamma_p^2 b^*_{p}b^*_{-p}+2 \gamma_p\sigma_{p} b^*_{p}b_p+\sigma_p^2b_pb_{-p}\\
        &\;\hspace{2.5cm}+d_p\big( \gamma_p  b_{-p}+\sigma_p b^*_{p}\big)+\big(\gamma_p b_p +\sigma_p b^*_{-p}\big)d_{-p} + \text{h.c.}\Big]  + \wt{\cE}_2 
        \end{split}
        \end{equation}
where the error $\wt{\cE}_{2}$ is such that 
        \[\pm \wt{\cE}_{2}  \leq CN^{-1/2} (\cV_N+\cN_++1)(\cN_++1)\]

To show (\ref{eq:G40order1}), we write 
\[ \frac{( N+1)}{2N^2} \sum_{\substack{p,q \in \Lambda_+^*, r \in \Lambda^*:\\ r \not = -p,-q}} \widehat{V} (r/N) e^{-B(\eta)} b_{p+r}^*b_q^*  b_{p} b_{q+r} e^{B(\eta)} = \text{V}_0+ \text{V}_1+\text{V}_2+\text{V}_3+\text{V}_4   \]
with 
        \begin{equation}\label{eq:def-G40V0-1}\begin{split}
        \text{V}_0:=\;& \frac{( N+1)}{2N^2}\sum_{p,q,r}^{*}\widehat{V} (r/N) \Big[  \gamma_{p+r}\gamma_qb^*_{p+r}  b^*_q +  \gamma_{p+r}\sigma_qb^*_{p+r}b_{-q}+\sigma_{p+r}\sigma_qb_{-p-r}b_{-q}  \\
        &\hspace{1.2cm}+\sigma_{p+r}\gamma_q\big (b^*_qb_{-p-r}-N^{-1}a^*_qa_{-p-r}\big) \Big] \, \!\Big[\sigma_p\sigma_{q+r}b^*_{-p}b^*_{-q-r} \\
        &\hspace{1.2cm}+ \sigma_p\gamma_{q+r}b^*_{-p}b_{q+r}+ \gamma_p \gamma_{q+r} b_p  b_{q+r} + \gamma_p\sigma_{q+r} \big(b^*_{-q-r}b_p-N^{-1}a^*_{-q-r}a_p\big)\Big] \\
        &+\frac{( N+1)}{2N^2} \sum_{p,q\in\Lambda_+^*}\widehat{V} ((p-q)/N)\sigma_q\gamma_q \Big[ \big(\gamma_p^2 b^*_{p}b^*_{-p} +2 \gamma_p\sigma_{p} b^*_{p}b_p -N^{-1}\gamma_p\sigma_{p}a^*_{p}a_p\\
        &\hspace{6.5cm}+\sigma_p^2b_pb_{-p}\big)\big(1-\cN_+/N\big) +\text{h.c.}\Big]\\
        &+ \frac{( N+1)}{2N^2} \sum_{p,q\in\Lambda_+^*}\widehat{V} ((p-q)/N)\sigma_q\gamma_q\sigma_p\gamma_p\,\big(1-\cN_+/N\big)^2 , \end{split} \end{equation}
\begin{equation}\label{eq:V1}\begin{split} 
        \text{V}_1:=\;& \frac{( N+1)}{2N^2}\sum_{p,q,r}^{*}\widehat{V} (r/N)  \Big[  \gamma_{p+r}\gamma_qb^*_{p+r}  b^*_q +  \gamma_{p+r}\sigma_qb^*_{p+r}b_{-q} \\ & \hspace{3cm} +\sigma_{p+r}\sigma_qb_{-p-r}b_{-q} +\sigma_{p+r}\gamma_q\big (b^*_qb_{-p-r}-N^{-1}a^*_qa_{-p-r}\big) \Big]\!\\ &\hspace{2cm} \times 
        \!  \Big[\big(\gamma_pb_p +\sigma_pb^*_{-p}\big)d_{q+r} +   d_p\big(\gamma_{q+r} b_{q+r} +\sigma_{q+r}b^*_{-q-r}\big) \Big] \\
        &+\frac{( N+1)}{2N^2} \sum_{p,q\in\Lambda_+^*}\widehat{V} ((p-q)/N)\sigma_q\gamma_q \big(1-\cN_+/N\big) \\ & \hspace{2cm} \times \Big[d_p\big( \gamma_p  b_{-p}+\sigma_p b^*_{p}\big)+\big(\gamma_p b_p +\sigma_p b^*_{-p}\big)d_{-p} \Big] +\text{h.c.}\; ,
        \end{split}\end{equation}
and
\begin{equation}\label{eq:def-G40V2-4}\begin{split}
        \text{V}_2:=\;& \frac{( N+1)}{2N^2}\sum_{p,q,r}^{*}\widehat{V} (r/N) \Big[ \big( \gamma_{p+r}b^*_{p+r} +\sigma_{p+r}b_{-p-r}  \big)d^*_q+  d^*_{p+r}\big(\gamma_qb^*_q +\sigma_qb_{-q} \big)\Big]\\
        &\hspace{3.5cm}\times \!\!\Big[\big( \gamma_pb_p +\sigma_pb^*_{-p} \big)d_{q+r}+  d_p\big( \gamma_{q+r}b_{q+r} +\sigma_{q+r}b^*_{-q-r}  \big) \Big] \\
        &+\frac{( N+1)}{2N^2} \sum_{p,q\in\Lambda_+^*}\widehat{V} ((p-q)/N)\sigma_q\gamma_q \Big[d^*_{-p}d^*_p\big(1-\cN_+/N\big)+ \big(1-\cN_+/N\big)d_p d_{-p} \Big]\; , \\
        \text{V}_3:=\;& \frac{( N+1)}{2N^2}\sum_{p,q,r}^{*}\widehat{V} (r/N) \Big[ \big( \gamma_{p+r}b^*_{p+r} +\sigma_{p+r}b_{-p-r}  \big)d^*_q+  d^*_{p+r}\big(\gamma_qb^*_q +\sigma_qb_{-q} \big )\Big]d_pd_{q+r}\\
        &+\text{h.c.}\;;\\
        \text{V}_4:=\;& \frac{( N+1)}{2N^2}\sum_{p,q,r}^{*}\widehat{V} (r/N)d^*_{p+r}d^*_q d_pd_{q+r}
        \end{split}\end{equation}
Here, we used the notation $\sum_{p,q,r}^{*}:=\sum_{p,q\in\Lambda_+^*,r\in\Lambda^*: r\neq -p,-q}$ for simplicity. Notice that the index of $\text{V}_j$ refers to the number of $d$-operators it contains. 

Let us consider $\text{V}_4$. Switching to position space and using (\ref{eq:dxy-bds}), we find 
\[ \begin{split} 
|\langle \xi , \text{V}_4 \xi \rangle | \leq \; &C \int dx dy \, N^2 V (N (x-y)) \| \check{d}_x \check{d}_y \xi \| \, \|  \check{d}_x \check{d}_y \xi \| \\ \leq \; &C \int dx dy \, V(N(x-y)) \Big[  \| (\cN_+ + 1)^2 \xi \| + N \| (\cN_+ + 1) \xi \| \\ & \hspace{2cm} + \| \check{a}_x (\cN_+ +1)^{3/2} \xi \|  + \| \check{a}_y (\cN_ + + 1)^{3/2} \xi \| + \| \check{a}_x \check{a}_y (\cN_+ + 1) \xi \| \Big]^2  \\ \leq \; &C N^{-1} \langle \xi , (\cV_N + \cN_+ + 1) ( \cN_+ + 1) \xi \rangle 
\end{split} \]

Next, we switch to the contribution $ \text{V}_3$, defined in (\ref{eq:def-G40V2-4}). Switching again to position space, using (\ref{eq:dxy-bds}) and the bound $\cN_+ \leq N$, we obtain 
\[  \begin{split} 
| \langle \xi , \text{V}_3 \xi \rangle | \leq \; &C \int dx dy \, N^2 V (N (x-y))  \| \check{d}_x \check{d}_y \xi \| \\ &\hspace{2cm} \times \Big[ \| \check{d}_y b (\check{\gamma}_x) \xi \| + \| \check{d}_y b^* (\check{\s}_x) \xi \| + \| b (\check{\g}_y) \check{d}_x \xi \| + \| b^* (\check{\s}_y) \check{d}_x \xi \| \Big]
 \\ \leq \; &C \int dx dy \, V (N (x-y)) \Big[ \| (\cN_+ + 1)^2 \xi \| + \| \check{a}_x (\cN_+ + 1)^{3/2} \xi \| \\ &\hspace{1cm} +  \| \check{a}_y (\cN_+ + 1)^{3/2} \xi \| + \| \check{a}_x \check{a}_y (\cN_+ + 1) \xi \| + N \| (\cN_+ + 1) \xi \| \Big]^2 \\ 
\leq \; &C N^{-1} \langle \xi , (\cV_N + \cN_+ + 1) ( \cN_+ + 1) \xi \rangle \end{split} \]

Proceeding similarly, $\text{V}_2$ can be bounded, switching to position space, by 
\[ \begin{split} 
| \langle \xi , \text{V}_2 \xi \rangle | \leq \; &C \int dx dy \, N^2 V (N (x-y)) \\ &\hspace{2cm} \times \Big[ \| \check{d}_y b (\check{\gamma}_x) \xi \| + \| \check{d}_y b^* (\check{\s}_x) \xi \| + \| b (\check{\g}_y) \check{d}_x \xi \| + \| b^* (\check{\s}_y) \check{d}_x \xi \| \Big]^2 \\ &+ C \int dx dy \, N^2 V (N (x-y)) |(\check{\s} * \check{\g}) (x-y)| \| (\cN_+ + 1)^{-1} \check{d}_x \check{d}_y \xi \| \| (\cN_+ + 1) \xi \| \\ 
\leq \; & C \int dx dy V(N(x-y))  \Big[ \| (\cN_+ + 1)^2 \xi \| + \| \check{a}_x (\cN_+ + 1)^{3/2} \xi \| \\ &\hspace{1cm} +  \| \check{a}_y (\cN_+ + 1)^{3/2} \xi \| + \| \check{a}_x \check{a}_y (\cN_+ + 1) \xi \| + N \| (\cN_+ + 1) \xi \| \Big]^2 \\ &+ C \int dx dy N V (N (x-y)) \| (\cN_ + +1) \xi \| \Big[ \| (\cN_+ + 1)^2 \xi \| + N \| (\cN_ + +1) \xi \| \\ &\hspace{2cm} + \| \check{a}_x (\cN_ + +1)^{3/2} \xi \| + \| \check{a}_y (\cN_+ + 1)^{3/2} \xi \| + \|  \check{a}_x \check{a}_y (\cN_+ + 1) \xi \| \Big] \\
\leq \; & C N^{-1} \langle \xi , (\cV_N + \cN_+ + 1) ( \cN_+ + 1) \xi \rangle \end{split} \]
Here we used the bound $\| \check{\s} * \check{\g} \|_\infty \leq C N$ from (\ref{eq:bdsg-x}).

Let us now study the term $\text{V}_1$. We write 
\begin{equation}\label{eq:deco-V1} \begin{split} 
\text{V}_1 = \; &\frac{1}{2N} \sum_{p,q\in\Lambda_+^*}\widehat{V} ((p-q)/N)\sigma_q\gamma_q \Big[d_p\big( \gamma_p  b_{-p}+\sigma_p b^*_{p}\big)+\big(\gamma_p b_p +\sigma_p b^*_{-p}\big)d_{-p} \Big] + \text{h.c.} \\ &+ \text{V}_{12} + \text{V}_{13} \end{split} \end{equation}
where $\text{V}_{13}$ denotes the first sum on the r.h.s. of (\ref{eq:V1}) and $\text{V}_{12}$ is the difference between the second term on the r.h.s. of (\ref{eq:V1}) and the term on the r.h.s. of (\ref{eq:deco-V1}). Switching to position space and using (\ref{eq:dxy-bds}), we find easily 
\[ \begin{split} 
|\langle \xi, \text{V}_{12} \xi \rangle | \leq \; &C \int dx dy \, N V (N (x-y)) |(\check{\s} * \check{\g}) (x-y)| \| (\cN_+ + 1) \xi \| \\ &\hspace{.5cm} \times \| \left[ \| \check{d}_x b (\check{\gamma}_y) \xi \| + \| \check{d}_x b^* (\s_y)  \xi \| + \| b (\check{\g}_x) \check{d}_y  \xi \| + \| b^* (\check{\s}_x) \check{d}_y \xi \| \right] 
\\ 
\leq \; &C N^{-1} \langle \xi , (\cV_N + \cN_+ + 1) (\cN_+ + 1) \xi \rangle 
\end{split} \]
and 
\[ \begin{split} |\langle \xi , \text{V}_{13} \xi \rangle | \leq \; &C \int dx dy N^2 V (N (x-y)) \\ & \hspace{.5cm} \times \Big[ \| (\cN_+ + 1) \xi \| + \| a_x (\cN_+ + 1)^{1/2} \xi \| + \| a_y (\cN_+ + 1)^{1/2} \xi \| + \| a_x a_y \xi \| \Big] \\ &\hspace{.5cm}  \times \Big[ \| b (\check{\gamma}_x)  \check{d}_y \xi \| + \| b^* (\check{\s}_x) \check{d}_y \xi \| + \| \check{d}_x b (\check{\gamma}_y) \xi \| + \| \check{d}_y b^* (\check{\s}_y) \xi \| \Big] 
%\\ 
%\leq \; &C \int dx dy N V (N (x-y)) \Big[ \| (\cN_+ + 1) \xi \| + \| a_x (\cN_+ + 1)^{1/2} \xi \| + \| a_y 
%(\cN_+ + 1)^{1/2} \xi \| + \| a_x a_y \xi \| \Big]  \Big[ \| a_x (\cN_+ + 1)^{3/2} \xi \| +  N \| (\cN_+ + 1) \xi %\|  + 
%\| a_y (\cN_ + + 1)^{3/2} \xi \| + \| a_x a_y (\cN_ + +1) \xi \|  + \| (\cN_+ + 1)^2 \xi \| \Big] 
\\ \leq \; &C N^{-1} \langle \xi , (\cV_N + \cN_+ + 1) (\cN_+ + 1) \xi \rangle \end{split} \] 
        
Finally, we analyse $ \text{V}_0$, as defined in (\ref{eq:def-G40V0-1}). We write $ \text{V}_0 = \text{V}_{01}+\text{V}_{02}+\text{V}_{03}$, where
        \begin{equation*}%\label{eq:def-V01-03}
        \begin{split} \text{V}_{01}:=& \;\frac{( N+1)}{2N^2} \sum_{p,q\in\Lambda_+^*}\widehat{V} ((p-q)/N)\sigma_q\gamma_q\sigma_p\gamma_p\,\big(1-\cN_+/N\big)^2; \\
         \text{V}_{02}:=&\; \frac{( N+1)}{2N^2} \sum_{p,q\in\Lambda_+^*}\widehat{V} ((p-q)/N)\sigma_q\gamma_q \Big[ \big(\gamma_p^2 b^*_{p}b^*_{-p}+2 \gamma_p\sigma_{p} b^*_{p}b_p -N^{-1}\gamma_p\sigma_{p}a^*_{p}a_p\\
        &\hspace{5.5cm}+\sigma_p^2b_pb_{-p}\big)\big(1-\cN_+/N\big) +\text{h.c.}\Big]; \\
        \text{V}_{03}:=&\; \frac{( N+1)}{2N^2}\sum_{p,q,r}^{*}\widehat{V} (r/N) \Big[  \gamma_{p+r}\gamma_qb^*_{p+r}  b^*_q +  \gamma_{p+r}\sigma_qb^*_{p+r}b_{-q}+\sigma_{p+r}\sigma_qb_{-p-r}b_{-q}  \\
        &\hspace{1.2cm}+\sigma_{p+r}\gamma_q\big (b^*_qb_{-p-r}-N^{-1}a^*_qa_{-p-r}\big) \Big]\! \!\Big[\sigma_p\sigma_{q+r}b^*_{-p}b^*_{-q-r} + \sigma_p\gamma_{q+r}b^*_{-p}b_{q+r}\\
        &\hspace{1.2cm} + \gamma_p \gamma_{q+r} b_p  b_{q+r} + \gamma_p\sigma_{q+r} \big(b^*_{-q-r}b_p-N^{-1}a^*_{-q-r}a_p\big)\Big] \\
        \end{split}
        \end{equation*}
Proceeding similarly as above, switching to position space and using (in the estimate for $\wt{\cE}_4$) the bound $\| \check{\s} * \check{\g} \|_\infty \leq C N$, we find that 
        \begin{equation*}%\label{eq:V01bnd1}
        \begin{split}
        \text{V}_{01} =&\; \frac{1}{2N} \sum_{p,q\in\Lambda_+^*}\widehat{V} ((p-q)/N)\sigma_q\gamma_q\sigma_p\gamma_p\,\big(1+1/N-2\cN_+/N\big) + \wt{\cE}_{3} \\
         \text{V}_{02}= &\; \frac{1}{2N} \sum_{p,q\in\Lambda_+^*}\widehat{V} ((p-q)/N)\sigma_q\gamma_q \Big[ \gamma_p^2 b^*_{p}b^*_{-p}+2 \gamma_p\sigma_{p} b^*_{p}b_p+\sigma_p^2b_pb_{-p} +\text{h.c.} \Big] + \wt{\cE}_4 \\
         \text{V}_{03} = &\; \cV_N + \wt{\cE}_5
        \end{split}
        \end{equation*}
Combining with (\ref{eq:deco-V1}) and with all other bounds for the error terms, we arrive at (\ref{eq:G40order1}). 

           {\it Step 2.} We claim that 
  \begin{equation}\label{eq:G41order1}
        \begin{split}  
        \frac{1}{N^2} &\sum_{\substack{p,q,u \in \Lambda_+^*, r \in \Lambda^*:\\ r \not = -p,-q}} \widehat{V} (r/N) e^{-B(\eta)} b_{p+r}^*b_q^* b^*_ub_u  b_{p} b_{q+r} e^{B(\eta)} \\
         = &\;\frac{1}{N^2} \sum_{p,q,u\in\Lambda_+^*}\widehat{V} ((p-q)/N)\sigma_q\gamma_q\sigma_p\gamma_p \Big[ \gamma_u^2 b^*_u b_u + \sigma_u^2 b^*_u b_u + \gamma_u\sigma_u b^*_u b^*_{-u}+\gamma_u\sigma_u b_u b_{-u}  \Big]\,\\
        &\;+\frac{1}{N^2} \sum_{p,q,u\in\Lambda_+^*}\widehat{V} ((p-q)/N)\sigma_q\gamma_q\sigma_p\gamma_p\sigma_u^2 + \wt{\cE}_6
        \end{split}
        \end{equation}
where the error $\wt{\cE}_6$ is such that, on $ \cF_+^{\leq N}$, 
        \[\pm \cE_{3,N}^{(4)}  \leq CN^{-1/2} (\cV_N+\cN_++1)(\cN_++1)\]

To show (\ref{eq:G41order1}), we split   
\begin{equation*}%\label{eq:G41dec}
\frac{1}{N^2} \sum_{\substack{p,q,u \in \Lambda_+^*, r \in \Lambda^*:\\ r \not = -p,-q}} \widehat{V} (r/N) e^{-B(\eta)} b_{p+r}^*b_q^* b^*_ub_u  b_{p} b_{q+r} e^{B(\eta)} 
= \text{W}_0+ \text{W}_1+\text{W}_2  \end{equation*}
where
        \begin{equation*}%\label{eq:def-G41W0-1}
        \begin{split}
        \text{W}_0:=\;&\frac{1}{N^2} \sum_{p,q,u\in\Lambda_+^*}\widehat{V} ((p-q)/N)\sigma_q\gamma_q\sigma_p\gamma_p\,\big[1-\cN_+/N\big]\big[e^{-B(\eta)} b^*_ub_u e^{B(\eta)}\big]\big[1-\cN_+/N\big];\\
        \end{split}\end{equation*}
and
\begin{equation*}%\label{eq:def-G41W2-4}
\begin{split}
        \text{W}_1:=\;& \frac{1}{N^2} \sum_{p,q,u\in\Lambda_+^*}\widehat{V} ((p-q)/N)\sigma_q\gamma_q \Big[ \gamma_p^2 b^*_{p}b^*_{-p} +2 \gamma_p\sigma_{p} b^*_{p}b_p -N^{-1}\gamma_p\sigma_{p}a^*_{p}a_p+\sigma_p^2b_pb_{-p}\\
        &\hspace{4.5cm}+\gamma_{p}b^*_{-p}d^*_p +\sigma_{p}b_{p} d^*_p+  \gamma_p d^*_{-p}b^*_p +\sigma_p d^*_{p}b_{p}+d^*_{p+r}d^*_q\Big] \\
        &\hspace{4.5cm}\times   \big[e^{-B(\eta)} b^*_ub_u e^{B(\eta)}\big]\big[1-\cN_+/N\big] +\text{h.c.};\\
        \text{W}_2:=\;& \frac{1}{N^2}\sum_{p,q,r,u}^{*}\widehat{V} (r/N) \Big[  \gamma_{p+r}\gamma_qb^*_{p+r}  b^*_q +  \gamma_{p+r}\sigma_qb^*_{p+r}b_{-q}+\sigma_{p+r}\sigma_qb_{-p-r}b_{-q}\\
        &\hspace{0.5cm}+\sigma_{p+r}\gamma_q\big (b^*_qb_{-p-r}-N^{-1}a^*_qa_{-p-r}\big)+ \big( \gamma_{p+r}b^*_{p+r} +\sigma_{p+r}b_{-p-r}  \big)d^*_q\\
        &\hspace{0.5cm}+  d^*_{p+r}\big(\gamma_qb^*_q +\sigma_qb_{-q} \big)+d^*_{p+r}d^*_q\Big]\times \Big[e^{-B(\eta)} b^*_ub_u e^{B(\eta)}\Big]\times \Big[ \sigma_p\sigma_{q+r}b^*_{-p}b^*_{-q-r}\\
        &\hspace{0.5cm}+ \sigma_p\gamma_{q+r}b^*_{-p}b_{q+r} + \gamma_p \gamma_{q+r} b_p  b_{q+r}+ \gamma_p\sigma_{q+r} \big(b^*_{-q-r}b_p-N^{-1}a^*_{-q-r}a_p\big)  \\
        &\hspace{0.5cm}+ \big( \gamma_pb_p +\sigma_pb^*_{-p} \big)d_{q+r} +  d_p\big( \gamma_{q+r}b_{q+r} +\sigma_{q+r}b^*_{-q-r}  \big)+d_pd_{q+r}\Big]
        \end{split}\end{equation*}
Here, we introduced the notation $\sum_{p,q,r,u}^{*}:=\sum_{p,q,u\in\Lambda_+^*,r\in\Lambda^*: r\neq -p,-q}$ for simplicity. Using Lemma \ref{lm:Ngrow} to get rid of the factor $\sum_{u \in \Lambda_+^*} e^{-B(\eta)} b_u^* b_u e^{B(\eta)} =  e^{-B(\eta)} \cN_+ (1- \cN_+ /N) e^{B(\eta)}$ and then proceeding similarly as in Step 1, we obtain that 
\[ \begin{split} 
|\langle \xi , \text{W}_1 \xi \rangle | &\leq C N^{-1/2} \langle \xi , (\cV_N + \cN_+ + 1) ( \cN_+ + 1) \xi \rangle \\
|\langle \xi , \text{W}_2 \xi \rangle | &\leq C N^{-1} \langle \xi , (\cV_N + \cN_+ + 1) ( \cN_+ + 1) \xi \rangle \end{split} \]

As for $\text{W}_0$, we write 
        \[\begin{split}\text{W}_0 = \; &\frac{1}{N^2} \sum_{p,q,u\in\Lambda_+^*}\widehat{V} ((p-q)/N)\sigma_q\gamma_q\sigma_p\gamma_p \\ &\hspace{2cm} \times \Big[ \gamma_u^2 b^*_u b_u + \sigma_u^2 b^*_u b_u + \gamma_u\sigma_u b^*_u b^*_{-u}+\gamma_u\sigma_u b_u b_{-u} + \s_u^2 \Big] \\ &+ \wt{\cE}_7 
        \end{split}\]
Using (\ref{eq:ebe}) to decompose  
\[  e^{-B(\eta)} b^*_ub_u e^{B(\eta)} = \big(\gamma_u b^*_u + \sigma_ub_{-u} + d^*_u  \big) \big(\gamma_u b_u + \sigma_u b^*_{-u} + d_u\big)  \]
and then the bounds (\ref{eq:d-bds}), it is easy to estimate the remainder operator $\wt{\cE}_7$, on $\cF_+^{\leq N}$, by  
\[ \pm \wt\cE_7 \leq C N^{-1} (\cN_+ + 1) \] 
Hence, we obtain (\ref{eq:G41order1}). 

{\it Step 3. Conclusion of the proof.} Combining (\ref{eq:G40order1}) and (\ref{eq:G41order1}) with (\ref{eq:def-G40G41}), we conclude that 
 \begin{equation}\label{eq:concl-pr}
        \begin{split}  \cG_N^{(4)} = &\;\cV_N+ \frac{1}{2N} \sum_{p,q\in\Lambda_+^*}\widehat{V} ((p-q)/N)\sigma_q\gamma_q\sigma_p\gamma_p\,\big(1+1/N-2\;\cN_+/N\big)\\
        &\; + \frac{1}{2N} \sum_{p,q\in\Lambda_+^*}\widehat{V} ((p-q)/N)\sigma_q\gamma_q \Big[ \gamma_p^2 b^*_{p}b^*_{-p}+2 \gamma_p\sigma_{p} b^*_{p}b_p+\sigma_p^2b_pb_{-p}+d_p\big( \gamma_p  b_{-p}+\sigma_p b^*_{p}\big)\\
        &\;\hspace{5cm}+\big(\gamma_p b_p +\sigma_p b^*_{-p}\big)d_{-p} + \text{h.c.}\Big]\\
        &\;+\frac{1}{N^2} \sum_{p,q,u\in\Lambda_+^*}\widehat{V} ((p-q)/N)\sigma_q\gamma_q\sigma_p\gamma_p \\ & \hspace{2cm} \times \Big[ \gamma_u^2 b^*_u b_u + \sigma_u^2 b^*_u b_u + \gamma_u\sigma_u b^*_u b^*_{-u}+\gamma_u\sigma_u b_u b_{-u}  + \s_u^2  \Big]\\ &+ \wt{\cE}_8
        \end{split}
        \end{equation}
with an error $ \wt{\cE}_8$ such that, on $ \cF_+^{\leq N}$, 
        \[\pm \wt{\cE}_8  \leq CN^{-1/2} (\cV_N+\cN_++1)(\cN_++1) \]
To conclude the proof of (\ref{eq:lmGN4}), we just observe that in the term appearing on the second line on the r.h.s. of (\ref{eq:concl-pr}), we can replace the product $\s_q \g_q$ simply by $\eta_q$. Since $| \s_q \g_q - \eta_q| \leq C |q|^{-6}$ (or, in position space, $\| (\check{\s} * \check{\g}) - \check{\eta} \|_\infty  \leq C$), it is easy to show that the difference can be incorporated in the error term. Similarly, in the term appearing on the third and fourth lines on the r.h.s. of (\ref{eq:concl-pr}), we can replace $\s_p \g_p \s_q \g_q$ by $\eta_p \eta_q$; also in this case, the contribution of the difference is small and can be included in the remainder. This concludes the proof of the lemma.
\end{proof}

%%%%%%%%%%%%%%%%%%%%%%%%%%%%%%%%%%%%%%%%%%%%%%%%%%%%%%%%%%%%%%%%%%%%%%%%%%%%%%%%%%%%%%%%%%%%%%%%%%%%%%%%%%%%%%%%%%%%%%%%%%%%%%%%%%%%%%%%%%%%%%%%%%%%%%%%%%%%%%%%%%%%%%%%%%%%%
%%%%%%%%%%%%%%%%%%%%%%%%%%%%%%%%%%%%%%%%%%%%%%%%%%%%%%%%%%%%%%%%%%%%%

%%%%%%%%%%%%%%%%%%%%%%%%%%%%%%%%%%%%%%%%%%%%%%%%%%%%%%%%%%%%%%%%%%%%%%%%%%%%%%%%%%%%%%%%%%%%%%%%%%%%%%%%
%%%%%%%%%%%%%%%%%%%%%%%%%%%%%%%%%%%%%%%%%%%%%%%%%%%%%%%%%%%%%%%%%%%%%%%%%%%%%%%%%%%%%%%%%%%%%%%%%%%%%%%%%%%%%%%%%%%%%%%%%%%%%%%%%%%%%%%%%%%%%%%%%%%%%%%%%%%%%%%%%%%%%%%%%%%%%
%%%%%%%%%%%%%%%%%%%%%%%%%%%%%%%%%%%%%%%%%%%%%%%%%%%%%%%%%%%%%%%%%%%%%%%%%%%%%%%%%%%%%%%%%%%%%%%%%%%%%%%%%%%%%%%%%%%%%%%%%%%%%%%%%%%%%%%%%%%%%%%%%%%%%%%%%%%%%%%%%%%%%%%%%%%%%
%%%%%%%%%%%%%%%%%%%%%%%%%%%%%%%%%%%%%%%%%%%%%%%%%%%%%%%%%%%%%%%%%%%%%%%%%%%%%%%%%%%%%%%%%%%%%%%%%%%%%%%%%%%%%%%%%%%%%%%%%%%%%%%%%%%%%%%%%%%%%%%%%%%%%%%%%%%%%%%%%%%%%%%%%%%%%
%%%%%%%%%%%%%%%%%%%%%%%%%%%%%%%%%%%%%%%%%%%%%%%%%%%%%%%%%%%%%%%%%%%%%%%%%%%%%%%%%%%%%%%%%%%%%%%%%%%%%%%%%%%%%%%%%%%%%%%%%%%%%%%%%%%%%%%%%%%%%%%%%%%%%%%%%%%%%%%%%%%%%%%%%%%%%
\subsection{Proof of Proposition \ref{prop:gene}}

Collecting the results of (\ref{eq:GN0}), Lemma \ref{prop:G2}, Lemma \ref{lm:GN-3} and Lemma \ref{lm:GN-4}, we obtain that 
        \begin{equation} \label{eq:GNpf1}
        \cG_N =  \wt{C}_{\cG_N} + \wt{\cQ}_{\cG_N} + \cD_N+ \cH_N + \cC_N  + \wt{\cE}_{\cG_N} \end{equation}
where $\cC_N$ is the cubic term defined in \eqref{eq:cCN}, $\wt{\cE}_{\cG_N}$ an error term controlled by 
        \[\pm \wt{\cE}_{\cG_N} \leq CN^{-1/2} (\cH_N+\cN_+^2)(\cN_++1),\]
and where $\wt{C}_{\cG_N}$, $\wt{\cQ}_{\cG_N}$ and $\cD_N$ are given by 
        \begin{equation}\label{eq:defC0}
        \begin{split}
        \wt{C}_{\cG_N} =& \;\frac{(N-1)}{2} \widehat{V} (0) +\sum_{p\in\L^*_+}\Big[p^2\sigma_p^2\left(1+\frac{1}{N}\right)+\widehat{V}(p/N)\left(\s_p\g_p +\s_p^2\right)\Big]\\
        &+\frac{1}{2N} \sum_{q\in\Lambda_+^*}\widehat{V} ((p-q)/N)\sigma_q\gamma_q\sigma_p\gamma_p\,\big(1+1/N\big) \\&+\frac{1}{N}\sum_{u\in\L^*_+} \s_u^2 \sum_{p\in\L^*_+}\big[p^2\s_p^2+\frac{1}{N}\sum_{q\in\L^*_+}\widehat{V}((p-q)/N)\eta_p\eta_q\big]
        \end{split}
        \end{equation}
        \begin{equation}\label{eq:defQN1}
        \begin{split}
        \wt{\cQ}_{\cG_N} =& \;\sum_{p\in\L^*_+}b^*_pb_p \big[ 2\s_p^2 p^2 +\widehat{V}(p/N)\left(\g_p+\s_p\right)^2+\frac{2}{N}\g_p\s_p\sum_{q\in\L^*}\widehat{V}((p-q)/N) \eta_q\big]\\
        &+\sum_{p\in\L^*_+}\left(b^*_pb^*_{-p}+ b_pb_{-p}\right)      \\ &\hspace{1cm} \times     \big[p^2\s_p\g_p+\frac{1}{2}\widehat{V}(p/N)(\g_p+\s_p)^2 +\frac{1}{2N}\sum_{q\in\L^*}\widehat{V}((p-q)/N) \eta_q (\g_p^2+\s_p^2)\big]\\
        &-\frac{\cN_+}{N}\sum_{p\in\L^*_+}\big[p^2\s_p^2+\widehat{V}(p/N)\g_p\s_p+\frac{1}{N}\sum_{q\in\L^*_+}\widehat{V}((p-q)/N)\g_p\s_p\g_q\s_q\big]\\
        &+\frac{1}{N}\sum_{u\in\L^*_+}\big[(\g_u^2+\s_u^2)b^*_ub_u+\g_u\s_u(b^*_ub^*_{-u}+b_ub_{-u})\big]\\
        &\hspace{5.7cm}\times\sum_{p\in\L^*_+}\big[p^2\s_p^2+\frac{1}{N}\sum_{q\in\L^*_+}\widehat{V}((p-q)/N)\eta_p\eta_q\big]\\
                \end{split}
        \end{equation}
        and
        \begin{equation}\label{eq:defQN2}
        \begin{split}
        \cD_N =& \;  \sum_{p\in\Lambda_+^*}\Big[p^2\eta_pb_pd_{-p} + \frac{1}{2} \widehat{V}(p/N)b_pd_{-p}+\frac{1}{2N}\big(\widehat{V} (\cdot/N)\ast \eta \big)_p  b_pd_{-p} +\text{h.c.}\Big]\\
         &\;+ \frac{1}{2} \sum_{p\in\Lambda_+^*}  (\widehat{V} (./N) * \widehat{f}_{\ell,N}) (p) 
         \Big[(\gamma_p-1) b_pd_{-p} +\sigma_p b^*_{-p}d_{-p}+\text{h.c.}\Big]\\
        &\;+  \frac{1}{2} \sum_{p\in\Lambda_+^*}  (\widehat{V} (./N) * \widehat{f}_{\ell,N}) (p) 
       \Big[  \gamma_p  d_pb_{-p}+\sigma_p d_pb^*_{p}+\text{h.c.}\Big]\\
        \end{split}
        \end{equation}
with $\widehat{f}_{\ell,N}$ defined as in (\ref{eq:fellN}). Next, we analyse the operator $\cD_N$, which still contains $d$-operators, to extract the important contributions. To this end, we write $\cD_N = \text{D}_1 + \text{D}_2 + \text{D}_3$, where $\text{D}_1, \text{D}_2, \text{D}_3$ denote the operators on the first, second and, respectively, third line on the r.h.s. of (\ref{eq:defQN2}).

Using the relation (\ref{eq:eta-scat}) and the bound (\ref{eq:d-bds}), we find 
\[ \begin{split} 
|\langle \xi , \text{D}_1 \xi \rangle | \leq \; & \sum_{p \in \L^*_+}  |(\widehat{\chi}_\ell \ast \widehat{f}_{\ell,N} ) (p)| \| (\cN_+ + 1) \xi \| \| (\cN_+ + 1)^{-1/2} d_{-p} \xi \| \\ \leq \; &\frac{C}{N}  \| (\cN_+ + 1) \xi \| \sum_{p \in \L^*_+}  \left[ \frac{1}{p^4} \| (\cN_+ + 1) \xi \| + \frac{1}{p^2} \| b_p (\cN_+ + 1)^{1/2} \xi \| \right] \\ \leq \; & \frac{C}{N} \| (\cN_+ + 1) \xi \|^2 \end{split} \]
Similarly, using (\ref{eq:d-bds}), we find 
\[ \pm \text{D}_2 \leq C N^{-1} (\cN_+ + 1)^2 \]
Thus, we switch to $\text{D}_3$. We split $\text{D}_3 = \text{D}_{31} + \text{D}_{32} + \text{D}_{33}$, with 
\[\begin{split}
        \text{D}_{31} :=& \; \frac{1}{2}  \sum_{p\in\Lambda_+^*}  (\widehat{V} (./N) * \widehat{f}_{\ell,N}) (p)  \gamma_p  d_pb_{-p}+\text{h.c.} \\
         \text{D}_{32} :=& \;  \frac{1}{2}    \sum_{p\in\Lambda_+^*}  (\widehat{V} (./N) * \widehat{f}_{\ell,N}) (p)  (\sigma_p-\eta_p) d_pb^*_{p}+\text{h.c.} \\
        \text{D}_{33} :=& \;  \frac{1}{2}   \sum_{p\in\Lambda_+^*} (\widehat{V} (./N) * \widehat{f}_{\ell,N}) (p)  \eta_p  d_pb^*_{p}+\text{h.c.} \, 
        \end{split} \]
Switching to position space and using (\ref{eq:dxy-bds}), we observe that 
\[ \begin{split} 
| \langle \xi , \text{D}_{31} \xi \rangle | \leq \; &C \int dx dy N^3 V(N (x-y)) f_\ell (N(x-y))  \| (\cN_+ + 1) \xi \| \| (\cN_+ + 1)^{-1} \check{d}_x b ( \check{\gamma}_y) \xi \|  \\
\leq \; &C \int dx dy N^2  V(N(x-y)) f_\ell (N(x-y)) \| (\cN_+ + 1) \xi \| \\ &\hspace{.5cm} \times \left[ \| (\cN_+ + 1) \xi \| + \| \check{a}_x (\cN_+ + 1)^{1/2} \xi \| + \| \check{a}_y (\cN_+ + 1)^{1/2} \xi \| + \| \check{a}_x \check{a}_y \xi \| \right] \\
\leq \; &C N^{-1/2}\langle \xi , (\cV_N + \cN_+ + 1) ( \cN_+ + 1) \xi \rangle \end{split} \]
As for $\text{D}_{32}$, we can use the decay of $|\s_p - \eta_p| \leq C |p|^{-6}$ to prove that
\[ \pm \text{D}_{32} \leq C N^{-1} (\cN_+ + 1)^2 \]
We are left with $\text{D}_{33}$; here we cannot apply (\ref{eq:d-bds}) because of the lack of decay in $p$. This term contains contributions that are relevant in the limit of large $N$. To isolate these contributions, it is useful to rewrite the remainder operator $d_p$ as 
\begin{equation}\label{eq:dp-dec} \begin{split} 
d_p = \; &e^{-B(\eta)} b_p e^{B(\eta)} - \g_p b_p - \s_p b_{-p}^* \\ = \; &(1-\g_p) b_p - \s_p b_{-p}^* +\eta_p  \int_0^1 ds \, e^{-sB(\eta)} \frac{N-\cN_+}{N} b_{-p}^* e^{s B(\eta)} \\ &\hspace{4cm}- \frac{1}{N} \int_0^1 ds \sum_{q \in \L^*_+} \eta_q e^{-sB(\eta)} b_q^* a_{-q}^* a_p e^{sB(\eta)} \\ = \; & \eta_p \int_0^1 ds \, d^{(s)*}_{-p}  - \frac{\eta_p}{N} \int_0^1 ds \;e^{-s B(\eta)} \cN_+ b_{-p}^* e^{s B(\eta)}  \\ &\hspace{4cm} - \frac{1}{N} \int_0^1 ds \sum_{q \in \L^*_+} \eta_q e^{-sB(\eta)} b_q^* a_{-q}^* a_p e^{sB(\eta)} \end{split} \end{equation}
where, in the last step, we wrote $e^{-sB(\eta)} b_{-p}^* e^{s B(\eta)} = \g^{(s)}_p b_{-p}^* + \s_p^{(s)} b_p + d_{-p}^{(s)*}$ (the label $s$ indicates that the coefficients $\g_p^{(s)}$, $\s_p^{(s)}$ and the operator $d_{-p}^{(s)*}$ are defined with $\eta$ replaced by $s \eta$, for an $s \in [0;1]$) and we integrated $\g_p^{(s)}$ and $\s_p^{(s)}$ over $s \in [0;1]$. Inserting (\ref{eq:dp-dec}) into $\text{D}_{33}$ and using the additional factor $\eta_p$ appearing in the first two terms on the r.h.s. of (\ref{eq:dp-dec}), we conclude that 
\[ \text{D}_{33} = -\frac{1}{2N} \int_0^1 ds \,  \sum_{p,q \in \L^*_+} (\widehat{V} (./N) \ast \widehat{f}_{\ell,N}) (p) \eta_p \eta_q \Big[e^{-s B(\eta)} b_q^* a_{-q}^* a_p e^{s B(\eta)} b_p^* + \text{h.c.}\Big] + \wt{\cE}_1 \]
with an error operator $\wt{\cE}_1$ such that 
\begin{equation}\label{eq:err-wtce1} \pm \wt{\cE}_1 \leq C N^{-1} (\cN_+ + 1)^2 \end{equation}
We expand 
		\[ -  e^{-sB(\eta)} a^*_{-q} a_p  e^{sB(\eta)}= - a^*_{-q} a_p-\int_0^s dt\; e^{-tB(\eta)} \big(\eta_p b^*_{-q} b^*_{-p} +\eta_q b_{q} b_p \big) e^{tB(\eta)} \]
Again, the contribution containing the additional factor $\eta_p$ is small. Hence, we have 
\[\begin{split}
		\text{D}_{33} = & \;- \frac{1}{2N} \int_0^1ds \; \sum_{p,q\in\Lambda_+^*} \big(\widehat{V} (\cdot/N)\ast \widehat{f}_{\ell,N} \big) (p) \,\Big[  \eta_p\eta_q e^{-sB(\eta)} b^*_q e^{sB(\eta)}\\
		&\hspace{4.5cm}\times \big(a^*_{-q} a_pb^*_{p} + \int_0^s dt \; \eta_q e^{-tB(\eta)}b_{q} b_p e^{tB(\eta)}b^*_{p}\big)+\text{h.c.}\Big] + \wt{\cE}_2 \\
		\end{split}\]
where, similarly to (\ref{eq:err-wtce1}), $\pm \wt{\cE}_2 \leq CN^{-1} (\cN_+ + 1)^2$. In the contribution proportional to $a^*_{-q} a_p b_p^*$ we commute $b_p^*$ to the left. In the other term, we expand $e^{-tB(\eta)} b_p e^{t B(\eta)} = \g^{(t)}_p b_p + \s^{(t)}_p b_{-p}^* + d^{(t)}_p$ using the notation introduced after (\ref{eq:dp-dec}) and we commute the contribution $\g^{(t)}_p b_p$ to the right of $b_p^*$. We obtain $\text{D}_{33} = \text{D}_{331} + \text{D}_{332} + \wt{\cE}_2$, with 
		\begin{equation*}%\label{eq:defQ331Q332}
        \begin{split}
        \text{D}_{331} :=& \; - \frac{1}{2N} \int_0^1ds\; \sum_{p,q\in\Lambda_+^*}
        \big(\widehat{V} (\cdot/N)\ast \widehat{f}_{\ell,N} \big)_p\eta_p\eta_q e^{-sB(\eta)} b^*_qe^{sB(\eta)} b^*_{-q} \\
        &\; - \frac{1}{2N} \int_0^1ds\int_0^s dt\; \sum_{p,q\in\Lambda_+^*}
        \big(\widehat{V} (\cdot/N)\ast \widehat{f}_{\ell,N}\big)_p\eta_p\eta_q^2 e^{-sB(\eta)} b^*_qe^{sB(\eta)} e^{-tB(\eta)} b_{q}e^{tB(\eta)}  \\ &\;+\text{h.c.} , \\
         \text{D}_{332} :=& \;- \frac{1}{2N} \int_0^1ds\; \sum_{p,q\in\Lambda_+^*} 
         \big(\widehat{V} (\cdot/N)\ast \widehat{f}_{\ell,N}\big)_p\eta_p\eta_q e^{-sB(\eta)} b^*_qe^{sB(\eta)} b^*_{-q}a^*_pa_{p}  \\
         &\;- \frac{1}{2N} \int_0^1ds\int_0^s dt\; \sum_{p,q\in\Lambda_+^*}
         \big(\widehat{V} (\cdot/N)\ast \widehat{f}_{\ell,N} \big)_p\eta_p\eta_q^2e^{-sB(\eta)} b^*_qe^{sB(\eta)} e^{-tB(\eta)} b_{q}e^{tB(\eta)} \\
         &\;\hspace{0.5cm}\times\Big[b^*_pb_p - N^{-1}\cN_+ -N^{-1}a^*_pa_p  +\big(\gamma_p^{(t)}-1\big)b_pb^*_p +\sigma_p^{(t)}b^*_{-p}b^*_p + d_p^{(t)}b^*_p \Big] +\text{h.c.}
        \end{split}
        \end{equation*}
Since $(\gamma_p^{(t)}-1)\leq C \eta_p$ and $ \sigma_p^{(t)} \leq C \eta_p$ , we can bound
		\begin{equation*}
%\label{eq:Q332bnd1} 
\big|  \langle \xi, \text{D}_{332} \xi\rangle  \big| \leq \; CN^{-1} \langle \xi,( \cN_++1)^2\xi \rangle\end{equation*}
We are left with the operator $\text{D}_{331}$, which is quadratic in the $b$-fields. Expanding 
$e^{-sB(\eta)} b_q^* e^{sB(\eta)} = \g_q^{(s)} b_q^* + \s_q^{(s)} b_{-q} + d_q^{(s),*}$ and $e^{-tB(\eta)} b_q e^{tB(\eta)} = \g_q^{(t)} b_q + \s_q^{(t)} b_{-q}^* + d_q^{(t)}$ and using the bounds (\ref{eq:d-bds}), we obtain
		\begin{equation}\label{eq:D331}
		\begin{split}
		\text{D}_{331} = & \; - \frac{1}{2N} \int_0^1ds\; \sum_{p,q\in\Lambda_+^*} \big(\widehat{V} (\cdot/N)\ast \widehat{f}_{\ell,N} \big)_p\eta_p\eta_q \, \Big[\gamma_q^{(s)} b^*_q  b^*_{-q}+\sigma_q^{(s)} b^*_q  b_{q} + \sigma_{q}^{(s)} +\text{h.c.}\Big] \\
		& \; - \frac{1}{2N} \int_0^1ds\int_0^s dt\; \sum_{p,q\in\Lambda_+^*}
		\big(\widehat{V} (\cdot/N)\ast \widehat{f}_{\ell,N} \big)_p\eta_p\eta_q^2 \,
		\Big[\gamma_q^{(s)}\gamma_q^{(t)} b^*_q  b_{q}+\gamma_q^{(s)}\sigma_q^{(t)} b^*_q  b^*_{-q}  \\
		&\;\hspace{4cm}+\sigma_q^{(s)}\gamma_q^{(t)} b_q  b_{q}+\sigma_q^{(s)}\sigma_q^{(t)} b^*_q  b_{q}+\sigma_q^{(s)}\sigma_q^{(t)} +\text{h.c.}  \Big] + \wt{\cE}_2 
		\end{split}
		\end{equation}
where the operator $\wt{\cE}_2$ is such that $\pm \wt{\cE}_2 \leq C N^{-1} (\cN_+ + 1)^2$.
Integrating (\ref{eq:D331}) over $t$ and $s$, we conclude that 
		\begin{equation}\label{eq:cQN2final}\begin{split}
		 \cD_N =&\;  - \frac{1}{2N}\sum_{p\in\Lambda^*,q\in\Lambda_+^*}\big(\widehat{V} (\cdot/N)\ast \widehat{f}_{\ell,N}\big) (p) \, \eta_p\Big[\gamma_q\sigma_q \big(b^*_q  b^*_{-q} +b_q  b_{-q}\big)+(\s^2_q+\g^2_q)b^*_q  b_{q}+\sigma_q^2 \Big] \\
		 &+ \frac{1}{2N}\sum_{p\in\Lambda^*,q\in\Lambda_+^*}\big(\widehat{V} (\cdot/N)\ast \widehat{f}_{\ell,N}\big) (p) \, \eta_p\,b^*_q  b_{q}  +\wt{\cE}_3  
		 \end{split}\end{equation}
where the error $\wt{\cE}_3$ satisfies 
		\begin{equation*} 
		\pm\, \wt{\cE}_3 \leq C N^{-1/2} ( \cV_N+\cN_++1)( \cN_++1) \end{equation*}
Notice that, in (\ref{eq:cQN2final}), we are summing also over the mode $p=0$ (this contribution is small, it can be inserted in the operator $\wt{\cE}_3$).   

Inserting (\ref{eq:cQN2final}) into (\ref{eq:GNpf1}), we obtain the decomposition (\ref{eq:deco2-GN}) of the Hamiltonian $\cG_N$. In fact, combining (\ref{eq:defC0}) with the constant contribution in (\ref{eq:cQN2final}), we find
\begin{equation}\label{eq:comp-const} \begin{split} 
\wt{C}_{\cG_N} - &\frac{1}{2N}\sum_{p\in\Lambda^*,q\in\Lambda_+^*}\big(\widehat{V} (\cdot/N)\ast \widehat{f}_{\ell,N}\big) (p) \, \eta_p \sigma_q^2
\\ = \; &\frac{(N-1)}{2}  \widehat{V} (0) +\sum_{p\in\L^*_+}\Big[p^2 \sigma_p^2 +\widehat{V}(p/N)\left(\s_p\g_p +\s_p^2\right) \Big] \\ & +\frac{1}{2N} \sum_{q\in\Lambda_+^*}\widehat{V} ((p-q)/N)\sigma_q\gamma_q\sigma_p\gamma_p +\frac{1}{N} \sum_{p \in \L_+^*} \left[ p^2 \eta_p^2 + \frac{1}{2N} (\widehat{V} (./N) \ast \eta) (p) \eta_p \right]  \\ &
+\frac{1}{N}\sum_{u \in\L^*_+} \s_u^2 \sum_{p\in\L^*_+}\left[p^2\eta_p^2 - \frac{1}{2} \widehat{V} (p/N) \eta_p + \frac{1}{2N}\sum_{q\in\L^*_+} (\widehat{V} (./N)\ast \eta) (p) \, \eta_p \right]
+ \cO (N^{-1})   \end{split} \end{equation}
(the error $\cO(N^{-1})$ arises from the substitution $\s_p \to \eta_p$ in the terms appearing 
at the end of the third and on the fourth line). Using the relation (\ref{eq:eta-scat0}), we have 
\begin{equation}\label{eq:rel-fin7} p^2 \eta_p^2 -  \frac{1}{2} \widehat{V} (p/N) \eta_p + \frac{1}{2N}\sum_{q\in\L^*_+} (\widehat{V}\ast \eta) (p) \, \eta_p =-  \widehat{V} (p/N) \eta_p + N^3 \lambda_\ell (\widehat{\chi}_\ell * \widehat{f}_{\ell,N}) (p) \eta_p \end{equation}
Since $N^3 \lambda_\ell = \cO (1)$ and $\| (\widehat{\chi}_\ell * \widehat{f}_{\ell,N}) \eta \|_1 \leq \| \widehat{\chi}_\ell * \widehat{f}_{\ell,N} \|_2 \| \eta \|_2 = \| \chi_\ell f_\ell \|_2 \| \eta \|_2 \leq \| \chi_\ell \|_2 \| \eta \|_2 \leq C$, uniformly in $N$, we conclude that the r.h.s. (\ref{eq:comp-const}) coincides with (\ref{eq:CN}), up to errors of order $N^{-1}$. Similarly, combining the quadratic term (\ref{eq:defQN1}) with the quadratic terms on the r.h.s. of (\ref{eq:cQN2final}) (and using again the relation (\ref{eq:rel-fin7}), we obtain (\ref{eq:QN}), up to terms that can be incorporated in the error. We omit these last  details. 

%%%%%%%%%%%%%%%%%%%%%%%%%%%%%%%%%%%%%%%%%%%%%%%%%%%%%%%%%%%%%%%%%%%%%%%%%%%%%%%%%%%%%%%%%%%%%%%%%%%%%%%%%%%%%%%%%%%%%%%%%%%%%%%%%%%%%%%%%%%%%%%%%%%%%%%%%%%%%%%%%%%%%%%%%%%%%
%%%%%%%%%%%%%%%%%%%%%%%%%%%%%%%%%%%%%%%%%%%%%%%%%%%%%%%%%%%%%%%%%%%%%%%%%%%%%%%%%%%%%%%%%%%%%%%%%%%%%%%%%%%%%%%%%%%%%%%%%%%%%%%%%%%%%%%%%%%%%%%%%%%%%%%%%%%%%%%%%%%%%%%%%%%%%
%%%%%%%%%%%%%%%%%%%%%%%%%%%%%%%%%%%%%%%%%%%%%%%%%%%%%%%%%%%%%%%%%%%%%%%%%%%%%%%%%%%%%%%%%%%%%%%%%%%%%%%%%%%%%%%%%%%%%%%%%%%%%%%%%%%%%%%%%%%%%%%%%%%%%%%%%%%%%%%%%%%%%%%%%%%%%
%%%%%%%%%%%%%%%%%%%%%%%%%%%%%%%%%%%%%%%%%%%%%%%%%%%%%%%%%%%%%%%%%%%%%%%%%%%%%%%%%%%%%%%%%%%%%%%%%%%%%%%%%%%%%%%%%%%%%%%%%%%%%%%%%%%%%%%%%%%%%%%%%%%%%%%%%%%%%%%%%%%%%%%%%%%%%
%%%%%%%%%%%%%%%%%%%%%%%%%%%%%%%%%%%%%%%%%%%%%%%%%%%%%%%%%%%%%%%%%%%%%%%%%%%%%%%%%%%%%%%%%%%%%%%%%%%%%%%%%%%%%%%%%%%%%%%%%%%%%%%%%%%%%%%%%%%%%%%%%%%%%%%%%%%%%%%%%%%%%%%%%%%%%
%%%%%%%%%%%%%%%%%%%%%%%%%%%%%%%%%%%%%%%%%%%%%%%%%%%%%%%%%%%%%%%%%%%%%%%%%%%%%%%%%%%%%%%%%%%%%%%%%%%%%%%%%%%%%%%%%%%%%%%%%%%%%%%%%%%%%%%%%%%%%%%%%%%%%%%%%%%%%%%%%%%%%%%%%%%%%

\section{Analysis of the excitation Hamiltonian $\cJ_N $}
\label{sec:wtGN}

In this section we analyse the excitation Hamiltonian
\begin{equation*}
 \cJ_N = e^{-A}e^{-B(\eta)} U_N H_N U^*_N e^{B(\eta)} e^{A} = e^{-A} \cG_N e^A 
\end{equation*}
to show Proposition  \ref{thm:tGNo1}. The starting point is part b) of Proposition \ref{prop:gene}, stating that 
\[
\cG_N = C_{\cG_N} + \cQ_{\cG_N} + \cC_N + \cH_N + \cE_{\cG_N} \,,
\]
where $C_{\cG_N}$, $\cQ_{\cG_N}$ and $\cC_N$ are  defined  in \eqref{eq:CN}, \eqref{eq:QN} and \eqref{eq:cCN}   and the error term $\cE_{\cG_N}$ is such that  
\begin{equation*} 
 \pm \cE_{\cG_N} \leq \frac C {\sqrt{N}}\; \big[ (\cN_++1)(\cH_N +1) + (\cN_++1)^3 \big]\,.
 \end{equation*}
From Proposition \ref{prop:cN} and Proposition \ref{prop:cNcH} we conclude that 
\[ \pm e^{-A} \cE_{\cG_N} e^A \leq \frac{C}{\sqrt{N}} \big[ (\cN_++1)(\cH_N +1) + (\cN_++1)^3 \big] \]
In the following sections we study the action of $e^A$ on $\cQ_{\cG_N}$, $\cC_N$ and $\cH_N$ separately. At the end, in Section \ref{sec:pf-thmA}, we combine these results to prove Theorem \ref{thm:tGNo1}.

\subsection{Analysis of $e^{-A} \cQ_{\cG_N} e^A$.}

The action of $e^A$ on the quadratic operator $\cQ_{\cG_N}$ defined in Prop. \ref{prop:gene} is determined by the next proposition. 
\begin{prop} \label{prop:SQS}
Let  $A$ and $\cQ_{\cG_N}$ be defined as in \eqref{eq:defA} and, respectively, \eqref{eq:QN}. Then, under the assumptions of Proposition~\ref{thm:tGNo1}, we have
\[ \begin{split}  %\label{eq:propSQS1} 
 e^{-A} \cQ_{\cG_N} \, e^A & = \cQ_{\cG_N}   + \cE^{(\cQ)}_{N} \,,
\end{split}\]
where the error $\cE_{N}^{(\cQ)}$ is such that, for a constant $C>0$, 
\[
 \begin{split} %\label{eq:propSQS2} 
\pm\, \cE_{N}^{(\cQ)} &\leq  \frac C {\sqrt N}\, (\cN_++1)^2
\end{split}
\]
\end{prop}

To prove the proposition we use the following lemma. 
\begin{lemma} \label{lm:commQA} Let A be defined as in \eqref{eq:defA} and let $\Phi_p$ and $\G_p$ such that $|\Phi_p| \leq C$ and $|\G_p| \leq C |p|^{-2}$. Then, under the assumptions of Proposition~\ref{thm:tGNo1}, 
\begin{align}
\pm \, \sum_{p\in\L^*_+} \Phi_p\;  [b^*_p b_p , A]  \leq  \frac C {\sqrt N}\, (\cN_++1)^2   \label{eq:commQAd} \\
\pm \, \sum_{p\in\L^*_+} \G_p\;  [(b_p b_{-p} + b^*_p b^*_{-p} ), A]  \leq  \frac C {\sqrt N}\, (\cN_++1)^2  \label{eq:commQAnd}
\end{align}
\end{lemma}

\begin{proof}[Proof of Proposition \ref{prop:SQS}]  We write
\[  \begin{split} % \label{eq:SQS}
 e^{-A} \cQ_{\cG_N} \, e^{A}= \cQ_{\cG_N} + \int_0^1 ds \, e^{-s A}\, [\cQ_N, A] e^{s A} \,. 
\end{split}\]
We recall the expression for $\cQ_{\cG_N}$ in \eqref{eq:QN}, where the coefficients of the diagonal and off-diagonal terms are given by $\Phi_p$ in \eqref{eq:phi} and, respectively, by $\Gamma_p$ in \eqref{eq:gamma}. With \eqref{eq:eta-scat}, one can show that $|\Phi_p|\leq C$ and $|\Gamma_p|\leq C |p|^{-2}$. Hence, Proposition \ref{prop:SQS} follows from Lemma \ref{lm:commQA} and Proposition~\ref{prop:cN}. 
\end{proof}

\begin{proof}[Proof of  Lemma \ref{lm:commQA}.]  We start from the proof of \eqref{eq:commQAd}. We use the formula
\be \begin{split}  \label{bpA}
& [b_p^*, A] \\
&=  \frac1{\sqrt N}\sum_{\substack{r \in P_H \\ v \in P_L}}\eta_r \, \Big[ \; (\g_v b_v^* + \s_v b_{-v}) b_{-r} \Big(1-\frac {\cN_+} N \Big)   \d_{p,r+v} \\[-0.5cm]
&  \hskip3cm +(\g_v b_v^* + \s_v b_{-v}) b_{r+v}\Big(1-\frac {\cN_+-1} N \Big)  \d_{p,-r} \\
&  \hskip3cm + \s_v\Big(1-\frac {\cN_+} N \Big) b_{r+v}b_{-r}  \d_{p,-v}  - \g_v  b^*_{r+v}b^*_{-r} \Big(1-\frac {\cN_+} N \Big) \d_{p,v}\; \Big]   \\[0.2cm]
&\hspace{0.5cm}- \frac1{N \sqrt N} \sum_{\substack{r \in P_H \\ v \in P_L }}\eta_r \big[\; (\g_v b_v^* + \s_v b_{-v})  ( b_{-r} a^*_p a_{r+v} + a^*_p a_{-r} b_{r+v})    \\[-0.5cm]
& \hskip 4.2cm + \s_v a^*_p a_{-v} b_{-r}b_{r+v}  -\g_v b^*_{r+v} b^*_{-r} a^*_p a_v \; \big] 
\end{split} \ee
and the fact that $[b_p, A]= [b^*_p, A]^*$ to compute $[b^*_p b_p, A]$. We get:
\[
\sum_{p\in\L^*_+} \Phi_p\;  [b^*_p b_p , A]  = \sum_{j=1}^8 \D_j+\hc
\]
where
\[\begin{split}
\D_{1} =\; & \frac 1 {\sqrt N} \sum_{\substack{r \in P_H,\, v \in P_L }} \Phi_{r+v} \eta_r \left(1-\frac{\cN_+-1} N\right) b^*_{r+v} b^*_{-r} (\g_v b_v +\s_v b^*_{-v})  \\
\D_{2} =\; & \frac 1 {\sqrt N} \sum_{\substack{r \in P_H,\, v \in P_L }} \Phi_{r} \eta_r \left(1-\frac{\cN_+-2} N\right) b^*_{-r} b^*_{r+v} (\g_v b_v +\s_v b^*_{-v}),  \\
\D_{3} =\; & \frac 1 {\sqrt N} \sum_{\substack{r \in P_H,\, v \in P_L}} \Phi_{v} \eta_r  \left(1-\frac{\cN_+-3} N\right)  b^*_{r+v} b^*_{-r}   \s_v  b^*_{-v}  \\
\D_{4} =\; & - \frac 1 {\sqrt N} \sum_{\substack{r \in P_H,\, v \in P_L }} \Phi_{v} \eta_r  \left(1-\frac{\cN_+-2} N\right)     b^*_{-r} b^*_{r+v} \g_vb_{v}    \\
%%%
\D_{5} =\; &- \frac 1 {N^{3/2}}  \sum_{p \in \L_+^*} \sum_{\substack{r \in P_H,\, v \in P_L }}  \Phi_p  \eta_r\, b^*_p a^*_{r+v} a_p b^*_{-r} (\g_v b_v +\s_v b^*_{-v})  \\
\D_{6} =\; & - \frac 1 {N^{3/2}}  \sum_{p \in \L_+^*} \sum_{\substack{r \in P_H,\, v \in P_L }}  \Phi_p  \eta_r\, b^*_p b^*_{r+v}a^*_{-r}a_p (\g_v b_v +\s_v b^*_{-v})\\
\D_{7} =\; & - \frac 1 {N^{3/2}}  \sum_{p \in \L_+^*} \sum_{\substack{r \in P_H,\, v \in P_L }}  \Phi_p  \eta_r \s_v\,  b^*_p b^*_{r+v} b^*_{-r} a^*_{-v} a_p  \\
\D_{8} =\; & \frac 1 {N^{3/2}}  \sum_{p \in \L_+^*} \sum_{\substack{r \in P_H,\, v \in P_L }}  \Phi_p  \eta_r \g_v\, b^*_p a^*_v a_p b_{-r} b_{r+v} 
 \end{split}
\]
Using $|\Phi_p| \leq C$, $|\eta_r| \leq C |r|^{-2}$, $|\s_v| \leq C |v|^{-2}$, we estimate, for any $\xi \in \cF_+^{\leq N}$, 
\[
 \begin{split}
  |\langle\xi,\D_1\xi\rangle|&\leq \frac C {\sqrt N} \sum_{\substack{r \in P_H,\, v \in P_L }}|\eta_r|\|b_{-r}b_{r+v}\xi\|\big( |\s_v| \| (\cN_+ + 1)^{1/2} \xi \|+\|b_v\xi\|\big)\\
  &\leq \frac C {\sqrt N} \|(\cN_++1)\xi\|^2
 \end{split}
\]
The terms $\D_j$ with $j=2,3,4$ are bounded in a similar way.  To bound $\D_5$ we first move $b^*_{-r}$ to the left, obtaining
\[ \begin{split}
\D_5=&- \frac 1 {N^{3/2}} \sum_{p\in\L^*_+}  \sum_{\substack{r \in P_H,\, v \in P_L }} \Phi_p \eta_r\, b^*_pb^*_{-r} a^*_{r+v} a_p (\g_v b_v +\s_v b^*_{-v}) \\
&- \frac 1 {N^{3/2}}\sum_{\substack{r \in P_H,\, v \in P_L}}\Phi_r\eta_r b^*_{-r}b^*_{r+v} (\g_v b_v +\s_v b^*_{-v}) =\D_5^{(1)}+\D_5^{(2)} 
\end{split}\]
The cubic term $\D_5^{(2)} $ can be estimated similarly as $\D_{1}$, while
\[
 \begin{split}
  |\langle\xi,\D_5^{(1)}\xi\rangle|&\leq  \frac C {N^{3/2}} \sum_{p\in\L^*_+}  \sum_{\substack{r \in P_H,\, v \in P_L }} |\eta_r|\|a_{r+v}b_{-r}b_p\xi\|\|a_p(\g_vb_v+\s_vb^*_{-v})\xi\|\\
  &\leq  \frac C {N} \Big(  \frac 1 N \sum_{p\in\L^*_+}  \sum_{\substack{r \in P_H,\, v \in P_L }}  \|a_{r+v}b_{-r}b_p\xi\|^2\Big)^{1/2}\\
  & \hskip 2cm \times \Big(\sum_{p\in\L^*_+}  \sum_{\substack{r \in P_H,\, v \in P_L }} |\eta_r|^2 \left[ \|a_p b_v \xi \|^2 + |\s_v|^2 \| a_p b^*_{-v} \xi\|^2  \right] \Big)^{1/2} \\
  &\leq \frac CN \|(\cN_++1)\xi\|^2\,.
 \end{split}
\]
(In the last step, to bound the contribution proportional to $\| a_p b_{-v}^* \xi \|$, we first estimated the sum over $r,p$ with fixed $v$ by $|\s_v|^2 \| \cN_+^{1/2} b_{-v}^* \xi \|^2 \leq |\s_v|^2 \| (\cN_+ +1) \xi \|$ and then we summed over $v$). The terms $\D_j$ with $j=6,7,8$ can be treated as $\D_5^{(1)}$.  Hence for all $j=1, \ldots, 8$ we have
\[
\pm\, \big( \D_j + \hc \big) \leq \frac C {\sqrt N}\, (\cN_++1)^2 \, . 
\]
This concludes the proof of \eqref{eq:commQAd}.

To show \eqref{eq:commQAnd} we use \eqref{bpA} and its conjugate to compute
\[\begin{split}
\sum_{p\in\L^*_+} \Gamma_p\; &  [(b_p b_{-p} +b^*_p b^*_{-p}), A]  = \sum_{j=1}^9 \Upsilon_j +\hc
\end{split}\]
where
\[\begin{split}
\Upsilon_1  =\; & \frac 1 {\sqrt N} \sum_{\substack{r \in P_H,\, v \in P_L}} \Gamma_{r+v}\eta_r\,  \left(1-\frac{\cN_++1} N\right) (b^*_{-r} b_{-r-v} - \frac 1 N a^*_{-r} a_{-r-v}) (\g_v b_v +\s_v b^*_{-v})\\
\Upsilon_2  =\; &   \frac 1 {\sqrt N} \sum_{\substack{r \in P_H,\, v \in P_L}} \Gamma_{r+v}  \eta_r\, \left(1-\frac{\cN_+} N\right) b^*_{-r} (\g_v b_v +\s_v b^*_{-v}) b_{-r-v} \\
\Upsilon_3  =\; & \frac 1 {\sqrt N} \sum_{\substack{r \in P_H,\, v \in P_L }}  \Gamma_{r}  \eta_r \left(1-\frac{\cN_+} N\right) (b^*_{r+v} b_r - \frac 1 N a^*_{r+v}a_r)  (\g_v b_v +\s_v b^*_{-v}) \\
\Upsilon_4  =\; & \frac 1 {\sqrt N} \sum_{\substack{r \in P_H,\, v \in P_L }}  \Gamma_{r}  \eta_r  \left(1-\frac{\cN_+-1} N\right) b^*_{r+v}  (\g_v b_v +\s_v b^*_{-v})b_{r} \\
\Upsilon_5  =\; &  \frac 1 {\sqrt N} \sum_{\substack{r \in P_H,\, v \in P_L }} \Gamma_{v} \eta_r \s_v \,  \left(1-\frac{\cN_+-1} N\right) ( b^*_{r+v} b_v - \frac 1 N a^*_{r+v}a_v)b^*_{-r}   \\
\Upsilon_6 =\; &   \frac 1 {\sqrt N} \sum_{\substack{r \in P_H,\, v \in P_L }} \Gamma_{v} \eta_r \s_v  \left(1-\frac{\cN_+-2} N\right)  b^*_{r+v} b^*_{-r}   b_{v} \\
\end{split}\]
and
\[\begin{split}
\Upsilon_7  =\; &  - \frac 1 {\sqrt N} \sum_{\substack{r \in P_H,\, v \in P_L }} \Gamma_{v} \eta_r \g_v  \left[\left(1-\frac{\cN_+} N\right) +\left(1-\frac{\cN_++1} N\right)  \right]   b_{r+v} b_{-r}b_{-v}  \\
\Upsilon_8  =\; & - \frac 1 {N\sqrt N} \sum_{p\in\L^*_+}  \sum_{\substack{r \in P_H,\, v \in P_L }}  \Gamma_p \eta_r  \big[\;  b_p (a^*_{r+v} a_{-p} b^*_{-r} + b^*_{r+v}a^*_{-r}a_{-p})(\g_v b_v +\s_v b^*_{-v})  \\[-0.5cm]
& \hskip 5.5cm + \s_v b_p b^*_{r+v} b^*_{-r} a^*_{-v} a_{-p} - \g_v b_p a^*_v a_{-p} b_{-r} b_{r+v}\; \big]  \\
\Upsilon_9  =\; &- \frac 1 {N\sqrt N} \sum_{p\in\L^*_+}  \sum_{\substack{r \in P_H,\, v \in P_L }} \Gamma_p\eta_r  \big[ \; (a^*_{r+v} a_p b^*_{-r} + b^*_{r+v}a^*_{-r}a_p)(\g_v b_v +\s_v b^*_{-v}) b_{-p}  \\[-0.5cm]
& \hskip 5.5cm + \s_v  b^*_{r+v} b^*_{-r} a^*_{-v} a_p b_{-p} - \g_v  a^*_v a_p b_{-r} b_{r+v} b_{-p} \;\big]   \\
\end{split}\]
We show now that for all $j=1,\ldots, 9$ we have, on $\cF_+^{\leq N}$, 
\be \label{eq:commQAnd2}
\pm\, \big( \Upsilon_j + \hc \big) \leq \frac C {\sqrt N}\, (\cN_++1)^2 
\ee
We observe that 
\[
 \begin{split}
  |\langle\xi,\Upsilon_{1} \xi\rangle|&\leq \frac 1 {\sqrt N}\sum_{\substack{r \in P_H,\, v \in P_L }}|\Gamma_{r+v}||\eta_r|  \|a_{-r}\xi\| \| a_{-r-v} (\g_v b_v + \s_v b^*_v) \xi\|   \\
  &\leq \frac C {\sqrt N} \Big(\sum_{\substack{r \in P_H,\, v \in P_L }}|\Gamma_{r+v}|^2  \|a_{-r}\xi\|^2  \Big)^{1/2}    \\ &\hspace{2cm} \times \Big(\sum_{\substack{r \in P_H,\, v \in P_L }} |\eta_r|^2 \left[  \| b_v (\cN_+ + 1)^{1/2} \xi \|^2 + |\s_v|^2 \| (\cN_+ + 1) \xi\|^2 \right]\Big)^{1/2}   \\
  &\leq \frac C{\sqrt N}\, \|(\cN_++1)\xi\|^2,
 \end{split}
\]
and similar bounds hold for $\Upsilon_j$ with j=2,3,4. 
%For $\Upsilon_3$ we have
%\[
% \begin{split}
%  |\langle\xi,\Upsilon_{3} \xi\rangle|&\leq \frac 1 {\sqrt N}\sum_{\substack{r \in P_H,\, v \in P_L }}|\Gamma_{r}||\eta_r|  \|b_{r+v}\xi\| \| a_{r} (\g_v b_v + \s_v b^*_v) \xi\|   \\
%  &\leq \frac C {\sqrt N} \Big(\sum_{\substack{r \in P_H,\, v \in P_L }}|\Gamma_{r}|^2  \|a_{r+v}\xi\|^2  \Big)^{1/2}    \Big(\sum_{\substack{r \in P_H,\, v \in P_L }} |\eta_r|^2  \| a_{r} (\g_v b_v + \s_v b^*_v) \xi\|^2 \Big)^{1/2}   \\
%  &\leq \frac C{\sqrt N}\, \|(\cN_++1)\xi\|^2,
% \end{split}
%\]
Next, we bound 
\[
 \begin{split}
  |\langle\xi,\Upsilon_{5} \xi\rangle|&\leq \frac 1 {\sqrt N}\sum_{\substack{r \in P_H,\, v \in P_L }}|\Gamma_{v}||\eta_r| |\s_v|  \|a_{r+v}\xi\| \| a_{v} b^*_{-r} \xi\|   \\
  &\leq \frac C {\sqrt N} \Big(\sum_{\substack{r \in P_H,\, v \in P_L }}|\Gamma_{v}|^2  \|a_{r+v}\xi\|^2  \Big)^{1/2}    \Big(\sum_{\substack{r \in P_H,\, v \in P_L }} |\eta_r|^2  |\s_v|^2 \| (\cN_++1) \xi\|^2 \Big)^{1/2}   \\
  &\leq \frac C{\sqrt N}\, \|(\cN_++1)\xi\|^2\,.
 \end{split}
\]
and similarly  for $\Upsilon_6$. To bound $\Upsilon_7$ we use
\[
 \begin{split}
  |\langle\xi,\Upsilon_{7} \xi\rangle|&\leq \frac 1 {\sqrt N}\sum_{\substack{r \in P_H,\, v \in P_L }}|\Gamma_{v}||\eta_r| \| b_{r+v}\, \cN_+^{-1/2} b_{-r} b_{-v}\xi\| \| (\cN_++1)^{1/2} \xi\|   \\
  &\leq \frac C {\sqrt N} \Big(\sum_{\substack{r \in P_H,\, v \in P_L }}\| b_{-r} b_{-v}\xi\| ^2  \Big)^{1/2}    \Big(\sum_{\substack{r \in P_H,\, v \in P_L }}  |\G_v|^2  |\eta_r|^2 \| (\cN_++1) \xi\|^2 \Big)^{1/2}   \\
  &\leq \frac C{\sqrt N}\, \|(\cN_++1)\xi\|^2\,.
 \end{split}
\]
We consider now $\Upsilon_{8}$. In the first term we move $b^*_{-r}$ to the left:
\[\begin{split}
\Upsilon_{8}^{(1)}&:= -  \frac 1 {N^{3/2}} \sum_{p\in\L_+^*} \sum_{\substack{r \in P_H,\, v \in P_L}}  \Gamma_p  \eta_r \, b_p a^*_{r+v} a_{-p} b^*_{-r} (\g_v b_v +\s_v b^*_{-v})  \\
&\,= -  \frac 1 {N^{3/2}} \sum_{p\in\L_+^*} \sum_{\substack{r \in P_H,\, v \in P_L }} \Gamma_p   \eta_r \, b_p  b^*_{-r}a^*_{r+v} a_{-p} (\g_v b_v +\s_v b^*_{-v})  \\
&\quad -  \frac 1 {N^{3/2}}  \sum_{\substack{r \in P_H,\, v \in P_L}} \Gamma_r \eta_r \, b_r b^*_{r+v} (\g_v b_v +\s_v b^*_{-v})  =\Upsilon_{8}^{(1a)} + \Upsilon_{8}^{(1b)}
\end{split}\]
To bound the quintic term we use that $|P_L| \leq C N^{3/2}$  and $\sum_{r \in P_H} |\eta_r|^2 \leq C N^{-1/2}$:
\[\begin{split}
|\langle\xi,\Upsilon_{8}^{(1a)}  \xi\rangle|&\leq  \frac 1 {N^{3/2}}  \sum_{p\in\L_+^*} \sum_{\substack{r \in P_H,\, v \in P_L}}   |\Gamma_p|  |\eta_r| \, \| b^*_p a_{r+v} \xi \| \| a_{-p} b^*_{-r} (\g_v b_v +\s_v b^*_{-v})  \xi\|
\\
& \leq  \frac 1 {N}  \Big(  \sum_{p\in\L_+^*}  \sum_{\substack{r \in P_H,\, v \in P_L }} |\G_p|^2  \| (\cN_++1)^{1/2} a_{r+v} \xi \|^2 \Big)^{1/2} \\
& \hspace{.5cm} \times\Big( \frac 1 N  \sum_{p\in\L_+^*} \sum_{\substack{r \in P_H,\, v \in P_L }} |\eta_r|^2 \left[ \| a_{-p} b_{-r}^* b_v \xi \|^2 + |\s_v|^2 \| a_{-p} b_{-r}^* b_{-v}^* \xi \|^2 \right] \Big)^{1/2} \\
&\leq \frac C {\sqrt N} \,\|(\cN_++1)\xi\|^2
\end{split}\]
where, in the last step, to bound the contribution proportional to $\| a_{-p} b_{-r}^* b_v \xi \|^2$, we first estimated the sum over $p$ by $\| \cN_+^{1/2} b_{-r}^* b_v \xi \| \leq \| b_v (\cN_+ +1) \xi \|$ and then we summed over $r$ and $v$ (and similarly for the term proportional to $\| a_{-p} b_{-r}^* b_{-v}^* \xi \|$). As for the cubic term, we have 
\[\begin{split}
|\langle\xi,\Upsilon_{8}^{(1b)}  \xi\rangle|&\leq \frac 1 {N^{3/2}}   \Big(\sum_{\substack{r \in P_H,\, v \in P_L }}|\Gamma_r|^2 \| b^*_r b_{r+v}\xi \|^2   \Big)^{1/2} \\ &\hspace{2cm} \times  \Big( \sum_{\substack{r \in P_H,\, v \in P_L }} |\eta_r|^2 \left[ \| b_v \xi \|^2 + |\s_v|^2 \| (\cN_+ + 1)^{1/2} \xi\|^2 \right] \Big)^{1/2}  \\
&\leq  \frac 1 {N^{3/2}}  \|(\cN_++1)\xi\|^2
\end{split}\]
The remaining terms in $\Upsilon_8$ and $\Upsilon_9$ can be treated with the same arguments shown above, thus concluding the proof of \eqref{eq:commQAnd2}.
\end{proof}

\subsection{Analysis of $e^{-A} \cC_N e^A$.}

In this section, we analyze the action of the cubic exponential on the cubic term $\cC_N$, defined in \eqref{eq:cCN}.

\begin{prop} \label{prop:SCS} Let $A$ be defined as in \eqref{eq:defA} and let $\cC_N$ be defined as in \eqref{eq:cCN}. Then, under the assumptions of Proposition~\ref{thm:tGNo1}, we have
\begin{equation}\label{eq:lm-SCS} \begin{split}
e^{-A} \cC_N\, e^A =\; &\cC_N  + \frac 2 N \sum_{r \in P_H, v \in P_L} \k \big( \widehat V(r/N) +\widehat V((r+v)/N)\big) \eta_r\\ &\hspace{2.5cm} \times  \Big[ \s^2_v + ( \g_v^2 + \s^2_v)\, b^*_v b_{v}  +\g_v \s_v\,  \big( b_v b_{-v} + b_v^* b_{-v}^*\big)\Big] \\
&+ \cE_{N}^{(\cC)}\,,
\end{split}\end{equation}
where the error $\cE_{N}^{(\cC)}$  satisfies
\be \begin{split}\label{eq:propSCS2} 
\pm \cE_{N}^{(\cC)} &\leq \frac C {\sqrt{N}}\; \big[ (\cN_++1)(\cH_N +1) + (\cN_++1)^3 \big]\,.
\end{split}\ee
\end{prop}

To prove the proposition we use the following lemma.
\begin{lemma} \label{lm:commCA}
Let $A$ be defined as in \eqref{eq:defA} and $\cC_N$  be defined as in \eqref{eq:CN}. Then, under the assumptions of Proposition~\ref{thm:tGNo1}, 
\be \label{eq:commCA}
[\cC_N, A] = \sum_{j=0}^{14} \Xi_j + \emph{h.c.}
\ee
where
\[ \begin{split}
\Xi_0 = \; & \frac \k N \sum_{\substack{r \in P_H\\ v \in P_L}} \big( \widehat V(r/N) +\widehat V((r+v)/N)\big) \eta_r\, \Big[ \s^2_v + ( \g_v^2 + \s^2_v)\, b^*_v b_{v}  +\g_v \s_v\,  \big( b_v b_{-v} + b_v^* b_{-v}^*\big)\Big]  \\
%\end{split}\]
%
%\[  \begin{split}
\Xi_{1} = \; & \frac \k N \sum_{\substack{r \in P_H \\
v \in P_L }}  \widehat V((r+v)/N)  \eta_r\, (\g_v b^*_v + \s_v b_{-v}) \left[\left( 1- \frac {\cal N_+} N \right)^2 -1 \right] (\g_v b_v + \s_v b^*_{-v}) \\
+ \;&   \frac \k N \sum_{\substack{r \in P_H\\ v \in P_L  }}  \widehat V(r/N)  \eta_r\, (\g_v b^*_v + \s_v b_{-v}) \left[\left( 1- \frac {\cN_+ +1} N \right) \left( 1- \frac {\cal N_+} N \right) -1 \right]  (\g_v b_v + \s_v b^*_{-v})   \\
\Xi_2=\;&   \frac \k N \sum_{\substack{r \in P_H,\, v \in P_L }}  \widehat V(r/N)  \eta_r\, \s_v\, \left( 1- \frac {\cN_++1} N \right) \left( 1- \frac {\cN_+} N \right)  b_{r+v}(\g_{r+v} b_{-r-v} + \s_{r+v} b^*_{r+v})  \\
\Xi_3 =\; &  \frac \k N \sum_{\substack{r \in P_H,\, v \in P_L}}  \widehat V((r+v)/N)  \eta_r\,  \s_v \left( 1- \frac {\cal N_+} N \right)^2  b_{-r} (\g_r b_r + \s_r b^*_{-r})\\
\Xi_4= \; & - \frac \k {N^2} \sum_{\substack{r \in P_H,\, v \in P_L}}  \widehat V((r+v)/N)  \eta_r\,  \s_v \left( 1- \frac {\cal N_+} N \right)  b_{-r} (\g_r b_r + \s_r b^*_{-r})  \\
 \Xi_5 =\; &\frac \k N \sum_{\substack{r \in P_H \\ v \in P_L }}    \sum_{\substack{p \in \L_+^*\\ p \neq r+v}} \widehat V(p/N)   \eta_r  (\g_v b^*_v + \s_v b_{-v})   \left( 1- \frac {\cN_++1} N \right) (b^*_{-p} b_{-r} - \frac 1 N a^*_{-p} a_{-r}) \\[-9pt]
 & \hskip 7cm \times (\g_{r+v-p} b_{r+v-p} + \s_{r+v-p} b^*_{p-r-v}) \\
\end{split}\]
and
\[  \begin{split}
 \Xi_6 =\; & \frac \k N \sum_{\substack{r \in P_H \\ v \in P_L }}    \sum_{\substack{p \in \L_+^*\\ p \neq -r}}  \widehat V(p/N)   \eta_r  (\g_v b^*_v + \s_v b_{-v})   \left( 1- \frac {\cN_+} N \right) (b^*_{-p} b_{r+v} - \frac 1 N a^*_{-p} a_{r+v}) \\[-9pt]
 & \hskip 7cm \times  (\g_{r+p} b_{-r-p} + \s_{r+p} b^*_{r+p}) \\
 \Xi_7 =\; & \frac \k N \sum_{\substack{r \in P_H \\ v \in P_L }}    \sum_{\substack{p \in \L_+^*\\ p \neq -v}}  \widehat V(p/N)   \eta_r  \left( 1- \frac {\cal N_+} N \right) \s_v   (b^*_{-p} b_{r+v} - \frac 1 N a^*_{-p} a_{r+v}) b_{-r}  \\[-9pt]
 & \hskip 7cm \times  (\g_{p+v} b_{-p-v} + \s_{p+v} b^*_{p+v}) \\
 \Xi_8 =\;   & -  \frac \k N \sum_{\substack{r \in P_H \\ v \in P_L }}    \sum_{\substack{p \in \L_+^*\\ p \neq v}}   \widehat V(p/N)   \eta_r   \left( 1- \frac {\cN_+-2} N \right) \g_v \,  b^*_{r+v} b^*_{-r} b^*_{-p}(\g_{p-v} b_{-p+v} + \s_{p-v} b^*_{p-v})\\
 \Xi_9 =\; & -\frac \k {N^2} \sum_{\substack{r \in P_H \\ v \in P_L }}    \sum_{\substack{p \in \L_+^*\\ p \neq -v}} \widehat V(p/N)   \eta_r   \left( 1- \frac {\cN_+} N \right)  \s_v\,  a^*_{-p} a_{-r} b_{r+v} \,  (\g_{p+v} b_{-p-v} + \s_{p+v} b^*_{p+v}) \\
\Xi_{10} =\; &  \frac \k {N^2}  \sum_{\substack{r \in P_H \\ v \in P_L }}  \sum_{\substack{p,q \in \L_+^* \\ p \neq -q}} \widehat V(p/N)    \eta_r  (\g_v b_v^* + \s_v b_{-v}) (a^*_{p+q} a_{r+v} b_{-r} + b_{r+v}a^*_{p+q} a_{-r}) b^*_{-p}  \\ &\hspace{10cm} \times    (\g_q b_q + \s_q b^*_{-q}) \\
\end{split}\]
\[ \begin{split} 
\Xi_{11} =\; & \frac \k {N^2}  \sum_{\substack{r \in P_H \\ v \in P_L }}  \sum_{\substack{p,q \in \L_+^* \\ p \neq -q}} \widehat V(p/N)    \eta_r  \big(
 \g_v b^*_{r+v} b^*_{-r} a^*_{p+q} a_v + \s_v a^*_{p+q} a_{-v} b_{-r}b_{r+v} \big) b^*_{-p}  \\ &\hspace{10cm} \times     (\g_q b_q + \s_q b^*_{-q}) \\
\end{split}\]
and
\begin{equation}\label{eq:Xi1214} \begin{split} 
\Xi_{12} =\; & \frac \k {\sqrt N} \sum_{\substack{p, q \in \L_+^* \\p\neq -q}}  \widehat V(p/N)\,  b_{p+q}^*\, [b^*_{-p},A]\, \big( \g_q b_{q} + \s_q b^*_{-q} \big)  \\
\Xi_{13} =\; &  \frac \k {\sqrt N} \sum_{\substack{p, q \in \L_+^* \\p\neq -q}}  \widehat V(p/N)\, \g_q\, b_{p+q}^*\, b^*_{-p}\, [ b_{q},A ]  \\
\Xi_{14} =\; & \frac \k {\sqrt N} \sum_{\substack{p, q \in \L_+^* \\p\neq -q}}  \widehat V(p/N)\, \s_q \, b_{p+q}^*\, b^*_{-p} \,[b^*_{-q} ,A]  
\end{split}\end{equation}
For all $j=1,\ldots, 14$ (but not for $j=0$) we have
\be \label{eq:EC1}
\pm\, \big(\Xi_j + \hc \big) \leq  \frac C {\sqrt N}\, \Big[ (\cN_++1)(\cK+1) + (\cN_++1)^3 \Big]\,.
\ee
Moreover,
\be \label{eq:EC2}
\pm\,  [\Xi_0, A]   \leq \frac C {\sqrt N}\, \big(\cN_++1\big)^2\,.
\ee
\end{lemma}

\begin{proof}[Proof of Prop. \ref{prop:SCS}.] We write
\be  \begin{split} \label{eq:SCS}
 e^{-A} \cC_N\, e^{A}= \cC_N+ \int_0^1 ds \, e^{-s A}\, [\cC_N, A] e^{s A} \,. 
\end{split}\ee
We set $\tl \cE_{N}^{(\cC)}:= [\cC_N, A] - 2 \Xi_0 = \sum_{j=1}^{14} (\Xi_j + \hc)$ and rewrite \eqref{eq:SCS} as $e^{-A}   \cC_N\, e^{A}  =  \cC_N +2 \Xi_0 + \cE_N^{(\cC)}$, with
  \[
  \cE_N^{(\cC)}=   \int_0^1 ds_1 \int_0^{s_1} ds_2 \, e^{-s_2 A}\, \big[ 2 \Xi_0 , A\big] e^{s_2 A}   + \int_0^1 ds \, e^{-s A}\, \tl \cE_{N}^{(\cC)} e^{s A} \,.
\]
Lemma \ref{lm:commCA} together with Proposition \ref{prop:cN} and Proposition \ref{prop:cNcH} 
imply \eqref{eq:propSCS2}; with the definition of $\Xi_0$ we obtain (\ref{eq:lm-SCS}). 

\end{proof}

\begin{proof}[Proof of Lemma \ref{lm:commCA}.]
We have 
\be  \begin{split} \label{eq:calCA1}
[\cC_N, A]   =\; & \frac \k {\sqrt N} \sum_{\substack{p, q \in \L_+^* \\p\neq -q}}  \widehat V(p/N)\,  [b_{p+q}^*,A]\, b^*_{-p} \big( \g_q b_{q} + \s_q b^*_{-q} \big) + \sum_{j=12}^{14}\Xi_j + \hc \\
\end{split}\ee
We use  the formula \eqref{bpA} to compute the first term on the r.h.s. of \eqref{eq:calCA1}. Putting in  normal order the quartic terms (but leaving unchanged the parenthesis $(\g_v b^*_v  + \s_v b_{-v})$ and its conjugate) we obtain
\[  \begin{split} 
& \frac \k {\sqrt N} \sum_{\substack{p, q \in \L_+^* \\p\neq -q}}  \widehat V(p/N)\,  [b_{p+q}^*,A]\, b^*_{-p} \big( \g_q b_{q} + \s_q b^*_{-q} \big)  \\
& =  \frac \k N \sum_{\substack{r \in P_H\\ v \in P_L}}  \widehat V(r/N)  \eta_r\, (\g_v b^*_v + \s_v b_{-v}) \left( 1- \frac {\cN_+ +1} N \right) \left( 1- \frac {\cal N_+} N \right)   (\g_v b_v + \s_v b^*_{-v}) \\
&  \quad + \frac \k N \sum_{\substack{r \in P_H\\ v \in P_L}}  \widehat V((r+v)/N)  \eta_r\, (\g_v b^*_v + \s_v b_{-v}) \left( 1- \frac {\cal N_+} N \right)^2   (\g_v b_v + \s_v b^*_{-v}) + \sum_{j=2}^{11} \Xi_j
\end{split}\]
The first two terms on the r.h.s. of the last equation can be further decomposed as $\Xi_0 + \Xi_1$. Combining \eqref{eq:calCA1} with the last equation we obtain the decomposition \eqref{eq:commCA}. 

Next, we prove the bound \eqref{eq:EC1}. With $\sum_{r \in \L_+^*} \widehat V(r/N) \eta_r \leq C N$ and $ \sum_{v \in P_L} |\s_v|^2  \leq C$, we obtain that $\pm \Xi_1 \leq C N^{-1} (\cN_+ + 1)^2$. As for $\Xi_2$, we find  
\[  \begin{split}
|\bmedia{\xi, \Xi_{2} \xi}|&  \leq  \frac C N \Big(  \sum_{\substack{r \in P_H\\ v \in P_L }} |\eta_r|^2 |\s_v| \| \cN_+^{1/2} \xi \|^2 \Big)^{1/2} \\ &\hspace{1cm} \times \Big(  \sum_{\substack{r \in P_H\\ v \in P_L }} |\s_v| \left[  \| b_{-r-v} \xi \|^2 + |\s_{r+v}|^2  \| (\cN_+ + 1)^{1/2} \xi\|^2 \right]  \Big)^{1/2} \\ &\leq \frac C {\sqrt N} \, \| (\cN_++1)^{1/2} \xi \|^2
\end{split}\]
using $\sum_{v \in P_L} |\s_v| \leq C N^{1/2}$. The terms $\Xi_3$ and $\Xi_4$ can be bounded analogously. As for the term $\Xi_5= \Xi_5^{(1)}+ \Xi_5^{(2)}$ we use that $|P_L| \leq C N^{3/2} $ and $\sum_{r \in P_H} |\eta_r|^2 \leq C N^{-1/2}$, hence 
\[  \begin{split}
& |\bmedia{\xi, \Xi_{5}^{(1)} \xi}|\\
&  \leq  \frac \k N \sum_{\substack{r \in P_H \\ v \in P_L}} \sum_{\substack{p \in \L_+^*\\ p \neq r+v}} | \widehat V(p/N)| |\eta_r|  \| b_v  b_{-p} \xi \| \| b_{-r} (\g_{r+v-p}  b_{r+v-p}  + \s_{r+v-p}  b^*_{p-r-v}) \xi \| \\
& \leq  \frac \k {N}    \, \biggl(\;  \sum_{\substack{r \in P_H \\ v \in P_L}}   \sum_{\substack{p \in \L_+^*\\ p \neq r+v}}    |\eta_r|^2   \|b_v  b_{-p} \xi \|^2\biggr)^{1/2}  \\
& \times \biggl[\biggl(\;  \sum_{\substack{r \in P_H \\ v \in P_L}}   \sum_{\substack{p \in \L_+^*\\ p \neq r+v}}  \|b_{-r}  b_{r+v-p} \xi\|^2  \biggr)^{1/2} +  \biggl(\;  \sum_{\substack{r \in P_H \\ v \in P_L}}   \sum_{\substack{p \in \L_+^*\\ p \neq r+v}}  |\s_{r+v-p}|^2 \|b_{-r}  b_{p-r-v}^* \xi \|^2  \biggr)^{1/2}\; \biggr] \\
& \leq \frac C {\sqrt N}\; \| (\cN_++1) \xi \|^2
\end{split}\]
In the last step, to bound the term proportional to $\| b_{-r} b_{p-r-v}^* \xi \|$, we first shifted $p \to p+r+v$, then we estimated the sum over $r$ by $\| \cN_+^{1/2} b_p^* \xi \| \leq \| (\cN_+ + 1) \xi \|$ and at the end we summed over $v$ and $p$ (using the factor $|\s_p|^2$). The bound for $\Xi_{5}^{(2)}$ and for the terms $\Xi_j$ with $j=6,7,8,9$ can be obtained similarly.  As for the term $\Xi_{10}$ we first move the operator $b^*_p$ to the left, obtaining
\begin{equation}\label{eq:Xi10dec}  \begin{split} 
 \Xi_{10} =\; & \frac \k {N^2} \sum_{\substack{r \in P_H\\v \in P_L }}   \sum_{\substack{p,q \in \L_+^* \\ p \neq -q}}  \widehat V(p/N)   \eta_r \,  (\g_v b_v^* + \s_v b_{-v})\,a^*_{p+q} a_{r+v} \big(b^*_{-p} b_{-r} - \frac 1 N a^*_{-p} a_{-r} \big) \\[-0.5cm] &\hspace{10cm} \times (\g_q b_q + \s_q b^*_{-q})\\
 & + \frac \k {N^2} \sum_{\substack{r \in P_H\\v \in P_L }}   \sum_{\substack{p,q \in \L_+^* \\ p \neq -q}}  \widehat V(p/N)   \eta_r \,  (\g_v b_v^* + \s_v b_{-v})\,b_{r+v} b^*_{-p} a^*_{p+q}  a_{-r}(\g_q b_q + \s_q b^*_{-q})\\
 &  + \frac \k {N^2} \sum_{\substack{r \in P_H\\v \in P_L }}   \sum_{\substack{q \in \L_+^* \\ q \neq -r}}  \widehat V(r/N)   \eta_r (\g_v b_v^* + \s_v b_{-v})\,a^*_{r+q} a_{r+v}  (\g_q b_q + \s_q b^*_{-q}) \\
 &  + \frac \k {N^2} \sum_{\substack{r \in P_H\\v \in P_L }}   \sum_{\substack{q \in \L_+^* \\ q \neq -r}}  \widehat V(r/N)   \eta_r (\g_v b_v^* + \s_v b_{-v})\,b_{r+v} b^*_{r+q}  (\g_q b_q + \s_q b^*_{-q}) \\
 = \; &\sum_{j=1}^4 \Xi_{10}^{(j)}
\end{split}
\end{equation}
To bound $\Xi_{10}^{(1)}$, we commute the operator $a_{r+v}$ to the right of $b_{-p}^*$. we find
\begin{equation}\label{eq:Xi101} \begin{split}  |\langle \xi , \Xi^{(1)}_{10} \xi \rangle | \leq\; & \frac{C}{N^2} \sum_{\substack{p,q \in \L_+^* \\ p \neq -q}} \sum_{\substack{r \in P_H\\v \in P_L }} |\widehat{V} (p/N)| |\eta_r|  \\ &\hspace{2cm} \times \| b_{-p} a_{p+q} (\g_v b_v + \s_v b_{-v}^*) \xi \| \| a_{r+v} b_{-r} (\g_q b_q + \s_q b_{-q}^* ) \xi \| \\ &+ \frac{C}{N^2} \sum_{\substack{p,q \in \L_+^* \\ p \neq -q}} \sum_{\substack{v \in P_L \\ p+v \in P_H} } |\widehat{V} (p/N)| 
|\eta_{p+v}|  \\ &\hspace{2cm} \times  \| b_{p+q} (\g_v b_v + \s_v b_{-v}^*) \xi \| \| b_{p+v} (\g_q b_q + \s_q b_{-q}^*) \xi \|  \end{split} \end{equation}
By Cauchy-Schwarz and using the bound $\sum_{p \in \Lambda^*_+} |\widehat{V} (p/N)|^2 \leq C N^3$, we obtain 
\begin{equation}\label{eq:Xi102} \begin{split} 
& |\langle \xi , \Xi^{(1)}_{10} \xi \rangle | \\
&\leq  \frac{C}{N^2} \Big( \sum_{\substack{p,q \in \L_+^* \\ p \neq -q}} \sum_{\substack{r\in P_H \\v \in P_L }} |\eta_r|^2 \left[  \| a_{-p} a_{p+q} a_v \xi \|^2 + |\s_v|^2 \| a_{-p} a_{p+q} a_{-v}^* \xi \|^2 \right] \Big)^{1/2} \\
&\hspace{1cm} \times \Big( \sum_{\substack{p,q \in \L_+^* \\ p \neq -q}} \sum_{\substack{r \in P_H\\v \in P_L }} |\widehat{V} (p/N)|^2 \left[ \| a_{r+v} a_{-r} a_q \xi \| + |\s_q|^2 \| a_{r+v} a_{-r} a_q^* \xi \|^2 \right] \Big)^{1/2}  \\
&\qquad + \frac{C}{N^2} \Big( \sum_{\substack{p,q \in \L_+^* \\ p \neq -q}} \sum_{v \in P_L } |\eta_{p+v}|^2 \left[ \| a_{p+q}  b_v  \xi \|^2 + | \s_v|^2 \| a_{p+q} b_{-v}^* \xi \|^2 \right]   \Big)^{1/2} \\ 
&\qquad \hspace{1cm} \times \Big( \sum_{\substack{p,q \in \L_+^* \\ p \neq -q}} \sum_{\substack{v \in P_L }} \left[ \| a_{p+v}  b_q  \xi \|^2 + | \s_q|^2 \| a_{p+v} b_{-q}^* \xi \|^2 \right]   \Big)^{1/2}\\
& \leq  \frac{C}{\sqrt{N}} \langle \xi , (\cN_+ + 1)^3 \xi \rangle  
\end{split} \end{equation}
The bound for  $\Xi_{10}^{(2)}$ is similar. As for the quartic operators $\Xi_{10}^{(3)}$ and $\Xi_{10}^{(4)}$, they can be handled like the second term on the r.h.s. of (\ref{eq:Xi101}) (in $\Xi_{10}^{(4)}$ we first commute $b_{r+v}$ and $b^*_{r+q}$). We obtain that 
\[ \pm \Xi_{10} \leq \frac{C}{\sqrt{N}} (\cN_+ + 1)^3 \]

The operator $\Xi_{11}$ can be controlled similarly as $\Xi_{10}$. To estimate $\Xi_{12}$, $\Xi_{13}$ and $\Xi_{14}$, we insert (\ref{bpA}) into (\ref{eq:Xi1214}); this produces several terms. The contributions arising from $\Xi_{12}$ are similar to the terms $\Xi_1, \dots, \Xi_{11}$ considered above and their expectation can be estimated analogously.  
On the other hand, to bound some of the contributions to $\Xi_{13}$ and $\Xi_{14}$ we need to use the kinetic energy operator. To explain this step, let us compute $\Xi_{13}$ explicitly. We find $\Xi_{13} =\sum_{j=1}^6 \Xi_{13}^{(j)}$ with
\[ \begin{split}
\Xi_{13}^{(1)} & = \frac \k {N} \sum_{\substack{r \in P_H\\ v \in P_L }} \sum_{\substack{p \in \L_+^* \\ p \neq -r-v}}   \widehat V(p/N)\,\eta_r\, \g_{r+v} b_{p+r+v}^*\, b^*_{-p}\, \Big(1 - \frac{\cN_+}N \Big) b^*_{-r} (\g_v b_v + \s_v b^*_{-v}) \\
\Xi_{13}^{(2)} & = \frac \k {N} \sum_{\substack{r \in P_H\\ v \in P_L }}   \sum_{\substack{p \in \L_+^* \\ p \neq r}}  \widehat V(p/N)\,\eta_r \, \g_{r} b_{p-r}^*\, b^*_{-p}\, \Big(1 - \frac{\cN_+-1}N \Big) b^*_{r+v} (\g_v b_v + \s_v b^*_{-v}) \\
\Xi_{13}^{(3)} & = \frac \k {N}\sum_{\substack{r \in P_H\\ v \in P_L}}  \sum_{\substack{p \in \L_+^*\\ p \neq v}}  \widehat V(p/N)\,\eta_r\, \g_v\s_v \, b^*_{p-v}\, b^*_{-p}\,  b^*_{r+v} b^*_{-r} \Big(1 - \frac{\cN_+}N \Big) \\
\Xi_{13}^{(4)} & = -\frac \k {N} \sum_{\substack{r \in P_H\\ v \in P_L }}  \sum_{\substack{p \in \L_+^* \\ p \neq -v }}   \widehat V(p/N)\,\eta_r\, \g^2_v \, b^*_{p+v}\, b^*_{-p}\, \Big(1 - \frac{\cN_+}N \Big) b_{-r} b_{r+v} \\
\Xi_{13}^{(5)}& = -\frac \k {N^2} \sum_{\substack{r \in P_H\\ v \in P_L}}  \sum_{\substack{p, q \in \L_+^* \\p\neq -q}}  \widehat V(p/N)\,\eta_r \, \g_{q} b_{p+q}^*\, b^*_{-p}\, ( a^*_{r+v} a_q b^*_{-r} + b^*_{r+v}a^*_{-r}a_q ) (\g_v b_v + \s_v b^*_{-v}) \\
\Xi_{13}^{(6)}& = \frac \k {N^2} \sum_{\substack{r \in P_H\\ v \in P_L}}  \sum_{\substack{p, q \in \L_+^* \\p\neq -q}}  \widehat V(p/N)\,\eta_r \, \g_{q} b_{p+q}^*\, b^*_{-p}(  \g_v a^*_v a_q b_{-r} b_{r+v}  - \s_v b^*_{r+v} b^*_{-r} a^*_{-v} a_q )\,.
\end{split}\]
To bound $\Xi_{13}^{(1)}$ we use Cauchy-Schwarz. We find (with appropriate shifts of the summation variables)  
\[  \begin{split}
& |\bmedia{\xi, \Xi_{13}^{(1)} \xi}|\\
& \leq  \frac C {N}    \, \biggl(\;  \sum_{\substack{r \in P_H \\ v \in P_L}}   \sum_{\substack{p \in \L_+^*\\ p \neq r+v}}   |p|^2  \, \|   b_{-r}  b_{-p} b_{p+r+v}\, (\cN_+ + 1)^{-1/2} \xi \|^2\biggr)^{1/2}  \\
& \qquad \times \biggl(\;  \sum_{\substack{r \in P_H \\ v \in P_L}}   \sum_{\substack{p \in \L_+^*\\ p \neq r+v}} \frac{| \widehat V(p/N)|^2}{|p|^2}\, |\eta_r|^2  \left[  \|  b_{v}\, (\cN_+ + 1)^{1/2} \xi\|^2 + 
 |\s_{v}|^2 \| (\cN_++1) \xi \|^2 \right]  \biggr)^{1/2}  \\
& \leq \frac C {\sqrt N}\; \Big[ \| (\cN_++1)^{1/2} (\cK +1)^{1/2} \xi \|^2 + \| (\cN_++1) \xi \|^2 \Big]\,.
\end{split}\]
The bounds for $\Xi_{13}^{(j)}$ with $j=2,3,4$ can be obtained similarly. As for the terms $\Xi_{13}^{(5)}$ and $\Xi_{13}^{(6)}$, they can be estimated proceeding as we did for $\Xi_{10}$. We conclude that
\begin{equation}\label{eq:Xi13fin} \pm \Xi_{13} \leq \frac{C}{\sqrt{N}} (\cN_+ + 1) (\cK + 1) + (\cN_+ + 1)^3 \end{equation}
Also the term $\Xi_{14}$ can be controlled analogously. To avoid repetitions, we skip the details. 

To conclude the proof of the lemma it remains to show \eqref{eq:EC2}, which follows from Lemma \ref{lm:commQA} since, for any $v \in P_L$, we have 
\[ \begin{split}
 \Big| \frac \k N \sum_{r \in P_H } \big( \widehat V(r/N) +\widehat V((r+v)/N)\big) \eta_r\,  ( \g_v^2 + \s^2_v) \Big | &\leq C, \\
  \Big|  \frac \k N \sum_{r \in P_H} \big( \widehat V(r/N) +\widehat V((r+v)/N)\big) \eta_r\,\g_v \s_v\, \Big | & \leq \frac C {v^2}\,
  \end{split}
\]
\end{proof}

%%%%%%%%%%%%%%%%%%%%%%%%%%%%%%%%
%%%%%%%%%%%%%%%%%%%%%%%%%%%%%%%%
%%%%%%%%%%%%%%%%%%%%%%%%%%%%%%%%
%%%%%%%%%%%%%%%%%%%%%%%%%%%%%%%%

\subsection{Analysis of $e^{-A} \cH_N e^{A}$.}

In this section, we analyse the action of the cubic exponential on $\cH_N =\cK + \cV_N$.

\begin{prop} \label{prop:SHS}
Let $A$ be defined as in \eqref{eq:defA} and $\cH_N$ as defined after (\ref{eq:KcVN}). Then, under the assumptions of Proposition~\ref{thm:tGNo1}, we have
\[ \begin{split} 
 e^{-A} \cH_N\, e^{A}   =\; & \cH_N - \frac 1 {\sqrt{N}} \sum_{\substack{r \in P_H\\ v \in P_L }}  \k \widehat{V}(r/N)\, \Big[ b^*_{r+v} b^*_{-r} \big( \g_v b_v + \s_v b^*_{-v} \big) + \emph{h.c.} \Big] \\
 & - \frac 1 N \sum_{\substack{r \in P_H\\ v \in P_L}} \k  \big( \widehat V(r/N) +\widehat V((r+v)/N) \big)  \eta_r \\ &\hspace{2cm} \times \Big[ \s_v^2 + (\g^2_v + \s^2_v)\, b^*_v b_{v} + \g_v \s_v\,  \big( b_v b_{-v} + b_v^* b_{-v}^*\big)  \Big]\\
&  + \cE_{N}^{(\cH)}\,,
\end{split} \]
where the error $\cE_{N}^{(\cH)}$  satisfies
\[ \begin{split}
\pm \cE_{N}^{(\cH)} &\leq C N^{-1/4}  \Big[(\cN_++1) (\cH_N +1) + (\cN_++1)^3 \Big]\,.
\end{split}\]
\end{prop}

To show the proposition we use the following lemma.
\begin{lemma} \label{lm:commT0A}
Let $A$ be defined as in \eqref{eq:defA} and let 
\[ 
\Theta_0 = 
-\frac{1}{\sqrt{N}} \sum_{\substack{r \in P_H,\, v \in P_L }}  \k \widehat{V}(r/N)\, b^*_{r+v} b^*_{-r} \big( \g_v b_v + \s_v b^*_{-v} \big)\, 
\]
as defined in Lemma \ref{lm:commHA}. Then, under the assumptions of 
Proposition~\ref{thm:tGNo1}, 
\[
[\Theta_0 + \Theta_0^*, A] =  \sum_{j=0}^{12} \Pi_j  + \emph{h.c.} 
\]
with 
\[  \begin{split}
\Pi_0= \; & - \frac \k N \sum_{\substack{r \in P_H \\v \in P_L }} \big(\; \widehat V(r/N) +  \widehat V((r+v)/N) \;\big)  \eta_r  \big[\; \s_v^2 + (\g^2_v + \s^2_v) b^*_vb_v + \g_v \s_v (b_v b_{-v} + b^*_v b^*_{-v})\;\big] \\ 
\Pi_1=\; &  - \frac \k N \sum_{\substack{r \in P_H \\ v \in P_L}}  \widehat V((r+v)/N)  \eta_r (\g_v b^*_v + \s_v b_{-v}) \left[ \left(1 - \frac{\cN_+}N \right)^2 -1 \right] (\g_v b_v + \s_v b^*_{-v})  \\
& - \frac \k N \sum_{\substack{r \in P_H \\ v \in P_L }}  \widehat V(r/N)  \eta_r (\g_v b^*_v + \s_v b_{-v})  \left[\left(1 - \frac{\cN_+ +1 }N \right) \left(1 - \frac{\cN_+}N \right) -1 \right] \\ &\hspace{10cm} \times (\g_v b_v + \s_v b^*_{-v})  \\
\Pi_2 =\; & - \frac \k N \sum_{\substack{r \in P_H \\ v \in P_L}} \widehat V(r/N)  \eta_r \s_v \left(1 - \frac{\cN_+ +1 }N \right) \left(1 - \frac{\cN_+}N \right)  b_{r+v}(\g_{r+v} b_{-r-v} + \s_{r+v} b^*_{r+v})  \\
%%%
 \Pi_3 =\; & -\frac \k N \sum_{\substack{r \in P_H \\ v \in P_L}}  \sum_{\substack{w \in P_L \\-w +r+v \in P_H}}  \widehat{V}((r+v-w)/N) \eta_r  (\g_v b^*_v + \s_v b_{-v}) \left(1 - \frac{\cN_+ +1}N \right) \\
 & \hskip 4cm \times\big(b^*_{w-r-v} b_{-r} - \frac 1 N a^*_{w-r-v} a_{-r} \big)(\g_w b_w + \s_w b^*_{-w}) \\[9pt]
 \Pi_4 =\; & -\frac \k N \sum_{\substack{r \in P_H \\ v \in P_L}} \sum_{\substack{w \in P_L \\w+r \in P_H}}  \widehat{V}((r+w)/N) \eta_r  (\g_v b^*_v + \s_v b_{-v}) \left(1 - \frac{\cN_+}N \right) \\
 & \hskip 4cm \times\big(b^*_{r+w} b_{r+v} - \frac 1 N a^*_{r+w} a_{r+v} \big)(\g_w b_w + \s_w b^*_{-w})  \\
 \Pi_5 =\; & -\frac \k N \sum_{\substack{r \in P_H \\ v \in P_L}} \sum_{\substack{w \in P_L \\w+v \in P_H}} \widehat{V}((v+w)/N) \eta_r  \s_v  \left(1 - \frac{\cN_+}N \right) \\
 & \hskip 4cm \times\big(b^*_{w+v} b_{r+v} - \frac 1 N a^*_{w+v} a_{r+v} \big) b_{-r} (\g_w b_w + \s_w b^*_{-w})  \\[9pt]
 \Pi_6 =\; & \frac \k {N^2} \sum_{\substack{r \in P_H \\ v \in P_L}} \sum_{\substack{w \in P_L \\w+v \in P_H}} \widehat{V}((v+w)/N) \eta_r  \s_v  \left(1 - \frac{\cN_+}N \right) a^*_{v+w} a_{-r} b_{r+v} (\g_w b_w + \s_w b^*_{-w}) \\
\Pi_7  = \; & \frac \k N  \sum_{\substack{r \in P_H \\ v \in P_L}} \sum_{\substack{w \in P_L \\v-w \in P_H}} \widehat{V}((v-w)/N) \eta_r \g_v b^*_{r+v} b^*_{-r} \left(1 - \frac{\cN_+}N \right)   b^*_{w-v}(\g_w b_w + \s_w b^*_{-w})  \\
%\end{split}\]
%\[ \begin{split}
\Pi_8  =\; &\frac \k {N^2} \sum_{\substack{r,s \in P_H \\v,w \in P_L }}   \widehat{V}(s/N) \eta_r   (\g_v b^*_v + \s_v b_{-v}) (b_{-r} a^*_{s+w} a_{r+v} + a^*_{s+w} a_{-r} b_{r+v} ) b^*_{-s} (\g_w b_w + \s_w b^*_{-w}) \\
 \Pi_9 =\; &\frac \k {N^2} \sum_{\substack{r,s \in P_H \\v,w \in P_L }}  \widehat{V}(s/N) \eta_r  (\s_v a^*_{s+w} a_{-v} b_{-r} b_{r+v} - \g_v b^*_{r+v} b^*_{-r} a^*_{s+w} a_v \;\Big]\, b^*_{-s}\,(\g_w b_w + \s_w b^*_{-w})    \end{split}\]
and 
 \begin{equation}\label{eq:Pi1012} \begin{split}
 %%%%
\Pi_{10} =\; &  - \frac \k {\sqrt N} \sum_{\substack{r \in P_H,\, v \in P_L }} \widehat{V}(r/N) b^*_{r+v} [b^*_{-r}, A]  \big(  \g_v b_v + \s_v b^*_{-v} \big)  \\
\Pi_{11} =\; &  - \frac \k {\sqrt N} \sum_{\substack{r \in P_H,\, v \in P_L }} \widehat{V}(r/N) \g_v\,  b^*_{r+v} b^*_{-r}\, [ b_v , A]  \\
\Pi_{12}=\; &  - \frac \k {\sqrt N} \sum_{\substack{r \in P_H,\, v \in P_L}}  \widehat{V}(r/N)  b^*_{r+v} b^*_{-r} \s_v\, [ b^*_{-v}, A] 
\end{split} \end{equation}
For all $j=1,\dots , 12$ (but not for $j=0$) we have
\be \label{eq:EP1} 
\pm \big( \Pi_j + \text{h.c.} \big) \leq \frac{C}{\sqrt{N}} \Big[  
(\cN_++1) (\cK+1)+   (\cN_++1)^{3} \Big]  \,.
\ee
\end{lemma}

\vskip 0.5cm

\begin{proof}[Proof of Prop. \ref{prop:SHS}.] To show the proposition we write
\be  \begin{split} \label{eq:SHS}
 e^{-A} \cH_N\, e^{A}= \cH_N+ \int_0^1 ds \, e^{-s A}\, [\cH_N, A] e^{s A}  
\end{split}\ee
From  Lemma \ref{lm:commHA} we know that
\[
[\cH_N, A] = \Theta_0 + \Theta_0^* + \cE^{(\cH)}_{N,1}  
\]
where
\begin{equation}\label{eq:EHN1}
\pm \,\cE^{(\cH)}_{N,1} \leq  C N^{-1/4}\; \Big[(\cN_++1) (\cK +1) + (\cN_++1)^3 \Big]\,.
\end{equation}
Hence, \eqref{eq:SHS} implies that 
\[  \begin{split}
  e^{-A} & \cH_N\, e^{A } \\
%  =\; & \; \cH_N + \int_0^1 ds \, e^{-s A}\, (\Theta_0 + \Theta_0^*) \, e^{s A}    + \int_0^1 ds \, e^{-s A}\, %\cE_{N,1}^{(H)}\, e^{s A}  \\
%
 =\; &  \cH_N + \Theta_0 + \Theta_0^*\\
& \quad +\int_0^1 ds_1 \int_0^{s_1} ds_2 \, e^{-s_2 A}\, \big[(\Theta_0 + \Theta_0^*), A \big] e^{s_2 A}   + \int_0^1 ds \, e^{-s A}\, \cE_{N,1}^{(\cH)}\, e^{s A} \,.
\end{split}\]
Using Lemma \ref{lm:commT0A} and setting $\cE_{N,2}^{(\cH)} = \sum_{j=1}^{12} \Pi_j + \hc$ we finally obtain 
\[ \begin{split} 
  e^{-A} \cH_N\, e^{A}  =\; & \cH_N  + \Theta_0 + \Theta_0^* + \Pi_0 \\
  & +\int_0^1 ds_1 \int_0^{s_1} ds_2 \int_0^{s_2} d s_3 \, e^{-s_3 A}\, [2 \Pi_0, A] e^{s_3 A}  \\
& +\int_0^1 ds_1 \int_0^{s_1} ds_2 \, e^{-s_2 A}\, \cE_{N,2}^{(\cH)}  e^{s_2 A} + \int_0^1 ds e^{-s A}\, \cE_{N,1}^{(\cH)}\,  e^{s A} \,.
\end{split}\]
Proposition \ref{prop:SHS} now follows combining (\ref{eq:EHN1}) with the estimates (\ref{eq:EP1}) and with the observation that $\Pi_0 = -\Xi_0$, where $\Xi_0$ is defined in Lemma \ref{lm:commCA} and satisfies the bound \eqref{eq:EC2}.
\end{proof}

\begin{proof}[Proof of Lemma \ref{lm:commT0A}]  We write
\[
[\Theta_0, A] =  - \frac 1 {\sqrt N} \sum_{\substack{r \in P_H,\, v \in P_L \\r+v \neq 0}} \k \widehat{V}(r/N)   [b^*_{r+v}, A] b^*_{-r}\big( \g_v b_v + \s_v b^*_{-v}\big)  + \sum_{j=10}^{12} \Pi_{j}
\] 
Using \eqref{bpA} and normal ordering the quartic terms (with the exception of the factor $( \g_v b_v + \s_v b^*_{-v})$) we obtain that
\[ - \frac 1 {\sqrt N} \sum_{\substack{r \in P_H,\, v \in P_L \\r+v \neq 0}} \k \widehat{V}(r/N)   [b^*_{r+v}, A] b^*_{-r}\big( \g_v b_v + \s_v b^*_{-v}\big) = \sum_{j=0}^{9} \Pi_j \]

We now show \eqref{eq:EP1}. The bound for $\Pi_1$ follows from
\[ \begin{split} | \langle \xi, \Pi_1 \xi  \rangle | \leq\; & \frac{C}{N^2} \sum_{r \in P_H, v \in P_L} \frac{|\widehat{V} (r/N)|}{r^2} \,  \left[ \| b_v (\cN_+ +1)^{1/2} \xi \|^2 + |\s_v|^2 \| (\cN_+ + 1) \xi \|^2 \right] \\ \leq \; &\frac{C}{N} \| (\cN_+ + 1) \xi \|^2 \end{split} \]
To bound $\Pi_j$ with j=2,3,4,7 one uses that $|P_L| \leq C N^{3/2}$ and $\sum_{r \in P_H} |\eta_r|^2 \leq C N^{-1/2}$ . Hence
\[ \begin{split}
|\bmedia{\xi, \Pi_2 \xi}|  & \leq   \frac k N \sum_{\substack{r \in P_H \\ v \in P_L}} |\widehat V(r/N)|  |\eta_r| |\s_v|  \| b_{r+v}^* \xi \|  \|(\g_{r+v} b_{-r-v} + \s_{r+v} b^*_{r+v}) \xi\| \\
& \leq  \frac k N \Big(\sum_{\substack{r \in P_H \\ v \in P_L}}  |\eta_r|^2 |\s_v|^2  \| b_{r+v}^* \xi \|^2 \Big)^{1/2} \Big(  \sum_{\substack{r \in P_H \\ v \in P_L}} \|(\g_{r+v} b_{-r-v} + \s_{r+v} b^*_{r+v}) \xi\|^2 \Big)^{1/2} \\
& \leq \frac C {\sqrt N} \|(\cN_+ +1)^{1/2} \xi \|^2
\end{split}\]
Similarly, with Cauchy-Schwarz we can bound $\Pi_3$ by  
\[ \begin{split}
 |\bmedia{\xi, \Pi_3 \xi}|
 & \leq \frac C N \Big( \sum_{\substack{r \in P_H \\ v \in P_L}}  \sum_{\substack{w \in P_L : \\-w +r+v \in P_H}} \hskip -0.5cm  |\eta_r|^2  \left[ \| b_v  b_{w-r-v} \xi \|^2 + |\s_v|^2 \| b_{w-r-v} (\cN_+ + 1)^{1/2} \xi \|^2 \right] \Big)^{1/2} \\ &\hspace{2cm} \times \Big( \sum_{\substack{r \in P_H \\ v \in P_L}}  \sum_{\substack{w \in P_L : \\-w +r+v \in P_H}} \hskip -0.5cm  \left[ \| b_{-r} b_w \xi \|^2 + |\s_w|^2 \| b_{-r} (\cN_+ + 1)^{1/2} \xi \|^2 \right] \Big)^{1/2}  \\
& \leq \frac C {\sqrt N}\;  \|(\cN_++1) \xi \|^2
 \end{split}\]
 The terms $\Pi_4$ and $\Pi_7$ are bounded similarly. It is easy to check that $\Pi_5$ and $\Pi_6$ satisfy \eqref{eq:EP1}. For example, we have  
 \[ \begin{split}
 |\bmedia{\xi, \Pi_5 \xi}|&\leq \frac C N \sum_{\substack{r \in P_H \\ v \in P_L}} \sum_{\substack{w \in P_L \\w+v \in P_H}} |\widehat{V}((v+w)/N)| |\eta_r| | \s_v| \, \| b_{w+v} (\cN_++1)^{1/2} \xi \| \\
 & \hskip 1cm \times  \left[ \| b_{r+v}  b_{-r}  b_w (\cN_+ + 1)^{-1/2} \xi \| + |\s_w| \| b_{r+v}  b_{-r} b^*_{-w} (\cN_+ + 1)^{-1/2} \xi \| \right]  \\
 & \leq \frac C N \Big(\sum_{\substack{r \in P_H \\ v \in P_L}} \sum_{\substack{w \in P_L \\w+v \in P_H}} |\eta_r|^2 | \s_v|^2 \| b_{w+v} (\cN_++1)^{1/2} \xi \|^2  \Big)^{1/2} \\
 & \hskip 1cm \times \bigg(\sum_{\substack{r \in P_H, \\ v \in P_L}} \sum_{\substack{w \in P_L \\w+v \in P_H}} \bigg[ \| b_{r+v}  b_{-r}  b_w (\cN_+ + 1)^{-1/2} \xi \|^2 \\ &\hspace{6cm} + |\s_w| \| b_{r+v}  b_{-r} b^*_{-w} (\cN_+ + 1)^{-1/2} \xi \|^2 \bigg]  \bigg)^{1/2}  \\
 & \leq \frac C N \; \| (\cN_++1) \xi \|^2
  \end{split}\]
(In the term containing the creation operator $b_w^*$, we first sum over $\wt{v} = v+r$ and over $r$. This produces a factor $(\cN_+ + 1)$ which can be moved through $b_w^*$. At this point, we can estimate $b_w^*$ by an additional factor $(\cN_+ + 1)^{1/2}$; with this procedure, we do not have to compute the commutator between $b_w^*$ and the other annihilation operators). The bound for $\Pi_6$ is similar. As for $\Pi_8$, we decompose it as 
  \[ \begin{split}
 \Pi_8 = \; &\frac 1 {N^2} \sum_{\substack{r,s \in P_H \\v,w \in P_L }}  \k \widehat{V}(s/N) \eta_r   (\g_v b^*_v + \s_v b_{-v}) b_{-r} b^*_{-s} a^*_{s+w} a_{r+v} \,(\g_w b_w + \s_w b^*_{-w}) \\
 & +\frac 1 {N^2} \sum_{\substack{r\in P_H \\v,w \in P_L }}  \k \widehat{V}((r+v)/N) \eta_r   (\g_v b^*_v + \s_v b_{-v}) b_{-r} b^*_{-r-v +w}  \,(\g_w b_w + \s_w b^*_{-w}) \\
 & +\frac 1 {N^2} \sum_{\substack{r,s \in P_H \\v,w \in P_L }}  \k \widehat{V}(s/N) \eta_r   (\g_v b^*_v + \s_v b_{-v}) a^*_{s+w} a_{-r} b_{r+v}  b^*_{-s}\,(\g_w b_w + \s_w b^*_{-w}) \\
  = \; &\Pi_8^{(1)} + \Pi_8^{(2)}  +  \Pi_8^{(3)} 
  \end{split}\]
The term $\Pi^{(1)}_8$ can be bounded commuting first the operator $b_{-r}$ to the right, analogously to the estimates (\ref{eq:Xi101}) and (\ref{eq:Xi102}) for the term $\Xi_{10}^{(1)}$ in the proof of Lemma \ref{lm:commCA}. Also the terms $\Pi^{(2)}_8$ (which is similar to  $\Xi_{10}^{(3)}$ in (\ref{eq:Xi10dec})) and $\Pi^{(3)}_8$ can be treated similarly. We conclude that 
\[ \pm \Pi_8 \leq C N^{-1/2} (\cN_+ + 1)^3 \, . \]
The operator $\Pi_9$ can be controlled as $\Pi_8$.
   
Finally, to bound the terms $\Pi_{10}$, $\Pi_{11}$ and $\Pi_{12}$ in (\ref{eq:Pi1012}), we can expand them using (\ref{bpA}). The contributions arising from $\Pi_{10}$ are similar to the terms $\Pi_{1}, \dots, \Pi_{9}$  and can be estimated analogously. On the other hand, the terms arising from $\Pi_{11}$ and $\Pi_{12}$ are similar to those arising from $\Xi_{13}$ and $\Xi_{14}$ in the proof of  Lemma \ref{lm:commCA}, and can be handled proceeding as we did to show (\ref{eq:Xi13fin}), 
making use of the kinetic energy operator. 
\end{proof}
   
%%%%%%%%%%%%%%%%%%%%%%%%%%%%%%%%%%%%%%%%%%%%%%%%%%%%%%%%%%%%%%%%%%%%%%%%%%%%%%%%%%%%%%%%%%

\subsection{Proof of Theorem~\ref{thm:tGNo1}}
\label{sec:pf-thmA}

Combining the results of Prop.~\ref{prop:SQS},  Prop.~\ref{prop:SCS} and  Prop.~\ref{prop:SHS} we conclude that
\be \begin{split}  \label{eq:SGS}
\cJ_N = \; &e^{-A} \, \cG_N e^A    \\
   = \; &C_{\cG_N} + \cQ_{\cG_N} + \cH_N \\
 &+ \frac 1 N \sum_{\substack{r \in P_H, v \in P_L }} \k  \big(\widehat V(r/N) + \widehat V((r+v)/N)\big)   \eta_r \\ &\hspace{3cm} \times \Big[ \s^2_v + (\g^2_v + \s^2_v)\, b^*_v b_{v} + \g_v \s_v\,  \big( b_v b_{-v} + b_v^* b_{-v}^*\big)  \Big]\\
 &+\cC_N - \frac 1 {\sqrt{N}} \sum_{\substack{r \in P_H, v \in P_L}}  \k \widehat{V}(r/N)  \Big[ b^*_{r+v} b^*_{-r} \big( \g_v b_v + \s_v b^*_{-v} \big) + \hc \Big] \\
&+ \wt{\cE}_{\cJ_N}\,,
\end{split} \ee
with an error operator $\wt{\cE}_{\cJ_N}$ satisfying 
\[ \begin{split} 
\pm\, \wt{\cE}_{\cJ_N} &\leq  C N^{-1/4} \; \Big[(\cH_N +1)(\cN_++1)  + (\cN_++1)^3 \Big]
\end{split}\]
We show now that the sum of the cubic terms on the fifth line of \eqref{eq:SGS} also contributes to the error term. In fact we have
\[ \begin{split} 
\cC_N -  \frac 1 {\sqrt{N}} & \sum_{\substack{r \in P_H,\, v \in P_L }}  \k \widehat{V}(r/N)\, \Big[ b^*_{r+v} b^*_{-r} \big( \g_v b_v + \s_v b^*_{-v} \big) + \hc \Big]  \\
=\; & \frac 1 {\sqrt{N}} \sum_{\substack{ v \in  P_H, \,r \in \L^*_+ \\ r+v \neq 0}}  \k \widehat{V}(r/N)\, \Big[ b^*_{r+v} b^*_{-r} \big( \g_v b_v + \s_v b^*_{-v} \big) + \hc \Big]  \\
&  + \frac 1 {\sqrt{N}} \sum_{\substack{ v,r \in P_L\\ r+v \neq 0}}  \k \widehat{V}(r/N)\, \Big[ b^*_{r+v} b^*_{-r} \big( \g_v b_v + \s_v b^*_{-v} \big) + \hc \Big]  
\\ = \;& \text{Z}_1 + \text{Z}_2 
\end{split}\] 
To bound $\text{Z}_1$ we use that $|v|^{-1} \leq N^{-1/2}$ for $v \in P_H$ and $\sum_{v \in P_H} |\s_v|^2 \leq C N^{-1/2}$. We find  
\[ \begin{split} |\langle \xi, \text{Z}_1\xi\rangle | \leq \; &\frac{C}{\sqrt{N}} \sum_{\substack{ v \in P_H, \, r \in \L^*_+ \\ r+v \neq 0}} |\widehat{V} (r/N)| \| b_{r+v} b_{-r} \xi \| \left[ \| b_v \xi \| + |\s_v| \| (\cN_+ + 1)^{1/2} \xi \| \right]  \\ \leq \; &\frac{C}{\sqrt{N}} \bigg( \sum_{\substack{ v \in P_H, \, r \in \L^*_+ \\ r+v \neq 0}} \frac{r^2}{v^2} \| b_{r+v} b_{-r} \xi \|^2 \bigg)^{1/2} \bigg( \sum_{\substack{ v \in P_H, \, r \in \L^*_+ \\ r+v \neq 0}} \frac{|\widehat{V} (r/N)|^2}{r^2} v^2 \| b_v \xi \|^2 \bigg)^{1/2} \\ &+ \frac{C}{\sqrt{N}} \bigg( \sum_{\substack{ v \in P_H, \, r \in \L^*_+ \\ r+v \neq 0}} r^2 \| b_{r+v} b_{-r} \xi \|^2 \bigg)^{1/2} \bigg( \sum_{\substack{ v \in P_H, \, r \in \L^*_+ \\ r+v \neq 0}} \frac{|\widehat{V} (r/N)|^2}{r^2}  |\s_v|^2  \bigg)^{1/2}  \\ &\hspace{9cm} \times \| (\cN_+ + 1)^{1/2}  \xi \| \\
\leq \; &\frac{C}{N^{1/4}} \| (\cK+ 1)^{1/2} (\cN_+ + 1)^{1/2} \xi \|^2 
\end{split} \]
The term $\text{Z}_2$ is bounded using Cauchy-Schwarz and the estimate $\sum_{r \in P_L} |r|^{-2} \leq C N^{1/2}$. We obtain 
\[ \begin{split}
|\langle \xi, \text{Z}_2\xi\rangle | \leq \; &  \frac C {\sqrt N}\; \Big( \sum_{\substack{  r,v \in P_L \\ r+v \neq 0}}  r^2\| b_{r+v} b_{-r} \xi \|^2 \Big)^{1/2} \Big( \sum_{\substack{  r,v \in P_L,\\ r+v \neq 0}}  \frac{1}{r^2}  \left[ \| b_v \xi \|^2 + |\s_v|^2 \| (\cN_+ + 1)^{1/2} \xi \|^2 \right]  \Big)^{1/2} \\
\leq \; &\frac C {N^{1/4}}\; \|(\cN_++1)^{1/2} (\cK +1)^{1/2} \xi \|^2\,.
\end{split}\]
Similarly, we can show that, in the term on the third and fourth line in (\ref{eq:SGS}), the restriction $r \in P_H, v \in P_L$ can be removed producing only a negligible error. We conclude that 
\[ \begin{split} % \label{eq:SGS2}
 \cJ_N =\; &  C_{\cG_N} + \cQ_{\cG_N} + \cH_N \\
 & +\frac 1 N \sum_{\substack{p, q \in \L_+^*}}\k \widehat V((p+q)/N)  \eta_q  \Big[ \s^2_p + (\g^2_p + \s^2_p)\, b^*_p b_{p} + \g_p \s_p\,  \big( b_p b_{-p} + b_p^* b_{-p}^*\big)  \Big]  \\
  & +\frac 1 N \sum_{\substack{p, q \in \L_+^*}}\k \widehat V(q/N) \eta_q  \Big[ \s^2_p + (\g^2_p + \s^2_p)\, b^*_p b_{p} + \g_p \s_p\,  \big( b_p b_{-p} + b_p^* b_{-p}^*\big)  \Big]  \\
 &+ \bar{\cE}_{\cJ_N}
\end{split} \]
with
\[
\pm\, \bar{\cE}_{\cJ_N} \leq  \frac C {N^{1/4}}\; \Big[ (\cH_N+1)(\cN_++1) + (\cN_++1)^3\Big] \,.
\]
Theorem \ref{thm:tGNo1} now follows from the observation that 
\begin{equation*}
  \begin{split}
  &C_{\cG_N}+\frac 1 N \sum_{\substack{p, q \in \L_+^*}}\k  \widehat V((p+q)/N)  \eta_q  \s^2_p +\frac 1 N \sum_{\substack{p, q \in \L_+^*}}\k \widehat V(q/N) \eta_q  \s^2_p  =C_{\cJ_N} + \cO (N^{-1}) 
 \end{split}
\end{equation*}
and that 
\begin{equation}\label{eq:FinalSumQ}
  \begin{split}
  \cQ_N &+ \cK +\frac 1 N \sum_{\substack{p, q \in \L_+^*}}\k  \widehat V((p+q)/N)  \eta_q  \Big[(\g^2_p + \s^2_p)\, b^*_p b_{p} + \g_p \s_p\,  \big( b_p b_{-p} + b_p^* b_{-p}^*\big)  \Big]\\
  &+\frac 1 N \sum_{\substack{p, q \in \L_+^*}}\k  \widehat V(q/N)  \eta_q  \Big[(\g^2_p + \s^2_p)\, b^*_p b_{p} + \g_p \s_p\,  \big( b_p b_{-p} + b_p^* b_{-p}^*\big)  \Big] =\wt{\cQ}_N+ \wt{\cE}_{\cQ_N}, 
 \end{split}
\end{equation}
with 
\begin{equation}\label{eq:errorEQ}
\pm\, \wt{\cE}_{\cQ_N} \leq \frac CN (\cN_++1)\,.
\end{equation}
Here we used the fact that the contribution to $\cQ_N$ arising from the last term in (\ref{eq:phi}) and (\ref{eq:gamma}) cancels with the last sum on the l.h.s. of (\ref{eq:FinalSumQ}) (it is easy to check that the remainder corresponding to the momentum $q=0$ satisfies (\ref{eq:errorEQ})).

\appendix

\section{Condensate Depletion} 
\label{sec:deple}

The goal of this short appendix is to prove the formula (\ref{eq:depl0}) for the number of orthogonal excitations of the condensate, in the ground state of (\ref{eq:Ham0}). 

We start with the observation that 
\begin{equation}\label{eq:deple0} \begin{split} \langle  U_N \psi_N, \cN_+ U_N \psi_N \rangle = \; &\left\langle \left[ U_N \psi_N - e^{i\omega} e^{B(\eta)} e^{A} e^{B(\tau)} \Omega \right] , \cN_+ U_N \psi_N \right\rangle \\ 
&+ \left\langle   e^{i\omega} e^{B(\eta)} e^{A} e^{B(\tau)} \Omega, \cN_+ \left[  U_N \psi_N - e^{i\omega} e^{B(\eta)} e^{A} e^{B(\tau)} \Omega \right] \right\rangle \\ &+ \left\langle e^{B(\eta)} e^{A} e^{B(\tau)} \Omega, \cN_+ e^{B(\eta)} e^{A} e^{B(\tau)} \Omega \right\rangle \end{split} \end{equation}
From (\ref{eq:gs-appro}), Lemma \ref{lm:Ngrow} and Prop. \ref{prop:hpN} we conclude that 
\[ \left|   \langle  U_N \psi_N, \cN_+ U_N \psi_N \rangle - \left\langle e^{B(\eta)} e^{A} e^{B(\tau)} \Omega, \cN_+ e^{B(\eta)} e^{A} e^{B(\tau)} \Omega \right\rangle \right| \leq C N^{-1/8} \]

Proceeding as in Section \ref{sec:GN2} and recalling the notation $\g_p = \cosh \eta_p$, $\s_p = \sinh \eta_p$, we find 
\[
 e^{-B(\eta)} \cN_+  e^{B(\eta)} = \sum_{p \in \Lambda^*_+}  \left[\; (\g_p^2 + \s_p^2) b^*_p b_p + \g_p \s_p (b^*_p b^*_{-p} + b_p b_{-p}) + \s_p^2  \;\right] + \wt{\cE}_1
\]
where $\pm\, \wt{\cE}_1 \leq C N^{-1} (\cN_++1)^2$. By Lemma \ref{lm:commQA} and Proposition \ref{prop:cN} we have
\[ e^{-A}e^{-B(\eta)} \cN_+  e^{B(\eta)} e^{A} = 
 \sum_{p \in \Lambda^*_+}  \left[\; (\g_p^2 + \s_p^2) b^*_p b_p + \g_p \s_p (b^*_p b^*_{-p} + b_p b_{-p}) + \s_p^2  \;\right] + \wt{\cE}_2\]
with $\pm\,\wt{\cE}_2 \leq  C N^{-1/2} (\cN_++1)^2$. Conjugating with the generalized Bogoliubov transformation $e^{B(\tau)}$ and taking the vacuum expectation, we obtain 
\[ \begin{split} &\left\langle e^{B(\eta)} e^{A} e^{B(\tau)} \Omega, \cN_+ e^{B(\eta)} e^{A} e^{B(\tau)} \Omega \right\rangle \\ &\hspace{1cm} = \sum_{p \in \L^*_+} \left[ \s_p^2 + (\s_p^2 + \g_p^2) \sinh^2 \tau_p + 2 \g_p \s_p \sinh (\tau_p) \cosh (\tau_p) \right] + \cO (N^{-1/2}) \end{split} \]
With \eqref{eq:taup}, we find 
\[
2\sinh^2 \tau_p = \frac{F_p}{\sqrt{F_p^2 - G_p^2}} -1\,,  \qquad  2\sinh \tau_p \cosh \tau_p = \frac{-G_p}{\sqrt{F_p^2 - G_p^2}}\,.
\]
Using (\ref{eq:defFpGp}), we arrive at
\[ \begin{split} &\left\langle e^{B(\eta)} e^{A} e^{B(\tau)} \Omega, \cN_+ e^{B(\eta)} e^{A} e^{B(\tau)} \Omega \right\rangle \\ &\hspace{1cm} = \sum_{p \in \L^*_+}
  \frac{p^2   + \big(\widehat{V} (\cdot/N)\ast\widehat{f}_{\ell,N}\big)_p - \sqrt{ p^{4} + 2p^2 \big(\widehat{V} (\cdot/N)\ast\widehat{f}_{\ell,N}\big)_p }}{2 \sqrt{p^{4} + 2p^2 \big(\widehat{V} (\cdot/N)\ast\widehat{f}_{\ell,N}\big)_p }}
+ \cO (N^{-1/2}) \end{split} \]
with $\widehat {f}_{\ell,N}$ as in (\ref{eq:fellN}). Proceeding as in the proof of (\ref{eq:EBogconv}), we conclude that 
\[ \left\langle e^{B(\eta)} e^{A} e^{B(\tau)} \Omega, \cN_+ e^{B(\eta)} e^{A} e^{B(\tau)} \Omega \right\rangle = \sum_{p \in \L^*_+}
  \frac{p^2   + 8\pi \frak{a}_0  - \sqrt{ p^{4} + 16 \pi \frak{a}_0  p^2}}{2\sqrt{p^4 + 16\pi \frak{a}_0  p^2}} + \cO (N^{-1/2}) \]
Eq. (\ref{eq:depl0}) follows by combining (\ref{eq:deple0}) with the last equation, since 
\[ 1- \langle \ph_0, \gamma^{(1)}_N \ph_0 \rangle = N^{-1} \langle U_N \psi_N, \cN_+ U_N  \psi_N
 \rangle \]

\section{Properties of the scattering function}\label{appx:sceq}

In this appendix we give a proof of Lemma \ref{3.0.sceqlemma} containing the basic properties of the solution of the Neumann problem \eqref{eq:scatl}. 

\begin{proof}[Proof of Lemma \ref{3.0.sceqlemma}]
We adapt the proof of \cite[Lemma A.1]{ESY0} and \cite[Lemma 4.1]{BBCS3}. We start showing an upper bound for \eqref{eq:lambdaell}. We consider the solution $f$ of the zero energy scattering equation \eqref{eq:0en} on $\mathbb{R}^3$ with boundary condition $f(x)\to1$ as $|x|\to\infty$ (for the properties of $f$ we refer to \cite[Lemma D.1]{ESY2}). We set $r=|x|$ and $m(r)=rf(r)$. Clearly $m(r)$ satisfies
\begin{equation}\label{eq:zeroEm}
-m''(r) +\frac 1 2 V(r) m(r) =0\, 
\end{equation}
We define 
\begin{equation}\label{eq:upppsi}
 \psi(r)=\sin(km(r))/r
\end{equation}
with $k\in \mathbb{R}$.
From
\[
 \partial_r\psi(r)=\frac{1}{r^2}(m'(r)kr\cos(km(r))-\sin(km(r))
\]
we conclude that $\psi$ satisfies Neumann boundary conditions at $r=N\ell$ if and only if 
\begin{equation}\label{eq:tang}
 kN\ell=\tan(k(N\ell-\frak{a}_0))
\end{equation}
(recall that $m(r)=r-\frak{a}$ for $r$ outside the support of the potential).
We choose $k$ to be the smallest positive real number satisfying equation \eqref{eq:tang}. Expanding the tangent, we find a constant $C>0$ such that
\begin{equation}\label{eq:B0}
\frac{3\frak{a}_0}{(N\ell)^3} \left(1 + \frac 9 5 \frac{\frak{a}_0}{N \ell} - C \frac{\frak{a}_0^2}{(N\ell)^2} \right) \leq k^2 \leq \frac{3\frak{a}_0}{(N\ell)^3} \left(1+ \frac 9 5 \frac{\frak{a}_0}{ N\ell} + C \frac{\frak{a}_0^2}{(N\ell)^2} \right)
\end{equation}
We calculate now $\langle\ps, (-\D+V/2) \ps\rangle$. To this end we compute
\begin{equation}\label{eq:spp}
 -[\sin(k m(r))]'' = k^2 \big(m'(r)\big)^2 \sin(k m(r)) - k m''(r)  \cos(k m(r))
\end{equation}
Using that $\D \psi(r) = \frac 1 r \frac{\dpr^2}{\dpr r^2} (r \psi(r))$ and \eqref{eq:spp} we get
\begin{equation}\label{eq:trialS}
\begin{split}
 \ps (-\D+V/2)\ps  =\; & k^2 \ps^2 \\
& + k^2 ((m'(r))^2 -1) \left( \frac{\sin(k m(r))}{r^2} \right)^2 \\
&+\frac 1 {r^2} \Big[ - k m'' \sin(k m(r)) \cos(k m(r)) + \frac{1}{2} V (\sin(k m(r)))^2 \Big]  
\end{split}
\end{equation}
We denote with $R$ the radius of the support of $V$, so that 
$\text{supp} V\subset\{x\in\mathbb{R}^3:|x|\leq R\}$.
Notice that the terms on the second and  third line in \eqref{eq:trialS} are supported on $|x|\leq R$; the contribution to $\langle\ps, (-\D+V/2)\ps\rangle$ coming from the second line is bounded by 
\[
  4 \pi k^2 \int_0^{R} dr r^2\, \big[(m'(r))^2 -1 \big]   \frac{k^2 m^2(r)}{r^4} \leq  C k^4  
\]
where we used that $m(r)  \leq r$ (and the fact that $m'$ is bounded, as follows from \cite[Lemma D.1]{ESY2}). We consider now the contribution coming from the third line in \eqref{eq:trialS}. Using equation \eqref{eq:zeroEm} we notice that the summands of order $k^2$ in the square brecket cancel. We get
\begin{equation}\nonumber
\begin{split}
2\pi\int_0^{R} dr V(r) \big[ - k m(r) \sin(k m(r)) \cos(k m(r)) +  (\sin(k m(r)))^2 \big]  \leq C k^4 
\end{split}
\end{equation}
We have therefore
\begin{equation}\nonumber
\begin{split}
 \langle\ps, (-\D+V/2) \ps\rangle\leq k^2  \langle\ps, \ps\rangle + Ck^4 
\end{split}
\end{equation}
Using the estimate
\[
\sin(k m(r)) \geq C k m(r) \geq C k r \quad \text{for } \quad  2R < r < N \ell 
\]
(since $m(r)=r -a > r - R$ for $r>R$), we also get the lower bound
\[
\langle\ps,\ps\rangle = 4 \pi \int_0^{N\ell} dr\, \sin^2 (k m(r))   \geq C \int_{2R}^{N\ell} dr \,k^2 r^2 \geq C (N\ell)^3 k^2 \geq C 
\]
by the lower bound in (\ref{eq:B0}). The upper bound in (\ref{eq:B0}) implies therefore that 
\begin{equation}\label{eq:UBlambda}
\begin{split}
\l_\ell \leq \frac{\langle\ps, (-\D+V/2) \ps\rangle }{\langle\ps, \ps\rangle } \leq k^2+C k^4 
\leq  \frac{3\frak{a}_0}{(N\ell)^3} \left(1+ \frac 9 5 \frac{\frak{a}_0}{ N\ell} + C \frac{\frak{a}_0^2}{(N\ell)^2} \right) \end{split}
\end{equation}

We look now for a lower bound for $\l_\ell $. Given any function $\f$ satisfying Neumann boundary conditions at $|x|=N\ell$, we can write it as  $\f(x)=g(x)\ps(x)$, with $\ps(x)$ the trial function defined in (\ref{eq:upppsi}) and $g>0$ satisfying Neumann boundary condition at $|x|=N\ell$, too. From the identity
\begin{equation}\nonumber
\begin{split}
\big(-\D+V/2\big) \f = \left[\big(-\D+V/2\Big)  \ps\right] g - (\D g) \ps - 2 \nabla g \nabla \ps
\end{split}
\end{equation}
we have
\begin{equation}\nonumber
\begin{split}
\int_{|x|\leq N\ell} dx \, \bar{\f} \big(-\D+V/2 \big) \f = \int_{|x|\leq N\ell} dx \, |\nabla g|^2 \ps^2 + \int_{|x|\leq N\ell} dx \, |g|^2 \ps \big(-\D+V/2\big) \ps
\end{split}
\end{equation}
From \eqref{eq:trialS} we see that
\begin{equation}\nonumber
\begin{split}
\Big| \ps\big(-\D+V/2 \big) \ps - k^2 \ps^2 \Big| \leq C k^4 r^{-2} \chi(r \leq R)  
\end{split}
\end{equation}
Therefore
\begin{equation}\nonumber
\begin{split}
\int_{|x|\leq N\ell} dx\, \bar \f \big(-\D+V/2\big) \f \geq k^2\|\f\|^2_2 + \int_{|x|\leq N\ell} dx |\nabla g|^2 \ps^2 - C k^4 \int_{|x|\leq R} dx \frac{|g(x)|^2}{|x|^2} 
\end{split}
\end{equation}
where we indicated with $\|\cdot\|_2$ the $L^2$ norm on $B_{N\ell}=\{x\in\mathbb{R}^3:|x|\leq N\ell\}$.
In the second integral we can bound 
\begin{equation}\nonumber
\begin{split}
\ps(r) \geq \frac{k m(r)}{r} \geq c k 
\end{split}
\end{equation}
being $m(r)\geq c r$, for a constant $c >0$ (see \cite[Lemma D.1]{ESY2}). We get 
\begin{equation}\nonumber
\begin{split}
\int_{|x|\leq N\ell} dx\, \bar \f \big(-\D+V/2\big) \f \geq k^2\|\f\|^2 + c^2 k^2 \int_{|x|\leq N\ell} dx |\nabla g|^2 - C k^4 \int_{|x|\leq R} dx \frac{|g(x)|^2}{|x|^2} 
\end{split}
\end{equation}
The third integral can be bounded using the following finite volume version of the Hardy inequality (see, for example \cite[Lemma 5.1]{ESY0}), so that
\begin{equation}\nonumber
\begin{split}
\int_{|x|\leq R} \frac{|g(x)|^2}{|x|^2} dx \leq C \int_{|x|\leq N\ell} |\nabla g|^2 dx + \frac {C } { (N\ell)^3} \Big(\int_{|x|\leq R} \frac{dx}{|x|^2}\Big) \int_{|x|\leq N\ell} |g|^2 dx
\end{split}
\end{equation}
We have therefore
\[
\int_{|x|\leq N\ell} dx \bar \f \big(-\D+V/2\big) \f \geq k^2\|\f\|^2 + (c^2 k^2 - C  k^4) \int_{|x|\leq N\ell} dx |\nabla g|^2 -  \frac{C k^4}{(N\ell)^3} \int_{|x|\leq N\ell} dx |g|^2
\]
Using again $\ps \geq c k$ we get the lower bound
\begin{equation}\label{eq:LBlambda}
\begin{split}
\int_{|x|\leq N\ell} dx \bar \f \big(-\D+ V/2\big) \f & \geq k^2\|\f\|^2 -  \frac{C k^2}{ (N\ell)^3} \int_{|x|\leq N\ell} dx |g|^2 |\ps|^2  \\
& \geq  k^2\, \Big(1 -\frac{C k^2}{ (N\ell)^3}  \Big) \|\f\|^2     \\
& \geq \frac{3\frak{a}_0}{  (N\ell)^3} \left(1+ \frac 9 5 \frac{\frak{a}_0}{ N\ell} - C \frac{\frak{a}_0}{(N\ell)^2} \right)  \|\f\|^2
\end{split}
\end{equation}
The last estimate, together with inequality \eqref{eq:UBlambda}, proves i). 

We prove now part ii) and Eq. (\ref{eq:intw}) in part iii). The bounds $0\leq f_\ell, w_\ell \leq 1$ have been proved in  \cite[Lemma A.1]{ESY0}. We show \eqref{eq:Vfa0} and \eqref{eq:intw}. We set $r=|x|$ and $m_\ell(r)=rf_\ell(r)$, where $f_\ell(r)$ is now the solution of the Neumann problem \eqref{eq:scatl}. We rewrite \eqref{eq:scatl} as
\begin{equation}\label{eq:mlambda}
-m_\ell''(r)+\frac{1}{2}V(r)m_\ell(r)=\l_\ell m_\ell(r)
\end{equation}
For $r\in(R,N\ell]$, we find
\begin{equation}
 m_\ell(r)=\l_\ell^{-1/2}\sin(\l_\ell^{1/2}(r-N\ell))+N\ell\cos(\l_\ell^{1/2}(r-N\ell))
\end{equation}
Expanding up to order $\l_\ell^2$, we obtain
\begin{equation}\label{eq:mExpanded}
 m_\ell(r)=r-\mathfrak{a}_0+\frac{3}{2}\frac{\frak{a}_0}{N\ell}r-\frac{1}{2}\frac{\frak{a}_0}{(N\ell)^3}r^3+\cO(\frak{a}_0^2(N\ell)^{-1})
\end{equation}
Using the scattering equation we can write
\begin{equation}
\begin{split}
\int V(x) f_\ell (x) dx&=4\pi\int_0^{N\ell} dr\, rV(r)m_\ell(r)\\
&=8\pi\int_0^{N\ell}dr\, (rm_\ell''(r)+\l_\ell rm_\ell(r))\\
\end{split}
\end{equation}
The first contribution on the right-hand-side vanishes due to the boundary conditions. We evaluate the second contribution using expansion \eqref{eq:mExpanded}, to get
\begin{equation}
\begin{split}
8\pi\l_\ell\int_0^{N\ell}dr\,  rm_\ell(r)&= 8\pi\l_\ell\left(\frac{(N\ell)^3}{3}-\frac{\frak{a}_0}{10}(N\ell)^2+\cO(\frak{a}^2N\ell)\right)\\
\end{split}
\end{equation}
With \eqref{eq:lambdaell} we obtain
\begin{equation*}
\begin{split}
8\pi\l_\ell&\int_0^{N\ell}dr\,  rm_\ell(r)\\
&= 8\pi\frac{3\frak{a}_0 }{(\ell N)^3} \left(1 + \frac{9}{5}\frac{\frak{a}_0}{N\ell}+\mathcal{O} \big(\frak{a}_0^2  / (\ell N)^2\big) \right) \left(\frac{(N\ell)^3}{3}-\frac{\frak{a}_0}{10}(N\ell)^2+\cO(a^2N\ell)\right)\\
&= 8\pi \left(\frak{a}_0 + \frac{3}{2}\frac{\frak{a}_0}{N\ell}+\mathcal{O} \big(\frak{a}_0^3  / (\ell N)^2\big) \right) 
\end{split}
\end{equation*}
which proves \eqref{eq:Vfa0}.   Starting from the expansion \eqref{eq:mExpanded} and recalling that $w_\ell=1-f_\ell$, an easy calculation leads to \eqref{eq:intw}. Eq. (\ref{3.0.scbounds1}) in part iii) and part iv) have been shown in \cite[Lemma 4.1]{BBCS3}.
\end{proof}

%%%%%%%%%%%%%%%%%%%%%%%%%%%%%%%%%%%%%%%%%%%%%%%%%%%%%%%%%%%%%%%%%%%%%%%%%%%%%%%%%%%%%%%%%%%%%%%%%%%%%%%%%%%%%%%%%%%%%%%%%%%%%%%%%%%%%%%%%%%%%%%%%%%%%%%%%%%%%%%%%%%%%%%%%%%%%
%%%%%%%%%%%%%%%%%%%%%%%%%%%%%%%%%%%%%%%%%%%%%%%%%%%%%%%%%%%%%%%%%%%%%%%%%%%%%%%%%%%%%%%%%%%%%%%%%%%%%%%%%%%%%%%%%%%%%%%%%%%%%%%%%%%%%%%%%%%%%%%%%%%%%%%%%%%%%%%%%%%%%%%%%%%%%
%%%%%%%%%%%%%%%%%%%%%%%%%%%%%%%%%%%%%%%%%%%%%%%%%%%%%%%%%%%%%%%%%%%%%%%%%%%%%%%%%%%%%%%%%%%%%%%%%%%%%%%%%%%%%%%%%%%%%%%%%%%%%%%%%%%%%%%%%%%%%%%%%%%%%%%%%%%%%%%%%%%%%%%%%%%%%
%%%%%%%%%%%%%%%%%%%%%%%%%%%%%%%%%%%%%%%%%%%%%%%%%%%%%%%%%%%%%%%%%%%%%%%%%%%%%%%%%%%%%%%%%%%%%%%%%%%%%%%%%%%%%%%%%%%%%%%%%%%%%%%%%%%%%%%%%%%%%%%%%%%%%%%%%%%%%%%%%%%%%%%%%%%%%
%%%%%%%%%%%%%%%%%%%%%%%%%%%%%%%%%%%%%%%%%%%%%%%%%%%%%%%%%%%%%%%%%%%%%%%%%%%%%%%%%%%%%%%%%%%%%%%%%%%%%%%%%%%%%%%%%%%%%%%%%%%%%%%%%%%%%%%%%%%%%%%%%%%%%%%%%%%%%%%%%%%%%%%%%%%%%
%%%%%%%%%%%%%%%%%%%%%%%%%%%%%%%%%%%%%%%%%%%%%%%%%%%%%%%%%%%%%%%%%%%%%%%%%%%%%%%%%%%%%%%%%%%%%%%%%%%%%%%%%%%%%%%%%%%%%%%%%%%%%%%%%%%%%%%%%%%%%%%%%%%%%%%%%%%%%%%%%%%%%%%%%%%%%
%%%%%%%%%%%%%%%%%%%%%%%%%%%%%%%%%%%%%%%%%%%%%%%%%%%%%%%%%%%%%%%%%%%%%%%%%%%%%%%%%%%%%%%%%%%%%%%%%%%%%%%%%%%%%%%%%%%%%%%%%%%%%%%%%%%%%%%%%%%%%%%%%%%%%%%%%%%%%%%%%%%%%%%%%%%%%

\end{document}